\documentclass[a4paper,UKenglish,cleveref, autoref, thm-restate]{lipics-v2021}

\title{Typing Fallback Functions: A Semantic Approach to Type Safe Smart Contracts}

\author{Stian Lybech}
{Reykjavík University, Reykjavík, Iceland \and University of Southern Denmark (SDU), Odense, Denmark}
{stly@mmmi.sdu.dk}
{https://orcid.org/0000-0001-8219-2285}
{Supported by the Icelandic Research Fund Grant No.\@ 218202-05(1-3), the Department of Computer Science at Reykjavik University, and by the Software Engineering section at University of Southern Denmark.}

\author{Daniele Gorla}{Sapienza, Università di Roma, Rome, Italy}{gorla@di.uniroma1.it}{https://orchid.org/0000-0001-8859-9844}{}

\author{Luca Aceto}{Reykjavík University, Reykjavík, Iceland \and Gran Sasso Science Institute, L'Aquila, Italy}{luca@ru.is}{https://orchid.org/0000-0002-2197-3018}{Supported by the Icelandic Research Fund Grant No.\@ 218202-05(1-3).}

\authorrunning{
S. Lybech, D. Gorla and L. Aceto
}
\Copyright{
Stian Lybech, Daniele Gorla and Luca Aceto
}
\keywords{semantic typing, smart contracts, information flow control, non-interference}

\ccsdesc[500]{Theory of computation~Type theory}
\ccsdesc[500]{Security and privacy~Information flow control}

\category{} 

\relatedversion{} 
\relatedversiondetails{Full Version}{https://doi.org/10.48550/arXiv.2512.04755} 

\nolinenumbers 

\EventEditors{}
\EventLongTitle{}
\EventShortTitle{}
\EventAcronym{}
\EventYear{}
\EventDate{}
\EventLocation{}
\EventLogo{}
\SeriesVolume{}
\ArticleNo{}

\usepackage{calculi}
\usepackage{microtype}
\usepackage{listings}
\newcommand{\changed}[1]{#1}

\definecolor{strings}{rgb}{0,0.5,0}
\definecolor{emphs}{rgb}{0.64,0.08,0.08}
\definecolor{comments}{rgb}{0.17,0.57,0.68}
\colorlet{keywords}{blue!50!cyan}

\lstdefinestyle{tinysol}{
  language=C,
  captionpos=b,
  numbers=left,
  numberstyle=\tiny,
  frame=lines,
  showspaces=false,
  showtabs=false,
  breaklines=true,
  showstringspaces=false,
  breakatwhitespace=true,
  emph={contract,field,func,call,dcall,if,then,else,while,do},
  emphstyle={\rmfamily\bfseries\color{emphs}},
  commentstyle=\color{comments},
  morekeywords={contract, interface, account, this, var, field, func, method, send, value, balance, sender, then, fallback, call, dcall, args, id},
  keywordstyle={\bfseries\color{keywords}},
  stringstyle=\color{strings},
  basicstyle=\ttfamily\small,
  escapechar=@
}

\usetikzlibrary{
  patterns
}

\newbool{showedits}
\boolfalse{showedits}

\makeatletter
\def\edited{%
  \def\edited@temp{edited}%
  \ifx\@currenvir\edited@temp
    \@xp\edited@environmentcase 
  \else
    \@xp\edited@commandcase 
  \fi
}

\ifbool{showedits}{
  \newcommand{\edited@commandcase}[2][red]{\textcolor{#1}{#2}}
  \newcommand{\edited@environmentcase}[1][red]{\color{#1}}
}{
  \newcommand{\edited@commandcase}[2][\relax]{#2}
  \newcommand{\edited@environmentcase}[1][\relax]{}
}
\makeatother

\makeatletter
\newcommand{\replunderscores}[1]{\expandafter\@repl@underscores#1_\relax}
\def\@repl@underscores#1_#2\relax{%
  \ifx \relax #2\relax
    #1%
  \else
    #1%
    \textunderscore
    \@repl@underscores#2\relax
  \fi
}
\makeatother

\def\lstvdots{\raisebox{-1pt}[0pt][0pt]{$\vdots$}}

\def\TINT{\text{\normalfont\sffamily int}}

\def\TBOOL{\text{\normalfont\sffamily bool}}

\def\SUBS{\ensuremath{\mathrel{\text{\code{<:}}}}}

\def\TYPES{\SETNAME{T}}
\def\BASETYPES{\SETNAME{B}}

\def\TYPEENVS{\SETNAME{E}}

\def\SUCHTHAT{\ensuremath{\mathrel{.}}}
\newcommand{\TFAMILY}[1][\Gamma]{\SETNAME{R}_{#1}}
\newcommand{\WRONG}[2][\Gamma]{\ensuremath{\text{\normalfont\sffamily Wrong}_{#1}\kern-2pt\ARGS{#2}}}
\newcommand{\SAFE}[2][\Gamma]{\ensuremath{\text{\normalfont\sffamily Safe}_{#1}\kern-2pt\ARGS{#2}}}
\newcommand{\NSAFE}[2][\Gamma]{\ensuremath{\text{\normalfont\sffamily NSafe}_{#1}\kern-2pt\ARGS{#2}}}

\def\EPI{\ensuremath{\prescript{e}{}{\kern-1pt\PI}}}
\def\TINYSOL{\textsc{TinySol}}

\def\SWAP#1#2#3{\ensuremath{\ARGS{#1,#2}\kern-2pt\boldsymbol{\cdot}\kern-2pt #3}}

\makeatletter
\newcommand*\bigcdot{\mathpalette\bigcdot@{.5}}
\newcommand*\bigcdot@[2]{\mathbin{\raisebox{-.7ex}{\hbox{\scalebox{#2}{$\m@th#1\bullet$}}}}}
\makeatother

\def\CHANEQ{\ensuremath{\xleftrightarrow{\kern-1.5pt\bigcdot}}}

\def\CSOR{\ensuremath{~{\syntaxfont{[\kern-2pt]}}~}}

\def\TYPES{\ensuremath{\textbf{\textit{Types}}}}

\def\INCRSYM{\ensuremath{\text{\syntaxfont{+}}}}

\def\NCOMP#1#2{\ensuremath{#1\cdot#2}}

\def\NAMESPACE#1{\SETNAME{N}_{\kern-1pt\HOLE[#1]}}
\def\LNS#1{\prescript{\INCRSYM\kern-2pt}{}{\NAMESPACE{#1}}}

\def\LNC{\prescript{\INCRSYM\kern-2pt}{}{N}}

\newcommand{\LNSSUBST}[1][s_1,s_2]{\ensuremath{\prescript{\INCRSYM\kern-2pt}{}{\sigma_{#1}}}}

\def\TYPES{\ensuremath{\LANG{T}}}

\def\SUCHTHAT{\ensuremath{\;.\;}}

\newlength\QEDSymbolSpace
\newcommand{\code}[1]{\texttt{#1}}
\newcommand{\CONF}[1]{\ensuremath{\left<#1\right>}}
\newcommand{\UPDATE}[2]{\ensuremath{\kern-2pt\left[#1 \mapsto #2\right]}} 

\def\EXPR{\ensuremath{\text{\normalfont\sffamily Exp}}}
\def\STM{\ensuremath{\text{\normalfont\sffamily Stm}}}

\def\VALUES{\SETNAME{V}}

\newcommand{\GROUPNAMES}[1]{\ensuremath{\SETNAME{N}_{#1}}}
\def\ANAMES{\GROUPNAMES{A}} 
\def\MNAMES{\GROUPNAMES{M}} 
\def\VNAMES{\GROUPNAMES{V}} 
\def\FNAMES{\GROUPNAMES{F}} 
\def\TNAMES{\GROUPNAMES{I}} 

\newcommand{\ENV}[1]{\ensuremath{\text{\normalfont\sffamily env}_{#1}}}
\newcommand{\DEC}[1]{\ensuremath{\SETNAME{D}_{#1}}}
\newcommand{\SETENV}[1]{\ensuremath{\text{\normalfont\sffamily Env}_{#1}}}

\def\TRUE{\ensuremath{\text{\sffamily T}}}
\def\FALSE{\ensuremath{\text{\sffamily F}}}

\def\BOOLEANS{\ensuremath{\mathbb{B}}}
\def\INTEGERS{\ensuremath{\mathbb{Z}}}

\newcommand{\CALL}[4]{\ensuremath{\code{call $#1$.$#2$($#3$)\$$#4$}}}
\newcommand{\DCALL}[3]{\ensuremath{\code{dcall $#1$.$#2$($#3$)}}}

\def\SECLEVELS{\SETNAME{S}}
\def\ordleq{\sqsubseteq}

\def\SUBS{\mathrel{<:}} 

\newcommand{\TVAR}[1]{\ensuremath{\text{\normalfont\sffamily var}(#1)}}
\newcommand{\TCMD}[1]{\ensuremath{\text{\normalfont\sffamily cmd}(#1)}}
\newcommand{\TSPROC}[2]{\ensuremath{\text{\normalfont\sffamily proc}(#1)\code{:}#2}}

\def\STOP{\ensuremath{s_\top}}
\def\SBOT{\ensuremath{s_\bot}}
\def\ITOP{\ensuremath{I^{\texttt{top}}}}

\newcommand{\TRANSACT}[5]{\ensuremath{#1\leadsto#2\code{.}#3\code{(}#4\code{)}\$#5}}
\newcommand{\DEL}[1]{\ensuremath{\text{\normalfont\sffamily del}(#1)}}

\newcommand{\EXTEND}[2]{\ensuremath{\kern-2pt\left[#1 \mapsto #2\right]}} 
\newcommand{\EXTRACT}[1]{\ensuremath{\mathbf{E}(#1)}} 

\def\STACKS{\SETNAME{Q}}
\def\TSTACKS{\SETNAME{TQ}}

\def\TIDF{\text{\sffamily idf}}   
\def\TTYPE{\text{\sffamily type}} 
\def\TYPEENVS{\SETNAME{E}} 

\NewDocumentCommand{\ETYPI}  { O{B_s} O{\Delta} }{\ensuremath{\SETNAME{R}^{#1}_{\Sigma;\Gamma;#2} }}
\NewDocumentCommand{\ETYPING}{ O{B_s} O{\Delta} }{\ensuremath{\nu\SETNAME{R}^{#1}_{\Sigma;\Gamma;#2} }}

\NewDocumentCommand{\STYPI}  { O{R} }{\ensuremath{\SETNAME{#1}_{\Sigma,\Gamma,\ENV{T}}}}

\NewDocumentCommand{\STYPIS}  { O{S} }{\ensuremath{\SETNAME{#1}_{\Sigma,\Gamma,\ENV{T}}}}

\NewDocumentCommand{\STYPING}{ O{R} }{\ensuremath{\nu\SETNAME{R}_{\Sigma,\Gamma,\ENV{T}}}}

\NewDocumentCommand{\sstypi}  { O{s} O{\Delta} }{\ensuremath{\SETNAME{R}_{#2,#1} }} 
\newcommand{\CANDSTYPI}[1][R]{\ensuremath{\SETNAME{#1}}}

\newcommand{\esemDash}[1][B_s]{\ensuremath{\vDash_{#1}}}
\newcommand{\ssemDash}[1][s]{\ensuremath{\vDash_{#1}}}

\newcommand{\TYPEOF}[2][\Gamma]{\ensuremath{\text{\normalfont\textsc{TypeOf}}_{#1}(#2)}}

\def\BSVEC{\VEC{B}_{\VEC{s}}}                           
\newcommand{\REMS}[1]{\ensuremath{\text{\normalfont\textsc{rem}}(#1)}}  
\newcommand{\MERGE}[1]{\ensuremath{\text{\normalfont\textsc{mrg}}(#1)}} 
\newcommand{\REBIND}[2]{\ensuremath{\kern-2pt\left[#1 \mapsto #2\right]}} %

\newcommand{\FA}[1]{\ensuremath{\mathcal{F}_A(#1)}}
\newcommand{\FV}[1]{\ensuremath{\mathcal{F}_V(#1)}}

\def\ordgeq{\sqsupseteq} 
\def\conseq{\rightleftharpoons} 

\def\DNAMES{\ensuremath{\SETNAME{D}}} 

\newcommand{\FIRSTS}[1]{\ensuremath{\text{\normalfont\textsc{fst}}(#1)}}
\newcommand{\FINISH}[1]{\ensuremath{\text{\normalfont\textsc{fin}}(#1)}}

\def\ttrans{\ensuremath{\Mapsto}}           
\def\progress{\ensuremath{\rightarrowtail}} 
\newcommand{\UPTOSTYPI}[1][S]{\ensuremath{\SETNAME{#1}}}

\def\SETTRIPLETS{\ensuremath{\SETNAME{P}}}

\begin{document}
\sloppy

\maketitle

\begin{abstract}
This paper develops semantic typing in a 
smart-contract setting to ensure type safety of code that uses statically untypable language constructs, such as the fallback function. 
The idea is that the creator of a contract on the blockchain equips code containing such constructs with a formal proof of its type safety, given in terms of the semantics of types.  
Then, a user of the contract only needs to check the validity of the provided `proof certificate' of type safety. 
This is a form of proof-carrying code, which naturally fits with
the immutable nature of the blockchain environment.
%
As a concrete application of our approach, we focus on ensuring information flow control and non-interference for 
\TINYSOL{}, a distilled version of the Solidity language, through security types. 
We provide the semantics of types in terms of a typed operational semantics of \TINYSOL\ and we express the proofs of safety as coinductively-defined typing interpretations, which can be 
represented compactly  
via up-to techniques, similar to those used for bisimilarity. 
We also show how our machinery can be used to type the typical pointer-to-implementation pattern 
based on the fallback function
and to reject a distilled version of the infamous Parity Multisig Wallet Attack.
\end{abstract}

\section{Introduction}
In 2017, a vulnerability in a smart contract on the Ethereum platform led to the theft of approximately 31 million dollars' worth of Ether (see \cite{manning2018parity_hack} for a code example with comments explaining the details of the attack).
This vulnerability, known as the \emph{Parity Multisig Wallet attack} (PMW), involved an intricate interplay between two peculiar, low-level features of the Solidity smart-contract language \cite{solidity2022}: \emph{delegate calls} and \emph{fallback functions}, which can be used for developing utility libraries. 
The root cause was an unintended and unwanted flow of data from an untrusted source to a trusted variable, which was used by an attacker to access wallet contracts and drain them of their Ether. 
In 
security parlance, 
the \emph{integrity} of that variable was compromised.
Preserving integrity is crucial to guarantee \emph{non-interference} \cite{goguen_mesegier1982security_policies} and, in \cite{volpano2000secure_flow_typesystem}, Volpano, Irvine and Smith created a type system for ensuring such a property, using the concept of \emph{security types}---that is, types that control the information flow between memory segments of different security levels.
This line of work has been taken up in
\cite{AGL25,chinese2021flowtypes},
where type systems with security types for a Solidity-like language have been devised
to prevent such unwanted flows.
Unfortunately, the language used in \cite{AGL25} does not include the fallback function construct, and the authors of the type system in \cite{chinese2021flowtypes}, despite using the PMW attack as their motivating example, do not actually provide a typing rule for the fallback function, which was a key factor in enabling the aforementioned vulnerability. 

This is unsurprising, because the fallback function enables a limited form of \emph{reflection}, or at least introspection, which is notoriously difficult to handle using a syntactic type system.
The source of this problem is that the fallback function is invoked when a method call \emph{fails}, because the called method does not exist, and the body of the fallback function can then use information \emph{about} the failed call to determine how the computation shall proceed.
This information is only available at runtime, hence it cannot be assigned a static type able to distinguish between the safe and unsafe usages of this construct. 
We shall explain these problems better and illustrate our proposal by using a concrete example in Section~\ref{sec:motEx}.

The difficulty in handling the fallback function is a prime example of a more general shortcoming of the standard, syntactic approach to type soundness \cite{wright_felleisen1994subject_reduction}, namely that it tends to 
equate `well-typedness' (a syntactic property) with `safety' (a semantic, or behavioural, property).
Specifically, in the syntactic approach, one defines a syntactic predicate, well-typedness, by means of syntax-driven type rules.
These rules are defined \emph{a priori} and then subsequently shown to \emph{imply} safety, but they are not \emph{derived} from the notion of safety itself.
This makes the syntactic approach incapable of reasoning about programs or constructs that \emph{behave} safely at runtime, but for which the difference between safe and unsafe usages is not a purely syntactic property.
Examples of such constructs include inline assembly, pointer arithmetics, reflection, and, indeed, the fallback function.

This issue was originally noted in \cite{10.5555/10510,MARTINLOF1982153} and, more recently, by
Caires in \cite{caires2007} and Timany et al.\@ in \cite{timany/2024/acm/semantic_typing}, who instead advocate a different approach, known as the \emph{semantic approach to type soundness}, or just \emph{semantic typing} for short, which is capable of reasoning about even seemingly `unsafe' programming constructs.
%
Semantic typing is now a well-established technique and, among other properties, it has also been used to ensure non-interference (see, e.g., \cite{9470654,GregersenBTB21,RajaniG20}). 
However, as far as we know, the semantic approach to type soundness has not hitherto been applied in the setting of blockchains and smart contracts, even though it seems particularly well-suited for this purpose.
Suppose a contract makes use of a syntactically untypable construct, such as a fallback function, but does it in such a way that it can be shown to preserve type safety via a formal proof, given in terms of the semantics of types.
Such a proof can then be supplied along with the code of the contract containing the fallback function, when it is declared on the blockchain, as a form of proof-carrying code \cite{necula/1997/popl/pcc}.
The immutable nature of the blockchain environment ensures that neither the code nor the proof can be subsequently altered. 
Thus, as long as the proof can be verified by any user of that contract, 
the untypable fragment of code does not itself need to be type checked using the syntactic rules.
Of course, \emph{finding} such a proof may then be challenging, but this obligation now falls on the contract creator, rather than on the user of the contract, who only needs to check the validity of the provided `proof certificate' of type safety.
In this way, we may not be able to make the fallback function itself typable by syntactic means, but we can provide a method for reasoning about, and ensuring, the type safety of contracts \emph{using} that construct.

Our main contribution in this paper is the presentation of the theoretical developments necessary to make this approach work in a blockchain/smart-contract setting.
Similar to \cite{chinese2021flowtypes}, we shall focus on ensuring non-interference by means of types for secure information flows in the style of Volpano, Irvine and Smith \cite{volpano2000secure_flow_typesystem}, yet this should only be viewed as an example of how the semantic approach may be applied.
(Indeed, in Appendix~\ref{app:CI} we show the changes needed to ensure \emph{call integrity}, a property designed in  \cite{grishchenko2018} to avoid attacks similar to PMW.)
To keep the presentation simple, with a focus on the fallback function, we shall use the language \TINYSOL{} \cite{bartoletti2019tinysol}, which models some of the core features of the smart-contract language Solidity.
A syntactic version of the Volpano-Irvine-Smith type system has already been developed for a version of \TINYSOL{} \emph{without} fallback functions in \cite{AGL25}, so we can adapt this work by adding just the fallback function construct.
We present this extended version of \TINYSOL{} in Section~\ref{sec:typedSyntax}, and review the language of security types in Section~\ref{sec:types}.

There are different approaches to defining a semantic model of types and type judgments.
One is the step-indexed model of Appel and McAllester \cite{appel_mcallester2001semantic_types_pcc}, which was further extended by Ahmed in \cite{ahmed2004semantic_types_phd}.
This approach has been applied in a variety of settings, for example to reason about the safety of `unsafe code' in core Rust libraries \cite{jung/2017/popl/rustbelt}; see \cite{timany/2024/acm/semantic_typing} for an overview of work based on this approach.
Another approach is based on \emph{typing interpretations}, which are coinductively defined objects corresponding to fragments of transition systems.
It is reminiscent of the logical relations technique of Plotkin \cite{plotkin1973logical_relations} and Tait \cite{tait1967logical_relations}, and it is the method used by Caires in \cite{caires2007} for a simple type system for the \PI-calculus \cite{PICALC}. 
Informally, 
\changed{
a typing interpretation for a syntactic category (e.g.\@ programs $P$)
w.r.t.\@ a collection of type predicates $\Gamma$, 
is a set of programs $\SETNAME{R}_\Gamma$ 
}%
defined such that $P \in \SETNAME{R}_\Gamma$ implies that
(1) $P$ satisfies the type predicates in $\Gamma$, and 
(2) If $P \trans P'$ then $P' \in \SETNAME{R}_\Gamma$.
The largest typing interpretation for a given $\Gamma$, written $\nu\SETNAME{R}_\Gamma$, is the union of all typing interpretations for $\Gamma$, just as bisimilarity is the union of all bisimulations.
The typing interpretation 
$\nu\SETNAME{R}_\Gamma$ is the \emph{typing} of $\Gamma$, and the semantic typing relation $\Gamma \vDash P$ is then defined as $P \in \nu\SETNAME{R}_\Gamma$.
Proving semantic type safety of $P$ w.r.t.\@ a given $\Gamma$ thus becomes a matter of finding a typing interpretation for $\Gamma$ that contains $P$.
\changed{
We take this approach, as it is flexible and based on general coinductive techniques, which are well-studied (e.g.\@ in bisimulation theory).
Using this method, it is possible to equip a statically untypable piece of code with a proof of its type safety, to be stored on the blockchain, as a form of \emph{proof-carrying code}; the proof is simply the typing interpretation itself.
}

In Section~\ref{sec:semTypingTinySol}, we define the semantics of types in terms of a \emph{typed operational semantics} of \TINYSOL{}, given in Section~\ref{sec:semantic_types_stacks}.
It is based on an (untyped) small-step operational semantics of the language (given in Appendix~\ref{app:untyped_operational_semantics}) with type constraints inspired from the syntactic type system. This ensures a crucial property, which is needed for the definition of typing interpretations: \emph{a non-terminal, stuck configuration indicates a type error}.
In a sense, our approach may resemble \emph{gradual typing} \cite{CLPS19,SW07}, which allows a mixing of static and dynamic typing.
However, the crucial difference is that our typed operational semantics is defined solely for the purpose of giving a semantics to types.
The runtime semantics of \TINYSOL{} is still untyped, and our typing interpretations are intended to \emph{statically} witness the safety of statically untypable code, without runtime type-related checks.

The task of finding a typing interpretation falls on the contract creator, but any client \emph{using} a contract with such an associated proof must, of course, also be able to verify that it indeed is a typing interpretation.
Both tasks can be difficult, since typing interpretations can be large and, in limit cases, also infinite.
This also threatens the practical applicability of our proposed approach, since infinite objects cannot be stored on the blockchain, unless they have a finite representation.
Inspired by the similarity between typing interpretations and bisimulations, we therefore adapt the work of Pous and Sangiorgi \cite{pous_sangiorgi2011uptotechniques} on enhancements of the bisimulation proof-technique, and propose \emph{up-to techniques for typing interpretations}, which may be used to reduce their size. 
The benefits of this approach and its soundness are 
discussed in Section~\ref{sec:upto_techniques}.
To our knowledge, this is the first use of up-to techniques in a predicative (i.e., non-relational) setting; indeed, traditionally these techniques have been used to reduce the size of (inductively or coinductively defined) {\em equivalences} on programs (see, e.g., \cite{BGGP18,BKK17,CPV16,Dan18,SV19}), whereas typing interpretations are (unary) {\em predicates} on programs.
The possibility of exploiting up-to techniques is a key advantage of giving the semantics to our types by means of typing interpretations (instead of, e.g., logical relations \cite{JKJBBD2018,10.1007/978-3-662-54434}).

Finally, we conclude the paper in Section~\ref{sec:conclusion} with a discussion of some avenues for future research.
All proofs and many definitions are deferred to the appendices for space reasons. A synopsis of all the notations used in this paper is in Appendix~\ref{app:notations}.

\section{A Motivating Example}
\label{sec:motEx}
We highlight the key points of our typing approach using the code in Figure~\ref{fig:example_fallback_proxy}.
This is a typical example of Solidity-like programming written in the language \TINYSOL{} (to be formally introduced in the next section), which is a class-based language, where classes, called \emph{contracts}, have fields and methods.
Every contract has 
an interface ($I^P$ and $I^X$ in our example) and 
a security level, $s$. 
Both these annotations are related to typing: the former is a way to handle typing and subtyping of contract addresses; the latter is the key feature to devise a type system for secure information flow.

\begin{figure}
\begin{minipage}[t]{.5\textwidth}
\begin{lstlisting}[style=tinysol]
contract Proxy : @$I_s^P$@ {
  field owner := A;
  field impl := X;
  @$\lstvdots$@
  func update(x) {
    if sender = this.owner 
    then this.impl := x 
    else skip      
  }
  func fallback() {
    call this.impl.id(args)$value
  }
}
\end{lstlisting}
\end{minipage}\hfill
\begin{minipage}[t]{.45\textwidth}
\begin{lstlisting}[style=tinysol]
contract X : @$I_s^X$@ {
  @$\cdots$@
  func @$f_1$@(@$\VEC{x_1}$@) { @$\ldots$@ }
  @\lstvdots@
  func @$f_h$@(@$\VEC{x_h}$@) { @$\ldots$@ }
  func fallback() { skip }
}
\end{lstlisting}
\end{minipage}
\vspace*{-.5cm}
\caption{An example of the pointer-to-implementation pattern ($\VEC{x_i}$, $1\leq i\leq h$, denotes a sequence of distinct variables).}
\label{fig:example_fallback_proxy}
\end{figure}

As previously mentioned, a \emph{fallback function} is a special, low-level construct, which is used to handle calls to non-existing contract methods.
When a contract receives such a call, its fallback function is invoked instead, and the name and actual arguments of the failed call are then available within the fallback method body through the reserved variables \code{id} and \code{args}, respectively.
The construct was originally introduced as 
a mechanism for handling failed calls, but is commonly used for other purposes, such as enabling contract functionality to be \emph{upgradable}, despite the immutable nature of the blockchain (see \cite{ilham_phd} for an extensive discussion of this pattern for upgradability and related  ones).
This usage is implemented by the code in Figure~\ref{fig:example_fallback_proxy}. The \code{Proxy} contract uses its fallback function to forward all  calls to methods other than \code{update} to an implementation residing in the contract \code{impl}, which is initially \code{X}. In more detail, assume that we invoke \code{Proxy.$f$($\VEC{v}$)}, for some method $f$ different from \code{update} and for some sequence of values $\VEC{v}$. Then, \code{Proxy}'s function \code{fallback} is executed, which, in turn, forwards to \code{impl} the invocation of function $f$ with actual parameters $\VEC{v}$ through the use of \code{id} and \code{args}.\footnote{There is also some currency transfer as denoted by the `\code{\$value}' part of the invocation, but this is orthogonal to the issues of our typing approach, so we ignore it for the moment.} The \code{Proxy} contract provides another method, \code{update}, that allows the contract owner (having the address \code{A}) to update the \code{impl} field, which holds the reference to the implementation.
The functionality can then be updated by deploying a new implementation contract and updating the reference in \code{Proxy}. 

Figure~\ref{fig:example_fallback_proxy} is an example of the so-called \emph{pointer-to-implementation pattern} \cite{dickerson2018proof_carrying_smart_contracts}.
The implementation contract \code{X} provides methods $f_1$, \ldots, $f_h$ and a fallback function with \code{skip} as its body.
Thus, even though \code{Proxy} uses \code{id} and \code{args} to generically forward all method calls (except \code{update}) to \code{X}, we know that there are only finitely many possible continuations after any method call to \code{Proxy}:
either \code{update}, or one of $f_1$, \ldots, $f_h$, or \code{skip} (when the fallback function of \code{X} is invoked).
Therefore, if each of the methods $f_i$ is type safe (for some notion of type safety), then the fallback function of \code{Proxy} \emph{ought to} be type safe as well.
However, \emph{deciding} this, by using a standard compositional, syntactic type system, is difficult, because the body of \code{Proxy}'s fallback function does not contain calls to these methods.
In more general terms, the safety of the fallback function is not a purely syntactic property \emph{of} its body; it also depends on the set-up of these two contracts, and on the current state of the program, in which the address of \code{X} is held in the field \code{Proxy.impl}.
A purely syntactic analysis of the body of the fallback function will therefore not suffice.

The solution we propose is to formalise the above informal argument, by providing a definition of what it means to \emph{behave} type-safely:
this is what the typed semantics for \TINYSOL{} does.
A fragment of this typed transition system, called a \emph{typing interpretation} and which contains the body of the fallback function, can then serve as a witness of its safety.
The key idea is to think of a typing interpretation as a set of safe states that is closed under transitions. So, safety is an invariant for all states in a typing interpretation. 

A typing interpretation can---at least, in principle---be attached to the contract as a form of proof-carrying code, when it is committed to the blockchain.
However, one important detail to note here is that the value of \code{Proxy.impl} can \emph{change} at runtime, to allow the implementation to be updated; this is precisely the point of the pointer-to-implementation pattern.
If the associated transition system fragment were to witness the entire transition sequence of every possible continuation from an invocation of the fallback function, it would either have to be very large, which would make our proposed solution infeasible, or it would have to be updatable as well, which would defeat the purpose of static verification.

We handle this problem by
making typing interpretations modular using up-to techniques.
Through them,
 the witnessing set need not contain \emph{all} the following states of all continuations from an invocation of the fallback function; it suffices to require that all states in the witnessing set only have typed transitions to states that are known to be in \emph{some} typing interpretation.
Suppose, for example, that \code{Proxy.impl} were updated with the address of some contract \code{Y}.
To be well-typed, \code{Y} must also implement the interface $I^X$, ensuring that it will also contain methods $f_1, \ldots, f_h$ with the same signatures.
If these are all well-typed, then we know they are also type-safe, so there exist typing interpretations containing them.
Thus, the set witnessing the safety of the body of the fallback function only needs to contain the states corresponding to an invocation of any of these methods, which will be the same for \code{X} and \code{Y}.


\section{The \TINYSOL{} language and its security types}\label{sec:tinysol_syntax_semantics}
 
\subsection{Typed Syntax}
\label{sec:typedSyntax}

We use an extended version of the \TINYSOL{} language, based on the presentation in \cite{AGL25}, which, in turn, is based on the original presentation in \cite{bartoletti2019tinysol}.
Its syntax is given in Figure~\ref{fig:syntax_tinysol}, where the notation $\mathop{\ \widetilde\cdot\ }$ denotes (possibly empty) sequences of elements.
The syntax consists of declarations for contracts $DC$, fields $DF$, and methods $DM$, followed by values $v$, expressions $e$, and statements $S$.
%
A contract contains declarations of fields and methods, and of the so-called \emph{fallback function} (discussed in Section~\ref{sec:motEx} and in what follows).
The set of \emph{values} $\VALUES$, ranged over by $v$, consists of integers $\INTEGERS$, ranged over by $z$, booleans $\BOOLEANS = \SET{\TRUE, \FALSE}$, ranged over by $b$, method names $\MNAMES$, ranged over by $f$,\footnote{
Having method names as values is related to the fallback function and the special variable \code{id}.
} 
and address names $\ANAMES$, ranged over by $X, Y$.
Note that we use a \emph{typed} syntax, with type annotations $I_s$ and $\TVAR{B_s}$ appearing in declarations of contracts and of local variables; these will be described in Section~\ref{sec:types}.

\begin{figure}[t]
\begin{syntax}[h]
  DC \in \DEC{C} \IS \epsilon 
                 \OR \code{contract $X$ : $I_s$ \{ }                 \tabularnewline
                  & & \qquad\ \ \ \code{field balance := $z$; $DF$}  \tabularnewline
                  & & \qquad\ \ \ \code{func send() \{ skip \} $DM$} \tabularnewline
                  & & \qquad\ \ \ \code{func fallback() \{ $S$ \} }  \tabularnewline
                  & & \quad\ \code{\} $DC$} 
\tabularnewline
  DF \in \DEC{F} \IS \epsilon 
                 \OR \code{field p := $v$;} DF 
\tabularnewline
  DM \in \DEC{M} \IS \epsilon 
                 \OR \code{func $f(\VEC{x})$ \{ $S$ \} } DM 
\tabularnewline
  v \in \VALUES \IS \INTEGERS 
                \UNION \BOOLEANS 
                \UNION \ANAMES 
                \UNION \MNAMES    
\tabularnewline
  e \in \EXPR \IS v 
              \OR x 
              \OR \code{$e$.balance} 
              \OR \code{$e$.$p$} 
              \OR \op(\VEC{e})         
            \ISOR \code{this} 
              \OR \code{sender} 
              \OR \code{value} 
              \OR \code{id} 
              \OR \code{args} 
\tabularnewline 
  S \in \STM  \IS \code{skip} 
              \OR \code{throw} 
              \OR \code{$\TVAR{B_s}$ $x$ := $e$ in $S$} 
              \OR \code{$x$ := $e$} 
              \OR \code{this.$p$ := $e$} 
            \ISOR \code{$S_1$;$S_2$} 
              \OR \code{if $e$ then $S_{\TRUE}$ else $S_{\FALSE}$ } 
              \OR \code{while $e$ do $S$} 
            \ISOR \CALL{e_1}{m}{\VEC{e}}{e_2} 
              \OR \DCALL{e}{m}{\VEC{e}} 
\end{syntax}
\flushleft
\noindent where $x, y \in \VNAMES$ (variable names), $p, q \in \FNAMES$ (field names), $X, Y \in \ANAMES$ (address names),\\ $f \in \MNAMES$ (method names), $m \in \{\code{id}\} \UNION \MNAMES$, $I \in \TNAMES$ (interface names),  $s\in\SECLEVELS$ (security levels).
\caption{The syntax of \TINYSOL.}
\label{fig:syntax_tinysol}
\end{figure}

There are six `magic' keywords appearing in the syntax:
\begin{itemize}
  \item \code{balance}, an integer field that records the current balance of the contract;
    it can be read from, but not directly assigned to, except through method calls.
  \item \code{value}, an integer variable that is bound to the currency amount transferred with a call.
  \item \code{sender}, an address variable that is bound to the address of the caller of a method. 
  \item \code{this}, an address variable that is bound to the address of the contract that contains the method that is currently executed.  
  \item \code{id} and \code{args}, two variables that, when the fallback function is invoked, contain the name of the (non-existing) called method and the list of actual arguments.  
\end{itemize}
All except \code{balance} are local variables, and we collectively refer to them as `magic variables'.

The core of the language is the declaration of the expressions and of the statements, which are similar to those for the classic imperative language \code{While} \cite{nielson_nielson2007semantics_with_applications}, except for the aforementioned magic keywords, and for the method calls, which have three flavours:
\begin{itemize}
  \item In an ordinary method call $\,\CALL{\!\!e_1}{f}{\VEC{e}}{e_2}$, expression $e_1$ must evaluate to a contract address $X$ and $e_2$ must evaluate to an integer $z$.
    The latter is the amount of currency transferred from the caller to the callee along with the call; this amount will be deducted from the caller's \code{balance} and added to the callee's  \code{balance}. 
    All currency transfers are done through method calls; so, all contracts contain a  \code{send()} method (for \code{send} $\in \MNAMES$), which does nothing, apart from ensuring that the contract can always receive payments.

  \item In a delegate call $\,\DCALL{\!\!e}{f}{\VEC{e}}$, the method body $S$ is executed in the context of the caller, rather than the callee.
    This means that
    \code{this}, \code{sender} and \code{value} are not rebound, so the method executes \emph{as if} it were declared in the caller contract.
    Hence, there is also no currency transfer, since the execution is local.
    The purpose of this feature is to facilitate code reuse and allow for the development of library contracts (see, e.g., Figure~\ref{fig:example_fallback_combined}).

  \item Differently from \cite{AGL25,bartoletti2019tinysol} and similarly to Solidity \cite{solidity2022}, if a call $\,\CALL{\!\!X}{f}{\VEC{v}}{z}$ is issued to a contract $X$ that does \emph{not} contain a method named $f$, then the fallback function of $X$ will be executed.
In this case, 
\code{id} and \code{args} contain, respectively, the method name and the actual arguments of the failed call, i.e.\@ $f$ and $\VEC{v}$, which allow the programmer to perform some error handling, for example by forwarding the call to another contract $Y$: 
  \begin{lstlisting}[style=tinysol, frame=none, numbers=none]
    func fallback() { call Y.id(args)$value }
  \end{lstlisting}

  \noindent Hence, we use $m$ instead of $f$ in the syntax for method calls and delegate calls in Figure~\ref{fig:syntax_tinysol} to allow both \code{id}  and a method name $f$ to be used for invocations.
\end{itemize}


To give an operational semantics for \TINYSOL, the first thing to do is turning all declarations into some environments, to record the bindings of variables, fields, methods, and contracts.
Such environments are modelled as the following sets of partial functions:
\vspace*{-.1cm}
\[
\begin{array}{ll}
  \multicolumn{2}{l}{\ENV{V} \in \SETENV{V} : (\VNAMES \UNION \SET{\code{value},\code{sender},\code{this},\code{id},\code{args}}) \PARTIAL \VALUES}
  \vspace*{.2cm}
    \\
  \ENV{S} \in \SETENV{S} : \ANAMES \PARTIAL \SETENV{F}
  &
  \ENV{F} \in \SETENV{F} : (\FNAMES \UNION \SET{\code{balance}}) \PARTIAL \VALUES 
  \vspace*{.2cm}
  \\
  \ENV{T} \in \SETENV{T} : \ANAMES \PARTIAL \SETENV{M}  \qquad\qquad
  &
  \ENV{M} \in \SETENV{M} : (\MNAMES \UNION \SET{\code{fallback}}) \PARTIAL (\VNAMES^* \times \STM)
\end{array}
\vspace*{-.1cm}
\]
\noindent We regard each environment $\ENV{J}$, for any $J \in \SET{V, F, M, S, T}$, as a (finite) set of pairs $(d, c)$ where $d \in \DOM{\ENV{J}}$ and $c \in \CODOM{\ENV{J}}$.
The notation $\ENV{J}, d:c$ denotes the extension of $\ENV{J}$ with the new pair $(d,c)$; the update of $\ENV{J}$ where $d$ is mapped to $c$ is denoted as $\ENV{J}\EXTEND{d}{c}$; updates are left-associative.
We write $\ENV{J}^{\EMPTYSET}$ for the empty $J$-environment.
To simplify the notation, when two or more environments appear together, we shall use the convention of writing the subscripts together (e.g.\@ we write $\ENV{MF}$ instead of $\ENV{M}, \ENV{F}$).

The \emph{method table} $\ENV{T}$ and the \emph{state} $\ENV{S}$ map addresses to method and field environments, respectively.
Thus, for each contract, we have the list of methods it declares (together with the lists of formal parameters and the statement that forms their bodies) and its current state (its fields and their values).
The method table is constant, once all declarations are performed;
in contrast, the state will change during the evaluation of a program.
Lastly, we also have a \emph{variable environment} $\ENV{V}$, which contains bindings for locally declared variables, including the magic variables and formal parameters of methods;
hence, this too will change during the execution of a program.

Differently from \cite{AGL25} and similarly to  \cite{AGLM/2024/reocas/outofgas}, we rely on a small-step operational semantics, and we therefore also introduce \emph{stacks} (that, again, use type information---see Section~\ref{sec:types}):
\begin{syntax}[h]
  Q \in \STACKS \IS \bot \OR S ; Q \OR \DEL{x}; Q \OR (\ENV{V}, \Delta); Q \OR s; Q
\end{syntax}
Briefly, $\bot$ is the empty stack, and besides statements $S$, a stack can also contain: the end-of-scope symbol $\DEL{x}$, denoting deallocation of a local variable $x$; the method-return symbol $(\ENV{V}, \Delta)$, where $\ENV{V}$ and $\Delta$ are the bindings for local variables and their type annotations which were in use before the method call, that have to be restored at its completion; lastly, a security level $s$, since the current security level may be changed and restored during a computation (see Section~\ref{sec:typed_semantics}).
Stacks also obviate the need for an explicit definition of \emph{transactions}; in \cite{AGL25} these are sequences of method invocations, which can be easily rendered in the current setting as stacks consisting of  the sequential composition of commands corresponding to such calls.

The operational semantics uses transitions 
of the form $\CONF{Q, \ENV{TSV}} \trans \CONF{Q', \ENV{T},\ENV{SV}'}$, which in turn are based on judgements 
of the form $\CONF{e, \ENV{SV}} \trans v$ for evaluating expressions.
The 
rules are in Appendix~\ref{app:untyped_operational_semantics} and 
yield the real operational semantics of the language.
However, to define the semantics of types, we shall need a \emph{typed operational semantics}, which performs some type-related checks.
For space reasons, we omit the untyped semantics (which is essentially obtained from the typed one by removing all type-related checks) and postpone the definition of the typed one to Section~\ref{sec:typed_semantics}, 
after we have described the language of types. 

\changed{To conclude, we remark that \TINYSOL\ is complex because Solidity is complex. 
The former is already a dramatic simplification of the latter, being essentially a classic while-language, together with the key features for contract declarations and the three forms of method invocations. 
We could omit local variables declaration/assignment, but this would not reduce the complexity of the presentation, since parameters in method invocations are needed (e.g.\@ the \code{id} and \code{args} variables, which are typical of the fallback function). 
We could avoid \code{while}-loops, because recursion is allowed, but we think that this simplification would not have any major impact. 
No other sensible simplification is possible without affecting the link with Solidity and the ability to code the key examples illustrating our approach (e.g.\@ the pointer-to-implementation pattern of Figure~\ref{fig:example_fallback_proxy}).
}

\subsection{Types and Type Environments for Information Flow Control and Non-Interference
}\label{sec:types}
We assume a finite lattice $(\SECLEVELS, \ordleq)$ of security levels $s$, with $\STOP$ as the highest level and $\SBOT$ as the lowest level.
For ease of notation, we shall use the following convention when comparing multiple levels against each other using the ordering relation $\ordleq$:
\vspace*{-.1cm}
\[
  s \ordleq s_1, \ldots, s_h \DEFSYM s \ordleq s_1 \land \ldots \land s \ordleq s_h
  \qquad\qquad
  s_1, \ldots, s_h \ordleq s_1', \ldots s_h' \DEFSYM s_1 \ordleq s_1' \land \ldots \land s_h \ordleq s_h'
\]
The crucial property related to this framework is {\em information flow control}, stipulating that information at level $s$ cannot affect information at levels \emph{strictly below}, or incomparable to, $s$  (w.r.t. $\ordleq$). This, in turn, can be used to obtain {\em non-interference}: if we run a piece of code in memories differing only for information at level $s$ or higher, the resulting memories still differ only for information of level $s$ or higher.
As usual, the security levels can be used to express both secrecy and integrity. In the first case, higher levels (w.r.t. $\ordleq$) denote secret information, that should not influence lower level (i.e., less secret) information. The second case is dual: higher levels denote untrusted information, that should not influence lower level (i.e., trusted) information.

Our language of types is similar to the one given in \cite{AGL25} for secure information flows, inspired by the pioneering work from~\cite{volpano2000secure_flow_typesystem}.

\begin{definition}[Types and environments]
\label{def:types}
We define types, base types and typing environments:
\begin{center}
\begin{minipage}{0.55\textwidth}
\begin{syntax}[h]
  T \in \TYPES      \IS B_s \OR \TVAR{B_s} \OR \TSPROC{\BSVEC}{s} \tabularnewline
  \Gamma, \Delta \in \TYPEENVS \IS \NAMES \PARTIAL \TYPES \UNION \TYPEENVS 
\end{syntax}
\end{minipage}
\begin{minipage}{0.43\textwidth}
\begin{syntax}[h]
  B \in \BASETYPES  \IS \TIDF \OR \TINT \OR \TBOOL \OR I          \tabularnewline
  \Sigma                       \IS \TNAMES \PARTIAL \TNAMES 
\end{syntax}
\end{minipage}
\end{center}
where 
$\NAMES$, ranged over by $n$, is defined as $\VNAMES \UNION \FNAMES \UNION \MNAMES \UNION \ANAMES \UNION \TNAMES$. 
\end{definition}

For the base types ($\BASETYPES$),  $\TIDF$ is the type assigned to method names contained in the magic variable \code{id} and interfaces are associated to contract addresses. The type of an address $X$ is therefore an interface $I$ with the signatures of the methods and fields of that contract;\footnote{
	The interface could be extracted from the contract definition itself (as a form of structural typing); however, we prefer to define interfaces separately (thus adopting a nominal typing), to let multiple contracts implement the same interface.
}
moreover, interfaces are hierarchically organized and we define them as follows:

\begin{definition}[Interface definition language]
\label{def:interfaces}
Interface declarations $ID$ are given by:
\begingroup\normalfont
\begin{syntax}[h]
  ID \IS \epsilon \OR \mbox{\code{interface}}\  I_1 : I_2\ \{ IF\ IM \}\ ID \\
  IF \IS \epsilon \OR p : \TVAR{B_s}, IF
  \qquad\qquad
  IM ::= \epsilon \OR f : \TSPROC{\BSVEC}{s}, IM
\end{syntax}
\endgroup

\noindent where $I_1$ is the name of the interface and $I_2$ is the name of its parent in the hierarchy.
\end{definition}

The intended meaning of the types $T$ is as follows.
The type $B_s$ (interchangeably written as $(B,s)$) is given to expressions $e$: it denotes that the value of $e$ is of type $B$ and all values read to evaluate $e$ are from containers, i.e.\@ variables and fields, of level $s$ \emph{or lower}.
The type $\TVAR{B_s}$ is given to containers: it denotes that a field $p$ or variable $x$ is capable of storing a value of type $B$ of level $s$ \emph{or lower}.
We do not include a type for commands in the language of types; 
however, we say that a command $S$ is safe to execute as $\TCMD{s}$ if (1) all field assignments made during the execution of $S$ are to containers of level $s$ \emph{or higher}, and (2) all guarded commands (i.e., \code{if} and \code{while}) within $S$ must be safe to execute at the same level as their guards (or higher), and this level must be equal to, or higher than, $s$.\footnote{
	As usual, this restriction is necessary to ward against indirect information flows, since a guarded statement could produce observable effects that are indirectly influenced by the expression guard $e$.
}
Finally, the type $\TSPROC{\BSVEC}{s}$ is given to methods $f$: it denotes that the body of $f$ is safe to execute as $\TCMD{s}$, given that the formal arguments have types $\TVAR{\BSVEC}$ (that can be alternatively written as $\TVAR{\VEC{(B,s)}}$), which denotes a sequence of types $\TVAR{B_1, s_1}, \ldots, \TVAR{B_h, s_h}$.

%

We have three kinds of type environments:
\begin{itemize}
  \item $\Sigma$ records the interface hierarchy:
    if we have the declaration \code{interface\,$I_1$\,:\,$I_2$\,\{$IF$\ $IM$\}}, then $\Sigma(I_1) = I_2$.
    For the sake of simplicity, we shall assume that members common to both interfaces also appear in both the parent and the child; i.e.\@ members are not automatically inherited from the parent.
    We also assume that all contracts implement only a single interface.

  \item $\Gamma$ records the statically declared interface definitions and statically assigned types of contracts.
    Thus, if a contract $X$ implements an interface $I$, we have that $\Gamma(X) = I_s$, for some $s$, and $\Gamma(I) = \Gamma_I$, where $\Gamma_I = IF, IM$ (and these correspond to the fields and methods declared in $X$).

  \item $\Delta$ records the types of local variables.
    We record these in a separate environment, since local variables are allocated/deallocated at runtime; consequently,
    this environment may change during the evolution of the program (whereas the other environments are static).
\end{itemize}

Given the intended usage of $\Gamma$, only addresses $X$ and interface names $I$  should appear in its domain.
We express this requirement with the following well-formedness criterion and only consider well-formed environments in what follows.

\begin{definition}[Well-formedness of $\Gamma$]\label{def:wellformedness_gamma}
A type environment $\Gamma$ is \emph{well-formed} if:
\begin{itemize}
  \item Each contract name $X$ in $\DOM{\Gamma}$ has an interface type $I$; i.e.\@ $\Gamma(X) = I_s$ and $I \in \DOM{\Gamma}$.
  \item Each interface name $I$ in $\DOM{\Gamma}$ has a type environment entry $\Gamma_I$; i.e.\@ $\Gamma(I) = \Gamma_I$.
  \item No name in $\VNAMES$, $\FNAMES$ or $\MNAMES$ has an entry in $\Gamma$. 
\end{itemize}
\end{definition}

We remark that we only allow single inheritance among interfaces, and we shall therefore assume the existence of an interface $\ITOP$, corresponding to the basic implementation of a contract, which acts as the single root of the inheritance tree.
It has the following definition:

\begin{lstlisting}[style = tinysol]
interface @$\ITOP$@ : @$\ITOP$@ {
  balance  : @$\TVAR{\TINT, \STOP}$@
  send     : @$\TSPROC{}{\SBOT}$@
  fallback : @$\TSPROC{}{\SBOT}$@
}
\end{lstlisting}

Note that this definition specifies $\ITOP$ as its own parent; since $\ITOP$ is the root of the inheritance tree, it alone can inherit from nothing besides itself.
Indeed, although self-inheritance is permitted by the syntax, a simple consistency check on $\Sigma$ and $\Gamma$ prevents self-inheritance in any cases other than $\ITOP$ (see Appendix~\ref{app:subtyping_consistency}).
It also ensures that all interface members in a parent interface indeed are supertypes of the corresponding members in their children, where
the key rules in the definition of the subtyping relation $\SUBS$ are the following:
\begin{center}
\begin{semantics}
  \RULE[sub-field][typerules_sub_field]
    { \Sigma \vdash B_1 \SUBS B_2 \qquad s_1 \ordleq s_2}
    { \Sigma \vdash \TVAR{B_1, s_1} \SUBS \TVAR{B_2, s_2} }
\end{semantics}
\begin{semantics}
  \RULE[sub-proc][typerules_sub_proc]
    { \Sigma \vdash \VEC{B_2} \SUBS \VEC{B_1} \qquad  s_2 \ordleq s_1  \qquad \VEC{s_2} \ordleq \VEC{s_1}}
    { \Sigma \vdash \TSPROC{\VEC{B_1, s_1}}{s_1} \SUBS \TSPROC{\VEC{B_2, s_2}}{s_2} }
\end{semantics}
\end{center}

\noindent These rules stipulate that reads are covariant (all fields are read-only outside of the contract in which they are declared), whereas writes are contravariant (indeed, methods are comparable to `write-only' fields).\footnote{
\changed{Since fields can be both read and written, one might have expected subtyping to be invariant, since reading is co-variant but writing is contra-variant. 
However, fields are actually only writable within the contract in which they are declared (i.e. by using \code{this.p := e}). 
When accessed from outside the contract (i.e. through the contract’s address), a field can only be read. 
Consequently, all such accesses where subtyping could be involved (e.g.~storing a contract address in a variable \code{x} and then accessing through \code{x.p}) would be reads. 
Therefore, we can allow subtyping of fields to be co-variant.}
}
Since we use static inheritance, and do not allow dynamic creation of new interfaces, these checks only need to be performed once, when $\Sigma$ and $\Gamma$ are built.
In the remainder of this paper, we shall therefore assume that any $\Sigma$ and $\Gamma$ we consider are consistent in this sense.

Finally, for integer and boolean values, the base type can be determined directly from observing the value itself.
However, as the type of a value has both a \emph{data} aspect and a \emph{security} aspect, we also need to assign a security level:
we choose the lowest possible level, $\SBOT$, for integers and booleans, since they have no inherent security level attached, and it is always safe to treat a piece of data as having a higher security level than it actually does.  
With addresses $X \in \ANAMES$, 
we have to retrieve the base type associated to $X$ from the given typing environment $\Gamma$, since addresses are \emph{declared} to have a type and a security level.

\begin{definition}[\textsc{TypeOf}]\label{def:typeof}
We define the partial function $\TYPEOF{\cdot}$ from values to types:
\begin{equation*}
  \TYPEOF{v} = 
  \begin{cases}
    (\TIDF, \SBOT)  & \text{if $v \in \MNAMES$}   \\
    (\TINT, \SBOT)  & \text{if $v \in \INTEGERS$} \\
    (\TBOOL, \SBOT) & \text{if $v \in \BOOLEANS$} \\
    \Gamma(v)       & \text{if $v \in \ANAMES$ and $v \in \DOM{\Gamma}$} \\
    \text{undefined} & \text{otherwise.}
  \end{cases}
\end{equation*}
\end{definition}

\section{Semantic Typing for \TINYSOL{}} 
\label{sec:semTypingTinySol}

The key steps in our approach are the following: 
(1) define runtime type safety in terms of a typed operational semantics; 
(2) define typing interpretations as sets of configurations that are runtime type safe (for all steps); 
(3) define typing as the union of all typing interpretations; 
(4) show that the syntactic type rules of \cite{AGL25} are compatible with the semantic interpretation.
These steps have to be performed both for expressions and for stacks.

\subsection{Typed operational semantics}\label{sec:semantic_types_stacks}\label{sec:typed_semantics}
We give the semantics of our types in terms of a \emph{typed operational semantics} of \TINYSOL{}, which we create by adding 
to the (untyped) operational rules for expressions, statements and stacks (Appendix~\ref{app:untyped_operational_semantics})
some type-related conditions from the static type rules of \cite{AGL25} (reported in Appendix~\ref{app:syntactic_type_system}).
For expressions, judgements take the form
$\Sigma; \Gamma; \Delta \esemDash \CONF{e, \ENV{SV}} \trans v$ and
 should be read as: $e$ evaluates to the value $v$ of type $B_s$, or a subtype thereof, given $\ENV{SV}$ and the type assumptions in $\Sigma; \Gamma; \Delta$.
The details are instructive but, for space reasons, we are forced to defer them to Appendix~\ref{app:semantic_types_expressions}.
%
For stack configurations, a transition is of the form
$\Sigma; \Gamma; \Delta \ssemDash \CONF{Q, \ENV{TSV}} \trans \Sigma; \Gamma; \Delta' \ssemDash[{s'}] \CONF{Q', \ENV{T};\ENV{SV}'}$
and it denotes that we execute the element at the top of the stack $Q$ at security level $s$, relative to the code, state and variable environments $\ENV{TSV}$ and type environments $\Sigma; \Gamma; \Delta$.
As a result of the execution, the state and variable environments may have changed to $\ENV{SV}'$, 
the type environment for local variables may turn into $\Delta'$, if a new variable was dynamically allocated or deallocated, and the security level may have been altered from $s$ to $s'$, if e.g.\@ the transition steps in a context of heightened security.
The rules are given in Figures~\ref{fig:typed_semantics_stacks1} and~\ref{fig:typed_semantics_stacks2}.
By fixing a security level $s$, the rules ensure that (1) all containers written to within a command are of level $s$ or \emph{higher}, and (2) branching, loops and method calls only depend on information that is of a level \emph{lower} than, or equal to, the level of the variables they modify.
There are a few key points worth noting.

\begin{figure}[t]\centering
\begin{semantics}
  \RULE[r-skip][stm_rt_skip]
    { }
    { 
             \Sigma; \Gamma; \Delta \ssemDash \CONF{\code{skip} ; Q, \ENV{TSV}} 
      \trans \Sigma; \Gamma; \Delta \ssemDash \CONF{Q, \ENV{TSV}} 
    }
  \RULE[r-if][stm_rt_if]( s \ordleq s' )
    { \Sigma; \Gamma; \Delta \esemDash[(\TBOOL, s')] \CONF{e, \ENV{SV}} \trans b \in \BOOLEANS }
    { 
              \Sigma; \Gamma; \Delta \ssemDash \CONF{\code{if $e$ then $S_{\TRUE}$ else $S_{\FALSE}$} ; Q, \ENV{TSV}} 
      \trans  \Sigma; \Gamma; \Delta \ssemDash[{s'}] \CONF{S_b ; s; Q, \ENV{TSV}} 
    }
  \RULE[r-while$_\TRUE$][stm_rt_whiletrue]( s \ordleq s' )
    { \Sigma; \Gamma; \Delta \esemDash[(\TBOOL, s')] \CONF{e, \ENV{SV}} \trans \TRUE }
    { 
             \Sigma; \Gamma; \Delta \ssemDash \CONF{\code{while $e$ do $S$} ; Q, \ENV{TSV}}
      \trans \Sigma; \Gamma; \Delta \ssemDash[{s'}] \CONF{S ; s ; \code{while $e$ do $S$} ; Q, \ENV{TSV}} 
    } 
  \RULE[r-while$_\FALSE$][stm_rt_whilefalse]( s \ordleq s' )
    { \Sigma; \Gamma; \Delta \esemDash[(\TBOOL, s')] \CONF{e, \ENV{SV}} \trans \FALSE }
    { 
             \Sigma; \Gamma; \Delta \ssemDash \CONF{\code{while $e$ do $S$} ; Q, \ENV{TSV}}
      \trans \Sigma; \Gamma; \Delta \ssemDash \CONF{Q, \ENV{TSV}} 
    } 
  \RULE[r-decv][stm_rt_decv](x \notin \DOM{\ENV{V}}, \DOM{\Delta})
    { \Sigma; \Gamma; \Delta \esemDash[(B, s')] \CONF{e, \ENV{SV}} \trans v }
    {    \begin{array}{r @{~} l}
             & \Sigma; \Gamma; \Delta \ssemDash \CONF{\code{$\TVAR{B, s'}$ $x$ := $e$ in $S$} ; Q, \ENV{TSV}} \\
      \trans & \Sigma; \Gamma; \Delta, x:\TVAR{B, s'} \ssemDash \CONF{S ; \DEL{x} ; Q, \ENV{TS}; \ENV{V}, x:v} 
    \end{array}
    }
  \RULE[r-assv][stm_rt_assv]({ 
    \begin{array}{r @{~} l}
      x         & \in \DOM{\ENV{V}} \\
      \Delta(x) & = \TVAR{B, s'}    \\
      s         & \ordleq s'
    \end{array}
  }) 
    { \Sigma; \Gamma; \Delta \esemDash[(B, s')] \CONF{e, \ENV{SV}} \trans v }
    { 
    \begin{array}{r @{~} l}
             & \Sigma; \Gamma; \Delta \ssemDash \CONF{\code{$x$ := $e$} ; Q, \ENV{TSV}} \\
      \trans & \Sigma; \Gamma; \Delta \ssemDash \CONF{Q, \ENV{TS}; \ENV{V}\REBIND{x}{v}} 
    \end{array}
    }
  \RULE[r-assf][stm_rt_assf]({ 
  \begin{array}{r @{~} l}
    s_1 & \ordleq s' \\
    s   & \ordleq s' 
  \end{array}
  }) 
    { \Sigma; \Gamma; \Delta \esemDash[(B, s')] \CONF{e, \ENV{SV}} \trans v }
    { 
    \begin{array}{r @{~} l}
             & \Sigma; \Gamma; \Delta \ssemDash \CONF{\code{this.$p$ := $e$} ; Q, \ENV{TSV}} \\
      \trans & \Sigma; \Gamma; \Delta \ssemDash \CONF{Q, \ENV{T}; \ENV{S}\REBIND{X}{\ENV{F}\REBIND{p}{v}}; \ENV{V}} 
    \end{array}
    }
    \WHERE{}{\hspace*{-1cm}
    \begin{array}{lll}
    \ENV{V}(\code{this}) = X
    &
    \ENV{S}(X) = \ENV{F}
    &
    p \in \DOM{\ENV{F}}
    \\
    \Delta(\code{this}) = \TVAR{I, s_1}
    &
    \Gamma(I)(p) = \TVAR{B, s'}
    \end{array}
    }
  \RULE[r-delv][stm_rt_delv]
    { }
    { 
             \Sigma; \Gamma; \Delta, x:\TVAR{B_{s'}} \ssemDash \CONF{\DEL{x} ; Q, \ENV{TS}; \ENV{V}, x:v } 
      \trans \Sigma; \Gamma; \Delta \ssemDash \CONF{Q, \ENV{TSV}} 
    }
  \RULE[r-return][stm_rt_return]
    { \Sigma; \Gamma; \Delta' \vdash \ENV{V}' }
    { 
             \Sigma; \Gamma; \Delta \ssemDash \CONF{(\ENV{V}', \Delta') ; Q, \ENV{TSV}} 
      \trans \Sigma; \Gamma; \Delta' \ssemDash \CONF{Q, \ENV{TS}; \ENV{V}'} 
    }
  \RULE[r-restore][stm_rt_restore](s' \ordgeq s)
    { }
    {  
             \Sigma; \Gamma; \Delta \ssemDash[{s'}] \CONF{s; Q, \ENV{TSV}} 
      \trans \Sigma; \Gamma; \Delta \ssemDash \CONF{Q, \ENV{TSV}}
    }
\end{semantics}
\caption{Typed operational semantics of stacks (1).}
\label{fig:typed_semantics_stacks1}
\end{figure}

\begin{figure}\centering
\begin{semantics}
  \RULE[r-call][stm_rt_call]({
  \begin{array}{r @{~} l}
    s_1                     & \ordleq s' \\
    s                       & \ordleq s' \\
    z \neq 0 \Rightarrow s' & \ordleq s_3, s_4 
  \end{array}
  })
    {
    \begin{array}{r @{~} l}
      \Sigma; \Gamma; \Delta & \esemDash[(I^Y, s')]           \CONF{e_1, \ENV{SV}}     \trans Y \\
      \Sigma; \Gamma; \Delta & \esemDash[(\TINT, s')]         \CONF{e_2, \ENV{SV}}     \trans z \\
      \Sigma; \Gamma; \Delta & \esemDash[(\VEC{B}, \VEC{s}')] \CONF{\VEC{e}, \ENV{SV}} \trans \VEC{v}
    \end{array}
    }
    {
    \begin{array}{r @{~} l}
             & \Sigma; \Gamma; \Delta \ssemDash \CONF{ \CALL{e_1}{m}{\VEC{e}}{e_2} ; Q, \ENV{TSV}}    \\
      \trans & \Sigma; \Gamma; \Delta' \ssemDash[{s'}] \CONF{S ; (\ENV{V}, \Delta) ; s ; Q, \ENV{T}; \ENV{SV}'} 
    \end{array}
    }
    \WHERE{}{\hspace*{-1cm}
      \begin{array}{lll}
        \ENV{V}(\code{this}) = X & \ENV{S}(X) = \ENV{F}^X & \Delta(\code{this}) = \TVAR{I^X, s_1}  \\
        \Gamma(Y) = (I^Y, s_2)   & \ENV{S}(Y) = \ENV{F}^Y & \\ 
        f = 
        \begin{cases} 
          \ENV{V}(m) & \text{if $m = \code{id}$} \\ 
                   m & \text{otherwise} 
        \end{cases} 
	                         & \ENV{T}(Y)(f)  = (\VEC{x}, S)      & \Gamma(I^Y)(f) = \TSPROC{\BSVEC}{s'} \\
       |\VEC{x}| = |\VEC{v}| = |\VEC{B}| = |\VEC{s}| = |\VEC{s'}| = h & \Gamma(I^X)(\code{balance}) = \TVAR{\TINT, s_3}  
       & \Gamma(I^Y)(\code{balance}) =  \TVAR{\TINT, s_4} 
    \end{array}      
    }
    \WHERE{ \ENV{V}' }           { = \code{this}:Y, \code{sender}:X, \code{value}:z, x_1:v_1, \ldots ,x_h:v_h }
    \WHERE{ \ENV{S}' }           { = \ENV{S}\REBIND{X}{\ENV{F}^X [\code{balance -= } z]}\REBIND{Y}{\ENV{F}^Y [\code{balance += } z]} }
    \WHERE{ \Delta' }            { = \code{this}:\TVAR{I^Y, s_2}, \code{sender}:\TVAR{I^X, s_1}, \code{value}:\TVAR{\TINT, s'}, x_1:\TVAR{B_1, s_1'}, \ldots, x_h:\TVAR{B_h, s_h'} }
  \RULE[r-dcall][stm_rt_dcall]({ 
  \begin{array}{r @{~} l}
    s_1 & \ordleq s' \\
    s   & \ordleq s' 
  \end{array}
  })
    { 
    \begin{array}{r @{~} l @{\AND} l }
      \Sigma; \Gamma; \Delta & \esemDash[(I^Y, s')]           \CONF{e, \ENV{SV}}       \trans Y       & \\
      \Sigma; \Gamma; \Delta & \esemDash[(\VEC{B}, \VEC{s}')] \CONF{\VEC{e}, \ENV{SV}} \trans \VEC{v} & \Sigma \vdash I^X \SUBS I^Y 
    \end{array}
    }
    { 
    \begin{array}{r @{~} l}
             & \Sigma; \Gamma; \Delta \ssemDash \CONF{\DCALL{e}{m}{\VEC{e}} ; Q, \ENV{TSV}}       \\
      \trans & \Sigma; \Gamma; \Delta' \ssemDash[{s'}] \CONF{S ; (\ENV{V}, \Delta); s ; Q, \ENV{TS}; \ENV{V}'}  
    \end{array}
    }
    \WHERE{ f }                  { = 
                                 \begin{cases} 
                                   \ENV{V}(m) & \text{if $m = \code{id}$} \\ 
                                   m & \text{otherwise} 
                                 \end{cases} 
    \hspace*{1cm} 
    \ENV{T}(Y)(f)  = (\VEC{x}, S)
    \hspace*{1cm} 
    \Gamma(I^Y)(f) = \TSPROC{\BSVEC}{s'}
     }
    \WHERE{ \Delta(\code{this}) }{ = \TVAR{I^X, s_1} 
    \hspace*{2.9cm}     
     |\VEC{x}| = |\VEC{v}| = |\VEC{B}| = |\VEC{s}| = |\VEC{s'}| = h } 
    \WHERE{\ENV{V}'}             { = \code{this}:\ENV{V}(\code{this}), \code{sender}:\ENV{V}(\code{sender}), \code{value}:\ENV{V}(\code{value}), x_1:v_1, \ldots, x_h:v_h }
    \WHERE{ \Delta' }            { = \code{this}:\Delta(\code{this}), \code{sender}:\Delta(\code{sender}), \code{value}:\Delta(\code{value}), x_1:\TVAR{B_1, s_1'}, \ldots, x_h:\TVAR{B_h, s_h'} }
    \WHERE{ }{ \forall q \in (\DOM{\Gamma(I^Y)} \INTERSECT (\FNAMES \UNION \{\code{balance}\})).\ \Gamma(I^Y)(q) = \Gamma(I^X)(q) }
  \RULE[r-fcall][stm_rt_fcall]({
  \begin{array}{r @{~} l}
    s_1                     & \ordleq s' \\
    s                       & \ordleq s' \\
    z \neq 0 \Rightarrow s' & \ordleq s_3, s_4 
  \end{array}
  })
    { 
    \begin{array}{r @{~} l}
      \Sigma; \Gamma; \Delta & \esemDash[(I^Y, s')]           \CONF{e_1, \ENV{SV}}     \trans Y \\
      \Sigma; \Gamma; \Delta & \esemDash[(\TINT, s')]         \CONF{e_2, \ENV{SV}}     \trans z \\
      \Sigma; \Gamma; \Delta & \esemDash[(\VEC{B}, \VEC{s}')] \CONF{\VEC{e}, \ENV{SV}} \trans \VEC{v}
    \end{array}
    }
    {
    \begin{array}{r @{~} l}
             & \Sigma; \Gamma; \Delta \ssemDash \CONF{\CALL{e_1}{m}{\VEC{e}}{e_2} ; Q, \ENV{TSV}}             \\
      \trans & \Sigma; \Gamma; \Delta' \ssemDash[{s'}] \CONF{S ; (\ENV{V}, \Delta) ; s ; Q, \ENV{T}; \ENV{SV}'} 
    \end{array}
    }
        \WHERE{ } { \hspace*{-1cm}
        \begin{array}{lll}
        \ENV{V}(\code{this}) = X
    &
	\ENV{S}(X) = \ENV{F}^X
    &
        \Delta(\code{this}) = \TVAR{I^X, s_1}
    \\
	\Gamma(Y) = (I^Y, s_2)  
    &
    \ENV{S}(Y) = \ENV{F}^Y    
     &
    \\
    f =  \begin{cases} 
                                   \ENV{V}(m) & \text{if $m = \code{id}$} \\ 
                                    m & \text{otherwise} 
                                 \end{cases} 
	&
	 f \notin \DOM{\ENV{T}(Y)}
	&
	 \ENV{T}(Y)(\code{fallback}) = (\epsilon, S) 
     \\
      \Gamma(I^Y)(\code{fallback}) = \TSPROC{}{s'}
    &
	 \Gamma(I^X)(\code{balance}) = \TVAR{\TINT, s_3}
        &
	\Gamma(I^Y)(\code{balance}) = \TVAR{\TINT, s_4} 
    \end{array}
    }              
    \WHERE{ \ENV{V}' }         { = \code{this}:Y, \code{sender}:X, \code{value}:z, \code{id}:f, \code{args}:\VEC{v} }   
    \WHERE{ \ENV{S}' }         { = \ENV{S}\REBIND{X}{\ENV{F}^X [\code{balance -= z} ]}\REBIND{Y}{\ENV{F}^Y [\code{balance += } z]} }
    \WHERE{ \Delta' }          { = \code{this}:\TVAR{I^Y, s_2}, \code{sender}:\TVAR{I^X, s_1}, \code{value}:\TVAR{\TINT, s'}, \code{id}:\TVAR{\TIDF, s'}, \code{args}:\TVAR{\VEC{B}, \VEC{s}'} }
\end{semantics}
\caption{Typed operational semantics of stacks (2): method calls}
\label{fig:typed_semantics_stacks2}
\end{figure}

 In rules \nameref{stm_rt_if}, \nameref{stm_rt_whiletrue}, \nameref{stm_rt_call}, \nameref{stm_rt_dcall}, and \nameref{stm_rt_fcall}, we allow the security level to be raised from $s$ to $s'$, for some level $s' \ordgeq s$.
    The original level $s$ is pushed onto the stack \emph{before} the command inside the guarded statement or the method body, to be restored after the completion of such a command by \nameref{stm_rt_restore}.
    This rule has again the side condition $s' \ordgeq s$, which is 
    necessary, since otherwise it might be possible to start with a stack of the form $S; s; Q$ and execute it at some level $s'$ that is stricly lower than (or incomparable to) $s$: 
    this would violate the property that all containers written to within a command are of level $s$ or higher.
   To avoid this, the typed operational rules ensure that
    the security levels in a \emph{well-formed} stack must occur in \emph{descending} order (by reading the stack top-down, i.e., from left to right). 
    Consequently, 
    we can choose the execution level to be greater than or equal to the \emph{first} (highest) security level occurring in $Q$.
    If $Q$ does not contain any security level, we can choose the execution level arbitrarily.

  In \nameref{stm_rt_if}, we only run the selected branch, whilst discarding the other.
    In terms of type safety, this means that only the selected branch needs to be type safe, since the other branch is not executed.
    In contrast, the corresponding syntactic type rule must type-check both branches;  therefore, it can potentially reject a program if one of the branches were ill-typed, even if that branch will never be executed.
    The rule \nameref{stm_rt_whilefalse} presents a similar case, since the body of the loop is not checked, because it is not run, whereas in the corresponding syntactic type rule it is.
    These are examples of the \emph{slack} of the syntactic type system.

In rules \nameref{stm_rt_call} and \nameref{stm_rt_fcall}, we require that $s' \ordleq s_3, s_4$ must hold only if the value to be transferred is not 0;
 indeed, if $z=0$, then no value is transferred and the balance fields of the caller and callee are not modified.\footnote{\label{fn:no_exceptions}%
Notice that we should also check that the amount of currency transferred in the call ($z$) is at most the value currently owned by the caller in its \code{balance} field; if this is not the case, the method invocation should fail and an exception should be raised. 
To keep the presentation simple (in particular, to avoid exception handling), in this paper we let method invocations always succeed, irrespectively of the amount of currency associated with them.
}
Moreover, we also perform a lookup in the environment (namely $\Gamma(Y) = (I^Y, s_2)$) to obtain the actual security level $s_2$ of $Y$, rather than just the level $s'$, which is required in the evaluation of $e_1$ and which therefore might be higher.
The actual level is needed when we create the type binding for \code{sender} in the new type environment $\Delta'$, to prevent the level from increasing beyond necessity.

The rule \nameref{stm_rt_fcall} is similar to \nameref{stm_rt_call}, except for some minor differences in the construction of $\Delta'$.
    Firstly, we use the `dummy type' $\TIDF$ for the variable \code{id}. 
    We also assign the security level $s'$ to this variable, which is the same as the security level obtained from the type of the fallback function.
    For the sake of simplicity, we assume that this identifier is not passed as an argument to a method call; i.e.\@ we assume \code{id} is only ever read from within the body of the fallback function, and that its contents are not passed out of this scope.
    Secondly, the \code{args} magic variable is really an \emph{array} of values, which we assume is expanded at runtime.
    Therefore, we also assume that its contents are never assigned to another variable, or otherwise modified or read; i.e.\@ \code{args} is not allowed to appear either as the left-hand side of an assignment, or in an ordinary expression.
    It can only be used as an argument to another method call.
    Therefore, we also do not give it a single security level or a base type, but just retain the individual types of its elements, namely $\TVAR{\VEC{B}, \VEC{s}'}$.

Finally, rule \nameref{stm_rt_dcall} has three main differences w.r.t. \nameref{stm_rt_call}, all related to the fact that the body is executed as if it were present in the caller contract. 
First, there is no currency transfer.
Second, the magic variables \code{this}, \code{sender} and \code{value} are not changed to handle the call. 
Third, we require that the interface of the caller $I_X$ must be a subtype of the interface of the callee $I_Y$, and that all fields common to both must have exactly the same type. This requirement effectively makes subtyping invariant for fields and is needed to prevent an illegal information flow, and also to ensure that the execution cannot be stuck due to a missing field or method in the caller, which might be referenced in the method body.

\subsection{Properties of typed transitions} 
We start by proving a collection of properties of typed transitions, which ensure that the type-related checks indeed guarantee safety in accordance with the intended meaning of the types.
One property
is \emph{$s$-equivalence} between two memory states, written $\Gamma; \Delta \vdash \ENV{SV}^1 =_s \ENV{SV}^2$, as in the classic work on non-interference by Volpano, Irvine and Smith \cite{volpano2000secure_flow_typesystem}. 
Informally, it denotes that $\ENV{SV}^1$ and $\ENV{SV}^2$ contain exactly the same values in all memory cells (fields 
and variables) \emph{of level $s$ or lower}, according to their types in $\Gamma$ and $\Delta$.
The formal definition is straightforward (see Appendix~\ref{app:equivalence_states}).
Another property
is \emph{well-typedness} of memory states, written $\Sigma; \Gamma; \Delta \vdash \ENV{SV}$, which denotes that all values stored in $\ENV{SV}$ indeed are of the respective types of the fields/variables according to the type assignments in $\Sigma; \Gamma; \Delta$.
This definition is also straightforward (see Appendix~\ref{app:syntactic_type_system}) and it is similar to the one used in 
\cite{AGL25}.

Well-typedness alone ensures that all contract addresses occurring in $\ENV{SV}$, either as identifiers or values, have a corresponding interface type $I$ in $\Gamma$.
However, it does \emph{not} ensure that the address also contains a contract with a corresponding implementation matching $I$.
This property is obviously also desirable to avoid stuck configurations, so we introduce the notion of a \emph{consistency} check, written $\Gamma \vdash \ENV{TSV}$, 
to ensure it 
(see Appendix~\ref{app:environment_consistency}).

Given a typed configuration $\Sigma; \Gamma; \Delta \ssemDash \CONF{Q, \ENV{TSV}}$, well-typedness and consistency together ensure that $\Sigma; \Gamma; \Delta$ and $\ENV{SV}$ are compatible.
We must also impose a similar restriction on stacks $Q$, which ensures that the free variable names, field names and addresses occurring in $Q$ are in the domains of $\ENV{SV}$.
Furthermore, since a transition can restore a variable environment from the stack $Q$, we must also naturally require that all return symbols $(\ENV{V}, \Delta)$ occurring in $Q$ are also well-typed and consistent.
Lastly, if a stack $Q$ represents a (partially evaluated) actual program, then all security levels $s_1, \ldots, s_h$ in $Q$ must occur in \emph{descending} order, i.e. $s_1 \ordgeq \ldots \ordgeq s_h$, and these must be lower than the current execution level $s$, i.e. $s \ordgeq s_1$.
We ensure these by a \emph{well-formedness} check on stacks, written $\Sigma; \Gamma; \ENV{SV}, s \vdash Q$ (see Appendix~\ref{app:wellformedness_stack} for a formal definition).
We write $\FIRSTS{Q}$ for the first security level occurring in $Q$, which will also be the highest if $Q$ is well-formed.
If no security level occurs in $Q$, we let $\FIRSTS{Q} = \SBOT$.
Together, well-typedness, consistency and well-formedness ensure that the various components of the configuration are not simply chosen randomly, but in fact constitute an executable program.
We can now formulate a first property of the typed operational semantics, i.e. that the coercion property holds for stacks (the proof is in Appendix~\ref{app:typed_semantics_stacks}).

\begin{proposition}[Coercion property for stacks]\label{theorem:stacks_coercion_property}
Assume $Q$ satisfies that all security levels in $Q$ occur in descending order.
For all type environments $\Sigma$, $\Gamma$ and $\Delta$, security levels $s_1 \ordgeq s_2 \ordgeq \FIRSTS{Q}$, and environments $\ENV{TSV}$, it holds that, for every $s_1'$ such that $\Sigma; \Gamma; \Delta \ssemDash[{s_1}] \CONF{Q, \ENV{TSV}} \trans \Sigma; \Gamma; \Delta' \ssemDash[{s_1'}] \CONF{Q', \ENV{T}; \ENV{SV}'}$, there exists an $s_2'$ such that $\Sigma; \Gamma; \Delta \ssemDash[{s_2}] \CONF{Q, \ENV{TSV}} \trans \Sigma; \Gamma; \Delta' \ssemDash[{s_2'}] \CONF{Q', \ENV{T}; \ENV{SV}'}$.
\end{proposition}

Note that Proposition~\ref{theorem:stacks_coercion_property} does not directly relate the two resulting security levels $s_1'$ and $s_2'$.
Indeed, several of the rules allow a higher security level to be chosen; so, depending on the structure of the security lattice, different choices could be made in the two transitions. 
However, if we assume that the same choice is made in both executions (e.g.\@ the highest, or lowest, possible security level for which the transition can be concluded), then 
$s_1' \ordgeq s_2'$.

We can also show a result that allows us to abstract away from the specific values stored in the environments $\ENV{SV}$.
Indeed, given some $\Sigma$, $\Gamma$ and $\Delta$, we can construct environments $\ENV{SV}$ such that well-typedness and consistency are ensured to hold:

\begin{definition}[Construction of $\ENV{SV}$]\label{def:construction_envsv}
Given $\Sigma$, $\Gamma$ and $\Delta$, let $\DNAMES \DEFSYM \DOM{\Gamma} \INTERSECT \ANAMES$ denote the set of \emph{declared} address names in $\Gamma$.
\begin{itemize}
  \item We construct $\ENV{V}$ as a list of pairs $x_1:v_1, \ldots, x_h:v_h$ such that $\SET{x_1, \ldots, x_h} = \DOM{\Delta}$ and each $v_i$ is such that, if $\Delta(x_i) = \TVAR{B_2, s_2}$, then $\TYPEOF{v_i} = (B_1, s_1)$ for $\Sigma \vdash B_1 \SUBS B_2$ and $s_1 \ordleq s_2$, and $v_i \in \DNAMES$, whenever $v_i \in \ANAMES$.

  \item We construct $\ENV{S}$ as a list of pairs $X_1:\ENV{F}^1, \ldots, X_h:\ENV{F}^h$ such that $\SET{X_1, \ldots, X_h} = \DNAMES$ and each $\ENV{F}^i$ is constructed as follows:
Assume that $\Gamma(X_i) = I_s$ and $\Gamma(I) = \Gamma_F^I, \Gamma_M^I$, then, $\ENV{F}^i$ is a list of pairs $p_1:v_1, \ldots, p_h:v_h$ such that $\SET{p_1, \ldots, p_h} = \DOM{\Gamma_F^I}$ and each $v_i$ is such that, if $\Gamma_F^I(x_i) = \TVAR{B_2, s_2}$, then $\TYPEOF{v_i} = (B_1, s_1)$ for $\Sigma \vdash B_1 \SUBS B_2$ and $s_1 \ordleq s_2$, and $v_i \in \DNAMES$, whenever $v_i \in \ANAMES$.
\end{itemize}
\end{definition}

\changed{
Notice that the introduction of the above definition is a consequence of our starting point being the untyped operational semantics, on top of which we create the  typed one (which therefore is on configurations containing $\ENV{SV}$). Definition~\ref{def:construction_envsv} then lets us abstract away from the concrete values stored in $\ENV{SV}$ and allows us to focus only on the types of those values. In this way, we are incrementally working our way upwards, to a higher abstraction level, whilst retaining the correspondence with the real, untyped operational semantics.
}

\begin{lemma}[Correctness of the construction]
\label{lemma:construction_envsv}
Given $\Sigma$, $\Gamma$ and $\Delta$, if $\ENV{SV}$ is built according to Definition~\ref{def:construction_envsv}, then 
$\Sigma; \Gamma; \Delta \vdash \ENV{SV}$ (well-typedness), 
and $\Gamma \vdash \ENV{SV}$ (consistency). 
\end{lemma}

\begin{corollary}[Safe state abstraction]\label{theorem:stacks_runtime_states_abstraction}
For all environments $\Sigma$, $\Gamma$, $\Delta$ and $\ENV{T}$, security level $s$, and $\ENV{SV}^1, \ENV{SV}^2$ built according to Definition~\ref{def:construction_envsv}, it holds that, for every $s_1'$ such that $\Sigma; \Gamma; \Delta \ssemDash \CONF{Q, \ENV{T}; \ENV{SV}^1} \trans \Sigma; \Gamma; \Delta_1' \ssemDash[{s_1'}] \CONF{Q', \ENV{T}; \ENV{SV}^{1'}}$, there exists an $s_2'$ such that $\Sigma; \Gamma; \Delta \ssemDash \CONF{Q, \ENV{T}; \ENV{SV}^2} \trans \Sigma; \Gamma; \Delta_2' \ssemDash[{s_2'}] \CONF{Q', \ENV{T}; \ENV{SV}^{2'}}$. 
\end{corollary}

We can now ascertain that the typed semantics does ensure the intended notion of type safety; the following two theorems are indeed the first cornerstone results of our approach.

\begin{theorem}[Runtime preservation]\label{theorem:stacks_runtime_preservation}
Assume that 
$\Sigma; \Gamma; \Delta \vdash \ENV{SV}$ (Well-typedness),
$\Gamma \vdash \ENV{TSV}$ (Consistency), and 
$\Sigma; \Gamma; \ENV{SV}; s \vdash Q$ (Well-formedness).
If $\Sigma; \Gamma; \Delta \ssemDash \CONF{Q, \ENV{TSV}} \trans \Sigma; \Gamma; \Delta' \ssemDash[{s'}] \CONF{Q', \ENV{T}; \ENV{SV}'}$,
then 
$\Sigma; \Gamma; \Delta' \vdash \ENV{SV}'$ (Well-typedness),
$\Gamma \vdash \ENV{T};\ENV{SV}'$ (Consistency),
$\Sigma; \Gamma; \ENV{SV}'; s' \vdash Q'$ (Well-formedness), and 
 $\Gamma \vdash \ENV{S} =_{s''} \ENV{S}'$ ($s''$-equivalence), for all $s''$ such that $s \not\ordleq s''$.
\end{theorem}

Theorem~\ref{theorem:stacks_runtime_preservation} assures us that the properties of well-typedness, consistency, and well-formedness are preserved by the typed semantics; additionally, any change to $\ENV{S}$ must be to a field that is of level $s$ or higher.
The latter is in accordance with the meaning of the static type judgment $\TCMD{s}$, corresponding to the parameterisation of the turnstile $\ssemDash$ in the typed semantics.
All variables of a level \emph{strictly below}, or incomparable to, $s$ will be unaffected by the transition steps; so, the environments $\ENV{S}$ and $\ENV{S}'$ agree on all such entries.
This is the safety property induced by the types, since executing the stack at level $s$ means that no variables of a level \emph{lower} than, or incomparable to, $s$ may be modified (cf.\@ rules \nameref{stm_rt_assf}, \nameref{stm_rt_call} and \nameref{stm_rt_fcall}), 
and information is not permitted to flow from a higher level to a lower one.
Indeed, we can also use Theorem~\ref{theorem:stacks_runtime_preservation} to show that the typed semantics ensures non-interference for stack executions.

\begin{theorem}[Runtime non-interference for stacks]\label{theorem:stacks_runtime_noninterference}
Assume that
$\Sigma; \Gamma; \Delta \vdash \ENV{SV}^1$ and $\Sigma; \Gamma; \Delta \vdash \ENV{SV}^2$ (Well-typedness), 
$\Gamma \vdash \ENV{T}; \ENV{SV}^1$ and $\Gamma \vdash \ENV{T}; \ENV{SV}^2$ (Consistency), 
$\Sigma; \Gamma; \ENV{SV}^1, s \vdash Q$ and $\Sigma; \Gamma; \ENV{SV}^2, s \vdash Q$ (Well-formedness), and 
$\Gamma; \Delta \vdash \ENV{SV}^1 =_s \ENV{SV}^2$ ($s$-equivalence). 
If $\Sigma; \Gamma; \Delta \ssemDash \CONF{Q, \ENV{T}; \ENV{SV}^i} \trans \Sigma; \Gamma; \Delta' \ssemDash[{s'}] \CONF{Q_i, \ENV{T}; \ENV{SV}^{i'}}$, for $i \in \{1,2\}$, 
then $\Gamma; \Delta' \vdash \ENV{SV}^{1'} =_s \ENV{SV}^{2'}$. 
\end{theorem}

It is worth noting that the executions of the configurations $\CONF{Q, \ENV{T}; \ENV{SV}^1}$ and $\CONF{Q, \ENV{T}; \ENV{SV}^2}$ actually allow for two different methods to be called.
Suppose for example that $Q$ is of the form $\CALL{x}{f}{\VEC{e}}{e_2}$, and $x$ is a variable of type $\TVAR{I, s'}$ for some level $s'$ strictly higher than $s$.
This would allow it to contain two different addresses, e.g.\@ $X_1$ and $X_2$ in the two executions, which would lead to two different methods being called, albeit with the same signature.
Well-typedness and consistency of the two environments ensure that either both methods exist or neither does, because both $X_1$ and $X_2$ must be addresses of contracts that implement all the methods of the interface $I$.
Thus, the two resulting variable environments, $\ENV{V}^{1'}$ and $\ENV{V}^{2'}$, will still contain exactly the same variable entries, and with values agreeing at all levels not higher than $s$.

On the other hand, calling two different methods would also mean that the two different \code{balance} fields of the two callees would be different in $\ENV{S}^{1'}$ and $\ENV{S}^{2'}$ (e.g., in $\ENV{S}^{1'}$ it might be $X_1.\code{balance}$ that was modified, but $X_2.\code{balance}$ in $\ENV{S}^{2'}$).
However, this does not contradict  
Theorem~\ref{theorem:stacks_runtime_noninterference} because of the side condition of the rules \nameref{stm_rt_call} and \nameref{stm_rt_fcall}, which are the only two rules that could have been used to conclude such a transition (rule \nameref{stm_rt_dcall} cannot have been used, since it does not modify $\ENV{S}$):
this side condition states that $s' \ordleq s_3, s_4$, where $s_3$ and $s_4$ are the levels of the \code{balance} fields of the caller and callee, respectively.
As we assumed that $s'$ is strictly greater than $s$, then so are both $s_3$ and $s_4$, which means that this difference will not affect the $s$-equality of the two environments.

\subsection{Typing interpretations and semantic typing rules}\label{sec:typing_interpretation_stm_stack}
Theorem~\ref{theorem:stacks_runtime_preservation} ensures that a configuration $\CONF{Q, \ENV{TSV}}$ does not get stuck due to a malformed stack or an attempt to access a contract member on an address that does not contain an implementation.
Once we have ruled out such `trivial' type errors, what remains is that a stuck configuration can only derive from three situations:
(1) $Q$ is of the form $\bot$, meaning we have reached the bottom of the stack, so the execution terminates normally;
(2) $Q$ is of the form $\code{throw}; Q'$, meaning an exception was thrown, so the execution terminates abnormally; and
(3) the corresponding untyped transition would cause a violation of \emph{type safety} by allowing an insecure information flow.
Of these, the first two forms represent terminal stacks, defined as 
$\TSTACKS \triangleq \{\bot\} \UNION \{\code{throw};Q\ |\ Q \in \STACKS\}$;
they are obviously safe to execute, since they have no transitions, and can be detected simply by examining the structure of $Q$.
Thus, the key point is that \emph{a non-terminal stuck configuration indicates a runtime type error}.
We shall use this idea to formulate a notion of \emph{typing interpretation} for stacks $Q$, in such a way that $Q$  must either be one of the first two forms or have a typed transition (with some appropriately shaped $\ENV{SV}$) whose reduct is again contained in the typing interpretation.

However, the type constraints represented by $\Delta$ and the current security level $s$ can \emph{change} after a transition step.
Hence, also $\Delta$ and $s$ must be included in the typing interpretation, together with $Q$.
Our notion of typing interpretations will therefore not be for stacks alone, but for \emph{stack type triplets} $(Q, \Delta, s)$, where $\Delta$ and $s$ represent the currently active variable type environment and security level under which $Q$ is executing.
To make this clearer, we first define a new transition system directly on the triplets themselves.

\begin{definition}[Type lifted transition]\label{def:type_lifted_transitions}
Assume $\Sigma$, $\Gamma$ and $\ENV{T}$ are fixed.
A \emph{type lifted transition system} is a tuple $(\SETTRIPLETS, \ttrans)$, where $\SETTRIPLETS = 
\STACKS \times \TYPEENVS \times \SECLEVELS$ is the set of stack type triplets $(Q, \Delta, s)$ and the transition relation  $\Sigma; \Gamma; \ENV{T} \vDash (Q, \Delta, s) \ttrans (Q', \Delta', s')$ holds whenever there exist an $\ENV{SV}$ built from $\Sigma; \Gamma; \Delta$ according to Definition~\ref{def:construction_envsv} and an $\ENV{SV}'$ such that $\Sigma; \Gamma; \Delta; \ENV{T} \ssemDash \CONF{Q, \ENV{TSV}} \trans \Sigma; \Gamma; \Delta'; \ENV{T} \ssemDash[{s'}] \CONF{Q', \ENV{T};\ENV{SV}'}$. 
\end{definition}

By Corollary~\ref{theorem:stacks_runtime_states_abstraction}, we know that, if there exists an $\ENV{SV}^1$ such that the configuration $\Sigma; \Gamma; \Delta \ssemDash \CONF{Q, \ENV{T}; \ENV{SV}^1}$ has a transition, then a transition will also exist for all similarly shaped (according to Definition~\ref{def:construction_envsv}) $\ENV{SV}^2$.
All such $\ENV{SV}$ are ensured to be well-typed, consistent and compatible with $(Q, \Delta, s)$ by construction, and one way to view this quantification over all such environments is to think of it as a quantification over all possible, appropriately shaped and typed \emph{inputs} to the program on the stack.
This is quite natural, since we obviously want a notion of safety that does not depend on the \emph{actual} values being passed to the program, but only on the \emph{types} of those values.
It is in this sense that the transition system is \emph{type lifted:} the shape of the transition system is not determined by actual values, but by the \emph{types} of those values.

Different inputs may, of course, yield different reducts, e.g.\@ in case the stack $Q$ had an \code{if-else} statement at the top, and the environments yield different values for the guard expression.
Hence, there may also be more than one possible reduct for a given triplet $P$ in the type lifted transition system, and each of them must be contained in the typing interpretation.
We define the notion of typing interpretation as follows.

\begin{definition}[Typing interpretations]\label{def:typing_interpretation_stm}
A \emph{typing interpretation} for stack type triplets w.r.t.\@ environments $\Sigma$, $\Gamma$ and $\ENV{T}$, written $\STYPI$, is a set of stack type triplets $P = (Q, \Delta, s)$ such that, for every $(Q, \Delta, s) \in \STYPI$, it holds that 
\begin{enumerate}
  \item\label{case:typi1} $Q \in \TSTACKS$, or

  \item\label{case:typi2} $\exists P_1'$ such that $\Sigma; \Gamma; \ENV{T} \vDash (Q, \Delta, s) \ttrans P_1'$ and
                          $\forall P_2'$ such that $\Sigma; \Gamma; \ENV{T} \vDash (Q, \Delta, s) \ttrans P_2'$ it holds that $P_2' \in \STYPI$.
\end{enumerate}
The \emph{typing} of stack type triplets w.r.t.\@ $\Sigma$, $\Gamma$ and $\ENV{T}$, written $\STYPING$, is the union of all typing interpretations $\STYPI$.
We write $\Sigma; \Gamma; \Delta; \ENV{T} \vDash Q : \TCMD{s}$ if 
$(Q, \Delta, s)\in\STYPING$, and 
$\Sigma; \Gamma; \Delta; \ENV{T} \vDash S : \TCMD{s}$ if 
$(S; \bot, \Delta, s)\in\STYPING$.
\end{definition}

Notice that the coinductive nature of typing interpretations allows for the possibility of divergent executions.
Moreover, the existence of a typing interpretation $\STYPI$ containing $(Q, \Delta, s)$ expresses that $Q$ is safe to execute at level $s$ or lower (down to $\FIRSTS{Q}$) for all steps, according to the type constraints in $\Sigma$, $\Gamma$ and $\Delta$. 
Note also that typing interpretations for stack type triplets are formulated relative to a \emph{particular} choice of code environment $\ENV{T}$.
This is necessary because of the presence of fallback functions, which can place (statically) untypable code on the stack.
All we therefore can require is that such code must be runtime safe, if it is executed, which is what the existence of a typing interpretation expresses.

By using the type-lifted transition system in Definition~\ref{def:typing_interpretation_stm}, we abstract away from any particular choice of $\ENV{SV}$:
we only require that one can be built from $\Sigma$ and $\Gamma$, and from the $\Delta$ in the stack type triplet.
In the definition, we discard the environments $\ENV{SV}'$ from the configuration in the reduct of the underlying typed transition, since we use the type-lifted transition system, but we require that the triplet $P_2' = (Q_2', \Delta_2', s_2')$  must again be in the typing interpretation.
This works because, by Theorem~\ref{theorem:stacks_runtime_preservation}, we know that the properties of well-typedness, consistency and well-formedness are preserved by the typed semantics, so this $\ENV{SV}'$ will again be amongst those that can be built from $\Sigma$, $\Gamma$ and $\Delta_2'$ in the next step.

Using Definition~\ref{def:typing_interpretation_stm}, we can now show some properties of typing interpretations for stacks; all the proofs are in Appendix~\ref{app:semantic_typing_stacks}. 
The first one can be used as a proof technique for showing that a stack cannot be typed under the prescriptions of the given  environments.

\begin{lemma}\label{lemma:untypable}
If  $\Sigma; \Gamma; \ENV{T} \vDash (Q, \Delta, s) \ttrans^* (Q', \Delta', s') \not\ttrans$, for some $\Delta'$, $s'$ and $Q' \not\in \TSTACKS$, then $(Q, \Delta, s) \not\in \STYPING$.
\end{lemma}

We can now also prove that typing is closed under various operations, including, importantly, the syntactic constructors of the language, given certain assumptions. 
As an example of what is to come, we can now show that typing is closed under coercions.

\begin{lemma}[Typing coercion for stacks and statements]\label{lemma:compatibility_stacks_coercion}\label{lemma:compatibility_stm_coercion}
\ 
\begin{enumerate}
\item If  $\Sigma; \Gamma; \Delta; \ENV{T} \vDash Q : \TCMD{s_1}$ and 
$s_1 \ordgeq s_2 \ordgeq \FIRSTS{Q}$, then $\Sigma; \Gamma; \Delta; \ENV{T} \vDash Q : \TCMD{s_2}$.
\item If $\Sigma; \Gamma; \Delta; \ENV{T} \vDash S : \TCMD{s_1}$ and $s_1 \ordgeq s_2$, then $\Sigma; \Gamma; \Delta; \ENV{T} \vDash S : \TCMD{s_2}$.
\end{enumerate}
\end{lemma}


\begin{figure}\centering
\begin{semantics}
  \RULE[st-stack-sub]({ s_1 \ordgeq s_2 \ordgeq \FIRSTS{Q} })
    { \Sigma; \Gamma; \Delta; \ENV{T} \vDash Q : \TCMD{s_1} }
    { \Sigma; \Gamma; \Delta; \ENV{T} \vDash Q : \TCMD{s_2} }
  \RULE[st-stm-sub][stm_st_sub]( s_1 \ordgeq s_2 )
    { \Sigma; \Gamma; \Delta; \ENV{T} \vDash S : \TCMD{s_1} }
    { \Sigma; \Gamma; \Delta; \ENV{T} \vDash S : \TCMD{s_2} }
\end{semantics}
\caption{Semantic subtyping rules for stacks and statements.}
\label{fig:semantic_subtype_rules_stacks_stm}
\end{figure}

The previous results are instructive for two reasons.
Firstly, they are proved by coinduction: we construct a new candidate typing interpretation $\CANDSTYPI$ from an existing one, and show that $\CANDSTYPI$ is still a typing interpretation (that is contained in the typing of $\Sigma; \Gamma; \ENV{T}$, since $\STYPING$ is the union of all typing interpretations).
Secondly, both results have a particular form, with a number of premises involving typing that imply a single conclusion, still involving typing; hence, they can alternatively be written in the form of \emph{inference rules}, as shown in Figure~\ref{fig:semantic_subtype_rules_stacks_stm} (in particular, 
\nameref{stm_st_sub} resembles the contravariant coercion rule for typing statements in \cite{AGL25}). 
However, unlike the typing rules from \cite{AGL25}, these rules have been \emph{proved admissible} based on the definition of typing, rather than being defined a priori; this is one of the key differences between the syntactic and the semantic approaches.

Similarly, we have a collection of inference rules for each of the stack operations; 
as with the coercion rules, these rules closely resemble the static type rules for stacks (see Figure~\ref{fig:type_rules_stacks} in Appendix~\ref{app:syntactic_type_system}).
%
We now focus on the main rules, given in Figure~\ref{fig:semantic_type_rules_stm}, which handle the syntactic category of statements and allow us to infer
type safety assertions of the form $\Sigma; \Gamma; \Delta; \ENV{T} \vDash S : \TCMD{s}$
(proofs are again in Appendix~\ref{app:semantic_typing_stacks}).
Most rules are self-explanatory: compare the type checks with the corresponding ones in the typed operational semantics rules (or, again, to the static typing rules).
We discuss typing the fallback function in Section~\ref{sec:typing_fallback}.
For method and delegate calls, 
we define a semantic notion of type safety for the code environment $\ENV{T}$, using Definition~\ref{def:typing_interpretation_stm};
this is written $\Sigma; \Gamma; \ENV{T} \vDash \ENV{T}$ and states that there exists a typing interpretation for each method body $S$, with $\Delta$ and $s$ derived from the method signature in $\Gamma$ (see Figure~\ref{fig:semantic_type_rules_envtm} in Appendix~\ref{app:semantic_typing_stacks}).
Including the premise $\Sigma; \Gamma; \ENV{T} \vDash \ENV{T}$ in the following lemmas requires that every method declaration in $\ENV{T}$ (except fallback functions) must be semantically type safe.
This is of course a reasonable requirement, since a standard Subject-Reduction theorem  would also require well-typedness of all the environments.
However, it is worth pointing out that this is an over-approximation, since we actually only need the method $f$ (and its possible subtypes) declared on the interface $I^Y$ (and its possible subtypes) to be type safe, in order for the proof to work.
This restriction can thus be weakened without affecting the proof.
We preferred to use the more general requirement $\Sigma; \Gamma; \ENV{T} \vDash \ENV{T}$ merely to keep the premises of the rules simpler.

\begin{figure}\centering
\begin{semantics}
  \RULE[st-skip]
    { }
    { \Sigma; \Gamma; \Delta \vDash \code{skip} : \TCMD{s} }
  \RULE[st-throw]
    { }
    { \Sigma; \Gamma; \Delta \vDash \code{throw} : \TCMD{s} }
  \RULE[st-seq]
    { \Sigma; \Gamma; \Delta; \ENV{T} \vDash S_1 : \TCMD{s} \AND \Sigma; \Gamma; \Delta; \ENV{T} \vdash S_2 : \TCMD{s} }
    { \Sigma; \Gamma; \Delta; \ENV{T} \vDash S_1; S_2 : \TCMD{s} }
  \RULE[st-if][stm_st_if]
    {  \Sigma; \Gamma; \Delta \vDash e : \TBOOL_s \AND
       \Sigma; \Gamma; \Delta; \ENV{T} \vDash S_{\TRUE}  : \TCMD{s}  \AND
       \Sigma; \Gamma; \Delta; \ENV{T} \vDash S_{\FALSE} : \TCMD{s} 
    }
    { \Sigma; \Gamma; \Delta; \ENV{T} \vDash \code{if $e$ then $S_{\TRUE}$ else $S_{\FALSE}$} : \TCMD{s} }
  \RULE[st-while]
    { \Sigma; \Gamma; \Delta \vDash e : \TBOOL_s \AND \Sigma; \Gamma; \Delta; \ENV{T} \vDash S : \TCMD{s} }
    { \Sigma; \Gamma; \Delta; \ENV{T} \vDash \code{while $e$ do $S$} : \TCMD{s} }
  \RULE[st-decv]
    { 
      \Sigma; \Gamma; \Delta                              \vDash e : B_{s_1} \AND
      \Sigma; \Gamma; \Delta, x : \TVAR{B_{s_1}}; \ENV{T} \vDash S : \TCMD{s} 
    }
    { \Sigma; \Gamma; \Delta; \ENV{T} \vDash \code{$\TVAR{B_{s_1}}$ $x$ := $e$ in $S$} : \TCMD{s} }
  \RULE[st-assv]({ \Delta(x) = \TVAR{B_s} })
    { \Sigma; \Gamma; \Delta \vDash e : B_s }
    { \Sigma; \Gamma; \Delta; \ENV{T} \vDash \code{$x$ := $e$} : \TCMD{s} }
  \RULE[st-assf][stm_st_assf]({
  \begin{array}{r @{~} l}
    \Delta(\code{this}) & = \TVAR{I_{s_1}} \\
    \Gamma(I)(p)        & = \TVAR{B_s}     \\
  \end{array}
  })
    { \Sigma; \Gamma; \Delta \vDash e : B_s }
    { \Sigma; \Gamma; \Delta; \ENV{T} \vdash \code{this.$p$ := $e$} : \TCMD{s} }
  \RULE[st-call][stm_st_call]({ s_1 \ordleq s \ordleq s_3, s_4 })
    { 
    \begin{array}{l @{~} l}
      \Sigma; \Gamma; \Delta \vDash e_1 : I^Y_s   & \AND \Sigma; \Gamma; \Delta \vDash \VEC{e} : \BSVEC    \\
      \Sigma; \Gamma; \Delta \vDash e_2 : \TINT_s & \AND \Sigma; \Gamma; \ENV{T} \vDash \ENV{T}
    \end{array}
    }
    { \Sigma; \Gamma; \Delta; \ENV{T} \vDash \CALL{e_1}{f}{\VEC{e}}{e_2} : \TCMD{s} }
    \WHERE{ }{\hspace*{-.7cm}
    \begin{array}{ll}
    \Gamma(I^Y)(f) = \TSPROC{\BSVEC}{s}
    &
    \Delta(\code{this}) = \TVAR{I^X_{s_1}}
    \\
    \Gamma(I^X)(\code{balance}) =  \TVAR{\TINT_{s_3}} \qquad\qquad
    &
    \Gamma(I^Y)(\code{balance}) = \TVAR{\TINT_{s_4}}
    \end{array}
    }
  \RULE[st-dcall][stm_st_dcall]({ s_1 \ordleq s })
    { 
    \begin{array}{l @{~} l}  
      \Sigma; \Gamma; \Delta \vDash e : I^Y_s        & \AND \Sigma; \Gamma; \Delta \vDash \VEC{e} : \BSVEC  \\
       \Sigma \vdash I^X \SUBS I^Y & \AND \Sigma; \Gamma; \ENV{T} \vDash \ENV{T} 
    \end{array}
    }
    { \Sigma; \Gamma; \Delta; \ENV{T} \vdash \DCALL{e}{f}{\VEC{e}} : \TCMD{s} } 
    \WHERE{ }{\hspace*{-.7cm}
    \begin{array}{ll}
    \Gamma(I^Y)(f) = \TSPROC{\BSVEC}{s}\ \qquad\qquad\qquad
    &
    \Delta(\code{this}) = \TVAR{I^X_{s_1}}
    \end{array}
    }
    \WHERE{ \forall q \in }{ (\DOM{\Gamma(I^Y)} \INTERSECT (\FNAMES \UNION \{\code{balance}\})).\ \Gamma(I^Y)(q) = \Gamma(I^X)(q) }
  \RULE[st-fcall][stm_st_fcall]({ s_1 \ordleq s \ordleq s_3, s_4 })
    { 
      \Sigma; \Gamma; \Delta \vDash e_1 : I^Y_s  
      \AND
      \Sigma; \Gamma; \Delta \vDash \VEC{e} : \BSVEC
      \AND 
      \Sigma; \Gamma; \Delta \vDash e_2 : \TINT_s 
    }
    { \Sigma; \Gamma; \Delta; \ENV{T} \vDash \CALL{e_1}{f}{\VEC{e}}{e_2} : \TCMD{s} }
    \WHERE{ }{\hspace*{-.7cm}
    \begin{array}{ll}
    f \notin \DOM{\Gamma(I^Y)} 
    &
    \Delta(\code{this}) = \TVAR{I^X_{s_1}}
    \\
    \Gamma(I^X)(\code{balance}) =  \TVAR{\TINT_{s_3}} \qquad\qquad
    &
    \Gamma(I^Y)(\code{balance}) = \TVAR{\TINT_{s_4}}
    \end{array}
    }
    \WHERE{ \forall Y }         { \in \DOM{\ENV{T}} \SUCHTHAT \Gamma(Y) = I'_{s'} \land \Sigma \vdash I'_{s'} \SUBS I^Y_s  \implies \Sigma; \Gamma; \Delta'; \ENV{T} \vDash S : \TCMD{s} }
    \WHERE{ }                   {\!\!\!\text{where:} }
    \WHERE{ }                   {{
      \begin{array}{l}
        \ENV{T}(Y)(\code{fallback}) = (\epsilon, S) \\
        \Delta' = \code{this}:\TVAR{I'_{s'}}, \code{value}:\TVAR{\TINT_s}, \code{sender}:\TVAR{\ITOP_{\STOP}}, \code{id}:\TVAR{\TIDF_s}, \code{args}:\TVAR{\BSVEC}
      \end{array}
    }}
\end{semantics}
\caption{Semantic type rules for statements}
\label{fig:semantic_type_rules_stm}
\label{fig:semantic_type_rules_fallback}
\end{figure}


The similarity between the syntactic typing rules from \cite{AGL25} and the corresponding semantic ones 
allows us to relate the two judgements; this is our second cornerstone result.

\begin{theorem}[Compatibility for statements]\label{thm:compat_comms}
If $\Sigma; \Gamma; \EMPTYSET \vdash \ENV{T}$ and $\Sigma; \Gamma; \Delta \vdash S : \TCMD{s}$, then $\Sigma; \Gamma; \ENV{T} \vDash \ENV{T}$ and $\Sigma; \Gamma; \Delta; \ENV{T} \vDash S : \TCMD{s}$.
\end{theorem}

Notice that, also for the semantic typing, the typing rules 
provide only an \emph{approximation} to the set of type safe programs, just as the syntactic type rules only define well-typedness as an approximation to type safety.
Both must necessarily induce some \emph{slack}, that is, statements $S$ for which safety cannot be concluded by the rules, but which nevertheless behave in a type safe way at runtime (i.e.\@ whose typed transitions would not be stuck).
However, the important thing to note here is that using the semantic type rules is \emph{optional}:
there is always the possibility of showing safety for a statement $S$ directly by giving an explicit typing interpretation instead, to satisfy a premise anywhere in the system.
 In particular, it now becomes possible to approach safety for some usages of the fallback function.

\subsection{Safe and unsafe fallback calls}\label{sec:typing_fallback}
Armed with our notion of type safety, we can now return to the problem of how to ensure safety of the fallback function.
The semantic typing rule for fallback function calls is in Figure~\ref{fig:semantic_type_rules_fallback}.
The only difference w.r.t. \nameref{stm_st_call} is that we obtain the body $S$ of the fallback function and the type environment $\Delta'$ from an explicit quantification over addresses $Y$ in the side condition, rather than having the more general statement $\Sigma; \Gamma; \ENV{T} \vDash \ENV{T}$ in the premise.
That difference is unimportant in \emph{proving} admissibility of the rule, but it does make a difference in \emph{using} it for type checking a piece of code.
In rules \nameref{stm_st_call} and \nameref{stm_st_dcall}, the statement $\Sigma; \Gamma; \ENV{T} \vDash \ENV{T}$ is independent of the other premises; 
it does not have to be proved for each application of one of the aforementioned rules.
In contrast, in the rule \nameref{stm_st_fcall}, the quantification is over all addresses $Y$ satisfying a condition that depends on the type of the expression $e_1$.  This, in turn, determines the type environment $\Delta'$ and the level $s$ for which the body $S$ must be type safe.
The truth value of this side condition may thus have to be evaluated anew \emph{each time} the rule is applied, whereas $\Sigma; \Gamma; \ENV{T} \vDash \ENV{T}$ can be shown separately, once and for all.

In essence, rule \nameref{stm_st_fcall} states that the fallback functions of \emph{every} contract implementing the interface $I^Y$, \emph{or one of its descendants}, must be type safe at level $s$.
This extra step is necessary, since $I^Y_s$ is not guaranteed to be the actual type of $e_1$, but only its runtime type.
Any address having an actual type $I'_{s'}$, where $I'$ is a descendant of $I^Y$ and $s' \ordleq s$, is also in the set of values inhabiting the type $I^Y_s$; hence they too could be obtained from the evaluation of $e_1$, and they must therefore be included as well.

As with the other semantic type rules, rule \nameref{stm_st_fcall} only requires that the premises and side conditions must hold, but it imposes no restrictions on how they are shown.
A proof of type safety for a fallback function body $S$, w.r.t.\@ some type environment $\Delta$ and security level $s$, could simply be given in the form of an explicit typing interpretation $\STYPI$ containing $(S;\bot, \Delta, s)$.
Such a typing interpretation is just a set of triplets, and it could therefore in principle be provided by the contract creator and stored together with the code itself, to witness the type safety of the contract's fallback function.
This would then allow any user of the code, who might 
generate a method call that will result in an invocation of that fallback function, to use this typing interpretation to satisfy the premise of rule \nameref{stm_st_fcall} when type checking his own code; this can be seen as a form of \emph{proof-carrying code} \cite{necula/1997/popl/pcc}.

As an example of a safe usage of the fallback function, consider the code in Figure~\ref{fig:example_fallback_proxy}.
As we discussed in Section~\ref{sec:motEx}, there are only finitely many possible continuations after any method call to \code{Proxy}:
either \code{update}, or one of $f_1$, \ldots, $f_h$, or \code{skip} (when the fallback function is invoked).
Thus, if none of these methods can cause an insecure information flow, then any \code{call Proxy.$f$()\$$z$} for $f \neq \code{update}$ will be safe too.
This cannot be inferred directly, using the other semantic type rules, because the interface definition for $I^P$ will not contain a declaration of the method $f$.
However, with an explicit typing interpretation witnessing the safety of the execution of \code{call this.impl.id()\$value}, which will include the executions of $f_1$, \ldots, $f_h$, the call can nevertheless be shown to be safe by using rule \nameref{stm_st_fcall}.
Of course, \emph{finding} such a typing interpretation may be a monumental task, even with only finitely many possible continuations.
We return to this issue in Section~\ref{sec:upto_techniques}.

Of course, not all usages of the fallback function are type safe, even with semantic typing.
Suppose, for example, that we replaced the forwarding method call in the fallback function of \code{Proxy} with a \emph{delegate} call; i.e.
\begin{lstlisting}[style=tinysol, frame=none, numbers=none]
  func fallback() { dcall this.impl.id(args) }  @\hspace*{3cm}  ($\dagger$)@
\end{lstlisting}

\noindent Any attempt to invoke \code{Proxy.$f_i$()} would now result in $f_i$ being executed in the context of \code{Proxy}, rather than \code{X}, thereby allowing it to mutate the state of \code{Proxy}.
However, using the typed semantics, this execution would be stuck because of the type constraint on the rule \nameref{stm_rt_dcall} for delegate calls, which requires that the caller must be a subtype of the callee.
Clearly, it does not hold that $\Sigma \vdash I^P \SUBS I^X$, since \code{X} declares the methods $f_1, \ldots, f_h$ whereas \code{Proxy} does not contain similar declarations; so, this constraint cannot be satisfied.\footnote{Conversely, if \code{Proxy} \emph{did} declare methods $f_1, \ldots, f_h$ too, then none of these methods in \code{X} could be invoked through the fallback function of \code{Proxy}, since it is only executed on calls to \emph{non-existing} methods.}
Thus, no typing interpretation can be created to witness the safety of $(\dagger)$, cf.\@ Lemma~\ref{lemma:untypable}.
Altering our \code{Proxy} example with $(\dagger)$ is in fact a key ingredient for mounting an attack similar to the Parity Multisig Wallet (PMW) attack \cite{chinese2021flowtypes,manning2018parity_hack} mentioned in the Introduction, that thus is prevented by our approach.
The second ingredient is adding to the implementation contract \code{X} a method declaration 
\begin{lstlisting}[style=tinysol, frame=none, numbers=none]
func init(x) { this.owner := x }
\end{lstlisting}
which modifies a local field in \code{X} called \code{owner}.
However, since delegate calls are executed in the context of the \emph{caller}, this allows an attacker to call \code{Proxy.init(A')} with his own address \code{A'} as argument, which then results in the \code{owner} field of \code{Proxy} being updated instead.
Thus, the attacker can take over ownership of the \code{Proxy} contract and drain it of all funds.

\subsection{Variadic arguments}
Another issue pertains to the magic variable \code{args}, which we have hitherto avoided using.
This variable acts as a variadic argument list, i.e.\@ an array of unknown length with values of mixed types.
This can easily cause configurations to be stuck, due to a wrong number of arguments or mismatched types in method calls, if a method call is issued with \code{args} in place of an ordinary argument list.
However, \code{args} can also be encapsulated in a safe manner, although this will require some extensions to the type language, which we now briefly outline.

\begin{figure}
\begin{lstlisting}[style = tinysol]
func fallback() {
  if len(args) = 1 then 
    if typeof(args[0]) = (int,@$s_\top$@) then call X.@$f_1$@(args)$value 
    else if @$\ldots$@
    @\lstvdots@
  else skip
}
\end{lstlisting}
\vspace*{-.5cm}
\caption{An example with safe encapsulation of variadic arguments.}
\label{fig:example_fallback_variadic}
\end{figure}

As a first step, suppose we admit operators such as \code{len}, \code{typeof} and array indexing into the language of expressions, which can be used to extract information from \code{args}  at runtime.
Consider now the code in Figure~\ref{fig:example_fallback_variadic}.
By using these constructs, we can meticulously check the number of arguments present in \code{args} and their actual types, and then pass control to various handlers accordingly. 
Again, we ensure safety by reducing the infinite number of possible values of \code{args} to yield only a finite number of possible continuations.
One problem here is the operator \code{typeof}: we have hitherto assumed that all expression operators \code{op} can be given a type $\VEC{B} \to B$, but the \code{typeof} operator must be able to operate on a value of \emph{any} type.
Hence, it is polymorphic, and the language of types for operators must therefore be extended with polymorphic types; we must furthermore extend the language of base types $B$ with a type for type objects, e.g.\@ $\TTYPE$.
Thus, the type of \code{typeof} must be $\forall \alpha \SUCHTHAT \alpha \to \TTYPE$, where $\alpha$ is a type variable.
These are technical details that are, in a sense, orthogonal to the issue of semantic typing, so we shall not pursue them further here.
However, extending the type language with these features would be a relevant avenue for future work.

\begin{figure}
\begin{minipage}[t]{.45\textwidth}
\begin{lstlisting}[style=tinysol]
contract Proxy : @$I_s^P$@ {
  field owner := A;
  field impl := X;
  @$\cdots$@
  func @$f_1$@(@$\VEC{x}$@) { 
    dcall this.impl.@$f_1$@(@$\VEC{x}$@) 
  }
  @\lstvdots@
  func @$f_h$@(@$\VEC{x}$@) { 
    dcall this.impl.@$f_h$@(@$\VEC{x}$@) 
  }
  func fallback() {
    call this.impl.id(args)$0
  }
}
\end{lstlisting}
\end{minipage}\hfill
\begin{minipage}[t]{.48\textwidth}
\begin{lstlisting}[style=tinysol]
contract X : @$I_s^X$@ {
  @$\cdots$@
  func @$f_1$@(@$\VEC{x}$@) { @$\ldots$@ }
  @\lstvdots@
  func @$f_h$@(@$\VEC{x}$@) { @$\ldots$@ }
  func fallback() { 
    if len(args) = 1 then 
      if typeof(args[0]) = (int,@$s_\top$@) 
      then call Y.@$g_1$@(args)$0 
      else if @$\ldots$@
      @\lstvdots@
    else skip
  }
}
\end{lstlisting}
\end{minipage}
\vspace*{-.5cm}
\caption{A proxy contract using both explicit forwarding and fallback.}
\label{fig:example_fallback_combined}
\end{figure}

Finally, combining the preceding examples, we can show how to create a variant of the \code{Proxy} contract that allows both upgrading and extending the functionality, but avoids the vulnerability of the PMW attack.
Consider the code in Figure~\ref{fig:example_fallback_combined}: we now let the \code{Proxy} declare methods $f_1, \ldots, f_h$ and then \emph{explicitly} forward each call to the corresponding method in the implementation \code{X}.
This avoids the problem of allowing an unintended call to a method like \code{init}, as in the PMW attack, whilst still permitting the method bodies to mutate the state of \code{Proxy}.
All other method calls are then generically forwarded using the fallback function, but using the direct call style.
These calls thus cannot modify the state of \code{Proxy}, but they can still extend its functionality.

Because of the type restriction on the rule \nameref{stm_rt_dcall}, $I^P$ must be a subtype of $I^X$, hence \code{X} cannot declare any methods or fields that are not also declared in \code{Proxy}.
All generically forwarded calls, which necessarily cannot be to any of $f_1, \ldots, f_h$, will therefore invoke the fallback function of \code{X}.
It can then use the aforementioned \code{len} and \code{typeof} operators to ensure that the argument list has the correct number of elements and types, as in the preceding example, before forwarding these calls to another contract \code{Y} (not shown) containing the extension methods $g_1, g_2, \ldots$.
This extra level of indirection allows both the implementation in \code{X} and the extensions in \code{Y} to be upgraded.

Note that using explicit forwarding with delegate calls does not directly prevent the programmer from \emph{explicitly} forwarding a call that can modify the \code{owner} or \code{impl} fields.
For example, one of the forwarded $f_i$ methods could indeed be \code{update(x)}, which allows the implementation pointer to be updated.
The safety of the example would then depend on the implementation of this method.
Using integrity security settings $T$ (Trusted) and $U$ (Untrusted), with $T \ordleq U$, it would be natural for the \code{Proxy} interface $I^P$ to declare \code{impl} as $\TVAR{I^X, T}$.
However, if \code{X.update} were declared simply as
\begin{lstlisting}[style=tinysol, frame=none, numbers=none]
  func update(x) { this.impl := x } 
\end{lstlisting}
\noindent thus permitting an unrestricted information flow into \code{impl}, the appropriate security setting in the implementation interface $I^X$ would be to declare
\begin{lstlisting}[style=tinysol, frame=none, numbers=none]
  impl : @$\TVAR{I^X,U}$@        update : @$\TSPROC{(I^X,U)}{U}$@
\end{lstlisting}
\noindent In that case, forwarding \code{Proxy.update} to \code{X.update} with a delegate call would lead to an information flow from an Untrusted variable to a Trusted field, which is disallowed by the type settings and would lead to a stuck configuration in the type semantics (cf.\@ rule \nameref{stm_rt_assf}).
Hence, this cannot be safe, and it would be rejected by \nameref{stm_st_assf}.

\section{Up-to techniques for semantic typing}\label{sec:upto_techniques}
To manually prove type safety for a stack $Q$ w.r.t.\@ some type environment $\Delta$ and security level $s$, we must exhibit a typing interpretation $\STYPI$ containing the triplet $(Q, \Delta, s)$.
A major disadvantage of this approach is that it requires us to check the conditions of a typing interpretation for all the states reachable from $(Q, \Delta, s)$ in the type-lifted transition system (Definition~\ref{def:type_lifted_transitions}).
This poses practical problems in light of our proposal to store on.chain typing interpretations beside contract code, to witness type safety of fallback functions.
First, the contract creator, wishing to demonstrate type safety of his untypable code, must be able to \emph{build} a typing interpretation containing the code.
Second, a user, wishing to invoke a piece of untypable code and ensure (when writing a new contract) that the invocation is type safe, must be able to \emph{verify} that a given set of triplets indeed is a typing interpretation (this should preferably be doable by an automatic decision procedure, which can be called by the type checker itself, rather than requiring manual intervention). 
Not only may these tasks be laborious to do manually; they may even be impossible, since typing interpretations can contain infinitely many triplets in limit cases.
Furthermore, to be practically useful, it must be decidable whether a given set $\CANDSTYPI$ in fact \emph{is} a typing interpretation. 

The problem is similar to that of proving bisimilarity~\cite{Mi89} between two processes, by exhibiting a bisimulation relating them:
in that case, we also need to check that the bisimulation conditions hold for every pair in the candidate relation.
This problem is particularly well-known in the field of process calculi, such as the \PI-calculus \cite{PICALC, milner_walker_parrow1992picalc, parrow2001introduction}, where program equivalences such as bisimilarity are studied extensively; see e.g.\@ \cite{sangiorgi2011introduction_bisimulation_coinduction,sangiorgi2003pi}.
To alleviate this burden of proof, Pous and Sangiorgi \cite{pous_sangiorgi2011uptotechniques} present a number of \emph{up-to techniques}~\cite{Mi89,SangiorgiM92} for proving bisimilarity.
The idea is that it is enough to prove that a relation $\SETNAME{R}$ relating two processes is \emph{contained in} a bisimulation. 
Pous and Sangiorgi present a general theory of how to prove soundness of an up-to technique,
and this relies on greatest fixed-points on complete lattices; it can therefore also be adapted to work with other coinductively defined objects, like our typing interpretations.
To the best of our knowledge, this is the first use of up-to techniques in a concrete, non-relational setting.

\subsection{Another motivating example}\label{sec:upto_example}
\begin{figure}
\begin{lstlisting}[style = tinysol]
contract Y : @$I^Y_s$@ {
  @\lstvdots@
  func fallback() {
    if id = @$f_1$@ @$\lor$@ @$\ldots$@ @$\lor$@ id = @$f_h$@ then call X.id()$value else call X.@$f_0$@()$value
  }
}
\end{lstlisting}
\vspace*{-.5cm}
\caption{An example where \code{id} is used as the method name.}
\label{fig:example_fallback_upto}
\end{figure}

Consider the body of the fallback function in Figure~\ref{fig:example_fallback_upto}.
Calling methods by using the magic variable \code{id} makes the body untypable by the semantic type rules, as we have previously argued when discussing ($\dagger$).
However, the body consists of just a single \code{if-else} statement, with a method call in each branch.
Thus, this is a variant of the pattern presented in Figure~\ref{fig:example_fallback_proxy};
similarly to that example, also here, after just two transition steps, the body on the stack reduces to the body of one of the methods $f_1, \ldots, f_h$ or of the default handler $f_0$.
In terms of the type-lifted transition system, this body will yield a transition sequence of the form
\begin{align*}
\!\!\!\!\!  \Sigma; \Gamma; \ENV{T} & \vDash (\code{if\,\ldots\,};\bot, \Delta, s_1) \ttrans (\code{call\,\ldots\,};s_1;\bot, \Delta, s_2) \ttrans (S_i;(\ENV{V}, \Delta);s_2;s_1;\bot, \Delta_i, s_i)
\end{align*}

\noindent where $S_i$ is the body of the called handler method $f_i$, $\Delta_i$ is the constructed type environment, and $s_i \ordgeq s_2 \ordgeq s_1$. 
If we further assume that all the handler methods $f_i$ in fact are semantically typable, then there exists a typing interpretation $\STYPI^i$ containing $(S_i;\bot, \Delta_i, s_i)$.
From this and the above, we can obtain that $\Sigma; \Gamma; \Delta_i; \ENV{T} \vDash S_i;(\ENV{V}, \Delta);s_2;s_1;\bot : \TCMD{s_i}$; thus, there exists a typing interpretation $\STYPI^{{i'}}$ containing $(S_i;(\ENV{V}, \Delta);s_2;s_1;\bot, \Delta_i, s_i)$.

A typing interpretation for the body of the fallback function must contain the three triplets mentioned above for each of the possible other method calls $f_i$ and the union of all the typing interpretations $\STYPI^{{i'}}$.
The union of a collection of typing interpretations is itself a typing interpretation, so we can let $\STYPI' \DEFSYM \STYPI^{{0'}} \UNION \ldots \UNION \STYPI^{{h'}}$ and thus obtain a typing interpretation containing all the \emph{continuations} after the method call is issued.
However, the problem is that the set
\begin{equation*}
  \UPTOSTYPI \DEFSYM  \bigcup_{i \in \{0, \ldots, h\}} \SET{                   %
      (\code{if\,\ldots\,};\bot, \Delta, s_1), 
      (\code{call\,\ldots\,};s_1;\bot, \Delta, s_2),
      (S_i;(\ENV{V}, \Delta);s_2;s_1;\bot, \Delta_i, s_i)} , 
\end{equation*}

\noindent consisting of the union of the three triplets for each of the $f_i$ handler methods, is \emph{not} on its own a typing interpretation.
Hence, we cannot just discard $\STYPI'$ and only take $\UPTOSTYPI$ as our witness of type safety of the fallback function.
We need to take the whole union of $\UPTOSTYPI$ and $\STYPI'$, %
\changed{
which would require us to actually \emph{find} the  typing interpretations $\STYPI^{{0'}}, \ldots, \STYPI^{{h'}}$.
Furthermore, any user wishing to ascertain type safety of a call invoking this fallback function would also need to \emph{check} this whole union.
}
This is impractical, as we can clearly see that, if $(Q', \Delta', s') \in \UPTOSTYPI$ and $\Sigma; \Gamma; \ENV{T} \vDash (Q', \Delta', s') \ttrans (Q'', \Delta'', s'')$, then $(Q'', \Delta'', s'') \in \UPTOSTYPI \UNION \STYPI' $. 
Thus, although $\UPTOSTYPI$ is not itself a typing interpretation, it \emph{is} a typing interpretation \emph{up-to union} with $\STYPI'$.
Hence, \changed{it would be practical}, if we could use it in place of the full typing interpretation and let the remainder of the stack (viz., $S_i;(\ENV{V}, \Delta);s_2;s_1;\bot$) be typed by the ordinary semantic type rules.
Using $\UPTOSTYPI$ in place of $\UPTOSTYPI \UNION \STYPI'$ will drastically reduce the number of triplets that must be checked; this makes it easier for the contract creator to build the set $\UPTOSTYPI$ and for the user to verify it.

\subsection{Typing interpretations up-to union}
Following \cite{pous_sangiorgi2011uptotechniques}, we shall formulate up-to techniques for typing interpretations in terms of progressions, which we define as follows:

\begin{definition}[Progression]\label{def:progression}
Given 
environments $\Sigma$, $\Gamma$ and $\ENV{T}$, and two sets of stack type triplets $\UPTOSTYPI, \CANDSTYPI \subseteq \SETTRIPLETS$, we say that $\UPTOSTYPI$ \emph{progresses} to $\CANDSTYPI$, written $\UPTOSTYPI \progress \CANDSTYPI$, if $(Q, \Delta, s) \in \UPTOSTYPI$ implies
\begin{enumerate}
  \item\label{case:prog1} $Q \in \TSTACKS$, or
  \item\label{case:prog2} $\exists P_1'$ such that $\Sigma; \Gamma; \ENV{T} \vDash (Q, \Delta, s) \ttrans P_1'$ and 
        $\forall P_2'$ such that $\Sigma; \Gamma; \ENV{T} \vDash (Q, \Delta, s) \ttrans P_2'$ it holds that $P_2' \in \CANDSTYPI$. \qedhere
\end{enumerate}
\end{definition}
Note the similarity to the definition of typing interpretations (Definition~\ref{def:typing_interpretation_stm}).
When $\UPTOSTYPI$ and $\CANDSTYPI$ coincide, we have the ordinary definition of a typing interpretation; thus $\UPTOSTYPI \progress \UPTOSTYPI$ if and only if $\UPTOSTYPI$ is a typing interpretation.
The purpose of an up-to technique is to allow us to infer $\UPTOSTYPI \subseteq \STYPING$ from a progression $\UPTOSTYPI \progress \UPTOSTYPI \UNION \SETNAME{R}$, where $\SETNAME{R}$ is a set of redundant triplets.
If this is possible, then the up-to technique is \emph{valid} (sound); this is the third cornerstone result of our paper (for a proof, see Appendix~\ref{app:upto}).

\begin{definition}[Typing interpretation up-to union]\label{def:upto_union}
Let $\STYPI$ be a typing interpretation w.r.t.\@ environments $\Sigma$, $\Gamma$ and $\ENV{T}$.
A set of stack type triplets $\UPTOSTYPI$ is a \emph{typing interpretation up-to union} with $\STYPI$ if $\UPTOSTYPI \progress \UPTOSTYPI \UNION \STYPI$.
\end{definition}

\begin{theorem}[Validity of up-to union]\label{lemma:validity_upto_union}
Let $\STYPI$ be a typing interpretation.
If $\UPTOSTYPI$ is a typing interpretation up-to union with $\STYPI$ then $\UPTOSTYPI \UNION \STYPI$ is a typing interpretation. 
\end{theorem}

Going back to the example in Section~\ref{sec:upto_example}, this means that we indeed can use the set $\UPTOSTYPI$ in place of the full typing interpretation $\UPTOSTYPI \UNION \STYPI'$ to witness type safety of the body of the fallback function.
Instead of storing $\UPTOSTYPI \UNION \STYPI'$ on the blockchain, it will therefore suffice to store just $\UPTOSTYPI$, together with an indication that $\UPTOSTYPI$ is a typing interpretation up-to union with the typing interpretation corresponding to the semantic type judgments 
$\Sigma; \Gamma; \Delta_i; \ENV{T} \vDash S_i;(\ENV{V}, \Delta);s_2;s_1;\bot : \TCMD{s_i}$ which can be concluded by the ordinary semantic type rules.
This greatly reduces the number of triplets that must be stored on the blockchain and checked by (the type checker of) the user.
\changed{Notice that this is not only a matter of size: if we were to give an explicit typing interpretation for the body of the fallback function---i.e. if we were to actually \emph{find} $\STYPI'$---, then we would also have to find one  typing interpretation for each of the $f_i$ method calls. 
In contrast, the up-to union technique allows us to simply assume that the latter ones exist, without the need of finding them.
This allows for greater flexibility in how manual proofs can be mixed with use of the type rules.}

\section{Concluding remarks}\label{sec:conclusion}\raggedbottom
The difference between safe and unsafe usages of the fallback function is not a purely syntactic property; this makes the conventional syntactic approach to type soundness unsuitable to handle this construct.
Hence, in this paper, we have followed the \emph{semantic} approach to type soundness, by defining a semantics of types and type judgments in terms of a typed operational semantics, which gives us a definition of safety that does not depend on syntactic type rules.
To 
coin a slogan, we can say that semantic typing is a \emph{type-safety-first} approach.
The fallback function is a particularly problematic programming construct, both in terms of the security breaches to which it has given rise and the challenges it poses for a type-theoretic approach to ensuring program safety.
We have 
proposed a solution to the problem of how to recover type safety of programs with fallback functions, without simply rendering the construct useless by preventing its execution altogether.
However, these developments should only be seen as an example of the approach we advocate and
may be equally relevant for other notions of type safety (for another example, see Appendix~\ref{app:CI}), other programming constructs, and even other smart-contract languages besides \TINYSOL/Solidity.

We 
remark that our proposal is not the only possible way to formulate a semantics of types.
For example, in \cite{caires2007}, Caires gives a \emph{logical} characterisation of various notions of type safety in the \PI-calculus by expressing types as formulae in the language of spatial logics \cite{caires2004behavioural_spatial_logic,caires_cardelli2003spatial_logic1, caires_cardelli2003spatial_logic2}. 
This allows him to directly define the sets of safe  programs (processes) w.r.t.\@ a given type formula, corresponding to our typing interpretations, as simply those processes that satisfy the formula.
This naturally poses the question of what kind of logic might be appropriate for describing secure information flows.
Proposals for such logics are presented in, e.g.,~\cite{Balliu14,BeringerH07,FruminKB21,GregersenBTB21,SolovievBG24}. 
Studying a logic suitable for describing 
the Volpano-Irvine-Smith type system~\cite{volpano2000secure_flow_typesystem} and developing a corresponding logical characterisation of the present notion of safety are interesting areas for future work.


The adaptation of up-to techniques from Section~\ref{sec:upto_techniques} is interesting, as it also opens new research avenues.
Indeed, their use is not limited to our present \TINYSOL{} setting: such techniques are generally relevant for the semantic approach to type soundness based on typing interpretations.
Besides up-to union, Pous and Sangiorgi \cite{pous_sangiorgi2011uptotechniques} also define a number of other up-to techniques  
which also should be transferable to typing interpretations, since they pertain to the coinduction proof method. 
We can consider, for example, 
\emph{up-to environment extension}, which could be relevant if runtime contract creation were added to \TINYSOL.
Indeed, in our present setting, $\Sigma$, $\Gamma$ and $\ENV{T}$ are assumed to be static, but in the future we can  consider a more realistic setting where new contracts can be dynamically added.
It would therefore be important to ensure that the addition of a new contract to the blockchain cannot invalidate a typing interpretation that witnesses the safety of an existing contract.

One important issue that deserves more investigation is a formal study of the size of the typing interpretations. Indeed, these are checked only when a new contract is developed (hence, this does {\em not} consume any gas because it involves no runtime computation). However, the proofs must be permanently stored on the blockchain. This has a cost, due to the increased contract size, which depends on the size of the proof and its representation on-chain. As a promising benchmark, the size of the typing interpretation (up-to union) for the example in Section~\ref{sec:upto_example} is linear in the number of functions mentioned in the fallback code, so it is linear in the size of the program.

Last, but by no means least, it would be interesting to develop algorithms and tools for building and certifying typing interpretations, and to apply them to case studies from the domain of smart contracts that involve the use of the fallback function.

\bibliography{literature}

\newpage
\appendix

\section{Untyped Operational Semantics}\label{app:untyped_operational_semantics}
The rules for turning declarations into the environments $\ENV{J}$ (for $J \in \{F,M,S,T\})$ are straightforward; they are given in Figure~\ref{fig:semantics_declarations}.

\begin{figure}[h]
\begin{tabular}{cc}
\begin{minipage}{0.29\textwidth}
\center
\begin{semantics}
  \RULE[decf$_1$]
    { }
    { \CONF{\epsilon, \ENV{F}} \trans \ENV{F} }
  \RULE[decm$_1$]
    { }
    { \CONF{\epsilon, \ENV{M}} \trans \ENV{M} }
  \RULE[decc$_1$][ts_dec_c1]
    { }
    { \CONF{\epsilon, \ENV{ST}} \trans \ENV{ST} }
  \RULE[decc$_2$][ts_dec_c2]
    {
        \CONF{DF, \ENV{F}^{\EMPTYSET}} \trans \ENV{F} \AND
        \CONF{DM, \ENV{M}^{\EMPTYSET}} \trans \ENV{M} \AND
        \CONF{DC, \ENV{ST}} \trans \ENV{TS}'  
    }
    { \CONF{\code{contract $X$ \{ $DF$ $DM$ \} $DC$}, \ENV{ST}} \trans  (\ENV{T}'\ , X: \ENV{M}) \ ,\,  (\ENV{S}'\ , X : \ENV{F}) }
\end{semantics}
\end{minipage}
~
&
~
\begin{minipage}{0.66\textwidth}
\center
\begin{semantics}
  \RULE[decf$_2$]
    { \CONF{DF, \ENV{F}} \trans \ENV{F}' }
    { \CONF{\code{field $p$ := $v$;$DF$}, \ENV{F}} \trans \ENV{F}'\ , p: v}
  \RULE[decm$_2$]
    { \CONF{DM, \ENV{M}} \trans \ENV{M}' }
    { \CONF{\code{$f$($\VEC{x}$) \{ $S$ \} $DM$}, \ENV{M}} \trans \ENV{M}'\ , f : (\VEC{x}, S) }
\vspace*{2.2cm}
\end{semantics}
\end{minipage}
\end{tabular}
\caption{Semantics of declarations.}
\label{fig:semantics_declarations}
\end{figure}

The semantics of expression configurations has the form $\CONF{e, \ENV{SV}} \trans v$ and is standard; its defining rules are in Figure~\ref{fig:semantics_expressions}.
Here, we assume some semantics $\trans_{\op}$ to evaluate the result of applying an operation $\op$ to a sequence of values $\VEC{v}$; hence, such operations are assumed to be total and let $\op(\VEC{v}) \trans_{\op} v$, for some basic value $v$.

\begin{figure}[h]\centering
\begin{minipage}{.49\textwidth}
\begin{semantics}
  \RULE[exp-val][exp_sem_val]
    { }
    { \CONF{v, \ENV{SV}} \trans v }
  \RULE[exp-op][exp_sem_op]
    { \CONF{\VEC{e}, \ENV{SV}} \trans \VEC{v} \AND \op(\VEC{v}) \trans_{\op} v}
    { \CONF{\op(\VEC{e}), \ENV{SV}} \trans v }
\end{semantics}
\end{minipage}  
\begin{minipage}{.49\textwidth}
\begin{semantics}
  \RULE[exp-var][exp_sem_var]({ \ENV{V}(x) = v })
    { }
    { \CONF{x, \ENV{SV}} \trans v }
  \RULE[exp-field][exp_sem_field]({ \ENV{S}(X)(p) = v })
    { \CONF{e, \ENV{SV}} \trans X }
    { \CONF{e.p, \ENV{SV}} \trans v }
\end{semantics}
\end{minipage}
\caption{Semantics of expressions; for the sake of simplicity, we here assume that $x$ also ranges over the `magic variable' names \code{this}, \code{sender}, \code{value}, \code{id} and \code{args}.}
\label{fig:semantics_expressions}
\end{figure}

\begin{figure}[t]\centering
\begin{semantics}
  \RULE[skip][untyped_skip]
    { }
    { \CONF{\code{skip} ; Q, \ENV{TSV}} \trans \CONF{Q, \ENV{TSV}} }
  \RULE[if][untyped_if]
    { \CONF{e, \ENV{SV}} \trans b \in \BOOLEANS }
    { \CONF{\code{if $e$ then $S_{\TRUE}$ else $S_{\FALSE}$} ; Q, \ENV{TSV}} \trans \CONF{S_b ; Q, \ENV{TSV}} }
  \RULE[while$_\TRUE$][untyped_whileT]
    { \CONF{e, \ENV{SV}} \trans \TRUE }
    { \CONF{\code{while $e$ do $S$} ; Q, \ENV{TSV}} \trans \CONF{S ; \code{while $e$ do $S$} ; Q, \ENV{TSV}} } 
  \RULE[while$_\FALSE$][untyped_whileF]
    { \CONF{e, \ENV{SV}} \trans \FALSE }
    { \CONF{\code{while $e$ do $S$} ; Q, \ENV{TSV}} \trans \CONF{Q, \ENV{TSV}} } 
  \RULE[decv][untyped_decv]
    { \CONF{e, \ENV{SV}} \trans v \qquad x \notin \DOM{\ENV{V}}}
    {\CONF{\code{$\TVAR{B_s}$ $x$ := $e$ in $S$} ; Q, \ENV{TSV}} 
      \trans \CONF{S ; \DEL{x} ; Q, \ENV{TS}, (\ENV{V}\ ,x : (v, B_s) )} 
    }
  \RULE[assv][untyped_assv]
    { \CONF{e, \ENV{SV}} \trans v \qquad  x \in \DOM{\ENV{V}}}
    { \CONF{\code{$x$ := $e$} ; Q, \ENV{TSV}} \trans \CONF{Q, \ENV{TS}, \ENV{V}\REBIND{x}{v}} }
  \RULE[assf][untyped_assf]
    { \CONF{e, \ENV{SV}} \trans v }
    {\CONF{\code{this.$p$ := $e$} ; Q, \ENV{TSV}} \trans \CONF{Q, \ENV{T}, \ENV{S}\REBIND{X}{\ENV{F}\REBIND{p}{v}}, \ENV{V}} 
    }
    \WHERE{ X }      { = \ENV{V}(\code{this}) \qquad\quad \ENV{F} = \ENV{S}(X) \qquad\quad p \in \DOM{\ENV{F}} }
  \RULE[delv][untyped_delv]
    { }
    { \CONF{\DEL{x} ; Q, \ENV{TS}, ( \ENV{V}\ , x:(v, B_s))} \trans \CONF{Q, \ENV{TSV}} }
  \RULE[return][untyped_return]
    { }
    { \CONF{\ENV{V}' ; Q, \ENV{TSV}} \trans \CONF{Q, \ENV{TS}, \ENV{V}'} }
\end{semantics}
\caption{Small-step semantics of stacks (to be completed with the rules depicted in Figure~\ref{fig:semantics_stacks2}).}
\label{fig:semantics_stacks1}
\end{figure}

The semantics for stacks provides a small-step semantics for commands; most of the rules are standard and reported in Figure~\ref{fig:semantics_stacks1}. 
Note that, as we are using a typed syntax, with types appearing in declarations of local variables, we need a place to store these type assignments after the declaration has been performed, and the variable and its current value has been appended to the current variable environment $\ENV{V}$.
Thus, we record this type information directly in $\ENV{V}$ along with the value, as is done in \cite{AGLM/2024/reocas/outofgas}.
This has no practical significance for the semantics, but is necessary for the syntactic type system (given in Appendix~\ref{app:syntactic_type_system}), since variable environments can be pushed onto the stack in method calls, and recovered from the stack after a method returns (see rule \nameref{stack_t_ret} in Figure~\ref{fig:type_rules_stacks}).

The rules for the three forms of method calls in Figure~\ref{fig:semantics_stacks2} also deserve a few comments:
For rule \nameref{stm_s_call}, we first evaluate the address and the parameters (viz., $e_1$, $\VEC{e}$ and $e_2$), relatively to the current execution environment $\ENV{SV}$; we use the obtained address $Y$ of the callee to retrieve the field environment $\ENV{F}^Y$ for this contract and, through the method table, to extract the list of formal parameters $\VEC{x}$ and the body of the method $S$. 
Also, we have to check that the number of actual parameters is the same as the number of formal parameters. 
We then evaluate the currency amount passed through the call\footnote{
Notice that here we should also check that the amount of currency transferred in the call (viz., $z$ in such rules) is at most the value currently owned by the caller in its \code{balance} field; if this is not the case, the method invocation fails and an exception is raised.  To model this, we would need to complicate the semantics with an issue (namely, exception handling) that is orthogonal to the aims of this paper. Hence, we prefer to keep the presentation simple and let method invocations succeed, irrespectively of the amount of currency associated with them.
} and
update the state environment by subtracting $z$ from the balance of $X$ and adding $z$ to the balance of $Y$, in their respective field environments; this yields a new state $\ENV{S}''$.
We create the new execution environment $\ENV{V}''$ by creating new bindings for the special variables \code{this}, \code{sender} and \code{value}, and by binding the formal parameters $\VEC{x}$ to the values of the actual parameters $\VEC{v}$. 
Finally, we execute the statement $S$ in this new environment. 
This yields the new state $\ENV{S}'$, and also an updated variable environment $\ENV{V}'$, since $S$ may have modified the bindings in $\ENV{V}''$. 
However, these bindings are local to the method and therefore we discard them once the call finishes.
So, the result of this transition is the updated state ($\ENV{S}'$) and the original variable environment of the caller ($\ENV{V}$).


Rule \nameref{stm_s_fcall} handles the fallback function. Essentially, this is a normal call, with the exception that it is activated when the invoked method is not provided by the callee.
Moreover, the magic variables \code{id} and \code{args} are bound to the invoked method name and to the actual arguments passed, respectively,.

Rule \nameref{stm_s_dcall} accounts for delegate calls, which differ from normal calls in two aspects.
Since the method body $S$ is executed in the context of the caller, the magic variables \code{this}, \code{sender} and \code{value} are \emph{not} rebound and funds are not transferred to the callee with this type of call. 
This is consistent with the idea that the code is executed in the context of the caller.

\begin{figure}[t]\centering
\begin{semantics}
  \RULE[call][stm_s_call]
    { 
      \CONF{e_1, \ENV{SV}} \trans Y \AND
      \CONF{e_2, \ENV{SV}} \trans z \AND
      \CONF{\VEC{e}, \ENV{SV}} \trans \VEC{v}
    }
    { \CONF{ \CALL{e_1}{m}{\VEC{e}}{e_2} ; Q, \ENV{TSV}} \trans \CONF{S ; \ENV{V} ; Q, \ENV{T}, \ENV{SV}'} }
    \WHERE{ X }           { = \ENV{V}(\code{this}) \hspace*{2cm}  \ENV{F}^X  = \ENV{S}(X) \hspace*{1.5cm} \ENV{F}^Y  = \ENV{S}(Y) }
    \WHERE{ f }           { = 
                          \begin{cases} 
                            \ENV{V}(m) & \text{if $m = \code{id}$} \\ 
                            m         & \text{otherwise} 
                          \end{cases}
    \hspace*{.35cm} (\VEC{x}, S)  = \ENV{T}(Y)(f) 
    \hspace*{1.1cm} |\VEC{x}|  = |\VEC{v}| = h }
    \WHERE{ \ENV{V}' }{ = (\code{this}, Y), (\code{sender}, X) , (\code{value}, z) , (x_1, v_1) , \ldots , (x_h, v_h) }
    \WHERE{ \ENV{S}'}{ = \ENV{S}\REBIND{X}{\ENV{F}^X [\code{balance -= } z]}\REBIND{Y}{\ENV{F}^Y [\code{balance += } z]} }
  \RULE[fcall][stm_s_fcall]
    { 
      \CONF{e_1, \ENV{SV}} \trans Y \AND
      \CONF{e_2, \ENV{SV}} \trans z \AND
      \CONF{\VEC{e}, \ENV{SV}} \trans \VEC{v}
    }
    { \CONF{\CALL{e_1}{m}{\VEC{e}}{e_2}; Q, \ENV{TSV}} \trans \CONF{S ; \ENV{V} ; Q, \ENV{T}, \ENV{SV}'} }
    \WHERE{ X }           { = \ENV{V}(\code{this}) \hspace*{2cm}  \ENV{F}^X  = \ENV{S}(X) \hspace*{1.5cm} \ENV{F}^Y  = \ENV{S}(Y) }
    \WHERE{ f }            { = 
                           \begin{cases} 
                             \ENV{V}(m) & \text{if $m = \code{id}$} \\ 
                             m          & \text{otherwise} 
                           \end{cases} 
    \hspace*{.35cm} f \notin \DOM{\ENV{T}(Y)} 
    \hspace*{1cm} (\epsilon, S) = \ENV{T}(Y)(\code{fallback}) }
        \WHERE{ \ENV{V}' }{ = (\code{this}, Y), (\code{sender}, X), (\code{value}, z), (\code{id}, f), (\code{args}, \VEC{v}) }   
    \WHERE{ \ENV{S}' }     { = \ENV{S}\REBIND{X}{\ENV{F}^X [\code{balance -= } z]}\REBIND{Y}{\ENV{F}^Y [\code{balance += } z]} }
  \RULE[dcall][stm_s_dcall]
    { 
      \CONF{e, \ENV{SV}} \trans Y \AND
      \CONF{\VEC{e}, \ENV{SV}} \trans \VEC{v}      
    }
    { \CONF{\DCALL{e}{f}{\VEC{e}} ; Q, \ENV{TSV}} \trans \CONF{S ; \ENV{V} ; Q, \ENV{TS}, \ENV{V}'}  }
    \WHERE{ f }           { = 
                          \begin{cases} 
                            \ENV{V}(m) & \text{if $m = \code{id}$} \\ 
                            m         & \text{otherwise} 
                          \end{cases}
    \hspace*{.35cm} (\VEC{x}, S)  = \ENV{T}(Y)(f) 
    \hspace*{1.1cm} |\VEC{x}|  = |\VEC{v}| = h }
    \WHERE{\ENV{V}'}      { = (\code{this}, \ENV{V}(\code{this})), (\code{sender}, \ENV{V}(\code{sender})), (\code{value},\ENV{V}(\code{value})), (x_1, v_1), \ldots, (x_h, v_h)}
\end{semantics}
\caption{Small-step semantics of stacks: method calls}
\label{fig:semantics_stacks2}
\end{figure}

\clearpage 
\section{Subtyping and type environment consistency rules}\label{app:subtyping_consistency}

$\Sigma$ gives a \emph{static} representation of the subtype-supertype relationship for interfaces.
We shall therefore need a way to decide whether a given $\Sigma$ actually is \emph{consistent} w.r.t.\@ this relationship.
Thus, we shall give a set of rules for deciding whether a given interface $I_1$, declaring that it inherits from another interface $I_2$, indeed also is a \emph{subtype} of $I_2$.
The rules are given in Figure~\ref{fig:consistency_rules_sigma}.

\begin{figure}[h]\centering
\begin{semantics}
  \RULE[$\Sigma$-top][cons_sigma_top]
    { }
    { \Sigma; \Gamma \vdash (\ITOP, \ITOP) }
  \RULE[$\Sigma$-rec][cons_sigma_rec]
    { \Sigma; \Gamma \vdash \Sigma' \AND \forall n \in \DOM{\Gamma_2} \SUCHTHAT \Sigma \vdash \Gamma_1(n) \SUBS \Gamma_2(n) }
    { \Sigma; \Gamma \vdash (I_1, I_2), \Sigma' }
    \WHERE{ I_1 }{ \neq I_2
    \qquad \not\exists I' \SUCHTHAT \Sigma'(I') = I_1 }
    \WHERE{ \Gamma(I_1) }   { = \Gamma_1 
    \qquad \Gamma(I_2) = \Gamma_2 
    \qquad\qquad \DOM{\Gamma_2} \subseteq \DOM{\Gamma_1} }
\end{semantics}
\caption{Consistency rules for $\Sigma$.}
\label{fig:consistency_rules_sigma}
\end{figure}

The judgment is of the form $\Sigma; \Gamma \vdash \Sigma$, which expresses that the inheritance tree recorded in $\Sigma$ is \emph{consistent} with the actual declarations recorded in $\Gamma$.
The rules are quite simple, since they just iterate through the inheritance environment, examining each entry in turn.
$\Sigma$ only records the inheritance hierarchy, whereas the actual structures of the interfaces are recorded in $\Gamma$.
Hence, consistency of any sub-part of $\Sigma$ must be judged relative to $\Gamma$, as well as relative to the full $\Sigma$.

\begin{itemize}
  \item The base case is the rule \nameref{cons_sigma_top}, since we assume the interface $\ITOP$ always exists and inherits from nothing besides itself.
    Furthermore, the form of the conclusion ensures that this rule can only be applied to an inheritance tree consisting of the root node itself.

  \item The recursive case is the rule \nameref{cons_sigma_rec}, which examines the entry $(I_1, I_2)$.
    In the side condition, we require that $I_1 \neq I_2$, since no interface except $\ITOP$ may inherit from itself.
    We also require that there must not exist another interface $I'$ such that $I'$ inherits from $I_1$ according to the remainder of the inheritance environment $\Sigma'$.
    This ensures that \nameref{cons_sigma_rec} is only applicable to a \emph{leaf node} of the current inheritance tree, which is then trimmed away as we recur in the premise. 
    Hence it ensures that the environment indeed describes a \emph{tree}.
Finally, the premise of \nameref{cons_sigma_rec} invokes a subtyping judgment 
\begin{equation*}
  \Sigma \vdash \Gamma_1(n) \SUBS \Gamma_2(n) 
\end{equation*}
for each member of the interface $I_2$, i.e.\@ the supertype; the side condition 
\begin{equation*}
  \DOM{\Gamma_2} \subseteq \DOM{\Gamma_1} 
\end{equation*}
ensures that all members of the supertype also exist in the subtype. 
\end{itemize}

The rules defining the subtyping judgments for interface members are \nameref{typerules_sub_field} and \nameref{typerules_sub_proc}, given in Section~\ref{sec:types}.
Both rules invoke a subtyping judgment for types of the form 
\begin{equation*}
  \Sigma \vdash B_1 \SUBS B_2
\end{equation*}

\noindent in their premise: this, finally, is the proper subtyping relation, which is defined by the rules in Figure~\ref{fig:subtyping_expressions}.
We shall also need this relation later in the presentation.
Note that we shall sometimes use the abbreviation
\begin{equation*}
  \Sigma \vdash B_1 \SUBS B_2 \SUBS B_3 \DEFSYM \Sigma \vdash B_1 \SUBS B_2 \land \Sigma \vdash B_2 \SUBS B_3 
\end{equation*}

\noindent to simplify some premises.

%
%

\begin{figure}\centering
\begin{semantics}
  \RULE[sub-refl][typerules_sub_refl]
    { }
    { \Sigma \vdash B \SUBS B }
  \RULE[sub-trans][typerules_sub_trans]
    { \Sigma \vdash B_1 \SUBS B_2 \AND \Sigma \vdash B_2 \SUBS B_3 }
    { \Sigma \vdash B_1 \SUBS B_3 }
  \RULE[sub-name][typerules_sub_name]({\Sigma(I_1) = I_2})
    { }
    { \Sigma \vdash I_1 \SUBS I_2 }
\end{semantics}
\caption{The subtyping relation for expression types $B_s$}
\label{fig:subtyping_expressions}
\end{figure}

\clearpage
\section{Syntactic type rules}\label{app:syntactic_type_system}
In our presentation, we intentionally left the syntactic category of operations $\op$ undefined, and simply assumed that they can all be evaluated using some semantics $\trans_{\op}$, such that $\op(\VEC{v}) \trans_{\op} v$.
Now, we also assume that for each operation we can determine its signature, written $\vdash \op : \VEC{B} \to B$.
To tie these two together, we need to make one further assumption; 

\begin{figure}[h]\centering
\begin{semantics}
  \RULE[t-val][exp_t_val]({
  \begin{array}{r @{~} l}
    \TYPEOF{v} & = (B', s')  \\
    s'         & \ordleq s
  \end{array}
  }) 
    { \Sigma \vdash B' \SUBS B }
    { \Sigma; \Gamma; \Delta \vdash v : B_s }
  \RULE[t-var][exp_t_var]({ 
  \begin{array}{r @{~} l}
    \Delta(x) & = \TVAR{B', s'} \\
    s'        & \ordleq s
  \end{array}
  })
    { \Sigma \vdash B' \SUBS B }
    { \Sigma; \Gamma; \Delta \vdash x : B_s }
  \RULE[t-field][exp_t_field]({
    \begin{array}{r @{~} l}
      \Gamma(I)(p) & = \TVAR{B', s'} \\
      s'           & \ordleq s
    \end{array}
  })
    { \Sigma; \Gamma; \Delta \vdash e : I_s \AND \Sigma \vdash B' \SUBS B }
    { \Sigma; \Gamma; \Delta \vdash e.p : B_s }
  \RULE[t-op][exp_t_op]
    { \vdash \op : \VEC{B} \to B \AND \Sigma; \Gamma; \Delta \vdash \VEC{e} : \VEC{B}_{s, \ldots, s} }
    { \Sigma; \Gamma; \Delta \vdash \op(\VEC{e}) : B_s }
\end{semantics}
\caption{Type rules for expressions.}
\label{fig:type_rules_expr}
\end{figure}

\begin{definition}[Safety requirement for operations]\label{def:safety_operations}
For all operations $\op$ and argument list $\VEC{v}$, if
\begin{itemize}
  \item $\vdash \op : \VEC{B} \to B$, and
  \item $|\VEC{B}| = |\VEC{v}|$, and
  \item $\forall v_i \in \VEC{v} \SUCHTHAT \TYPEOF{v_i} = (B_i, s_i)$, 
\end{itemize}
then $\op(\VEC{v}) \trans_{\op} v$ and $\TYPEOF{v} = (B, \SBOT)$, assuming any addresses $X$ appearing in the argument list $\VEC{v}$ has an entry in $\Gamma$.
\end{definition}

Note that, by this definition, we allow addresses to appear as arguments to operations.
However, as we assume that an operation can never \emph{yield} an address, we also know that the resulting security level will be $\SBOT$ and that the resulting base type $B$ will be one of $\TINT$ or $\TBOOL$.

%

\begin{figure}\centering
\begin{semantics}
  \RULE[t-skip][stm_t_skip]
    { }
    { \Sigma; \Gamma; \Delta \vdash \code{skip} : \TCMD{s} }
  \RULE[t-throw][stm_t_throw]
    { }
    { \Sigma; \Gamma; \Delta \vdash \code{throw} : \TCMD{s} }
  \RULE[t-seq][stm_t_seq]
    { \Sigma; \Gamma; \Delta \vdash S_1 : \TCMD{s} \AND \Sigma; \Gamma; \Delta \vdash S_2 : \TCMD{s} }
    { \Sigma; \Gamma; \Delta \vdash S_1; S_2 : \TCMD{s} }
  \RULE[t-if][stm_t_if]( s \ordleq s' )
    { 
    \begin{array}{r @{~} l @{\qquad} r @{~} l}
                             &                         & \Sigma; \Gamma; \Delta & \vdash S_{\TRUE} : \TCMD{s'}  \\
      \Sigma; \Gamma; \Delta & \vdash e : (\TBOOL, s') & \Sigma; \Gamma; \Delta & \vdash S_{\FALSE} : \TCMD{s'} 
    \end{array}
    }
    { \Sigma; \Gamma; \Delta \vdash \code{if $e$ then $S_{\TRUE}$ else $S_{\FALSE}$} : \TCMD{s} }
  \RULE[t-while][stm_t_while](s \ordleq s')
    { \Sigma; \Gamma; \Delta \vdash e : (\TBOOL, s') \AND \Sigma; \Gamma; \Delta \vdash S : \TCMD{s'} }
    { \Sigma; \Gamma; \Delta \vdash \code{while $e$ do $S$} : \TCMD{s} }
  \RULE[t-decv][stm_t_decv]
    { 
      \Sigma; \Gamma; \Delta                   \vdash e : (B, s') \AND
      \Sigma; \Gamma; \Delta, x : \TVAR{B, s'} \vdash S : \TCMD{s} 
    }
    { \Sigma; \Gamma; \Delta \vdash \code{$\TVAR{B, s'}$ $x$ := $e$ in $S$} : \TCMD{s} }
  \RULE[t-assv][stm_t_assv]({ 
  \begin{array}{r @{~} l}
    \Delta(x) & = \TVAR{B, s'} \\
    s         & \ordleq s'
  \end{array}
  })
    { \Sigma; \Gamma; \Delta \vdash e : (B, s') }
    { \Sigma; \Gamma; \Delta \vdash \code{$x$ := $e$} : \TCMD{s} }
  \RULE[t-assf][stm_t_assf]({
  \begin{array}{r @{~} l}
    \Delta(\code{this}) & = \TVAR{I, s_1} \\
    \Gamma(I)(p)        & = \TVAR{B, s'}  \\
    s_1                 & \ordleq s' \\
    s                   & \ordleq s'
  \end{array}
  })
    { \Sigma; \Gamma; \Delta \vdash e : (B, s') }
    { \Sigma; \Gamma; \Delta \vdash \code{this.$p$ := $e$} : \TCMD{s} }
  \RULE[t-call][stm_t_call]({ 
  \begin{array}{r @{~} l}
    s_1 & \ordleq s'        \\
    s'  & \ordleq s_3, s_4  \\
    s   & \ordleq s'       
  \end{array}
  })
    { 
    \begin{array}{r @{~} l}
      \Sigma; \Gamma; \Delta & \vdash e_1 : (I^Y, s')    \\
      \Sigma; \Gamma; \Delta & \vdash e_2 : (\TINT, s') \AND \Sigma; \Gamma; \Delta \vdash \VEC{e} : \BSVEC
    \end{array}
    }
    { \Sigma; \Gamma; \Delta \vdash \CALL{e_1}{f}{\VEC{e}}{e_2} : \TCMD{s} } 
    \WHERE{ \TVAR{I^X, s_1} }    { = \Delta(\code{this}) \hspace*{2.6cm}  \TSPROC{\BSVEC}{s'} = \Gamma(I^Y)(f) }             
    \WHERE{ \TVAR{\TINT, s_3} }  { = \Gamma(I^X)(\code{balance})  \hspace*{1.5cm}  \TVAR{\TINT, s_4} = \Gamma(I^Y)(\code{balance}) }
  \RULE[t-dcall][stm_t_dcall]({ 
  \begin{array}{r @{~} l}
    s_1 & \ordleq s' \\
    s   & \ordleq s' 
  \end{array}
  })
    { 
      \Sigma \vdash I^X \SUBS I^Y                   \AND 
      \Sigma; \Gamma; \Delta \vdash e : (I^Y, s')   \AND
      \Sigma; \Gamma; \Delta \vdash \VEC{e} : \BSVEC
    }
    { \Sigma; \Gamma; \Delta \vdash \DCALL{e}{f}{\VEC{e}} : \TCMD{s} }
    \WHERE{ }{ \TVAR{I^X, s_1} = \Delta(\code{this}) \hspace*{2cm}  \TSPROC{\BSVEC}{s'}  = \Gamma(I)(f) }
    \WHERE{ \forall q }{ \in (\DOM{\Gamma(I^Y)} \INTERSECT (\FNAMES \UNION \{\code{balance}\})).\ \Gamma(I^Y)(q) = \Gamma(I^X)(q) }
    %
\end{semantics}
\caption{Type rules for statements}
\label{fig:type_rules_stm}
\end{figure}

\begin{figure}\centering
\begin{semantics}
  \RULE[t-env-$\EMPTYSET$](J \in \SET{T, M, S, F, V})
    { }
    { \Sigma; \Gamma; \Delta \vdash \ENV{J}^\EMPTYSET }
  \RULE[t-env-t][env_t_envt]
    { \Sigma; \Gamma; \EMPTYSET \vdash_X \ENV{M} \AND \Sigma; \Gamma; \EMPTYSET \vdash \ENV{T} }
    { \Sigma; \Gamma; \EMPTYSET \vdash \ENV{T}, X:\ENV{M} }
  \RULE[t-env-m][env_t_envm](s \ordleq s_1)
    { \Sigma; \Gamma; \EMPTYSET \vdash_X \ENV{M} \AND \Sigma; \Gamma; \Delta \vdash S : \TCMD{s} }
    { \Sigma; \Gamma; \EMPTYSET \vdash_X \ENV{M}, f:(x_1, \ldots, x_h, S) }
  \WHERE{ I_{s_1} }            { = \Gamma(X) \qquad\qquad \TVAR{\TINT_{s_2} } = \Gamma(I)(\code{balance})}
  \WHERE{ \Gamma(I)(f) }       {  = \TSPROC{(B_1, s_1'), \ldots, (B_n, s_n')}{s} }
  \WHERE{\Delta}               { = \code{this}:\TVAR{I_{s_1}}, \code{value}:\TVAR{\TINT_{s_2}}, \code{sender}:\TVAR{\ITOP_{\STOP}}, }
  \WHERE{}                     { \qquad x_1:\TVAR{B_1, s_1'}, \ldots, x_n:\TVAR{B_n, s_n'} }
  \RULE[t-envs][env_t_envs]
    { \Sigma; \Gamma; \EMPTYSET \vdash \ENV{S} \AND \Sigma; \Gamma; \EMPTYSET \vdash_X \ENV{F} }
    { \Sigma; \Gamma; \EMPTYSET \vdash \ENV{S}, X:\ENV{F} }
  \RULE[t-envf][env_t_envf]({ 
  \begin{array}{r @{~} l}
    \Gamma(X)    & = (I, s')   \\
    \Gamma(I)(p) & = \TVAR{B_s} 
  \end{array}
  })
    { \Sigma; \Gamma; \EMPTYSET \vdash_X \ENV{F} \AND \Sigma; \Gamma; \EMPTYSET \vdash v : B_s }
    { \Sigma; \Gamma; \EMPTYSET \vdash_X \ENV{F}, p:v }
  \RULE[t-envv$_t$][env_t_envv_t]({ \Delta(x) = \TVAR{B_s} })
    { \Sigma; \Gamma; \Delta \vdash \ENV{V} \AND \Sigma; \Gamma; \EMPTYSET \vdash v : B_s }
    { \Sigma; \Gamma; \Delta \vdash \ENV{V}, x:(v, B_s) }
  \RULE[t-envv$_u$][env_t_envv_u]({ \Delta(x) = \TVAR{B_s} })
    { \Sigma; \Gamma; \Delta \vdash \ENV{V} \AND \Sigma; \Gamma; \EMPTYSET \vdash v : B_s }
    { \Sigma; \Gamma; \Delta \vdash \ENV{V}, x:v }
  \RULE[t-envsv][env_t_envsv]
    { \Sigma; \Gamma; \EMPTYSET \vdash \ENV{S} \AND \Sigma; \Gamma; \Delta \vdash \ENV{V} }
    { \Sigma; \Gamma; \Delta \vdash \ENV{SV} }
\end{semantics}
\caption{Type rules for the method-, state-, and variable environments}
\label{fig:type_rules_env}
\end{figure}

The type rules for expressions $e$ are given in Figure~\ref{fig:type_rules_expr}; those for statements $S$ are given in Figure~\ref{fig:type_rules_stm}; those for the three environments $\ENV{TSV}$ are given in Figure~\ref{fig:type_rules_env}.
All of them are similar to the corresponding ones provided by the type system in \cite{AGL25}, with one exception: 
Since we store the type information for locally declared variables in $\ENV{V}$, we also need to slightly alter the type rule for $\ENV{V}$, so for this purpose we use the rule \nameref{env_t_envv_t} as is done in \cite{AGLM/2024/reocas/outofgas}. 
However, as we reuse the definition of well-typedness of environments in the runtime typed semantics and in later developments in the paper, where this type information instead is stored in a separate type environment $\Delta$, we also provide an `untyped' variant of this rule, \nameref{env_t_envv_u}.
Thus, for the syntactic type system with the untyped semantics, rule \nameref{env_t_envv_t} is to be used, and in the typed semantics and all later developments rule \nameref{env_t_envv_u} is to be used.
We do not directly distinguish between the two different definitions of well-typedness of $\ENV{V}$, since it is clear from the context which rule to use.

\begin{figure}
\begin{semantics}
  \RULE[t-bot][stack_t_bot]
    { }
    { \Sigma; \Gamma; \Delta \vdash \bot : \TCMD{s} }
  \RULE[t-stm][stack_t_stm]
    { \Sigma; \Gamma; \Delta \vdash S : \TCMD{s} \AND \Sigma; \Gamma; \Delta \vdash Q : \TCMD{s} }
    { \Sigma; \Gamma; \Delta \vdash S; Q : \TCMD{s} }
  \RULE[t-del][stack_t_del]
    { \Sigma; \Gamma, \Delta \vdash Q : \TCMD{s} }
    { \Sigma; \Gamma; \Delta, x:T \vdash \DEL{x}; Q : \TCMD{s} }
  \RULE[t-ret][stack_t_ret]
    { \Sigma; \Gamma; \EXTRACT{\ENV{V}} \vdash \ENV{V} \AND \Sigma; \Gamma; \EXTRACT{\ENV{V}} \vdash Q : \TCMD{s} }
    { \Sigma; \Gamma; \Delta \vdash \ENV{V}; Q : \TCMD{s} } 
\end{semantics}
\caption{Type rules for stacks.}
\label{fig:type_rules_stacks}
\end{figure}

The type rules for stacks $Q$ are given in Figure~\ref{fig:type_rules_stacks}.
These are mostly as expected, since none of the stack operations place any restrictions on the level $s$; hence, we do not need any side conditions to alter the security level.
The only slightly unobvious rule is \nameref{stack_t_ret}, corresponding to the `return' operation, where a method call ends, and we restore the previous variable environment $\ENV{V}$.
In this case, we need a way to recover the types of locally declared variables, which are stored in $\ENV{V}$ along with the values.
The extraction function $\EXTRACT{\cdot}$ simply extracts this information from the variable environment found on the stack, to build the type environment $\Delta'$ relative to which the remainder of the stack $Q$ is then typed. It is defined as follows:
\begin{align*}
  \EXTRACT{\epsilon}          & = \epsilon \\
  \EXTRACT{\ENV{V}'\ ,x:(v,T)} & = \EXTRACT{\ENV{V}'}\ , x:T
\end{align*}

\clearpage
\section{Equivalence of states}\label{app:equivalence_states}
The notion of an $s$-parametrised equivalence relation on states, written $\Gamma; \Delta \vdash \ENV{SV}^1 =_s \ENV{SV}^2$, is defined by the rules in Figure~\ref{fig:env_s_eq}.
It denotes that the environments $\ENV{SV}^1$ and $\ENV{SV}^2$ have the same values on all entries of level $s$ or lower.

\begin{figure}[h]\centering
\begin{semantics}
  \RULE[eq-env$_V^\EMPTYSET$]
    { }
    { \Delta \vdash \ENV{V}^\EMPTYSET =_s \ENV{V}^\EMPTYSET }
  \RULE[eq-env$_V$][env_s_eq_envv]({
    \begin{array}{l}
      \Delta(x) = \TVAR{B_{s'}} \\
      s' \ordleq s \implies v_1 = v_2
    \end{array}
  })
    { \Delta \vdash \ENV{V}^1 =_s \ENV{V}^2 }
    { \Delta \vdash \ENV{V}^1, x:v_1 =_s \ENV{V}^2, x:v_2 }
  \RULE[eq-env$_F^\EMPTYSET$]
    { }
    { \Gamma \vdash_I \ENV{F}^\EMPTYSET =_s \ENV{F}^\EMPTYSET }
  \RULE[eq-env$_F$][env_s_eq_envf]({
    \begin{array}{l}
      \Gamma(I)(p) = \TVAR{B_{s'}} \\
      s' \ordleq s \implies v_1 = v_2
    \end{array}
  })
    { \Gamma \vdash_I \ENV{F}^1 =_s \ENV{F}^2 }
    { \Gamma \vdash_I \ENV{F}^1, p:v_1 =_s \ENV{F}^2, p:v_2 }
  \RULE[eq-env$_S^\EMPTYSET$]
    { }
    { \Gamma \vdash \ENV{S}^\EMPTYSET =_s \ENV{S}^\EMPTYSET }
  \RULE[eq-env$_S$][env_s_eq_envs]({ \Gamma(X) = I_{s'} })
    { \Gamma \vdash \ENV{S}^1 =_s \ENV{S}^2 \AND \Gamma \vdash_I \ENV{F}^1 =_s \ENV{F}^2 }
    { \Gamma \vdash \ENV{S}^1, X:\ENV{F}^1 =_s \ENV{S}^2, X:\ENV{F}^2 }
  \RULE[eq-env$_{SV}$]
    { \Gamma \vdash \ENV{S}^1 =_s \ENV{S}^2 \AND \Delta \vdash \ENV{V}^1 =_S \ENV{V}^2 }
    { \Gamma; \Delta \vdash \ENV{SV}^1 =_s \ENV{SV}^2 }
\end{semantics}
\caption{Rules for the $s$-parameterised equivalence relation.}
\label{fig:env_s_eq}
\end{figure}

\clearpage
\section{Environment consistency}\label{app:environment_consistency}
Having a well-typed stack $Q$ and well-typed environments $\ENV{TSV}$ is unfortunately not enough to ensure that the combination $\CONF{Q, \ENV{TSV}}$ will also be safe to execute.
The names appearing in a stack $Q$ must obviously also appear within the state $\ENV{SV}$, since otherwise the transition might become stuck.
However, well-typedness does not ensure the \emph{existence} of a field or variable in $\ENV{SV}$, so it is entirely possible that two different states $\ENV{SV}^1$ and $\ENV{SV}^2$ can both be well-typed relative to the same $\Sigma; \Gamma; \Delta$, yet nevertheless have no fields or variables in common.
Then e.g.\@ $\CONF{Q, \ENV{TSV}^1}$ might be executable, but $\CONF{Q, \ENV{TSV}^2}$ might not, if the two states do not contain the same names.

Likewise, $\ENV{T}$ and $\ENV{S}$ represent two different parts of the contract declarations, namely the method declarations and the field declarations, respectively.
Yet, well-typedness itself does not ensure that an arbitrary $\ENV{T}$ and $\ENV{S}$ are actually derived from the \emph{same} set of contract declarations.

Furthermore, addresses can not only appear in the domain of $\ENV{S}$, but also as \emph{values} $v$ in $\ENV{SV}$, e.g.\@ as the value of a variable $x$.
Any such address could be `dereferenced' if $x$ is used in an expression or as the object path in a method call.
Well-typedness only ensures that the address has an appropriate interface type, but not that an appropriate \emph{implementation} of the interface also exists on that address.
Moreover, even if a contract $C$ declares that it implements an interface $I$, well-typedness does not ensure that $C$ implements exactly \emph{every member} of $I$.
It only ensures that the members that $C$ \emph{does} implement are in accordance with the corresponding type definitions in $I$.

To handle these problems, we shall need to impose some further consistency checks on the environments $\ENV{TSV}$ and their inner environments $\ENV{M}$ and $\ENV{F}$, containing the methods and fields of each individual contract.
Firstly, we shall define two different notions of free names:

\begin{figure}\centering
\begin{align*}
 \FV{v}                                      & = \EMPTYSET                              \\ 
 \FV{x}                                      & = \SET{x}                                \\ 
 \FV{e.p}                                    & = \FV{e}                                 \\
 \FV{\op(e_1, \ldots, e_n)}                  & = \FV{e_1} \UNION \ldots \UNION \FV{e_n} \\
                                             &                                          \\ 
 \FV{\code{skip}}                            & = \EMPTYSET                              \\
 \FV{\code{throw}}                           & = \EMPTYSET                              \\
 \FV{\code{var $x$ := $e$ in $S$}}           & = \FV{e} \UNION \left(\FV{S} \SETMINUS \SET{x} \right) \\
 \FV{\code{$x$ := $e$}}                      & = \SET{x} \UNION \FV{e}                  \\
 \FV{\code{this.$p$ := $e$}}                 & = \SET{\code{this}} \UNION \FV{e}        \\ 
 \FV{\code{$S_1$; $S_2$}}                    & = \FV{S_1} \UNION \FV{S_2}               \\ 
 \FV{\code{if $e$ then $S_1$ else $S_2$}}    & = \FV{e} \UNION \FV{S_1} \UNION \FV{S_2} \\
 \FV{\code{while $e$ do $S$}}                & = \FV{e} \UNION \FV{S}                   \\ 
 \FV{\CALL{e_1}{m}{e_1', \ldots, e_n'}{e_2}} & = \FV{e_1} \UNION \FV{e_1'} \UNION \ldots \UNION \FV{e_n'} \\ 
                                             & \quad \UNION \FV{e_2} \UNION \FV{m}  \\
 \FV{\DCALL{e}{m}{e_1', \ldots, e_n'}}       & = \FV{e} \UNION \FV{e_1'} \UNION \ldots \UNION \FV{e_n'} \UNION \SET{\code{id}} \\
 \FV{m}                                      & = %
   \begin{cases}
     \SET{ \code{id} } & \text{if $m$ = \code{id}} \\
     \EMPTYSET         & \text{otherwise}
   \end{cases} 
\end{align*}
\caption{The free variable names for expressions and statements.}
\label{fig:free_variable_names}
\end{figure}

\begin{definition}[Free names]
We use the following notation:
\begin{itemize}
 \item We write $\FA{S}$ and $\FA{e}$ for the set of \emph{addresses} $X$ occurring in a statement $S$ and an expression $e$, respectively.
   We omit a formal definition, since addresses cannot be bound, so they can be read directly from the syntax.

 \item We write $\FV{S}$ and $\FV{e}$ for the set of \emph{free variable names} $x$ occurring in a statement $S$ and an expression $e$.
   The formal definition is given in Figure~\ref{fig:free_variable_names}. 
   We assume here that $x$ also ranges over the magic variable names \code{this}, \code{sender}, \code{value}, \code{id}, and \code{args}. \qedhere
\end{itemize}
\end{definition}

\begin{figure}\centering
\begin{semantics}
 \RULE[c-envf$_1$][env_c_envf_empty]
   { }
   { \ENV{S} \vdash \ENV{F}^\EMPTYSET \conseq \EMPTYSET }
 \RULE[c-envf$_2$][env_c_envf_rec]
   { \ENV{S} \vdash \ENV{F}' \conseq \Gamma_F' }
   { \ENV{S} \vdash \ENV{F}', p:v \conseq \Gamma_F', p:\TVAR{B_s} }
   \WHERE{ v \in \ANAMES }{ \implies v \in \DOM{\ENV{S}} }
 \RULE[c-envm$_1$][env_c_envm_empty]
   { }
   { \ENV{M}^\EMPTYSET \conseq \EMPTYSET }
 \RULE[c-envm$_2$][env_c_envm_rec]
   { \ENV{S} \vdash \ENV{M}' \conseq \Gamma_F' }
   { \ENV{S} \vdash \ENV{M}', f:(\VEC{x}, S) \conseq \Gamma_M', f:\TSPROC{\BSVEC}{s} }  
   \WHERE{ \FA{S} }{ \subseteq \DOM{\ENV{S}} }
\end{semantics}
\caption{Consistency rules for $\ENV{F}$ and $\ENV{M}$.}
\label{fig:consistency_rules_envfm}
\end{figure}

Now, in Figure~\ref{fig:consistency_rules_envfm}, we give the rules for a consistency check for $\ENV{F}$ and $\ENV{M}$.
Both are defined in terms of a relation $\conseq$ between a code environment ($\ENV{F}$, resp.\@ $\ENV{M}$) and type environment ($\Gamma_F$, resp.\@ $\Gamma_M$), where the latter represents the interface definitions of fields and methods for a particular contract.
For $J \in \SET{F,M}$, the consistency judgments are of the form 
\begin{equation*}
 \ENV{S} \vdash \ENV{J} \conseq \Gamma_J .
\end{equation*}
In rule \nameref{env_c_envf_rec} we check that each field definition in $\Gamma_F$ has a corresponding field declaration in $\ENV{F}$; furthermore, if the value $v$ stored in the field $p$ is an address, then this address must also be in the domain of $\ENV{S}$.
   The first check ensures that all declared fields also have an implementation, whilst the second check ensures that all addresses appearing as values also contain a contract.
In rule \nameref{env_c_envm_rec} we check that each method signature has a corresponding implementation.
   Furthermore, in the side condition, we require that any free addresses occurring in the method body $S$ must also exist in $\ENV{S}$.
   This ensures that, even if a `hard-coded' address appears in $S$, there must still be a contract on that address.
   This is necessary, since well-typedness alone only ensures that the address has the correct type for its usage, but not that it contains an implementation.

\begin{figure}\centering
\begin{semantics}
 \RULE[c-envv$_1$][env_c_envv_empty]
   { }
   { \ENV{S} \vdash \ENV{V}^\EMPTYSET }
 \RULE[c-envv$_2$][env_c_envv_rec]
   { \ENV{S} \vdash \ENV{V} }
   { \ENV{S} \vdash \ENV{V}, x:v }
   \WHERE{ v \in \ANAMES }{ \implies v \in \DOM{\ENV{S}} }
 \RULE[c-envs$_1$][env_c_envs_empty]
   { }
   { \Gamma; \ENV{S} \vdash \ENV{S}^\EMPTYSET }
 \RULE[c-envs$_2$][env_c_envs_rec]({
 \begin{array}{r @{~} l}
   \Gamma(X)          & = I_s      \\ 
   \Gamma(I)          & = \Gamma_I
 \end{array}
 })
   { \Gamma; \ENV{S} \vdash \ENV{S}' \AND \ENV{S} \vdash \ENV{F} \conseq \Gamma_I |_{\FNAMES} }
   { \Gamma; \ENV{S} \vdash \ENV{S}', X:\ENV{F} }
 \RULE[c-envt$_1$][env_c_envt_empty]
   { }
   { \Gamma; \ENV{S} \vdash \ENV{T}^\EMPTYSET }
 \RULE[c-envt$_2$][env_c_envt_rec]({
 \begin{array}{r @{~} l}
   \Gamma(X)          & = I_s      \\ 
   \Gamma(I)          & = \Gamma_I
 \end{array}
 })
   { \Gamma; \ENV{S} \vdash \ENV{T}' \AND \ENV{S} \vdash \ENV{M} \conseq \Gamma_I |_{\MNAMES} }
   { \Gamma; \ENV{S} \vdash \ENV{T}', X:\ENV{M} }
 \RULE[c-envsv][env_c_envsv]
   { \Gamma; \ENV{S} \vdash \ENV{S} \AND \ENV{S} \vdash \ENV{V} }
   { \Gamma \vdash \ENV{SV} }
 \RULE[c-envtsv][env_c_envtsv]
   { \Gamma; \ENV{S} \vdash \ENV{T} \AND \Gamma; \ENV{S} \vdash \ENV{S} \AND \ENV{S} \vdash \ENV{V} }
   { \Gamma \vdash \ENV{TSV} }
 \WHERE{ \DOM{\ENV{T}} }{ = \DOM{\ENV{S}} }
\end{semantics}
\caption{Consistency rules for the environments $\ENV{TSV}$.}
\label{fig:consistency_rules_envtsv}
\end{figure}

Next, we can define the consistency rules for $\ENV{TSV}$, which are given in Figure~\ref{fig:consistency_rules_envtsv}.
Note that we use the notation $f|_{\SETNAME{D}}$ to denote the \emph{restriction} of the domain of a function $f$ to the set of values in $\SETNAME{D}$.
We use this to extract the field, resp.\@ method, signatures from an inner type environment $\Gamma_I$, corresponding to the declaration of an interface $I$, in rules \nameref{env_c_envs_rec} and \nameref{env_c_envt_rec}, which then invoke the consistency rules from Figure~\ref{fig:consistency_rules_envfm} in their premises.
The rules are otherwise straightforward:
\begin{itemize}
 \item Rules \nameref{env_c_envv_empty} and \nameref{env_c_envv_rec} are used to check the consistency of the variable environment $\ENV{V}$, which simply consists of ensuring that any address appearing as a value also exists in the domain of $\ENV{S}$, similar to the rule \nameref{env_c_envf_rec} above.
   Thus, the judgment is relative to $\ENV{S}$.

 \item Rules \nameref{env_c_envs_empty} and \nameref{env_c_envs_rec} are used to check the consistency of $\ENV{S}$, which contains all field declarations of all contracts.
   It must be judged consistent relative to \emph{itself}, as well as to $\Gamma$, because we need to ensure both that all declared fields have an implementation, \emph{and} that all addresses appearing as values refer to actual contracts.

 \item Rules \nameref{env_c_envt_empty} and \nameref{env_c_envt_rec} are used to check the consistency of $\ENV{T}$, which contains all method declarations of all contracts.
   Like \nameref{env_c_envs_rec}, we extract the relevant method declarations from $\Gamma$ for each collection of declarations $\ENV{M}$, which then must correspond.

 \item Finally, we have the two short-hand rules \nameref{env_c_envsv} and \nameref{env_c_envtsv}, which are used to judge consistency of a collection of environments $\ENV{SV}$, resp.\@ $\ENV{TSV}$, together.
   Notably, in both rules, $\ENV{S}$ must be judged consistent relative to \emph{itself}; moreover, in the latter rule we also require that the domains of $\ENV{T}$ and $\ENV{S}$ must coincide.
   This must be the case, if $\ENV{T}$ and $\ENV{S}$ are derived from the same collection of contracts.
\end{itemize}

Well-typedness and consistency together ensure that all addresses actually contain contracts with appropriate interface types, and that all contracts also implement all members of their interfaces.
Unless noted otherwise, in what follows, we shall only consider consistent environments $\ENV{TSV}$.

\clearpage
\section{Semantic Typing for Expressions}\label{app:semantic_types_expressions}
The typed operational semantic for expressions relies on judgements of the form
$\Sigma; \Gamma; \Delta \esemDash \CONF{e, \ENV{SV}} \trans v$.
This judgement should be read as: $e$ evaluates to the value $v$ of type $B_s$, or a subtype thereof, given $\ENV{SV}$ and the type assumptions in $\Sigma; \Gamma; \Delta$.
The rules are in Figure~\ref{fig:typed_semantics_expressions} and have some quite complicated side conditions.
In particular, to ensure that the coercion property holds, all rules have a side condition of the form $s' \ordleq s$, which ensures that the transition can be concluded for any level $s$ or \emph{higher}.

\begin{figure}[h]\centering
\begin{semantics}
  \RULE[r-val][exp_rt_val]({
  \begin{array}{r @{~} l}
    \TYPEOF{v} & = (B', s') \\ 
    s'         & \ordleq s
  \end{array}
  })
    { \Sigma \vdash B' \SUBS B }
    { \Sigma; \Gamma; \Delta \esemDash \CONF{v, \ENV{SV}} \trans v }
  \RULE[r-var][exp_rt_var]({ 
  \begin{array}{r @{~} l}
    \ENV{V}(x)  & = v               \\
    \TYPEOF{v}  & = (B_1, s_1)      \\
    \Delta(x)   & = \TVAR{B_2, s_2} \\
    s_1 \ordleq & s_2 \ordleq s
  \end{array}
  })
    { \Sigma \vdash B_1 \SUBS B_2 \SUBS B }
    { \Sigma; \Gamma; \Delta \esemDash \CONF{x, \ENV{SV}} \trans v }
  \RULE[r-field][exp_rt_field]({ 
  \begin{array}{r @{~} l} 
    \ENV{S}(X)(p) & = v               \\
    \TYPEOF{v}    & = (B_1, s_1)      \\
    \Gamma(I)(p)  & = \TVAR{B_2, s_2} \\ 
    s_1 \ordleq   & s_2 \ordleq s
  \end{array}
  })
    { 
      \Sigma \vdash B_1 \SUBS B_2 \SUBS B \qquad
      \Sigma; \Gamma; \Delta \esemDash[I_s] \CONF{e, \ENV{SV}} \trans X 
    }
    { \Sigma; \Gamma; \Delta \esemDash \CONF{e.p, \ENV{SV}} \trans v }
  \RULE[r-op][exp_rt_op]({ \vdash \op : \VEC{B} \to B })
    { \Sigma; \Gamma; \Delta \esemDash[\VEC{B}_{s, \ldots, s}] \CONF{\VEC{e}, \ENV{SV}} \trans \VEC{v} \AND \op(\VEC{v}) \trans_{\op} v }
    { \Sigma; \Gamma; \Delta \esemDash \CONF{\op(\VEC{e}), \ENV{SV}} \trans v }
\end{semantics}
\caption{Typed operational semantics for expressions.}
\label{fig:typed_semantics_expressions}
\end{figure}

In \nameref{exp_rt_val}, we are evaluating a hard-coded value, so we use $\TYPEOF{\cdot}$ to obtain its type.
    However, as this value could be an \emph{address} $X$, and its type therefore would be an interface $I$, we need to invoke subtyping on the actual type, in case the runtime base type $B$ in fact were a supertype of $I$, rather than $I$ itself.

In \nameref{exp_rt_var}, we have a similar situation to the one in \nameref{exp_rt_val}, but now we need to invoke subtyping twice:
    first, to ensure that the actual value $v$ in fact was able to be stored in the variable $x$; second, to obtain the runtime base type $B$, in case $v$ were an address and the runtime type were a supertype of its interface.
    The check on the value is necessary, since we cannot know whether $\ENV{V}$ in fact is well-typed, so we must check this at runtime as well whenever an entry is read.    
    Likewise, we must require that the three security levels form a chain $s_1 \ordleq s_2 \ordleq s$; first, we must have that $s_1 \ordleq s_2$ to ensure that the value in fact was of a lower level than the container $x$; second, we must have that $s_2 \ordleq s$ to ensure that the coercion property holds.

Rule \nameref{exp_rt_field} is similar to the situation in \nameref{exp_rt_var}, except that we must now, in addition, have that $e$ evaluates to an address $X$ of type $I_s$.
    Thus, we require a runtime type $I_s$ in the evaluation. 
    Now, because of subtyping and the coercion property, it might be the case that the \emph{actual} type of $X$ were some interface $I'$ and of level $s'$ such that $\Sigma \vdash I' \SUBS I$ and $s' \ordleq s$, but in that case the transition can be concluded for $I_s$ as well.
    Note that, by requiring the runtime level $s$ in the evaluation of $e$, we are in effect ensuring that $e$ can be evaluated to a level that is \emph{lower} than, or equal to, $s$, which ascertains that higher-level parts of the program cannot alter the path to the field being read.

In \nameref{exp_rt_op}, we require the runtime types $\VEC{B}_{s, \ldots s}$ for the operands $\VEC{v}$, and in the side condition we require that $\vdash \op : \VEC{B} \to B$.
    Thus, we make use of two assumptions: first, the safety requirement for operations (Definition~\ref{def:safety_operations}), which assures us that the resulting value $v$ in fact will be of type $B$; second, $v$ cannot be an address, since no operation is allowed to return an address value.
    Hence, we can let the type of $v$, obtained from the signature of $\op$, equal the runtime base type, because $B$ must be one of $\TINT$ or $\TBOOL$, on which subtyping can only be reflexive.
    Furthermore, by Definition~\ref{def:typeof}, the level of $v$ can only be $\SBOT$, and therefore it always holds that the level of $v$ is below, or equal to, the runtime level $s$, since by definition $\SBOT \ordleq s$.
    Lastly, the runtime levels $s, \ldots, s$ required for the operands are the same as the level $s$ in the conclusion.
    In effect, we thus require that all operands can be evaluated at \emph{some} levels $s_1, \ldots, s_n$, which must all be equal to, or lower than, $s$.

We now prove a few results about the types semantics of expressions.
The first two results state that every transition that can be concluded by the typed semantic rules can also be concluded by the `untyped' semantics and that the coercion property for expressions (relying on the ordering of security levels and on the subtyping relation) holds.

\begin{theorem}[Semantic compatibility for expressions]\label{theorem:expressions_semantic_compatibility}
If $\Sigma; \Gamma; \Delta \esemDash \CONF{e, \ENV{SV}} \trans v$, then $\CONF{e, \ENV{SV}} \trans v$. 
\end{theorem}
\label{proof:expressions_semantic_compatibility}
\begin{proof}
By induction on the inference for $\Sigma; \Gamma; \Delta \esemDash \CONF{e, \ENV{SV}} \trans v$.
There are two base cases:
\begin{itemize}
  \item Suppose \nameref{exp_rt_val} was used; then, $e = v$. As there are no premises or side conditions, $\CONF{v, \ENV{SV}} \trans v$ can then be concluded by \nameref{exp_sem_val}.

  \item Suppose \nameref{exp_rt_var} was used; then, $e = x$ and, from the side condition of this rule, we know that $\ENV{V}(x) = v$.
    This satisfies the side condition of \nameref{exp_sem_var}, which then allows us to conclude $\CONF{x, \ENV{SV}} \trans v$.

\end{itemize}
There are two inductive cases:
\begin{itemize}
  \item Suppose \nameref{exp_rt_field} was used; then $e = e'.p$ and, from the premise and side conditions, we know that 
    \begin{align*}
      \Sigma; \Gamma; \Delta & \esemDash[I_s] \CONF{e', \ENV{SV}} \trans X \\
      \ENV{S}(X)(p)          & = v 
    \end{align*}
By the induction hypothesis, the first fact implies that
    \begin{equation*}
 \CONF{e', \ENV{SV}} \trans X
    \end{equation*}
This, together with $\ENV{S}(X)(p) = v$ and rule \nameref{exp_sem_field}, allows us to  conclude that $\CONF{e.p, \ENV{SV}} \trans v$.

  \item Suppose \nameref{exp_rt_op} was used; then, $e=\op(\VEC{e})$ and, from the premises of the rule, we have that 
    \begin{align*}
      \Sigma; \Gamma; \Delta & \esemDash[(B_1, s)] \CONF{e_1, \ENV{SV}} \trans v_1 \\
                             & \vdots                                              \\
      \Sigma; \Gamma; \Delta & \esemDash[(B_n, s)] \CONF{e_n, \ENV{SV}} \trans v_n \\
                             & \op(v_1, \ldots, v_n) \trans_{\op} v 
    \end{align*}
By the induction hypothesis on each of the $n$ transitions, we obtain
$$
\CONF{e_1, \ENV{SV}} \trans v_1
      \qquad \ldots \qquad
\CONF{e_n, \ENV{SV}} \trans v_n
 $$
 These, together with $\op(v_1, \ldots, v_n) \trans_{\op} v$ and rule \nameref{exp_sem_op}, allow us to conclude that $\CONF{\op(\VEC{e}), \ENV{SV}} \trans v$. \qedhere
\end{itemize}
\end{proof}

\begin{theorem}[Coercion property for expressions]\label{theorem:expressions_coercion_property}
If $\Sigma \vdash B_1 \SUBS B_2$ and $s_1 \ordleq s_2$, then
$\Sigma; \Gamma; \Delta \esemDash[(B_1, s_1)] \CONF{e, \ENV{SV}} \trans v$ implies that
$\Sigma; \Gamma; \Delta \esemDash[(B_2, s_2)] \CONF{e, \ENV{SV}} \trans v$.
\end{theorem}
\begin{proof}
By induction on the inference for $\Sigma; \Gamma; \Delta \esemDash[(B_1, s_1)] \CONF{e, \ENV{SV}} \trans v$.
There are two base cases:
\begin{itemize}
  \item If \nameref{exp_rt_val} was used, then $e=v$; from the premise and side conditions, we know that 
    \begin{align*}
      \TYPEOF{v} & = (B', s')          \\
      \Sigma     & \vdash B' \SUBS B_1 \\
      s'         & \ordleq s_1
    \end{align*}
     By assumption, $\Sigma \vdash B_1 \SUBS B_2$ and $s_1 \ordleq s_2$, so, by transitivity, we have that 
    \begin{align*}
      \Sigma & \vdash B' \SUBS B_2 \\ 
      s'     & \ordleq s_2 
    \end{align*}
 This satisfies the premises of \nameref{exp_rt_val}, so we can conclude that
    \begin{equation*}
      \Sigma; \Gamma; \Delta \esemDash[(B_2, s_2)] \CONF{v, \ENV{SV}} \trans v 
    \end{equation*}

  \item If \nameref{exp_rt_var} was used, then $e=x$; from the premise and side conditions, we know that 
    \begin{align*}
      \ENV{V}(x)   & = v                              \\
      \TYPEOF{v}   & = (B_1', s_1')                   \\
      \Delta(x)    & = \TVAR{B_2', s_2'}              \\
      \Sigma       & \vdash B_1'  \SUBS B_2' \SUBS B_1 \\
      s_1' \ordleq & \ s_2' \ordleq s_1 
    \end{align*}
 By assumption, $\Sigma \vdash B_1 \SUBS B_2$ and $s_1 \ordleq s_2$, so by transitivity we have that 
    \begin{align*}
      \Sigma       & \vdash B_1' \SUBS B_2' \SUBS B_2 \\
      s_1' \ordleq &\ s_2' \ordleq s_2 
    \end{align*}
 This satisfies the premises of \nameref{exp_rt_var}, so we can conclude that
    \begin{equation*}    
      \Sigma; \Gamma; \Delta \esemDash[(B_2, s_2)] \CONF{x, \ENV{SV}} \trans v 
    \end{equation*}

\end{itemize}
There are two inductive cases:
\begin{itemize}
  \item If \nameref{exp_rt_field} was used, then $e=e'.p$; from the premise and side conditions, we know that 
    \begin{align*}
      \Sigma; \Gamma; \Delta & \esemDash[(I, s_1)] \CONF{e', \ENV{SV}} \trans X \\
      \ENV{S}(X)(p)          & = v                                             \\
      \TYPEOF{v}             & = (B_1', s_1')                                  \\
      \Gamma(I)(p)           & = \TVAR{B_2', s_2'}                             \\ 
      \Sigma                 & \vdash B_1' \SUBS B_2' \SUBS B_1                \\
      s_1' \ordleq           &\ s_2' \ordleq s_1                  
    \end{align*}
 By assumption, $\Sigma \vdash B_1 \SUBS B_2$ and $s_1 \ordleq s_2$, so by transitivity we have that 
    \begin{align*}
      \Sigma & \vdash B_2' \SUBS B_2 \\
      s_2'   & \ordleq s_2 
    \end{align*}
 By the induction hypothesis, 
    \begin{equation*}
      \Sigma; \Gamma; \Delta \esemDash[(I, s_2)] \CONF{e', \ENV{SV}} \trans X 
    \end{equation*}
 can be concluded as well.
    This satisfies the premises of \nameref{exp_rt_field}, so
    \begin{equation*}    
      \Sigma; \Gamma; \Delta \esemDash[(B_2, s_2)] \CONF{e'.p, \ENV{SV}} \trans v 
    \end{equation*}

  \item If \nameref{exp_rt_op} was used, then $e=\op(e_1, \ldots, e_n)$; from the premise and side conditions, we know that 
    \begin{align*}
      \Sigma; \Gamma; \Delta & \esemDash[(B_1', s_1)] \CONF{e_1, \ENV{SV}} \trans v_1 \\
                             & \vdots                                                 \\
      \Sigma; \Gamma; \Delta & \esemDash[(B_n', s_1)] \CONF{e_n, \ENV{SV}} \trans v_n \\
                             & \op(\VEC{v}) \trans_{\op} v                            \\
                             & \vdash \op : B_1', \ldots B_n' \to B_1
    \end{align*}
 By assumption, $s_1 \ordleq s_2$; thus, by the induction hypothesis, 
$$
      \Sigma; \Gamma; \Delta \esemDash[(B_1', s_2)] \CONF{e_1, \ENV{SV}} \trans v_1 
 \qquad \ldots \qquad
      \Sigma; \Gamma; \Delta \esemDash[(B_n', s_2)] \CONF{e_n, \ENV{SV}} \trans v_n
$$
Furthermore, no operation $\op$ can yield any other type of value than $\TBOOL$ or $\TINT$, for which the only form of subtyping, that can be concluded, is by reflexivity (rule \nameref{typerules_sub_refl}).
    Thus, since $\Sigma \vdash B_1 \SUBS B_2$ by assumption, we know that $B_1 = B_2$.
    Therefore, 
    \begin{equation*}
      \vdash \op : B_1', \ldots B_n' \to B_2
    \end{equation*}
also holds.
    This satisfies the premises of \nameref{exp_rt_op}, so 
    \begin{equation*}
      \Sigma; \Gamma; \Delta \esemDash[(B_2, s_2)] \CONF{\op(e_1, \ldots, e_n), \ENV{SV}} \trans v 
\qedhere
    \end{equation*}
\end{itemize}
\end{proof}

An expression configuration $\CONF{e, \ENV{SV}}$ is \emph{safe} precisely when it yields a value $v$ by the typed semantics; i.e.\@ the rules ensure both that $v$ indeed is of type $B$, or a subtype thereof, and that all containers \emph{read from} in the evaluation of $e$ are of security level $s$ or lower.
This precisely captures the intended meaning of safety implied by the expression type $B_s$.
Again, we can quite easily show that this property holds for transitions in the typed semantics:

\begin{theorem}[Runtime safety for expressions]\label{theorem:expressions_runtime_safety}
If $\Sigma; \Gamma; \Delta \esemDash \CONF{e, \ENV{SV}} \trans v$, then $\TYPEOF{v} = (B_1, s_1)$, for some $B_1$ and $s_1$ such that: (1) $\Sigma \vdash B_1 \SUBS B$; (2) $s_1 \ordleq s$; and (3) every container read from in the evaluation of $e$ is of type $\TVAR{B_2, s_2}$, with $s_2 \ordleq s$.
\end{theorem}
\begin{proof}
By induction on the inference of $\Sigma; \Gamma; \Delta \esemDash \CONF{e, \ENV{SV}} \trans v$.
There are two base cases:
\begin{itemize}
  \item Suppose the transition was concluded by \nameref{exp_rt_val}; then, $e=v$. As no variables or fields are evaluated, the requirement on their security level is vacuously satisfied.
    From the premise and side conditions, we directly obtain the desired 
    \begin{align*}
      \TYPEOF{v} & = (B_1, s_1)       \\ 
      \Sigma     & \vdash B_1 \SUBS B \\
      s_1        & \ordleq s
    \end{align*}

  \item Suppose the transition was concluded by \nameref{exp_rt_var}; then, $e=x$. From the premise and side conditions, we have that 
    \begin{align*}
      \ENV{V}(x)  & = v                          \\
      \TYPEOF{v}  & = (B_1, s_1)                 \\
      \Delta(x)   & = \TVAR{B_2, s_2}            \\
      \Sigma      & \vdash B_1 \SUBS B_2 \SUBS B \\
      s_1 \ordleq &\ s_2 \ordleq s
    \end{align*}
Hence, we read from a container, which has type $\TVAR{B_2, s_2}$ for $s_2 \ordleq s$, as required. 
    Furthermore, we can conclude $s_1 \ordleq s$ and $\Sigma \vdash B_1 \SUBS B$ as well by transitivity of the two relations.
\end{itemize}
There are two inductive cases:
\begin{itemize}

  \item Suppose the transition was concluded by \nameref{exp_rt_field}; then, $e=e'.p$. From the premise and side conditions, we have that
    \begin{align*}
      \Sigma; \Gamma; \Delta & \esemDash[I_s] \CONF{e', \ENV{SV}} \trans X \\
      \ENV{S}(X)(p)          & = v               \\
      \TYPEOF{v}             & = (B_1, s_1)      \\
      \Gamma(I)(p)           & = \TVAR{B_2, s_2} \\ 
      \Sigma                 & \vdash B_1 \SUBS B_2 \SUBS B \\
      s_1 \ordleq            & s_2 \ordleq s
    \end{align*}
By the induction hypothesis, the statement of this Lemma holds for the transition 
    \begin{equation*}
      \Sigma; \Gamma; \Delta \esemDash[I_s] \CONF{e', \ENV{SV}} \trans X  
    \end{equation*}
In particular, all containers read from in the evaluation of $e'$ are of level $s$ or lower.
Moreover, the field $p$ has type $\TVAR{B_2, s_2}$, and so the requirement $s_2 \ordleq s$ holds.
Finally, we can conclude $s_1 \ordleq s$ and $\Sigma \vdash B_1 \SUBS B$ as well, by transitivity of the two relations. 

  \item Suppose the transition was concluded by \nameref{exp_rt_op}; then $e=\op(e_1, \ldots, e_n)$. From the premise and side conditions, we have that 
    \begin{align*}
      \Sigma; \Gamma; \Delta & \esemDash[(B_1, s)] \CONF{e_1, \ENV{SV}} \trans v_1 \\ 
                             & \vdots                                              \\
      \Sigma; \Gamma; \Delta & \esemDash[(B_n, s)] \CONF{e_n, \ENV{SV}} \trans v_n \\
                             & \op(\VEC{v}) \trans_{\op} v                         \\
                             & \vdash \op : \VEC{B} \to B
    \end{align*}

    \noindent By $n$ applications of the induction hypothesis, we know that the statement of this Lemma holds for the $n$ transitions above.
Hence, no container of a level higher than $s$ is read from in the evaluations of $e_1$, \ldots, $e_n$.
    Finally, we know that $\TYPEOF{v} = (B, \SBOT)$ by the safety requirement for operations (Definition~\ref{def:safety_operations}), and $\SBOT \ordleq s$ by definition. \qedhere
\end{itemize}
\end{proof}


\begin{corollary}[Runtime non-interference for expressions]
\label{corollary:expressions_runtime_noninterference}
If $s_1 \ordleq s_2$, $\Gamma; \Delta \vdash \ENV{SV}^1 =_{s_2} \ENV{SV}^2$, 
$\Sigma; \Gamma; \Delta \esemDash[{B_{s_1}}] \CONF{e, \ENV{SV}^1} \trans v_1$, and 
$\Sigma; \Gamma; \Delta \esemDash[{B_{s_1}}] \CONF{e, \ENV{SV}^2} \trans v_2$,
then $v_1 = v_2$.
\end{corollary}
\begin{proof}
By assumption, all values of entries of level $s_2$ or lower in $\ENV{SV}^1$ and $\ENV{SV}^2$ agree; thus, by Theorem~\ref{theorem:expressions_runtime_safety}, only entries of level $s_1$ or lower are read in the transitions, since $s_1 \ordleq s_2$. Thus, the two expressions use the same values and therefore also produce the same result.
\end{proof}

Theorem~\ref{theorem:expressions_runtime_safety} ensures that, if an expression $e$ can be evaluated at a given type, then the resulting value $v$ will indeed be of that type.
However, suppose $v$ is an address $X$.
It is then technically possible that $X$ could have a type in $\Gamma$, but not actually contain an \emph{implementation} in $\ENV{S}$.
This trivial error is easily prevented by also requiring well-typedness and consistency of $\ENV{SV}$, and that $e$ does not contain any free names.

\begin{corollary}[Consistency of expression values]
\label{corollary:expressions_consistency_values}
If $\Sigma; \Gamma; \Delta \vdash \ENV{SV}$ (well-typedness), $\Gamma \vdash \ENV{SV}$ (consistency), and $\Sigma; \Gamma; \Delta \esemDash \CONF{e, \ENV{SV}} \trans v$,
then $v \in \ANAMES$ implies that $v \in \DOM{\ENV{S}}$.
\end{corollary}
\begin{proof}
By Theorem~\ref{theorem:expressions_runtime_safety}, we know that $\TYPEOF{v} = (B_1, s_1)$ such that $\Sigma \vdash B_1 \SUBS B$ and $s_1 \ordleq s$.
Thus $v$ is ensured to have a type.
If $v$ is an address $X$, then by the definition of $\TYPEOF{\cdot}$ (Definition~\ref{def:typeof}), $X$ must be an interface type $I$ from $\Gamma$.
Well-typedness and consistency of $\ENV{SV}$ then together ensure that every interface $I$ in $\Gamma$ also must have an implementation in $\ENV{S}$.
Thus we conclude that $v \in \DOM{\ENV{S}}$. 
\end{proof}


We now show a result that allows us to abstract away from the specific values stored in the state environment $\ENV{SV}$.
This ensures that, if an expression $e$ has a typed transition for some appropriately shaped $\ENV{SV}^1$, then any other similarly shaped $\ENV{SV}^2$ will also yield a transition.
Indeed, given some $\Sigma$, $\Gamma$ and $\Delta$, in Definition~\ref{def:construction_envsv} we show how to construct environments $\ENV{SV}$ such that well-typedness and consistency are ensured to hold, and thereby abstract away from a particular choice of these environments. 
Now, using the construction of Definition~\ref{def:construction_envsv}, we can now abstract away from the specific values used in the evaluation of an expression configuration $\CONF{e, \ENV{SV}}$ and thereby obtain a result about the safety of the expression $e$ itself:

\begin{theorem}[Safe state abstraction]\label{theorem:expressions_runtime_states_abstraction}
For all $\Sigma$, $\Gamma$ and $\Delta$, if $\ENV{SV}^1$ and $\ENV{SV}^2$ are built according to Definition~\ref{def:construction_envsv} and
$\Sigma; \Gamma; \Delta \esemDash \CONF{e, \ENV{SV}^1} \trans v_1$, then there exists a $v_2$ such that
$\Sigma; \Gamma; \Delta \esemDash \CONF{e, \ENV{SV}^2} \trans v_2$.
\end{theorem}
\begin{proof}
By induction on the rules of the typed semantics (Figure~\ref{fig:typed_semantics_expressions}).
By construction, $\ENV{SV}^1$ and $\ENV{SV}^2$ contain the same named entries (fields and variables), so since the configuration $\CONF{e, \ENV{SV}^1}$ was executable, the configuration $\CONF{e, \ENV{SV}^2}$ cannot be stuck due to a missing variable name.
A field or variable in $\ENV{SV}^2$ could contain an address $X$ not present in $\ENV{S}^1$, but by Lemma~\ref{lemma:construction_envsv} we know that $\ENV{SV}^2$ is well-typed and consistent w.r.t.\@ $\Sigma$, $\Gamma$ and $\Delta$, which by Corollary~\ref{corollary:expressions_consistency_values} ensures that $X$ must have an implementation in $\ENV{S}^2$, and by Theorem~\ref{theorem:expressions_runtime_safety} that $X$ has the appropriate type.
Thus, this cannot lead to the configuration being stuck either.
\end{proof}

The actual values $v_1$ and $v_2$ produced in the two transitions in Theorem~\ref{theorem:expressions_runtime_states_abstraction} may of course differ, but by Theorem~\ref{theorem:expressions_runtime_safety} above, their \emph{types} will be the same.
This intuitively follows because the typed semantics only restricts the \emph{types} of values read from $\ENV{SV}$, and not the actual values.
The existence of a typed transition with $\ENV{SV}^1$ implies that all the values read from this environment in the evaluation of $e$ are in accordance with the type restrictions in $\Sigma$, $\Gamma$, $\Delta$; and choosing any other state $\ENV{SV}^2$, which satisfies the same restrictions, will therefore still yield a transition.

%
%
%
Finally, Theorem~\ref{theorem:expressions_runtime_states_abstraction} allow us to abstract away from any \emph{specific} choice of environments $\ENV{SV}$ and instead focus on expressions that have a transition for \emph{any} sensible choice of these environments.
By this theorem, if an expression $e$ has a transition for some $\ENV{SV}^1$ constructed according to Definition~\ref{def:construction_envsv}, then it will have a transition for \emph{all} similarly constructed environments $\ENV{SV}^2$.
In a sense, these environments can be thought of as representing all meaningful `inputs' to the expression and, for any such appropriately shaped input, a safe expression will yield a value of the correct type.
We can now define the set of such safe expressions by using the notion of \emph{typing interpretations} in the style of Caires \cite{caires2007}:

\begin{definition}[Typing interpretation for expressions]\label{def:typing_interpretation_expr}
A \emph{typing interpretation} for expressions under a given $\Sigma$, $\Gamma$, $\Delta$ and $B_s$, written $\ETYPI$, is a set of expressions such that, if $e \in \ETYPI$, then  there exist an $\ENV{SV}$ built according to Definition~\ref{def:construction_envsv} and a value $v$ such that 
$\Sigma; \Gamma; \Delta \esemDash \CONF{e, \ENV{SV}} \trans v$.
The \emph{typing} of expressions w.r.t.\@ a given $\Sigma$, $\Gamma$, $\Delta$ and $B_s$, written $\ETYPING$, is the union of all typing interpretations $\ETYPI$.
%
We write $\Sigma; \Gamma; \Delta \vDash e : B_s$ if there exists a typing interpretation $\ETYPI$ containing $e$.
\end{definition}

Intuitively, any expression from $\ETYPING$ will safely \emph{behave as} a value of the type $B_s$ according to the typed semantics, regardless of the actual values bound to its names (as long as they are of the appropriate types).
This is exactly what we would expect it to \emph{mean}, when we say an expression inhabits the type $B_s$.
However, Definition~\ref{def:typing_interpretation_expr} does not tell us how to prove that an expression $e$ inhabits a given type $B_s$, apart from exhibiting a typing interpretation $\ETYPI$ containing $e$.
To remedy this situation, we shall therefore now state and show a theorem, which expresses that the syntactic type system is compatible with Definition~\ref{def:typing_interpretation_expr}, meaning 
that typability implies safety. 

\begin{theorem}[Compatibility for expressions]\label{theorem:compatibility_expressions}
If $\Sigma; \Gamma; \Delta \vdash e : B_s$, then $\Sigma; \Gamma; \Delta \vDash e : B_s$.
\end{theorem}
\begin{proof}
We must show that $\Sigma; \Gamma; \Delta \vdash e : B_s$ implies that there exists a typing interpretation $\ETYPI$ containing $e$.
The simplest such typing interpretation is just the singleton set $\SET{e}$, which by Definition~\ref{def:typing_interpretation_expr} is a typing interpretation w.r.t.\@ $\Sigma; \Gamma; \Delta$ and $B_s$ \emph{if} the transition 
\begin{equation*}
  \Sigma; \Gamma; \Delta \esemDash \CONF{e, \ENV{SV}} \trans v
\end{equation*}
can be concluded by the typed semantics for some state $\ENV{SV}$ built according to Definition~\ref{def:construction_envsv}, and which by Lemma~\ref{lemma:construction_envsv} therefore is ensured to be well-typed and consistent w.r.t.\@ $\Sigma; \Gamma; \Delta$.

Furthermore, from $\Sigma; \Gamma; \Delta \vdash e : B_s$, we know that all variable names, addresses and field names in $e$ must exist in $\Gamma$ and $\Delta$; so, by the construction of $\ENV{SV}$, we therefore also know that 
\begin{itemize}
  \item $\FA{e} \subseteq \DOM{\ENV{S}}$, and 
  \item $\FV{e} \subseteq \DOM{\ENV{V}}$.
\end{itemize}
The statement to be shown then simplifies to 
\begin{equation*}
  \Sigma; \Gamma; \Delta \vdash e : B_s \implies \Sigma; \Gamma; \Delta \esemDash \CONF{e, \ENV{SV}} \trans v
\end{equation*}
We show this by induction on the inference for $\Sigma; \Gamma; \Delta \vdash e : B_s$:
There are two base cases:
\begin{itemize}
  \item Suppose the type judgment was concluded by \nameref{exp_t_val}; then, $e=v$. From the premise and the side conditions, we know that 
  \begin{align*}
    \TYPEOF{v} & = (B', s') \\ 
    s'         & \ordleq s  \\
    \Sigma     & \vdash B' \SUBS B
  \end{align*}
  This satisfies the side conditions of \nameref{exp_rt_val}, so we can infer that
  \begin{equation*}
    \Sigma; \Gamma; \Delta \esemDash \CONF{v, \ENV{SV}} \trans v
  \end{equation*}

  \item Suppose the type judgment was concluded by \nameref{exp_t_var}; then, $e=x$. From the premise and side conditions, we know that 
  \begin{align*}
    \Delta(x) & = \TVAR{B_2, s_2} \\
    s_2       & \ordleq s    \\
    \Sigma    & \vdash B_2 \SUBS B
  \end{align*}
  By assumption, $\Sigma; \Gamma; \Delta \vdash \ENV{V}$, which was concluded by \nameref{env_t_envv_u}. 
  We know $x \in \DOM{\ENV{V}}$, so $\ENV{V}(x) = v$ by the assumption $\FV{e} \subseteq \DOM{\ENV{V}}$.
  Therefore, let $\ENV{V} = \ENV{V}', (x,v)$.

  From the premise and side condition of the aforementioned rule, we then know that 
  \begin{align*} 
    \Sigma; \Gamma; \Delta & \vdash v : (B_2, s_2) \\
    \Delta(x)              & = \TVAR{B_2, s_2}
  \end{align*}
  Now, $\Gamma; \Delta \vdash v : (B_2, s_2)$ must have been concluded by \nameref{exp_t_val}, and from its side condition we know that 
  \begin{align*}
    \TYPEOF{v} & = (B_1, s_1) \\ 
    s_1        & \ordleq s_2  \\
    \Sigma     & \vdash B_1 \SUBS B_2
  \end{align*}
  This satisfies all the side conditions of \nameref{exp_rt_var}, so we can infer that
  \begin{equation*}
    \Sigma; \Gamma; \Delta \esemDash \CONF{x, \ENV{SV}} \trans v
  \end{equation*}

\end{itemize}
There are two inductive cases:
\begin{itemize}

  \item Suppose the type judgment was concluded by \nameref{exp_t_field}; then, $e=e'.p$. From the premise and side conditions, we know that 
    \begin{align*}
      \Gamma; \Delta & \vdash e' : I_s    \\ 
      \Gamma(I)(p)   & = \TVAR{B_2, s_2} \\
      s_2            & \ordleq s         \\
      \Sigma         & \vdash B_2 \SUBS B
    \end{align*}
    By the induction hypothesis on $\Sigma; \Gamma; \Delta \vdash e' : I_s$, we obtain that
    \begin{equation*}
   \Gamma; \Delta \esemDash[{I_s}] \CONF{e', \ENV{SV}} \trans X
    \end{equation*}
 for some address $X$, since only addresses are given a type $I$ by the definition of \TYPEOF{\cdot} (Definition~\ref{def:typeof}).

    By assumption, $\Sigma; \Gamma; \Delta \vdash \ENV{S}$ holds, which was concluded by \nameref{env_t_envs}.
    We know that $X \in \DOM{\ENV{S}}$; so, $\ENV{S}(X) = \ENV{F}$ because of $\FA{e} \subseteq \DOM{\ENV{S}}$ and consistency.
    Therefore, let $\ENV{S} = \ENV{S}', X: \ENV{F}$.
    From the premises and side conditions of this rule, we know that $\Sigma; \Gamma; \Delta \vdash_X \ENV{F}$, which must have been concluded by \nameref{env_t_envf}.
    We know that $p \in \DOM{\ENV{F}}$, so $\ENV{F}(p) = v$ by consistency and well-typedness of $\ENV{SV}$.
    This establishes that $\ENV{S}(X)(p) = v$.

    Now, let $\ENV{F} = \ENV{F}', p:v$.
    From the premises and side conditions of \nameref{exp_t_field}, we know that
    \footnote{In principle, instead of  $\Gamma(X) = (I, s)$, we could have $\Gamma(X) = (I', s')$. However, by the coercion property (Theorem~\ref{theorem:expressions_coercion_property}), we then have that $\Sigma \vdash I' \SUBS I$ and $s' \ordleq s$, so we omit this extra complication.}
    \begin{align*}
      \Sigma; \Gamma; \Delta & \vdash v : (B_1, s_1) \\
      \Gamma(X)              & = (I, s)\\
      \Gamma(I)(p)           & = \TVAR{B_2, s_2}
    \end{align*}
    Judgement $\Sigma; \Gamma; \Delta \vdash v : (B_1, s_1)$ must have been concluded by \nameref{exp_t_val}, and from its side condition we know that 
    \begin{align*}
      \TYPEOF{v} & = (B_1, s_1) \\ 
      s_1        & \ordleq s_2  \\
      \Sigma     & \vdash B_1 \SUBS B_2
    \end{align*}
    This satisfies all the side conditions of \nameref{exp_rt_field}, so we can infer that
    \begin{equation*}
      \Sigma; \Gamma; \Delta \esemDash \CONF{e.p, \ENV{SV}} \trans v
    \end{equation*}

  \item Suppose the type judgment was concluded by \nameref{exp_t_op};
    then, $e= \op(\VEC{e})$. From the premises, we know that 
    \begin{align*}
                              & \vdash \op : \VEC{B} \to B  \\ 
      \Sigma; \Gamma; \Delta  & \vdash \VEC{e} : \BSVEC
    \end{align*}
    For each $e_i$ in $\VEC{e}$ and $(B_i, s)$, we can use the induction hypothesis to obtain that
    \begin{equation*}
    \Sigma; \Gamma; \Delta \esemDash[{(B_i, s)}] \CONF{e_i, \ENV{SV}} \trans v_i
    \end{equation*}
    Now, by the safety requirement for operations (Definition~\ref{def:safety_operations}), we know that it must be the case that $\op(\VEC{v}) \trans_{\op} v$ and $\TYPEOF{v} = (B, \SBOT)$.
    Furthermore, $\SBOT \ordleq s$ holds by definition of $\SBOT$.
    Thus, this satisfies all the premises and side conditions of \nameref{exp_rt_op}, so we can infer
    \begin{equation*}
      \Sigma; \Gamma; \Delta \esemDash \CONF{\op(\VEC{e}), \ENV{SV}} \trans v
      \qedhere
    \end{equation*}
\end{itemize}
\end{proof}

Combining Theorems~\ref{theorem:expressions_semantic_compatibility}, \ref{theorem:compatibility_expressions}, and \ref{theorem:expressions_runtime_safety} gives us the equivalent of the soundness theorem for expressions from \cite{AGL25}, which now follows as a simple corollary:

\begin{corollary}[Static soundness for expressions]\label{corollary:expressions_static_soundness}
If
\begin{itemize}
  \item $\Sigma; \Gamma; \Delta \vdash e : B_s$, and
  \item $\Sigma; \Gamma; \Delta \vdash \ENV{SV}$, and 
  \item $\Gamma \vdash \ENV{SV}$, and 
  \item $\FA{e} \subseteq \DOM{\ENV{S}}$, and 
  \item $\FV{e} \subseteq \DOM{\ENV{V}}$,
\end{itemize}

\noindent then there exists some $v$ such that $\CONF{e, \ENV{SV}} \trans v$, and $\TYPEOF{v} = (B_1, s_1)$ for some $B_1$ and $s_1$ such that $\Sigma \vdash B_1 \SUBS B$ and $s_1 \ordleq s$, and every container, that will be read from in the evaluation of $e$, is of type $\TVAR{B_2, s_2}$, for some $B_2$, such that $s_2 \ordleq s$.
\end{corollary}

\begin{proof}
By Theorem~\ref{theorem:compatibility_expressions}, we have that 
\begin{equation*}
  \Sigma; \Gamma; \Delta \vdash e : B_s \implies \Sigma; \Gamma; \Delta \vDash e : B_s
\end{equation*}

\noindent Unfolding the definition of $\Sigma; \Gamma; \Delta \vDash e : B_s$ (Definition~\ref{def:typing_interpretation_expr}) yields that there must exist a typing interpretation $\ETYPI$ containing $e$, which again means that the transition
\begin{equation*}
  \Sigma; \Gamma; \Delta \esemDash \CONF{e, \ENV{SV}} \trans v
\end{equation*}

\noindent can be concluded by the typed semantics for any consistent and well-typed $\ENV{SV}$ containing at least all the names used in $e$.
By Theorem~\ref{theorem:expressions_runtime_safety}, this transition implies that $\TYPEOF{v} = (B_1, s_1)$ and $s_1 \ordleq s$ and every container read from in $e$ is of level $s$ or lower. 
Finally, by Theorem~\ref{theorem:expressions_semantic_compatibility}, we have that 
\begin{equation*}
  \Sigma; \Gamma; \Delta \esemDash \CONF{e, \ENV{SV}} \trans v \implies \CONF{e, \ENV{SV}} \trans v \qedhere
\end{equation*}
\end{proof}

Finally, we can also obtain the equivalent of the non-interference theorem for expressions from \cite{AGL25} as a simple corollary of Corollaries~\ref{corollary:expressions_runtime_noninterference} and \ref{corollary:expressions_static_soundness}:

\begin{corollary}[Static non-interference for expressions]
If
\begin{itemize}
  \item $\Sigma; \Gamma; \Delta \vdash e : B_s$, and
  \item $\Sigma; \Gamma; \Delta \vdash \ENV{SV}^1$, $\Gamma \vdash \ENV{SV}^1$, $\FA{e} \subseteq \DOM{\ENV{S}^1}$, $\FV{e} \subseteq \DOM{\ENV{V}^1}$, and
  \item $\Sigma; \Gamma; \Delta \vdash \ENV{SV}^2$, $\Gamma \vdash \ENV{SV}^2$, $\FA{e} \subseteq \DOM{\ENV{S}^2}$, $\FV{e} \subseteq \DOM{\ENV{V}^2}$, and
  \item $\Gamma; \Delta \vdash \ENV{SV}^1 =_s \ENV{SV}^2$,
\end{itemize}

\noindent then $\CONF{e, \ENV{SV}^1} \trans v$ and $\CONF{e, \ENV{SV}^2} \trans v$.
\end{corollary}

\begin{proof}
From the assumptions, and by two applications of Corollary~\ref{corollary:expressions_static_soundness}, we get that 
\begin{align*}
  \CONF{e, \ENV{SV}^1} & \trans v_1 \\
  \CONF{e, \ENV{SV}^2} & \trans v_2
\end{align*}

\noindent such that no container of a level higher than $s$ was read from $\ENV{SV}^1$ resp.\@ $\ENV{SV}^2$ in the two evaluations of $e$.
By assumption, $\Gamma; \Delta \vdash \ENV{SV}^1 =_s \ENV{SV}^2$, so the two states agree on all entries of level $s$ or lower.
Thus, the two evaluations of $e$ will use the same values, and will therefore yield the same result, so $v_1 = v_2$.
\end{proof}

With this last result, we have recovered the properties ensured by syntactic well-typedness for expressions in \cite{AGL25}, but through a quite different route; namely by starting from a semantics of types for expressions, given in terms of a typed semantics, that is compatible with the untyped semantics and ensures that the restrictions imposed by the type predicates are in fact observed.

\clearpage
\section{Well-formedness of stacks}\label{app:wellformedness_stack}

The notion of \emph{well-formedness} of stacks w.r.t.\@ type environments $\Sigma; \Gamma$, a state $\ENV{SV}$, and a security level $s$, is defined by the rules in Figure~\ref{fig:wellformedness_stack}.

\begin{figure}[h]\centering
\begin{semantics}
  \RULE[wf-bot][wf_stack_bot]
    { }
    { \Sigma; \Gamma; \ENV{SV}; s \vdash \bot }
  \RULE[wf-stm][wf_stack_stm]({
  \begin{array}{r @{~} l}
    \FA{S} & \subseteq \DOM{\ENV{S}} \\
    \FV{S} & \subseteq \DOM{\ENV{V}} 
  \end{array}
  })
    { \Sigma; \Gamma; \ENV{SV}; s \vdash Q }
    { \Sigma; \Gamma; \ENV{SV}; s \vdash S; Q }
  \RULE[wf-del][wf_stack_del]
    { \Sigma; \Gamma; \ENV{SV}; s \vdash Q }
    { \Sigma; \Gamma; \ENV{SV}, x:v; s \vdash \DEL{x}; Q }
  \RULE[wf-ret][wf_stack_return]
    { 
      \Sigma; \Gamma; \Delta               \vdash \ENV{V}' \AND   
      \ENV{S}                              \vdash \ENV{V}' \AND   
      \Sigma; \Gamma; \ENV{S}; \ENV{V}'; s \vdash Q               
    }
    { \Sigma; \Gamma; \ENV{S}; \ENV{V}; s \vdash (\ENV{V}', \Delta); Q }
  \RULE[wf-sec][wf_stack_restore](s \ordgeq s')
    { \Sigma; \Gamma; \ENV{SV}; s' \vdash Q }
    { \Sigma; \Gamma; \ENV{SV}; s \vdash s'; Q }
\end{semantics}
\caption{Well-formedness of stacks $Q$.}
\label{fig:wellformedness_stack}
\end{figure}

In \nameref{wf_stack_stm} we require that any statement $S$ found on the stack must only contain names found in the domains of $\ENV{SV}$.
    Together with consistency of $\ENV{SV}$, this ensures that transitions cannot be stuck due to a free name, or a missing contract implementation.
In \nameref{wf_stack_del}, we remove $x$ from $\ENV{V}$ when we encounter an end-of-scope symbol $\DEL{x}$, since the domain of $\ENV{V}$ is used in the rule \nameref{wf_stack_stm}.
The rule \nameref{wf_stack_return} handles the of a return symbol $(\ENV{V}', \Delta)$ occurring on the stack.
In this case, we must here require that $\ENV{V}$ is both well-typed and consistent, and the remainder of the stack must then be well-formed relative to this new $\ENV{V}'$.
Finally, rule \nameref{wf_stack_restore} ensures that all security levels $s'$ occurring in the stack are in descending order.

Well-formedness ensures that $Q$ and $\ENV{SV}$ can meaningfully be combined in a configuration $\CONF{Q, \ENV{TSV}}$ and executed by the typed semantics at level $s$, provided that the initial state $\ENV{SV}$ is well-typed, and $\ENV{TSV}$ are consistent, and that $s$ is greater than, or equal to, the first security level occurring in $Q$.
Note that it is necessary to assume consistency of $\ENV{T}$ as well to ensure that well-formedness of $Q$ is preserved, since a transition can place code from $\ENV{T}$ on the stack.

\clearpage
\section{Typed Operational Semantics}\label{app:typed_semantics_stacks}
As was the case with the typed semantics of expressions, the rules for the typed semantics of stacks are obtained from the rules for the untyped semantics, but with added restrictions from the type rules for statements and stacks.
Every transition that can be concluded by the typed semantics can therefore also be concluded by the untyped semantics, with one exception: the transitions 
which restore a security level $s$ that had previously been placed on the stack.
There is nothing similar in the untyped semantics.
Hence, to show a statement corresponding to Theorem~\ref{theorem:expressions_semantic_compatibility} for stacks, we first need to remove all security levels $s$, and also merge the type information from $\Delta$ into variable environments $\ENV{V}$ placed on the stack.
This is done by functions $\REMS{Q}$ and $\MERGE{\ENV{V}; \Delta}$, defined on the syntax of return symbols and stacks, respectively:
\begin{align*}
  \MERGE{\ENV{V}^{\EMPTYSET}; \EMPTYSET} & = \ENV{V}^{\EMPTYSET}                  \\
  \MERGE{x:v, \ENV{V}; x:B_s, \Delta}    & = (x, v, B_s), \MERGE{\ENV{V}, \Delta}
\end{align*}

\begin{align*}
  \REMS{\bot}                            & = \bot                                 \\
  \REMS{S; Q}                            & = S; \REMS{Q}                          \\
  \REMS{\DEL{x}; Q}                      & = \DEL{x}; \REMS{Q}                    \\
  \REMS{(\ENV{V}, \Delta); Q}            & = \MERGE{\ENV{V}, \Delta}; \REMS{Q}    \\
  \REMS{s; Q}                            & = \REMS{Q}
\end{align*}
It will also be useful to know the security lever of the first symbol in a stack:
\begin{align*}
  \FIRSTS{\bot}                 & = \SBOT      \\ 
  \FIRSTS{S; Q}                 & = \FIRSTS{Q} \\ 
  \FIRSTS{\DEL{x}; Q}           & = \FIRSTS{Q} \\ 
  \FIRSTS{(\ENV{V}, \Delta); Q} & = \FIRSTS{Q} \\ 
  \FIRSTS{s; Q}                 & = s
\end{align*}
Notice that the last function returns $\SBOT$ for the bottom of the stack.
This corresponds to no restriction on the level at which the stack can be executed, so in that case we can choose the execution level arbitrarily low.

\begin{theorem}[Semantic compatibility for stacks]\label{theorem:stacks_semantic_compatibility}
Let all security levels in $Q$ occur in descending order.
For all type environments $\Sigma$, $\Gamma$, $\Delta$, and security level $s \ordgeq \FIRSTS{Q}$, and environments $\ENV{TSV}$, it holds that either
\begin{equation}\label{eq:semcomp1}
  \begin{split}
             & \Sigma; \Gamma; \Delta \ssemDash \CONF{Q, \ENV{TSV}} \trans \Sigma; \Gamma; \Delta' \ssemDash[{s'}] \CONF{Q', \ENV{T}; \ENV{SV}'} \\ 
    \implies & \CONF{\REMS{Q}, \ENV{TS}; \MERGE{\ENV{V}; \Delta}} \trans \CONF{\REMS{Q}, \ENV{T}; \ENV{S}'; \MERGE{\ENV{V}'; \Delta'}}
  \end{split}
\end{equation}

\noindent or otherwise there exists a sequence of transitions
\begin{equation*}
  \Sigma; \Gamma; \Delta \ssemDash \CONF{Q, \ENV{TSV}} \trans \ldots \trans \Sigma; \Gamma; \Delta \ssemDash[{s'}] \CONF{Q', \ENV{TSV}}
\end{equation*}

\noindent of finite length, to a state for which the implication in \eqref{eq:semcomp1} holds.
\end{theorem}

\begin{proof}
Firstly, suppose $Q = s'; Q'$.
Then $\FIRSTS{Q} = s'$ and $s \ordgeq s'$ by assumption.
Thus
\begin{equation*}
  \Sigma; \Gamma; \Delta \ssemDash \CONF{s'; Q', \ENV{TSV}} \trans \Sigma; \Gamma; \Delta \ssemDash[{s'}] \CONF{Q', \ENV{TSV}}
\end{equation*}

\noindent can then be concluded by \nameref{stm_rt_restore}.
It might now again be the case that $Q' = s''; Q''$, but then $s' \ordgeq s''$ by the assumption of ordering of the security levels.
Hence, the transition can again be concluded by \nameref{stm_rt_restore}, and so on.
As we know the stack is of finite length, it will eventually be the case that the symbol at the top of the stack is \emph{not} a security level, in which case the theorem holds if the statement in \eqref{eq:semcomp1} holds.

When $Q \neq s'; Q'$, we proceed by case analysis of the rules used to conclude the transition
\begin{equation*}
  \Sigma; \Gamma; \Delta \ssemDash \CONF{Q, \ENV{TSV}} \trans \Sigma; \Gamma; \Delta' \ssemDash[{s'}] \CONF{Q', \ENV{T}; \ENV{SV}'}
\end{equation*}

\noindent In each case, we obtain the required premises and side conditions needed to conclude the same transition in the untyped semantics from the rule of the typed semantics.
This is straightforward, since each typed semantic rule (with the exception of \nameref{stm_rt_restore}) has an untyped variant, whose premises and side conditions are a subset of the premises and side conditions of the typed variant.
\end{proof}

\begin{proof}[Proof of Proposition~\ref{theorem:stacks_coercion_property}]
By case analysis of the rules used to conclude the transition.
\begin{itemize}
  \item For rules \nameref{stm_rt_skip}, \nameref{stm_rt_whilefalse}, \nameref{stm_rt_decv}, \nameref{stm_rt_delv}, and \nameref{stm_rt_return}, the conclusion follows directly from the fact that these rules do not impose any restrictions on the security level.

  \item For rules \nameref{stm_rt_if}, \nameref{stm_rt_whiletrue}, \nameref{stm_rt_assv}, \nameref{stm_rt_assf}, \nameref{stm_rt_call}, \nameref{stm_rt_dcall}, and \nameref{stm_rt_fcall}, we have a side condition of the form $s_1 \ordleq s'$ restricting the security level, where $s'$ is derived from the premise of the rule. 
    As we know that $s_2 \ordleq s_1$, we therefore also have that $s_2 \ordleq s'$ by transitivity of $\ordleq$, which, for each rule, allows us to conclude the transition with $s_2$ as well.

  \item Finally, if \nameref{stm_rt_restore} was used to conclude the transition, we have in the side condition that $s_1 \ordgeq s'$, where $s'$ is the top element of the stack $Q$.
    Hence, $\FIRSTS{Q} = s'$, and by assumption $s_2 \ordgeq \FIRSTS{Q}$ so $s_2 \ordgeq s'$.
    Thus we can satisfy the side condition with $s_2$ as well and conclude the transition with \nameref{stm_rt_restore}. \qedhere
\end{itemize}
\end{proof}

\begin{proof}[Proof of Lemma~\ref{lemma:construction_envsv}]
The construction of Definition~\ref{def:construction_envsv} is a kind of ‘‘reverse reading'' of the well-typedness and consistency rules for $\ENV{SV}$ (see Figures~\ref{fig:type_rules_env} and~\ref{fig:consistency_rules_envtsv}).
The proofs of well-typedness and consistency are then obtained through an easy induction on the inferences of the respective judgements.
\end{proof}

\begin{proof}[Proof of Corollary~\ref{theorem:stacks_runtime_states_abstraction}]
By case analysis on the rules of the typed semantics (Figures~\ref{fig:typed_semantics_stacks1} and~\ref{fig:typed_semantics_stacks2}).
The environments $\ENV{SV}$ only affect the execution of a statement through variable and field names appearing in expressions, and for these cases we apply Theorem~\ref{theorem:expressions_runtime_states_abstraction} for the evaluation of all expressions occurring inside statements in $Q$.
\end{proof}

Finally, we can show a result concerning the type environment $\Delta$ and the security level $s$.
Both may have changed after a transition, in case the security level was raised, a local variable was declared, or a method call was issued.
However, any such change is ‘temporary', in the sense that all of them have resulted in a special symbol pushed onto the stack, which undoes the change whenever reached.
Thus, \emph{if} a stack $Q$ terminates normally (i.e., by reaching the bottom symbol $\bot$), then all such temporary changes have been discarded.
Assuming $Q$ is executed relative to some $\Delta$ and $s$, we can detect the type environment and security level $(\Delta', s')$ on which $Q$ will terminate, \emph{if} it terminates, as follows. This will be used in the proofs of Appendix~\ref{app:semantic_typing_stacks}.

\begin{definition}[Terminal $\Delta$ and $s$]\label{def:terminal_conf_types}
We define the function $\FINISH{\cdot}$ on triplets $(Q, \Delta, s)$ of stacks $Q$, type environments $\Delta$ and security levels $s$ as follows:
\belowdisplayskip=-12pt
\begin{align*}
  \FINISH{\bot, \Delta, s}                          & = (\Delta, s)            \\ 
  \FINISH{S; Q, \Delta, s}                          & = \FINISH{Q, \Delta, s}  \\
  \FINISH{\DEL{x}; Q, (\Delta, x:\TVAR{B_{s'}}), s} & = \FINISH{Q, \Delta, s}  \\
  \FINISH{(\ENV{V}, \Delta'); Q, \Delta, s}         & = \FINISH{Q, \Delta', s} \\
  \FINISH{s'; Q, \Delta, s}                         & = \FINISH{Q, \Delta, s'} 
\end{align*}\qedhere
\end{definition}


\begin{theorem}[Terminal type level]\label{theorem:stacks_term_level}
If $\Sigma; \Gamma; \Delta \ssemDash \CONF{Q, \ENV{TSV}} \trans* \Sigma; \Gamma; \Delta' \ssemDash[{s'}] \CONF{\bot, \ENV{T}, \ENV{SV}'}$, then $(\Delta', s') = \FINISH{Q, \Delta, s}$.
\end{theorem}
\begin{proof}
We proceed by induction on the length of $\trans*$.
In the base case, $Q=\bot$, $\Delta'=\Delta$ and $s'=s$, and it trivially holds by definition of $\FINISH\cdot$.
For the inductive case, let
\begin{equation*}
  \Sigma; \Gamma; \Delta \ssemDash \CONF{Q, \ENV{TSV}} \trans \Sigma; \Gamma; \Delta'' \ssemDash[{s''}] \CONF{Q'', \ENV{T}, \ENV{SV}''}
  \trans* \Sigma; \Gamma; \Delta' \ssemDash[{s'}] \CONF{Q', \ENV{T}, \ENV{SV}'}
\end{equation*}
and let us proceed by case analysis of the semantic rule used to conclude the first transition:
\begin{itemize}
  \item The cases for \nameref{stm_rt_skip}, \nameref{stm_rt_whilefalse}, \nameref{stm_rt_assv}, \nameref{stm_rt_assf} are immediately seen to hold, since neither $\Delta$ nor $s$ are modified.

  \item The cases for \nameref{stm_rt_if} and \nameref{stm_rt_whiletrue} are similar, so we shall only consider the case for \nameref{stm_rt_if}.
    We have that $Q = \code{if $e$ then $S_\TRUE$ else $S_\FALSE$}; \hat Q$ and  $Q''= S_b;s;\hat Q$.
    By induction and Definition~\ref{def:terminal_conf_types}, we can conclude 
    \begin{equation*}
      (\Delta',s') = \FINISH{S;s;\hat Q, \Delta, s''} = \FINISH{s;\hat Q, \Delta, s''} = \FINISH{\hat Q, \Delta, s} = \FINISH{Q, \Delta, s}
    \end{equation*}

  \item If \nameref{stm_rt_decv} was used, then $Q= \code{$\TVAR{B_{\bar s}}$ $x$ := $e$ in $S$} ; \hat Q$ and $Q''= S ; \DEL{x} ; \hat Q$; moreover, 
  $\Delta''= \Delta, x:\TVAR{B_{\bar s}}$ and $s''=s$. By induction and Definition~\ref{def:terminal_conf_types}, we can conclude
    \begin{align*}
      (\Delta',s') = ~   & \FINISH{ S ; \DEL{x} ; \hat Q, (\Delta, x:\TVAR{B_{\bar s}}), s} \\
      = ~& \FINISH{\DEL{x} ; \hat Q, (\Delta, x:\TVAR{B_{\bar s}}), s}      \\
      = ~& \FINISH{\hat Q, \Delta, s} \\
      = ~& \FINISH{Q, \Delta, s}
    \end{align*}

  \item The cases for \nameref{stm_rt_call}, \nameref{stm_rt_dcall} and \nameref{stm_rt_fcall} are similar, since the reduct has the same form in all three cases.
    We shall only consider the case for \nameref{stm_rt_call}, where
    we have that $Q=\CALL{e_1}{m}{\VEC{e}}{e_2}; \hat Q$ and $Q''=S ; (\ENV{V}, \Delta) ; s ; \hat Q$.
    By induction and Definition~\ref{def:terminal_conf_types}, 
    \begin{align*}
      (\Delta',s') =~& \FINISH{S;(\ENV{V}, \Delta);s;\hat Q, \Delta'', s''} \\
      = ~& \FINISH{(\ENV{V}, \Delta);s;\hat Q, \Delta'', s''}   \\
      = ~& \FINISH{s;\hat Q, \Delta, s''} \\
      = ~& \FINISH{\hat Q, \Delta, s} \\
      = ~& \FINISH{Q, \Delta, s}
    \end{align*}

  \item If \nameref{stm_rt_delv} was used, then $Q=\DEL{x};Q''$, $\Delta = \Delta'', x:\TVAR{B_{\bar s}}$ and $s'' = s$. Hence, by induction and Definition~\ref{def:terminal_conf_types}, we have
    \begin{equation*}
    (\Delta',s') =  \FINISH{Q'', \Delta'', s''} =  \FINISH{Q, \Delta, s}
    \end{equation*}

  \item If \nameref{stm_rt_return} was used, then $Q= (\ENV{V}'', \Delta'');Q''$ and $s'' = s$.
    Like before, by induction and Definition~\ref{def:terminal_conf_types}, we have
    \begin{equation*}
    (\Delta',s') =  \FINISH{Q'', \Delta'', s''} =  \FINISH{Q, \Delta, s}
    \end{equation*}

  \item If \nameref{stm_rt_restore} was used, then $Q=s''; Q''$ and $\Delta'' = \Delta$.
    Again by induction and Definition~\ref{def:terminal_conf_types}, 
    \begin{equation*}
    (\Delta',s') =  \FINISH{Q'', \Delta'', s''} =  \FINISH{Q, \Delta, s}
    \qedhere
    \end{equation*}
\end{itemize}
\end{proof}


We now provide a few auxiliary lemmas for the proof of Theorem~\ref{theorem:stacks_runtime_preservation}.
The lemmas are shown by induction on the length of the type inference (Figure~\ref{fig:type_rules_env}), or of the consistency inference (Figure~\ref{fig:consistency_rules_envtsv}), or of the well-formedness inference (Figure~\ref{fig:wellformedness_stack}).
In the substitution lemmas (Lemmas~\ref{lemma:substitution_envv}--\ref{lemma:substitution_envs}) we incorporate consistency in addition to well-typedness. 
They could equally well be shown separately, but since we shall always require both properties to hold for the $\ENV{SV}$ we shall consider, we prefer to combine them to avoid repeating essentially the same premises.

\begin{lemma}[Weakening of $\Delta$]\label{lemma:weakening_delta}
If $\Sigma; \Gamma; \Delta \vdash \ENV{V}$, and $x \notin \DOM{\ENV{V}}$, then $\Sigma; \Gamma; \Delta, x:T \vdash \ENV{V}$.
\end{lemma}

\begin{lemma}[Strengthening of $\Delta$]\label{lemma:strengthening_delta}
If $\Sigma; \Gamma; \Delta, x:T \vdash \ENV{V}$, and $x \notin \DOM{\ENV{V}}$, then $\Sigma; \Gamma; \Delta \vdash \ENV{V}$.
\end{lemma}

\begin{lemma}[Substitution in $\ENV{V}$]\label{lemma:substitution_envv}
Assume that
\begin{itemize}
  \item $\Gamma \vdash \ENV{SV}$ (consistency)
  \item $\Sigma; \Gamma; \Delta, x:\TVAR{B_s} \vdash \ENV{V}$ (well-typedness)
  \item $x \in \DOM{\ENV{V}}$
  \item if $\Sigma; \Gamma; \EMPTYSET \vdash v : B_s$ and $v \in \ANAMES$, then $v \in \DOM{\ENV{S}}$.
\end{itemize} 
Then, $\Sigma; \Gamma; \Delta, x:\TVAR{B_s} \vdash \ENV{V}'$ and $\Gamma \vdash \ENV{S}, \ENV{V}'$, where $\ENV{V}' = \ENV{V}\REBIND{x}{v}$.
\end{lemma}

\begin{lemma}[Substitution in $\ENV{S}$]\label{lemma:substitution_envs}
Assume that
\begin{itemize}
  \item $\Gamma; \ENV{S} \vdash \ENV{S}$ (consistency)
  \item $\Sigma; \Gamma; \EMPTYSET \vdash \ENV{S}$ (well-typedness)
  \item $X \in \DOM{\ENV{S}}$
  \item $\Gamma(X) = (I,s'')$
  \item $\Gamma(I)(p) = \TVAR{B_s}$
  \item if $\Sigma; \Gamma; \EMPTYSET \vdash v : (B_s)$ and $v \in \ANAMES$, then $v \in \DOM{\ENV{S}}$.
\end{itemize}
Then, $\Sigma; \Gamma; \EMPTYSET \vdash \ENV{S}'$ and $\Gamma; \ENV{S}' \vdash \ENV{S}'$, where $\ENV{S}' = \ENV{S}\REBIND{X}{\ENV{F}\REBIND{p}{v}}$.
\end{lemma}

\begin{lemma}[Properties of well-formedness]\label{lemma:wellformedness_stacks_properties}
~
\begin{itemize}
  \item If $\Sigma; \Gamma; \ENV{SV}; s \vdash Q$ then $s \ordgeq \FIRSTS{Q}$. 
  \item If $\Sigma; \Gamma; \ENV{SV}; s \vdash Q$ and $s' \ordgeq s$ then $\Sigma; \Gamma; \ENV{SV}; s' \vdash s; Q$.
  \item If $\Sigma; \Gamma; \ENV{SV}; s \vdash Q$ and $x \in \DOM{\ENV{V}}$ then $\Sigma; \Gamma; \ENV{S}, \ENV{V}\REBIND{x}{v}; s \vdash Q$ for any $v$.
  \item If $\Sigma; \Gamma; \ENV{SV}; s \vdash Q$ and $X \in \DOM{\ENV{S}}$ and $p \in \DOM{\ENV{S}(X)}$ then $\Sigma; \Gamma; \ENV{S}\REBIND{X}{\ENV{F}\REBIND{p}{v}} , \ENV{V}; s \vdash Q$ for any $v$.
\end{itemize}
\end{lemma}

%
%

\label{proof:stacks_runtime_preservation}
\begin{proof}[Proof of Theorem~\ref{theorem:stacks_runtime_preservation}]
By case analysis of the rules used to conclude the transition (Figs.~\ref{fig:typed_semantics_stacks1}--\ref{fig:typed_semantics_stacks2}).

\bigskip
If \nameref{stm_rt_skip} was used, then the transition is of the form 
\begin{equation*}
         \Sigma; \Gamma; \Delta \ssemDash \CONF{\code{skip} ; Q, \ENV{TSV}} 
  \trans \Sigma; \Gamma; \Delta \ssemDash \CONF{Q, \ENV{TSV}}  
\end{equation*}
and the properties are immediately seen to hold as follows:
\begin{itemize}
  \item Well-typedness and Consistency hold by assumption.
  \item Well-formedness follows from the assumption 
$\Sigma; \Gamma; \ENV{SV}; s \vdash \code{skip}; Q$ and rule \nameref{wf_stack_stm}.
  \item Equivalence holds, since the two environments are equal, so they agree on \emph{all} fields, regardless of their security level.
\end{itemize}

\bigskip
If \nameref{stm_rt_if} was used, then the transition is of the form
$$
\Sigma; \Gamma; \Delta \ssemDash \CONF{\code{if $e$ then $S_{\TRUE}$ else $S_{\FALSE}$} ; Q, \ENV{TSV}}
  \trans \Sigma; \Gamma; \Delta \ssemDash[{s'}] \CONF{S_b ; s; Q, \ENV{TSV}}  
$$
and from the premise and side condition we have that 
\begin{align*}
                       s & \ordleq s' \\
  \Sigma; \Gamma; \Delta & \esemDash[(\TBOOL, s')] \CONF{e, \ENV{SV}} \trans b \in \BOOLEANS .
\end{align*}
The properties are now seen to hold as follows:
\begin{itemize}
  \item Well-typedness and Consistency hold by assumption.
  \item For Well-formedness, we know by assumption that
    \begin{equation*}
      \Sigma; \Gamma; \ENV{SV}; s \vdash \code{if $e$ then $S_{\TRUE}$ else $S_{\FALSE}$} ; Q
    \end{equation*}
    This was concluded by \nameref{wf_stack_stm}, and from the premise we have that $\Sigma; \Gamma; \ENV{SV}; s \vdash Q$ holds.
    From $s' \ordgeq s$ and Lemma~\ref{lemma:wellformedness_stacks_properties} we then have that $\Sigma; \Gamma; \ENV{SV}; s' \vdash s; Q$ holds as well.
    Then $\Sigma; \Gamma; \ENV{SV}; s' \vdash S_b; s; Q$ can be concluded by \nameref{wf_stack_stm}. 

  \item Equivalence holds, since the two environments are equal.
\end{itemize}

\bigskip
If \nameref{stm_rt_whiletrue} was used, then the transition is of the form
$$
\Sigma; \Gamma; \Delta \ssemDash \CONF{\code{while $e$ do $S$} ; Q, \ENV{TSV}}
      \trans \Sigma; \Gamma; \Delta \ssemDash[{s'}] \CONF{S ; s ; \code{while $e$ do $S$} ; Q, \ENV{TSV}} 
$$
 and from the side condition we have that $s \ordleq s'$.
The properties are now seen to hold as follows:
\begin{itemize}
  \item Well-typedness and Consistency hold by assumption.
  \item For Well-formedness, we know that,
    by assumption,
    \begin{equation*}
      \Sigma; \Gamma; \ENV{SV}; s \vdash \code{while $e$ do $S$} ; Q
    \end{equation*}
    From $s' \ordgeq s$ and Lemma~\ref{lemma:wellformedness_stacks_properties}, we then have that 
    \begin{equation*}
      \Sigma; \Gamma; \ENV{SV}; s' \vdash s; \code{while $e$ do $S$}; Q
    \end{equation*} 
    Then, $\Sigma; \Gamma; \ENV{SV}; s' \vdash S; s; \code{while $e$ do $S$}; Q$ can be concluded by \nameref{wf_stack_stm}. 

  \item Equivalence holds, since the two environments are equal.
\end{itemize}

\bigskip
If \nameref{stm_rt_whilefalse} was used, then the transition is of the form
\begin{equation*}
             \Sigma; \Gamma; \Delta \ssemDash \CONF{\code{while $e$ do $S$} ; Q, \ENV{TSV}}  
      \trans \Sigma; \Gamma; \Delta \ssemDash \CONF{Q, \ENV{TSV}} 
\end{equation*}
and the properties are now seen to hold as follows: 
\begin{itemize}
  \item Well-typedness and Consistency hold by assumption.
  \item For Well-formedness, we know by assumption that
    \begin{equation*}
      \Sigma; \Gamma; \ENV{SV}; s \vdash \code{while $e$ do $S$} ; Q
    \end{equation*}
    This was concluded by \nameref{wf_stack_stm}, and from the premise we have the desired $\Sigma; \Gamma; \ENV{SV}; s \vdash Q$.

  \item Equivalence holds, since the two environments are equal.
\end{itemize}

\bigskip
If \nameref{stm_rt_decv} was used, then the transition is of the form
\begin{align*}
& \Sigma; \Gamma; \Delta \ssemDash \CONF{\code{$\TVAR{B, s'}$ $x$ := $e$ in $S$} ; Q, \ENV{TSV}}   \\
      \trans ~& \Sigma; \Gamma; \Delta, x:\TVAR{B, s'} \ssemDash \CONF{S ; \DEL{x} ; Q, \ENV{TS}; \ENV{V}, x:v} 
\end{align*}
and from the premise and side condition we know that 
\begin{align*}
                       x & \notin \DOM{\ENV{V}} \\
                       x & \notin \DOM{\Delta}  \\
  \Sigma; \Gamma; \Delta & \esemDash[(B, s')] \CONF{e, \ENV{SV}} \trans v .
\end{align*}
We now show how each of the properties hold, in turn:
\begin{itemize}
  \item Well-typedness:
    The goal is to show $\Sigma; \Gamma; \Delta, x:\TVAR{B, s'} \vdash \ENV{S}; \ENV{V}, x:v$.
    We know that $\Sigma; \Gamma; \Delta \vdash \ENV{S}$ holds by assumption.
    From the premise and side condition of \nameref{stm_rt_decv}, we can then conclude as follows:
    \begin{align*}
      \Sigma; \Gamma; \Delta                 & \vdash v : (B, s')           & \text{by Theorem~\ref{theorem:expressions_runtime_safety},} \\
      \Sigma; \Gamma; \Delta, x:\TVAR{B, s'} & \vdash \ENV{V}               & \text{by Lemma~\ref{lemma:weakening_delta},}                \\
      \Sigma; \Gamma; \Delta, x:\TVAR{B, s'} & \vdash \ENV{V}, x:v          & \text{by rule \nameref{env_t_envv_u},}                      \\
      \Sigma; \Gamma; \Delta, x:\TVAR{B, s'} & \vdash \ENV{S}; \ENV{V}, x:v & \text{by rule \nameref{env_t_envsv}.}
    \end{align*}

  \item Consistency:
    The goal is to show $\Gamma \vdash \ENV{S}; \ENV{V}, x:v$.
    We know that $\Gamma \vdash \ENV{SV}$ holds by assumption.
    From the premise and side condition of \nameref{stm_rt_decv}, we can then conclude as follows:
    \begin{align*}
      \Sigma; \Gamma; \Delta & \vdash v : (B, s')           & \text{by Theorem~\ref{theorem:expressions_runtime_safety},} \\
      v \in \ANAMES          & \implies v \in \DOM{\ENV{S}} & \text{by Corollary~\ref{corollary:expressions_consistency_values},} \\
      \ENV{S}                & \vdash \ENV{V}, x:v          & \text{by rule \nameref{env_c_envv_rec},} \\ 
      \Gamma                 & \vdash \ENV{S}; \ENV{V}, x:v & \text{by rule \nameref{env_c_envsv}.}
    \end{align*}


  \item Well-formedness:
    The goal is to show $\Sigma; \Gamma; \ENV{S}; \ENV{V}, x:v; s \vdash S ; \DEL{x} ; Q$.
    By assumption, we know that 
    \begin{equation*}
      \Sigma; \Gamma; \ENV{SV}, s \vdash \code{$\TVAR{B, s'}$ $x$ := $e$ in $S$} ; Q
    \end{equation*}
    which was concluded by by \nameref{wf_stack_stm}.
    From the premise of that rule, we have that $\Sigma; \Gamma; \ENV{SV}; s \vdash Q$ holds.
    We can then conclude as follows:
    \begin{align*}
      \Sigma; \Gamma; \ENV{S}; \ENV{V}, x:v; s & \vdash \DEL{x} ; Q    & \text{by rule \nameref{wf_stack_del},} \\
      \Sigma; \Gamma; \ENV{S}; \ENV{V}, x:v; s & \vdash S; \DEL{x} ; Q & \text{by rule \nameref{wf_stack_stm}.} 
    \end{align*}

  \item Equivalence holds, since the two environments are equal. 
\end{itemize}

\bigskip
If \nameref{stm_rt_assv} was used, then the transition is of the form
\begin{equation*}
             \Sigma; \Gamma; \Delta \ssemDash \CONF{\code{$x$ := $e$} ; Q, \ENV{TSV}} 
      \trans \Sigma; \Gamma; \Delta \ssemDash \CONF{Q, \ENV{TS}; \ENV{V}\REBIND{x}{v}} 
\end{equation*}
 and from the premise and side condition of \nameref{stm_rt_assv} we know that 
\begin{align*}
      x                  & \in \DOM{\ENV{V}} \\
      \Delta(x)          & = \TVAR{B, s'}    \\
      s                  & \ordleq s'        \\
  \Sigma; \Gamma; \Delta & \esemDash[(B, s')] \CONF{e, \ENV{SV}} \trans v .
\end{align*}
 We now show how each of the properties holds, in turn:
\begin{itemize}
  \item Well-typedness and consistency:
    The goal is to show 
	$$
	\Sigma; \Gamma; \Delta \vdash \ENV{S}; \ENV{V}\REBIND{x}{v}
	\qquad
	\Gamma \vdash \ENV{S}; \ENV{V}\REBIND{x}{v}
	$$
    We have the following:
    \begin{align*}
      \Sigma; \Gamma; \EMPTYSET & \vdash v : (B, s')           & \text{from premises of \nameref{stm_rt_assv} and Thm.~\ref{theorem:expressions_runtime_safety},} \\
      v \in \ANAMES             & \implies v \in \DOM{\ENV{S}} & \text{by Corollary~\ref{corollary:expressions_consistency_values}.}
    \end{align*}
     By Lemma~\ref{lemma:substitution_envv} we can then conclude 

  \item Well-formedness:
    The goal is to show $\Sigma; \Gamma; \ENV{S}; \ENV{V}\REBIND{x}{v}; s \vdash Q$.
    We know that 
    \begin{equation*}
      \Sigma; \Gamma; \ENV{S}; \ENV{V}; s \vdash \code{$x$ := $e$} ; Q
    \end{equation*}
 which was concluded by \nameref{wf_stack_stm}, and from the premise of this rule we have that $\Sigma; \Gamma; \ENV{S}; \ENV{V}; s \vdash Q$.
We can conclud by Lemma~\ref{lemma:wellformedness_stacks_properties}.

  \item Equivalence holds, since the two environments are equal.
\end{itemize}

\bigskip
If \nameref{stm_rt_assf} was used, then the transition is of the form
$$
\quad \Sigma; \Gamma; \Delta \ssemDash \CONF{\code{this.$p$ := $e$} ; Q, \ENV{TSV}} 
      \trans \Sigma; \Gamma; \Delta \ssemDash \CONF{Q, \ENV{T}; \ENV{S}\REBIND{X}{\ENV{F}\REBIND{p}{v}}; \ENV{V}} 
$$
 and from the premises and side condition of \nameref{stm_rt_assf} we know that
\begin{align*}
  X                      & = \ENV{V}(\code{this}) \\
  \ENV{F}                & = \ENV{S}(X)           \\
  p                      & \in \DOM{\ENV{F}}      \\
  \TVAR{I,s_1}           & = \Delta(\code{this})  \\
  \TVAR{B, s'}           & = \Gamma(I)(p)         \\
  \Sigma; \Gamma; \Delta & \esemDash[(B, s')] \CONF{e, \ENV{SV}} \trans v \\ 
  s_1                    & \ordleq s'             \\
  s                      & \ordleq s' .
\end{align*}
 We now show how each of the properties holds, in turn:
\begin{itemize}
  \item Well-typedness and consistency:
    The goals are to show 
	$$
      \Sigma; \Gamma; \Delta \vdash \ENV{S}\REBIND{X}{\ENV{F}\REBIND{p}{v}}; \ENV{V}
      \qquad
      \Gamma \vdash \ENV{S}\REBIND{X}{\ENV{F}\REBIND{p}{v}}; \ENV{V}
    $$
    From $\ENV{V}(\code{this}) = X$ and consistency of $\ENV{SV}$ we infer that $X \in \DOM{\ENV{S}}$.
    From the premises and Theorem~\ref{theorem:expressions_runtime_safety} we know that $\Sigma; \Gamma; \EMPTYSET \vdash v : (B, s')$, and by Corollary~\ref{corollary:expressions_consistency_values} that $v \in \ANAMES \implies v \in \DOM{\ENV{S}}$.
    By Lemma~\ref{lemma:substitution_envs} we can then conclude.

  \item Well-formedness:
    The goal is to show 
    \begin{equation*}
      \Sigma; \Gamma; \ENV{S}\REBIND{X}{\ENV{F}\REBIND{p}{v}}; \ENV{V}; s \vdash Q .
    \end{equation*}
    
    \noindent We know that 
    \begin{equation*}
      \Sigma; \Gamma; \ENV{S}; \ENV{V}; s \vdash \code{this.$p$ := $e$} ; Q 
    \end{equation*}

    \noindent which was concluded by \nameref{wf_stack_stm}, and from the premise of this rule we have that $\Sigma; \Gamma; \ENV{S}; \ENV{V}; s \vdash Q$.
    Then, we conclude by Lemma~\ref{lemma:wellformedness_stacks_properties}.

  \item Equivalence:
    The goal is to show 
    \begin{equation*}
      \forall s'' \SUCHTHAT s \not\ordleq s'' \implies \Gamma \vdash \ENV{S} =_{s''} \ENV{S}\REBIND{X}{\ENV{F}\REBIND{p}{v}} .
    \end{equation*}

    \noindent The only difference between the two environments is in the value stored at $p$, and from the premise and side conditions we know that $\Gamma(I)(p) = \TVAR{B, s'}$ and $s \ordleq s'$.
    Hence, the field $p$ is not amongst those for which the equivalence must hold.
\end{itemize}

\bigskip
Suppose one of \nameref{stm_rt_call}, \nameref{stm_rt_dcall}, \nameref{stm_rt_fcall} was used. 
We consider only the case for \nameref{stm_rt_call}, since the two other cases are similar or simpler.
If \nameref{stm_rt_call} was used, then the transition is 
$$
\Sigma; \Gamma; \Delta \ssemDash \CONF{ \CALL{e_1}{m}{\VEC{e}}{e_2} ; Q, \ENV{TSV}}   
      \trans \Sigma; \Gamma; \Delta' \ssemDash[{s'}] \CONF{S ; (\ENV{V}, \Delta) ; s ; Q, \ENV{T}; \ENV{SV}'} 
$$ 
where $S$ is the method body; $\ENV{S}'$ is modified because of the update of the \code{balance} fields; and $\ENV{V}'$ is the new variable environment containing bindings for the formal parameters of the method.
As there are many side conditions, we shall only mention those that are relevant for each property below.
We now show how each of the properties hold, in turn:

\begin{itemize}
  \item Well-typedness and consistency:
    The goal is to show: 
	$$
	\Sigma; \Gamma; \Delta' \vdash \ENV{SV}' 
	\qquad
      	\Gamma \vdash \ENV{SV}'
	$$
    We treat each environment separately:
    \begin{itemize}
      \item From the premise of \nameref{stm_rt_call}, we have that 
        \begin{equation*}
          \Sigma; \Gamma; \Delta \esemDash[(\TINT, s')] \CONF{e_2, \ENV{SV}} \trans z
        \end{equation*} 
 and by Theorem~\ref{theorem:expressions_runtime_safety}, we know that $z$ is an integer and all variables read in the evaluation are of level $s$ or lower.
        From the side condition $s' \ordleq s_3, s_4$, we have that the levels of the two \code{balance} fields, $s_3, s_4$ are higher than, or equal to $z$.
        From another side condition, we have that 
        \begin{equation*}
          \ENV{S}'  = \ENV{S}\REBIND{X}{\ENV{F}^X [\code{balance -= } z]}\REBIND{Y}{\ENV{F}^Y [\code{balance += } z]} 
        \end{equation*}
	 where $X$ is obtained from \code{this}, and $Y$ is obtained from $e_1$.
        This update is thus equivalent to executing the statements
        \begin{align*}
          \code{this.balance}  & \code{ := this.balance - $e_2$}  \\ 
          \code{$e_1$.balance} & \code{ := $e_1$.balance + $e_2$}
        \end{align*}
	 at level $s'$.
        We see that the transition in both cases can be concluded by rule \nameref{stm_rt_assf}.
        Hence, well-typedness and consistency of $\ENV{S}'$ holds by the case for \nameref{stm_rt_assf} above.
        
      \item From the side condition $\ENV{V}(\code{this}) = X$, and the (expanded) premise
        \begin{align*} 
          \Sigma; \Gamma; \Delta & \esemDash[(I^Y, s')]   \CONF{e_1, \ENV{SV}}  \trans Y   \\
          \Sigma; \Gamma; \Delta & \esemDash[(\TINT, s')] \CONF{e_2, \ENV{SV}}  \trans z   \\
          \Sigma; \Gamma; \Delta & \esemDash[(B_1, s_1')] \CONF{e_1', \ENV{SV}} \trans v_1 \\ 
                                 & \vdots                                                  \\
          \Sigma; \Gamma; \Delta & \esemDash[(B_h, s_h')] \CONF{e_h', \ENV{SV}} \trans v_h 
        \end{align*}

        \noindent and Theorem~\ref{theorem:expressions_runtime_safety}, we know that the bindings of the formal parameters in the new $\ENV{V}'$, and the type bindings in the new $\Delta'$ will correspond, as in the side condition we have 
        \begin{align*}
          \ENV{V}' & = \code{this}:Y, \code{sender}:X, \code{value}:z, x_1:v_1, \ldots ,x_h:v_h                   \\
          \Delta'  & = \code{this}:\TVAR{I^Y, s_2}, \code{sender}:\TVAR{I^X, s_1}, \code{value}:\TVAR{\TINT, s'}, \\ 
                   & \qquad x_1:\TVAR{B_1, s_1'}, \ldots, x_h:\TVAR{B_h, s_h'} 
        \end{align*}

        \noindent and the types of \code{this} and \code{sender} are obtained from the actual types of $X$ and $Y$ in $\Gamma$.
        Hence $\Sigma; \Gamma; \Delta' \vdash \ENV{V}'$ (well-typedness) holds.

        \item We know directly from the side conditions that $X$ and $Y$, which appear as values in $\ENV{V}$, have types in $\Gamma$; so, by the assumption of consistency of $\ENV{S}$, we also know that $X$ and $Y$ must appear in the domain of $\ENV{S}$ and therefore also in the domain of $\ENV{S}'$, since the domain is static. 
        Thus, consistency for these entries holds.

       \item For the values $v_1, \ldots, v_h$, we use Corollary~\ref{corollary:expressions_consistency_values} to conclude that 
        \begin{equation*}
          v_i \in \ANAMES \implies v_i \in \DOM{\ENV{S}}
        \end{equation*}
	 which then again implies $v_i \in \DOM{\ENV{S}'}$, since the domain is static.
        Hence, $\ENV{S}' \vdash \ENV{V}'$ (consistency) holds as well.
    \end{itemize}

    By combining the four results above, we prove well-typedness and consistency of $\ENV{SV}'$.

  \item Well-formedness:
    The goal is to show $\Sigma; \Gamma; \ENV{SV}'; s' \vdash S ; (\ENV{V}, \Delta) ; s ; Q$.
    By assumption, 
    \begin{equation*}
      \Sigma; \Gamma; \ENV{SV}, s \vdash \CALL{e_1}{m}{\VEC{e}}{e_2} ; Q
    \end{equation*}

    \noindent holds, which was concluded by \nameref{wf_stack_stm}.
    Let 
    \begin{equation*}
      \ENV{S}' = \ENV{S}\REBIND{X}{\ENV{F}^X [\code{balance -= } z]}\REBIND{Y}{\ENV{F}^Y [\code{balance += } z]} 
    \end{equation*}

    \noindent for ease of notation.
    We know that $s \ordleq s'$ from the side condition of \nameref{stm_rt_call}.
    We then conclude the following:
    \begin{align*} 
      \Sigma; \Gamma; \ENV{SV}, s           & \vdash Q    & \text{from premise of \nameref{wf_stack_stm},}                \\
      \Sigma; \Gamma; \ENV{S}'; \ENV{V}, s  & \vdash Q    & \text{by Lemma~\ref{lemma:wellformedness_stacks_properties},} \\
      \Sigma; \Gamma; \ENV{S}'; \ENV{V}, s' & \vdash s; Q & \text{by rule \nameref{wf_stack_restore}.}
    \end{align*}

    \noindent We also know the following:
    \begin{align*}
      \Sigma; \Gamma; \Delta & \vdash \ENV{V} & \text{by assumption} \\
      \ENV{S}'               & \vdash \ENV{V} & \text{by assumption, since $\DOM{\ENV{S}} = \DOM{\ENV{S}'}$.}
    \end{align*}

    \noindent Using this and the above, by \nameref{wf_stack_return} we can then conclude 
    \begin{equation*}
      \Sigma; \Gamma; \ENV{SV}'; s' \vdash S ; (\ENV{V}, \Delta) ; s ; Q
    \end{equation*}

  \item Equality:
    The goal is to show $\forall s'' \SUCHTHAT s \not\ordleq s'' \implies \Gamma \vdash \ENV{S} =_{s''} \ENV{S}'$.

    We know that the only difference between $\ENV{S}$ and $\ENV{S}'$ is that the fields \code{$X$.balance} and \code{$Y$.balance} may have been modified.
    From the side conditions we have that 
    \begin{align*}
      \TVAR{\TINT, s_3} & = \Gamma(I^X)(\code{balance}) \\
      \TVAR{\TINT, s_4} & = \Gamma(I^Y)(\code{balance}) \\ 
      s \ordleq s'      & \ordleq s_3, s_4     
    \end{align*}

    \noindent so by transitivity $s \ordleq s_3, s_4$.
    As neither field is strictly lower than, or unrelated to, $s$, we therefore have that $\forall s'' \SUCHTHAT s \not\ordleq s'' \implies \Gamma \vdash \ENV{S} =_{s''} \ENV{S}'$ holds.
\end{itemize}


\bigskip
If \nameref{stm_rt_delv} was used, then the transition is of the form 
\begin{equation*}
             \Sigma; \Gamma; \Delta, x:\TVAR{B_{s'}} \ssemDash \CONF{\DEL{x} ; Q, \ENV{TS}; \ENV{V}, x:v } 
      \trans \Sigma; \Gamma; \Delta \ssemDash \CONF{Q, \ENV{TSV}} 
\end{equation*}
 We now show how each of the properties holds, in turn:
\begin{itemize}
  \item Well-typedness:
    The goal is to show $\Sigma; \Gamma; \Delta \vdash \ENV{SV}$.
    We have that 
    \begin{align*}
      \Sigma; \Gamma; \EMPTYSET               & \vdash \ENV{S}      & \text{by assumption,} \\ 
      \Sigma; \Gamma; \Delta, x:\TVAR{B_{s'}} & \vdash \ENV{V}, x:v & \text{by assumption,} \\
      \Sigma; \Gamma; \Delta, x:\TVAR{B_{s'}} & \vdash \ENV{V}      & \text{from premise of \nameref{env_t_envv_u},} \\
      \Sigma; \Gamma; \Delta                  & \vdash \ENV{V}      & \text{by Lemma~\ref{lemma:strengthening_delta},} \\
      \Sigma; \Gamma; \Delta                  & \vdash \ENV{SV}     & \text{by \nameref{env_t_envsv}.}
    \end{align*}

  \item Consistency:
    The goal is to show $\Gamma \vdash \ENV{SV}$. This holds because
    \begin{align*}
      \Gamma; \ENV{S} & \vdash \ENV{S}      & \text{by assumption,} \\ 
      \Gamma; \ENV{S} & \vdash \ENV{V}, x:v & \text{by assumption,} \\ 
      \Gamma; \ENV{S} & \vdash \ENV{V}      & \text{from premise of \nameref{env_c_envv_rec}.} \\ 
      \Gamma          & \vdash \ENV{SV}     & \text{by \nameref{env_c_envsv}.}
    \end{align*}

  \item Well-formedness:
    The goal is to show $\Sigma; \Gamma; \ENV{SV}; s \vdash Q$. This holds because
    \begin{align*}
      \Sigma; \Gamma; \ENV{S}; \ENV{V}, x:v s & \vdash \DEL{x}; Q & \text{by assumption,} \\
      \Sigma; \Gamma; \ENV{S}; \ENV{V}, s     & \vdash Q          & \text{from premise of \nameref{wf_stack_del}.}
    \end{align*}

  \item Equivalence holds, since the two environments are equal.
\end{itemize}

\bigskip 
If \nameref{stm_rt_return} was used, then the transition is of the form
\begin{equation*}
             \Sigma; \Gamma; \Delta \ssemDash \CONF{(\ENV{V}', \Delta') ; Q, \ENV{TSV}} 
      \trans \Sigma; \Gamma; \Delta' \ssemDash \CONF{Q, \ENV{TS}; \ENV{V}'} . 
\end{equation*}

\noindent We now show how each of the properties holds, in turn:
\begin{itemize}
  \item Well-typedness:
    The goal is to show $\Sigma; \Gamma; \Delta \vdash \ENV{S}; \ENV{V}'$.
    We have that
    \begin{align*}
      \Sigma; \Gamma; \EMPTYSET   & \vdash \ENV{S}                 & \text{by assumption,} \\ 
      \Sigma; \Gamma; \ENV{SV}; s & \vdash (\ENV{V}', \Delta'); Q  & \text{by assumption,} \\ 
      \Sigma; \Gamma; \Delta'     & \vdash \ENV{V}'                & \text{from premise of \nameref{wf_stack_return},} \\ 
      \Sigma; \Gamma; \Delta'     & \vdash \ENV{S}; \ENV{V}'       & \text{by \nameref{env_t_envsv}.}
    \end{align*}

  \item Consistency:
    The goal is to show $\Gamma \vdash \ENV{SV}$.
    We have that
    \begin{align*}
      \Gamma; \ENV{S}             & \vdash \ENV{S}                 & \text{by assumption,} \\ 
      \Sigma; \Gamma; \ENV{SV}; s & \vdash (\ENV{V}', \Delta'); Q  & \text{by assumption,} \\ 
      \Gamma; \ENV{S}             & \vdash \ENV{V}'                & \text{from premise of \nameref{wf_stack_return},} \\
      \Gamma                      & \vdash \ENV{S}; \ENV{V}'       & \text{by \nameref{env_c_envsv}.}
    \end{align*}

  \item Well-formedness:
    The goal is to show $\Sigma; \Gamma; \ENV{S}; s \vdash Q$.
    We have that
    \begin{align*}
      \Sigma; \Gamma; \ENV{SV}; s           & \vdash (\ENV{V}', \Delta'); Q & \text{by assumption,} \\
      \Sigma; \Gamma; \ENV{S}; \ENV{V}'; s  & \vdash Q                      & \text{from premise of \nameref{wf_stack_return}.} 
    \end{align*}

  \item Equivalence holds, since the two environments are equal.
\end{itemize}

\bigskip 
If \nameref{stm_rt_restore} was used, then the transition is of the form
\begin{equation*}
             \Sigma; \Gamma; \Delta \ssemDash \CONF{s'; Q, \ENV{TSV}} 
      \trans \Sigma; \Gamma; \Delta \ssemDash[{s'}] \CONF{Q, \ENV{TSV}}
\end{equation*}

\noindent and we know from the side condition that $s \ordgeq s'$.
We now show how each of the properties holds:
\begin{itemize}
  \item Well-typedness and Consistency hold by assumption.

  \item Well-formedness:
    The goal is to show $\Sigma; \Gamma; \ENV{SV}; s' \vdash Q$.
    We have that
    \begin{align*}
      \Sigma; \Gamma; \ENV{SV}; s  & \vdash s'; Q & \text{by assumption,} \\
      \Sigma; \Gamma; \ENV{SV}; s' & \vdash Q     & \text{from premise of \nameref{wf_stack_restore}.} 
    \end{align*}

  \item Equivalence holds, since the two environments are equal.
  \qedhere
\end{itemize}
\end{proof}

\label{proof:stacks_runtime_noninterference}
\begin{proof}[Proof of Theorem~\ref{theorem:stacks_runtime_noninterference}]
We first consider the two field environments.
We know that 
\begin{equation*}
  \Gamma \vdash \ENV{S}^1 =_s \ENV{S}^2 .
\end{equation*}
By two applications of Theorem~\ref{theorem:stacks_runtime_preservation}, we have that, for all $s''$ such that $s \not\ordleq s''$:
$$
\Gamma \vdash \ENV{S}^1 =_{s''} \ENV{S}^{1'} 
\qquad
\Gamma \vdash \ENV{S}^2 =_{s''} \ENV{S}^{2'}
$$ 
So, both environments agree before and after the transition on all fields that are \emph{strictly lower} than, or unrelated to, $s$.
Hence, we also have that
\begin{equation*}
  \forall s'' \SUCHTHAT s \not\ordleq s'' \implies \Gamma \vdash \ENV{S}^{1'} =_{s''} \ENV{S}^{2'}
\end{equation*}

Suppose now that the two transitions modified a field that is exactly at level $s$ (the modified field must be the same in both transitions, since $Q$ is the same).
From the semantics (Figures~\ref{fig:typed_semantics_stacks1}--\ref{fig:typed_semantics_stacks2}) we see that there are only three rules that could have been used to conclude a transition that modifies $\ENV{S}$: \nameref{stm_rt_assf}, \nameref{stm_rt_call} and \nameref{stm_rt_fcall}.
The two call rules can modify the \code{balance} fields of the caller and callee, and the assignment rule can modify any other field.
We consider only \nameref{stm_rt_assf}, but the argument is the same also for the two other rules.
The transitions are of the form 
\begin{equation*}
  \begin{split}
              & \Sigma; \Gamma; \Delta \ssemDash \CONF{\code{this.$p$ := $e$} ; Q, \ENV{T}; \ENV{SV}^i} \\
      \trans ~& \Sigma; \Gamma; \Delta \ssemDash \CONF{Q, \ENV{T}; \ENV{S}^i\REBIND{X}{\ENV{F}\REBIND{p}{v_i}}; \ENV{V}^i} 
  \end{split}
\end{equation*}
 for $i \in \SET{1, 2}$, and from the premise of that rule we have that 
\begin{equation*}
  \Sigma; \Gamma; \Delta \esemDash[(B, s')] \CONF{e, \ENV{SV}^i} \trans v_i .
\end{equation*}
 By Corollary~\ref{corollary:expressions_runtime_noninterference}, if
$$
  \Sigma; \Gamma; \Delta \esemDash[(B, s')] \CONF{e, \ENV{SV}^1} \trans v_1 
  \qquad
  \Sigma; \Gamma; \Delta \esemDash[(B, s')] \CONF{e, \ENV{SV}^2} \trans v_2
$$
 then $v_1 = v_2$; so, the expression in the premise yields the same value $v$ in both executions.
Hence, the field $p$ at level $s$ may be changed in the two executions, but it will be changed to the \emph{same value}.
Thus $\ENV{S}^{1'}$ and $\ENV{S}^{2'}$ also agree at level $s$.

For $\ENV{V}^i$, we do not have a result similar to Theorem~\ref{theorem:stacks_runtime_preservation}, so we must examine all the rules that could have been used to conclude the transition:
\begin{itemize}
  \item If one of the rules \nameref{stm_rt_skip}, \nameref{stm_rt_if}, \nameref{stm_rt_whiletrue}, \nameref{stm_rt_whilefalse}, \nameref{stm_rt_assf}, or \nameref{stm_rt_restore} was used, then the result is immediate, as the transition does not modify $\ENV{V}^i$.

  \item If \nameref{stm_rt_delv} was used, the transitions are of the form 
    \begin{equation*}
      \begin{split}
              & \Sigma; \Gamma; \Delta, x:\TVAR{B_{s'}} \ssemDash \CONF{\DEL{x} ; Q, \ENV{T}; \ENV{S}^i; \ENV{V}^i, x:v_i } \\
      \trans ~& \Sigma; \Gamma; \Delta \ssemDash \CONF{Q, \ENV{T}; \ENV{S}^i; \ENV{V}^i} 
      \end{split}
    \end{equation*}
 So, for both $\ENV{V}^1$ and $\ENV{V}^2$ we are removing the entry for the variable $x$.
    Equality of the remaining entries, for level $s$ and lower, then follows from the initial assumption of $s$-equality of the environments.

  \item If \nameref{stm_rt_return} was used, the transitions are of the form 
	$$ \Sigma; \Gamma; \Delta \ssemDash \CONF{(\ENV{V}', \Delta') ; Q, \ENV{T}; \ENV{SV}^i} 
        \trans \Sigma; \Gamma; \Delta' \ssemDash \CONF{Q, \ENV{T}; \ENV{S}^i; \ENV{V}'} 
	$$ 
	i.e.\@ the same variable environment $\ENV{V}'$ is taken off the stack in both transitions, so they obviously agree since they are equal.

  \item If one of \nameref{stm_rt_decv} or \nameref{stm_rt_assv} was used, then the transition either extends the variable environment with a new entry, or modifies an existing entry.
    In either case, the argument is essentially the same.
    In the premise, we have an expression evaluation of the form 
    \begin{equation*}
      \Sigma; \Gamma; \Delta \esemDash[{B_{s'}}] \CONF{e, \ENV{SV}^i} \trans v_i .
    \end{equation*}

    \noindent where $s'$ is the level of the variable being assigned to.
    There are now two cases:
    \begin{enumerate}
      \item If $s' \not\ordleq s$, i.e.\@ $s'$ is \emph{strictly higher} than $s$, (or, in the case of \nameref{stm_rt_decv}, unrelated to $s$), then the result is immediate, since only variables of level $s$ or lower must agree in the two environments for the result to hold.

      \item If $s' \ordleq s$, we have by Corollary~\ref{corollary:expressions_runtime_noninterference} that if 
      $$
        \Sigma; \Gamma; \Delta \esemDash[{B_{s'}}] \CONF{e, \ENV{SV}^1} \trans v_1 
        \qquad
        \Sigma; \Gamma; \Delta \esemDash[{B_{s'}}] \CONF{e, \ENV{SV}^2} \trans v_2
      $$
      then $v_1 = v_2$; so, the expression in the premise yields the same value $v$ in both executions.
      Hence, the two environments will contain the same value $v$ for the variable $x$.
      Equality for all variables of level $s$ or lower then follows from the initial assumption of $s$-equality of the two environments.
    \end{enumerate}

  \item Finally, if one of the call rules \nameref{stm_rt_call}, \nameref{stm_rt_dcall}, or \nameref{stm_rt_fcall} was used, then the $\ENV{V}^{i'}$ environments will be two \emph{new} environments created in the transition, to contain the bindings for the formal parameters $\VEC{x}$ of the method call and for the magic variables.

    For each variable binding, the argument is then exactly the same as for \nameref{stm_rt_decv} and \nameref{stm_rt_assv} above:
    Either the variable is of a level strictly higher than, or unrelated to, $s$, in which case its value is allowed to differ in the two executions; or the variable is of level $s$ or lower, in which case the corresponding expression also must be evaluated at level $s$ or lower, and then we obtain the same value $v$ in both executions by Corollary~\ref{corollary:expressions_runtime_noninterference}.
    Thus, $s$-equality of the two environments is again ensured. \qedhere
\end{itemize}
\end{proof}

\clearpage
\section{Semantic typing of stacks and commands}\label{app:semantic_typing_stacks}
Here we provide the missing details and proofs for the results in Section~\ref{sec:typing_interpretation_stm_stack}.

\begin{proof}[Proof of Lemma~\ref{lemma:untypable}]
By induction on the length of $\ttrans^*$. The base case is an immediate consequence of Definition~\ref{def:typing_interpretation_stm}. For the inductive step, let us consider $\Sigma; \Gamma; \ENV{T} \vDash (Q, \Delta, s) \ttrans (Q'', \Delta'',s'') \ttrans^* (Q', \Delta', s')$; by induction, we know that $(Q'', \Delta'', s'') \not\in \STYPING$ and so, again by Definition~\ref{def:typing_interpretation_stm}, it cannot be that $(Q, \Delta, s) \in \STYPING$.
\end{proof}

In some of the following proofs, we shall need to concatenate two stacks, that is, take a stack $Q_1$, remove the tailing bottom symbol $\bot$, and then append the remainder on top of another stack $Q_2$. Thus, we write $q$ for a finite sequence of stack symbols \emph{except} $\bot$; 
i.e.
\begin{center}
\begin{syntax}[h]
  q \IS \epsilon 
    \OR S; q                 
    \OR \DEL{x}; q
    \OR (\ENV{V}, \Delta); q
    \OR s; q
\end{syntax}
\end{center}
So, we can alternatively write a stack $Q$ as $q; \bot$.

In the following proofs, we shall frequently use function $\FIRSTS\cdot$, defined in Definition~\ref{def:terminal_conf_types}.

\begin{proof}[Proof of Lemma~\ref{lemma:compatibility_stacks_coercion}]
We only prove the first claim of the Lemma, since the second one is an immediate corollary of the first. The proof is by coinduction.
From $\Sigma; \Gamma; \Delta; \ENV{T} \vDash Q : \TCMD{s_1}$ we know there exists a typing interpretation $\STYPI$ containing $(Q, \Delta, s_1)$.
We create a candidate typing interpretation $\CANDSTYPI$ containing $(Q, \Delta, s_2)$ thus:
\begin{equation*}
  \CANDSTYPI \DEFSYM \SET{ (Q', \Delta', s_2') | (Q', \Delta', s_1') \in \STYPI \land s_1' \ordgeq s_2' \ordgeq \FIRSTS{Q'} }
\end{equation*}

\noindent and clearly, since $(Q, \Delta, s_1) \in \STYPI$ and $s_1 \ordgeq s_2 \ordgeq \FIRSTS{Q}$, we also have that $(Q, \Delta, s_2) \in \CANDSTYPI$.

We must now show that $\CANDSTYPI$ indeed is a typing interpretation according to Definition~\ref{def:typing_interpretation_stm}.
Pick any $(Q', \Delta', s_2') \in \CANDSTYPI$;
there are now two cases:
\begin{enumerate}
  \item If $Q' \in \TSTACKS$, then it satisfies Case~\ref{case:typi1} of Definition~\ref{def:typing_interpretation_stm}.

  \item Otherwise, we know that $(Q', \Delta', s_1') \in \STYPI$ and, for all appropriately shaped $\ENV{SV}$, there exists at least one transition 
    \begin{equation*} 
      \Sigma; \Gamma; \Delta' \ssemDash[{s_1'}] \CONF{Q', \ENV{TSV}} \trans \Sigma; \Gamma; \Delta'' \ssemDash[{s_1''}] \CONF{Q'', \ENV{T}; \ENV{SV}'}
    \end{equation*}

    \noindent such that $(Q'', \Delta'', s_1'') \in \STYPI$. 
    By Proposition~\ref{theorem:stacks_coercion_property}, we then have that the transition 
    \begin{equation*} 
      \Sigma; \Gamma; \Delta' \ssemDash[{s_2'}] \CONF{Q', \ENV{TSV}} \trans \Sigma; \Gamma; \Delta'' \ssemDash[{s_2''}] \CONF{Q'', \ENV{T}; \ENV{SV}'}
    \end{equation*}
 can be concluded for any $s_2'$ such that $s_1' \ordgeq s_2' \ordgeq \FIRSTS{Q'}$, and by construction of $\CANDSTYPI$, we also have that $(Q'', \Delta'', s_2'') \in \CANDSTYPI$.
\qedhere
\end{enumerate}
\end{proof}

\begin{figure}\centering
\begin{semantics}
  \RULE[st-bot][stack_st_bot]
    { }
    { \Sigma; \Gamma; \Delta; \ENV{T} \vDash \bot : \TCMD{s} }
  \RULE[st-stm][stack_st_stm]
    { 
      \Sigma; \Gamma; \Delta; \ENV{T} \vDash S : \TCMD{s} \AND 
      \Sigma; \Gamma; \Delta; \ENV{T} \vDash Q : \TCMD{s} 
    }
    { \Sigma; \Gamma; \Delta \vdash S; Q : \TCMD{s} }
  \RULE[st-del][stack_st_del]
    { \Sigma; \Gamma, \Delta; \ENV{T} \vDash Q : \TCMD{s} }
    { \Sigma; \Gamma; \Delta, x:\TVAR{B_{s'}}; \ENV{T} \vDash \DEL{x}; Q : \TCMD{s} }
  \RULE[st-ret][stack_st_return]
    { 
      \Sigma; \Gamma; \Delta' \vdash \ENV{V} \AND 
      \Sigma; \Gamma; \Delta' \vDash Q : \TCMD{s} }
    { \Sigma; \Gamma; \Delta  \vDash (\ENV{V}', \Delta'); Q : \TCMD{s} }
  \RULE[st-res][stack_st_restore]( s \ordgeq s' )
    { \Sigma; \Gamma; \Delta \vDash Q : \TCMD{s'} }
    { \Sigma; \Gamma; \Delta \vDash s'; Q : \TCMD{s} }
\end{semantics}
\caption{Semantic type rules for stacks.}
\label{fig:semantic_type_rules_stacks}
\end{figure}

Stacks can be semantically typed by using the rules in Figure~\ref{fig:semantic_type_rules_stacks}. These rules are just a more compact and readable formulation of the next Lemmas, that essentially describe how to incrementally build semantically typeable stacks.

\begin{lemma}[Bottom]\label{lemma:compatibility_stacks_bottom}
$\Sigma; \Gamma; \Delta; \ENV{T} \vDash \bot : \TCMD{s}$.
\end{lemma}

\begin{lemma}[Statement]\label{lemma:compatibility_stacks_statement}
If $\Sigma; \Gamma; \Delta; \ENV{T} \vDash S : \TCMD{s}$ and $\Sigma; \Gamma; \Delta; \ENV{T} \vDash Q : \TCMD{s}$, 
then $\Sigma; \Gamma; \Delta; \ENV{T} \vDash S; Q : \TCMD{s}$.
\end{lemma}
\begin{proof}
By the definition of $\Sigma; \Gamma; \Delta \vDash S;Q : \TCMD{s}$, we must exhibit a typing interpretation $\STYPI$ containing the triplet $(S;Q, \Delta, s)$.

From $\Sigma; \Gamma; \Delta; \ENV{T} \vDash S : \TCMD{s}$ and $\Sigma; \Gamma; \Delta; \ENV{T} \vDash Q : \TCMD{s}$ we know there exist typing interpretations $\STYPI^1$ and $\STYPI^2$ such that 
\begin{align*}
  (S; \bot, \Delta, s) & \in \STYPI^1 \\
  (Q, \Delta, s)       & \in \STYPI^2
\end{align*}

\noindent We then construct the following candidate typing interpretation:
\begin{align*}
  \CANDSTYPI & \DEFSYM \SET{ (q; Q, \Delta', s') \mid (q; \bot, \Delta', s') \in \STYPI^1 \land \FINISH{q;\bot, \Delta', s'} = (\Delta, s) } \\
             & \UNION \STYPI^2
\end{align*}

\noindent where $q$ may be empty.
As $(S;\bot, \Delta, s) \in \STYPI^1$ and $\FINISH{S;\bot, \Delta, s} = (\Delta, s)$, we have by construction that $(S; Q, \Delta, s) \in \CANDSTYPI$. 

Now we must show that $\CANDSTYPI$ indeed is a typing interpretation.
Consider an arbitrary triplet $(Q_1, \Delta_1, s_1) \in \CANDSTYPI$.
If $Q_1 \in \TSTACKS$, then Case~\ref{case:typi1} of Definition~\ref{def:typing_interpretation_stm} holds.
Otherwise, there are three possibilities:
\begin{itemize}
  \item If $Q_1$ is of the form $q_1; Q$ for non-empty $q_1$, then $(q_1; \bot, \Delta_1, s_1) \in \STYPI^1$ by construction of $\CANDSTYPI$.
    Thus we know the transition 
    \begin{equation*}
      \Sigma; \Gamma; \Delta_1 \ssemDash[{s_1}] \CONF{q_1; \bot, \ENV{TSV}} \trans \Sigma; \Gamma; \Delta_1' \ssemDash[{s_1'}] \CONF{q_1';\bot, \ENV{T}; \ENV{SV}'}
    \end{equation*}

    \noindent can be concluded by some rule, and $(q_1';\bot, \Delta_1', s_1') \in \STYPI^1$.
    Therefore, the transition 
    \begin{equation*}
      \Sigma; \Gamma; \Delta_1 \ssemDash[{s_1}] \CONF{q_1; Q, \ENV{TSV}} \trans \Sigma; \Gamma; \Delta_1' \ssemDash[{s_1'}] \CONF{q_1';Q, \ENV{T}; \ENV{SV}'}
    \end{equation*}

    \noindent can be concluded by the same rule, and by construction of $\CANDSTYPI$, we have that $(q_1';Q, \Delta_1', s_1') \in \CANDSTYPI$. 

  \item A special case of the above is if $q_1'$ in the reduct is empty.
    Again we have that $(q_1; \bot, \Delta_1, s_1) \in \STYPI^1$ by construction of $\CANDSTYPI$, and we know the transition
    \begin{equation*}
      \Sigma; \Gamma; \Delta_1 \ssemDash[{s_1}] \CONF{q_1; \bot, \ENV{TSV}} \trans \Sigma; \Gamma; \Delta_1' \ssemDash[{s_1'}] \CONF{\bot, \ENV{T}; \ENV{SV}'}
    \end{equation*}

    \noindent can be concluded by some rule, and $(\bot, \Delta_1', s_1') \in \STYPI^1$.
Still by construction of $\CANDSTYPI$, we know that $\FINISH{q_1;\bot, \Delta_1, s_1} = (\Delta, s)$ and, since the execution terminates, we have by Theorem~\ref{theorem:stacks_term_level} that $\Delta_1' = \Delta$ and $s_1' = s$.
    By the same reasoning as above, the transition 
    \begin{equation*}
      \Sigma; \Gamma; \Delta_1 \ssemDash[{s_1}] \CONF{q_1; Q, \ENV{TSV}} \trans \Sigma; \Gamma; \Delta \ssemDash \CONF{Q, \ENV{T}; \ENV{SV}'}
    \end{equation*}

    \noindent can therefore be concluded and, since $(Q, \Delta, s) \in \STYPI^2$, we also have that $(Q, \Delta, s) \in \CANDSTYPI$ by construction.

  \item Otherwise, the triplet $(Q_1, \Delta_1, s_1)$ came from $\STYPI^2$.
    Hence we know that the transition 
    \begin{equation*}
      \Sigma; \Gamma; \Delta_1 \ssemDash[{s_1}] \CONF{Q_1, \ENV{TSV}} \trans \Sigma; \Gamma; \Delta_1' \ssemDash[{s_1'}] \CONF{Q_1', \ENV{T}; \ENV{SV}'}
    \end{equation*}

    \noindent can be concluded, and $(Q_1', \Delta_1', s_1') \in \STYPI^2$.
    Therefore we also have that $(Q_1', \Delta_1', s_1') \in \CANDSTYPI$ by construction.
\qedhere
\end{itemize}
\end{proof}

Notice that, in the construction of the candidate typing interpretation $\CANDSTYPI$ of the previous proof, we make use of two related facts:
\begin{itemize}
  \item Transitions only ever modify \emph{the top} of the stack; hence, we can remove the $\bot$ symbol from a stack $q;\bot$ and suffix another stack $Q$ onto it.
  \item Theorem~\ref{theorem:stacks_term_level} tells us that, \emph{if} $q;\bot$ eventually terminates normally when executed at level $s$ and with type environment $\Delta$, \emph{then} we can use the function $\FINISH{\cdot}$ to find the security level $s'$ and type environment $\Delta$ in the terminal state.
\end{itemize}

By the definition of $\FINISH{\cdot}$, we have that $\FINISH{S;\bot, \Delta, s} = (\Delta, s)$, so in the construction of $\CANDSTYPI$ we make sure to only pick those triplets $(q;\bot, \Delta', s')$ from $\STYPI^1$ for which it holds that $\FINISH{q;\bot, \Delta', s'} = (\Delta, s)$.
This ensures two things:
\begin{itemize}
  \item Every stack $q;\bot$ that derives from the execution of $S$ will be suffixed with $Q$ and included in the candidate typing interpretation.
  \item Only those stacks $q;\bot$ that will terminate on $(\Delta, s)$, and hence are \emph{compatible} with $Q$, will be suffixed with $Q$ and included.
    This ensures that the newly created stacks $q; Q$ remain well-formed.
\end{itemize}

It is of course possible that some stacks $q;\bot$ will not terminate, in which case Theorem~\ref{theorem:stacks_term_level} does not apply.
This is, however, not a problem, since the execution will then never reach the suffix $Q$ in the compound stacks $q; Q$.
The predicate $\FINISH{q; \bot, \Delta', s'} = (\Delta, s)$ still ensures that every compound stack $q;Q$ is well-formed, and by Theorem~\ref{theorem:stacks_runtime_preservation} we know that well-formedness is preserved by the typed semantics, so the construction is still safe, even if the $Q$ suffix is never reached.

\begin{lemma}[End of scope]\label{lemma:compatibility_stacks_del}
If $\Sigma; \Gamma; \Delta; \ENV{T} \vDash Q : \TCMD{s}$,
then $\Sigma; \Gamma; \Delta, x:\TVAR{B'_{s'}}; \ENV{T} \vDash \DEL{x}; Q : \TCMD{s}$.
\end{lemma}
\begin{proof}
By the definition of $\Sigma; \Gamma; \Delta,x:\TVAR{B'_{s'}}; \ENV{T} \vDash \DEL{x}; Q : \TCMD{s}$, we must exhibit a typing interpretation $\STYPI'$ which is such that 
\begin{equation*}
  (\DEL{x};Q, (\Delta, x:\TVAR{B'_{s'}}), s) \in \STYPI'.
\end{equation*}

From $\Sigma; \Gamma; \Delta; \ENV{T} \vDash Q : \TCMD{s}$, we know there exists a typing interpretation $\STYPI$ containing the triplet $(Q, \Delta, s)$.
We then construct a candidate typing interpretation $\CANDSTYPI$, such that it contains $(\DEL{x}; Q, (\Delta, x:\TVAR{B'_{s'}}), s)$ as follows:
\begin{equation*}
  \CANDSTYPI \DEFSYM \SET{ (\DEL{x}; Q, (\Delta, x:\TVAR{B'_{s'}}), s) } \UNION \STYPI
\end{equation*}

As we have only added a single triple, all that now is needed is to show that there exists at least one transition 
\begin{equation*}
  \Sigma; \Gamma; \ENV{T} \vDash (\DEL{x}; Q, (\Delta, x:\TVAR{B'_{s'}}), s) \ttrans P_1'
\end{equation*}

\noindent and that, for all such transitions, the reduct is again in $\CANDSTYPI$.

From $(Q, \Delta, s) \in \STYPI$, we know that an appropriately shaped $\ENV{SV}$ can be built from $\Sigma; \Gamma; \Delta$ for all transitions from $Q$ by the method in Definition~\ref{def:construction_envsv} (and if $Q$ has no transitions by Case~\ref{case:typi1} of Definition~\ref{def:typing_interpretation_stm}, then the shape of $\ENV{SV}$ is irrelevant).
Pick any such environment, and pick any value $v$ such that $\TYPEOF{v} = (B_1, s_1)$ and $\Sigma \vdash B_1 \SUBS B'$ and $s_1 \ordleq s'$.
Now, the transition 
\begin{equation*}
              \Sigma; \Gamma; \Delta, x:\TVAR{B'_{s'}} \ssemDash \CONF{\DEL{x} ; Q, \ENV{TS}; \ENV{V}, x:v } 
      \trans  \Sigma; \Gamma; \Delta \ssemDash \CONF{Q, \ENV{TSV}}
\end{equation*}

\noindent can be concluded by \nameref{stm_rt_delv} and, by Corollary~\ref{theorem:stacks_runtime_states_abstraction}, the transition can also be concluded for any other $\ENV{SV}'$ having the same structure and satisfying the requirements.
No other transition is possible, so we conclude that the transition
\begin{equation*}
  \Sigma; \Gamma; \ENV{T} \vDash (\DEL{x}; Q, (\Delta, x:\TVAR{B'_{s'}}), s) \ttrans (Q, \Delta, s)
\end{equation*}

\noindent exists and, as $(Q, \Delta, s) \in \STYPI$, we also have that $(Q, \Delta, s) \in \CANDSTYPI$ by construction.
\end{proof}

\begin{lemma}[Method return]\label{lemma:compatibility_stacks_return}
If $\Sigma; \Gamma; \Delta'; \ENV{T} \vDash Q : \TCMD{s}$ and $\Sigma; \Gamma; \Delta' \vdash \ENV{V}'$, 
then $\Sigma; \Gamma; \Delta; \ENV{T} \vDash (\ENV{V}', \Delta'); Q : \TCMD{s}$.
\end{lemma}
\begin{proof}
By the definition of $\Sigma; \Gamma; \Delta; \ENV{T} \vDash (\ENV{V}', \Delta'); Q : \TCMD{s}$, we must exhibit a typing interpretation $\STYPI'$ containing the triplet $((\ENV{V}', \Delta'); Q, \Delta, s)$.

From $\Sigma; \Gamma; \Delta'; \ENV{T} \vDash Q : \TCMD{s}$, we know there exists a typing interpretation $\STYPI$ containing the triplet $(Q, \Delta', s)$.
We then construct a candidate typing interpretation $\CANDSTYPI'$ containing $((\ENV{V}', \Delta'); Q, \Delta, s)$ as follows:
\begin{equation*}
  \CANDSTYPI \DEFSYM \SET{ ((\ENV{V}', \Delta'); Q, \Delta, s) } \UNION \STYPI
\end{equation*}

As we have only added a single triple, all that now is needed is to show that there exists at least one transition 
\begin{equation*}
  \Sigma; \Gamma; \ENV{T} \vDash ((\ENV{V}', \Delta');Q, \Delta, s) \ttrans P_1'
\end{equation*}

\noindent and that, for all such transitions, the reduct is again in $\CANDSTYPI$.
Pick any $\ENV{SV}$ built from $\Sigma; \Gamma; \Delta$ by the method in Definition~\ref{def:construction_envsv}.
As we know that $\Sigma; \Gamma; \Delta' \vdash \ENV{V}'$, the transition 
\begin{equation*}
             \Sigma; \Gamma; \Delta \ssemDash \CONF{(\ENV{V}', \Delta') ; Q, \ENV{TSV}} 
      \trans \Sigma; \Gamma; \Delta' \ssemDash \CONF{Q, \ENV{TS}; \ENV{V}'} 
\end{equation*}

\noindent can now be concluded by \nameref{stm_rt_return} and, by Corollary~\ref{theorem:stacks_runtime_states_abstraction}, the transition can also be concluded for any other $\ENV{SV}$ built according to Definition~\ref{def:construction_envsv}.
No other transition is possible, so we conclude that the transition
\begin{equation*}
  \Sigma; \Gamma; \ENV{T} \vDash ((\ENV{V}', \Delta');Q, \Delta, s) \ttrans (Q, \Delta', s)
\end{equation*}

\noindent exists and, as $(Q, \Delta', s) \in \STYPI$, we also have that $(Q, \Delta', s) \in \CANDSTYPI$ by construction.
\end{proof}

\begin{lemma}[Security level restore]\label{lemma:compatibility_stacks_restore}
If $\Sigma; \Gamma; \Delta; \ENV{T} \vDash Q : \TCMD{s'}$ and $s \ordgeq s'$, then $\Sigma; \Gamma; \Delta; \ENV{T} \vDash s'; Q : \TCMD{s}$.
\end{lemma}
\begin{proof}
By the definition of $\Sigma; \Gamma; \Delta; \ENV{T} \vDash s'; Q : \TCMD{s}$, we must exhibit a typing interpretation $\STYPI'$ containing the triplet $(s'; Q, \Delta, s)$.

From $\Sigma; \Gamma; \Delta; \ENV{T} \vDash Q : \TCMD{s'}$, we know there exists a typing interpretation $\STYPI$ containing the triplet $(Q, \Delta, s')$.
We then construct a candidate typing interpretation $\CANDSTYPI$, such that it contains $(s'; Q, \Delta, s)$ for some $s \ordgeq s'$ as follows:
\begin{equation*}
  \CANDSTYPI \DEFSYM \SET{ (s'; Q, \Delta, s) } \UNION \STYPI
\end{equation*}

As we have only added a single triple, all that now is needed is to show that there exists at least one transition 
\begin{equation*}
  \Sigma; \Gamma; \ENV{T} \vDash (s';Q, \Delta, s) \ttrans P_1'
\end{equation*}

\noindent and that, for all such transitions, the reduct is again in $\CANDSTYPI$.
Pick any $\ENV{SV}$ built from $\Sigma; \Gamma; \Delta$ by the method in Definition~\ref{def:construction_envsv}.
As we know that $s \ordgeq s'$, the transition 
\begin{equation*}
             \Sigma; \Gamma; \Delta \ssemDash \CONF{s'; Q, \ENV{TSV}} 
      \trans \Sigma; \Gamma; \Delta \ssemDash[{s'}] \CONF{Q, \ENV{TSV}}
\end{equation*}

\noindent can be concluded by \nameref{stm_rt_restore} and, by Corollary~\ref{theorem:stacks_runtime_states_abstraction}, the transition can also be concluded for any other $\ENV{SV}'$ having the same structure and satisfying the requirements.
No other transition is possible, so we conclude that the transition
\begin{equation*}
  \Sigma; \Gamma; \ENV{T} \vDash (s';Q, \Delta, s) \ttrans (Q, \Delta, s')
\end{equation*}

\noindent exists and, as $(Q, \Delta, s') \in \STYPI$, we also have that $(Q, \Delta, s') \in \CANDSTYPI$ by construction.
\end{proof}

The next Lemmas prove  soundness of the rules in Figure~\ref{fig:semantic_type_rules_stm}. Again, the rules depicted there are just a more readable formulation of the following results.

\begin{lemma}[Skip]\label{lemma:compatibility_stm_skip}
$\Sigma; \Gamma; \Delta; \ENV{T}  \vDash \code{skip} : \TCMD{s}$.
\end{lemma}
\begin{proof}
By the definition of $\Sigma; \Gamma; \Delta \vDash \code{skip} : \TCMD{s}$, we must exhibit a typing interpretation $\STYPI$ containing the triplet $(\code{skip}; \bot, \Delta, s)$.
Here is our candidate typing interpretation:
\begin{equation*}
  \CANDSTYPI \DEFSYM \SET{(\code{skip}; \bot, \Delta, s), (\bot, \Delta, s)}
\end{equation*}

\noindent and, clearly, $(\code{skip}; \bot, \Delta, s) \in \CANDSTYPI$.
We must now show that $\CANDSTYPI$ indeed is a typing interpretation by checking the criteria of Definition~\ref{def:typing_interpretation_stm}.
We examine each of the pairs we have added:
\begin{itemize}
  \item By rule \nameref{stm_rt_skip} we can conclude the transition 
  \begin{equation*}
           \Sigma; \Gamma; \Delta \ssemDash \CONF{\code{skip} ; \bot, \ENV{TSV}} 
    \trans \Sigma; \Gamma; \Delta \ssemDash \CONF{\bot, \ENV{TSV}} 
  \end{equation*}

  \noindent for any $\Delta$, $s$, and $\ENV{SV}$.
  There are no other possible transitions, so we conclude that the transition
  \begin{equation*}
    \Sigma; \Gamma; \ENV{T} \vDash (\code{skip};\bot, \Delta, s) \ttrans (\bot, \Delta, s)
  \end{equation*}

  \noindent exists, and $(\bot, \Delta, s) \in \CANDSTYPI$ by construction.

  \item The triplet $(\bot, \Delta, s)$ belongs to $\CANDSTYPI$ by Case~\ref{case:typi1}  of Definition~\ref{def:typing_interpretation_stm}.
\qedhere
\end{itemize}
\end{proof}

\begin{lemma}[Throw]\label{lemma:compatibility_stm_throw}
$\Sigma; \Gamma; \Delta; \ENV{T} \vDash \code{throw} : \TCMD{s}$.
\end{lemma}
\begin{proof}
By the definition of $\Sigma; \Gamma; \Delta \vDash \code{throw} : \TCMD{s}$, we must exhibit a typing interpretation $\STYPI$ containing the triplet $(\code{throw}; \bot, \Delta, s)$.
Here is our candidate typing interpretation:
\begin{equation*}
  \CANDSTYPI \DEFSYM \SET{(\code{throw}; \bot, \Delta, s)}
\end{equation*}
that is a typing interpretation by Case~\ref{case:typi1}  of Definition~\ref{def:typing_interpretation_stm}.
\end{proof}

\begin{lemma}[Sequential composition]\label{lemma:compatibility_stm_seq}
If $\Sigma; \Gamma; \Delta; \ENV{T} \vDash S_1 : \TCMD{s}$ and $\Sigma; \Gamma; \Delta; \ENV{T} \vDash S_2 : \TCMD{s}$, then $\Sigma; \Gamma; \Delta; \ENV{T} \vDash S_1; S_2 : \TCMD{s}$.
\end{lemma}
\begin{proof}
By Definition~\ref{def:typing_interpretation_stm}, 
\begin{equation*}
  \Sigma; \Gamma; \Delta; \ENV{T} \vDash S_2 : \TCMD{s} = \Sigma; \Gamma; \Delta; \ENV{T} \vDash S_2; \bot : \TCMD{s}
\end{equation*} 

\noindent and $S_2; \bot$ is a stack $Q$.
The result then follows directly from Lemma~\ref{lemma:compatibility_stacks_statement}.
\end{proof}

\begin{lemma}[If-then-else]\label{lemma:compatibility_stm_if}
If $\Sigma; \Gamma; \Delta \vDash e : \TBOOL_s$, 
$\Sigma; \Gamma; \Delta; \ENV{T} \vDash S_\TRUE  : \TCMD{s}$ and 
$\Sigma; \Gamma; \Delta; \ENV{T} \vDash S_\FALSE : \TCMD{s}$, then $\Sigma; \Gamma; \Delta; \ENV{T} \vDash \code{if $e$ then $S_\TRUE$ else $S_\FALSE$} : \TCMD{s}$.  
\end{lemma}
\begin{proof}
By the definition of $\Sigma; \Gamma; \Delta; \ENV{T} \vDash \code{if $e$ then $S_\TRUE$ else $S_\FALSE$} : \TCMD{s}$, we must exhibit a typing interpretation $\STYPI$ which is such that it contains the triplet $(\code{if $e$ then $S_\TRUE$ else $S_\FALSE$}; \bot, \Delta, s)$.

From $\Sigma; \Gamma; \Delta; \ENV{T} \vDash S_\TRUE : \TCMD{s}$ and $\Sigma; \Gamma; \Delta; \ENV{T} \vDash S_\FALSE : \TCMD{s}$ we know that there exist typing interpretations $\STYPI^1$ and $\STYPI^2$ such that 
\begin{align*}
   (S_\TRUE;  \bot, \Delta, s) & \in \STYPI^1 \\
   (S_\FALSE; \bot, \Delta, s) & \in \STYPI^2 .
\end{align*}

\noindent We then construct the following candidate typing interpretation, where $q_1$ and $q_2$ may be empty:
\begin{align*}
    \CANDSTYPI & \DEFSYM \SET{ (\code{if $e$ then $S_\TRUE$ else $S_\FALSE$}; \bot, \Delta, s) }            \\
               & \UNION  \SET{ (q_1; s; \bot, \Delta_1, s_1) \mid (q_1; \bot, \Delta_1, s_1) \in \STYPI^1 \land \FINISH{q_1;\bot, \Delta_1, s_1} = (\Delta, s) } \\
               & \UNION  \SET{ (q_2; s; \bot, \Delta_2, s_2) \mid (q_2; \bot, \Delta_2, s_2) \in \STYPI^2 \land \FINISH{q_2;\bot, \Delta_2, s_2} = (\Delta, s) } \\ 
               & \UNION  \SET{ (\bot, \Delta, s) }
\end{align*}
We must now show that $\CANDSTYPI$ in fact is a typing interpretation.
First we consider the two single triples that we added:
\begin{itemize}
  \item The triplet $(\bot, \Delta, s)$ satisfies Case~\ref{case:typi1} of Definition~\ref{def:typing_interpretation_stm}.

  \item From $\Sigma; \Gamma; \Delta \vDash e : \TBOOL_s$ and Definition~\ref{def:typing_interpretation_expr}, we know that 
    \begin{equation*}
      \Sigma; \Gamma; \Delta \esemDash[{\TBOOL_s}] \CONF{e, \ENV{SV}} \trans b
    \end{equation*}

    \noindent where $b$ is a boolean of level $s$ for any $\ENV{SV}$ built from $\Sigma; \Gamma; \Delta$ according to Definition~\ref{def:construction_envsv}.
    Thus, the transition 
    \begin{align*}
              \Sigma; \Gamma; \Delta \ssemDash \CONF{\code{if $e$ then $S_{\TRUE}$ else $S_{\FALSE}$} ; \bot, \ENV{TSV}} 
      \trans \Sigma; \Gamma; \Delta \ssemDash \CONF{S_b ; s; \bot, \ENV{TSV}}     
    \end{align*}

    \noindent can be concluded by rule \nameref{stm_rt_if}, where $S_b$ can be either $S_\TRUE$ or $S_\FALSE$.
    Hence, there are two possible transitions:
    \begin{align*}
      \Sigma; \Gamma; \ENV{T} & \vDash (\code{if $e$ then $S_{\TRUE}$ else $S_{\FALSE}$};\bot, \Delta, s) \ttrans (S_\TRUE;s;\bot, \Delta, s) \\
      \Sigma; \Gamma; \ENV{T} & \vDash (\code{if $e$ then $S_{\TRUE}$ else $S_{\FALSE}$};\bot, \Delta, s) \ttrans (S_\FALSE;s;\bot, \Delta, s) 
    \end{align*}
    
    \noindent which by Corollary~\ref{theorem:stacks_runtime_states_abstraction} can be concluded for any $\ENV{SV}$ built from $\Sigma; \Gamma; \Delta$ according to Definition~\ref{def:construction_envsv}.
    There are no other possible transitions.
    Finally, since 
    \begin{align*}
      (S_\TRUE;  \bot, \Delta, s) & \in \STYPI^1 \land \FINISH{S_\TRUE;  \bot, \Delta, s} = (\Delta, s) \\
      (S_\FALSE; \bot, \Delta, s) & \in \STYPI^2 \land \FINISH{S_\FALSE; \bot, \Delta, s} = (\Delta, s)
    \end{align*}

    \noindent we have that $(S_\TRUE; s; \bot, \Delta, s) \in \CANDSTYPI$ and $(S_\FALSE; s; \bot, \Delta, s) \in \CANDSTYPI$ by construction.
\end{itemize}

Next we consider an arbitrary triplet $(q_1; s; \bot, \Delta_1, s_1)$ from the set 
\begin{equation*}
  \SET{ (q_1; s; \bot, \Delta_1, s_1) | (q_1; \bot, \Delta_1, s_1) \in \STYPI^1 \land \FINISH{q_1;\bot, \Delta_1, s_1} = (\Delta, s) } .
\end{equation*}
(the case for triplets from the other set is similar).
As $q_1$ can be empty, we have three cases to consider:
\begin{itemize}
  \item If $q_1$ is empty, then it must have come from a triplet $(\bot, \Delta, s) \in \STYPI^1$, since $\FINISH{\bot, \Delta, s} = (\Delta, s)$.
    Thus, the triplet in $\CANDSTYPI$ is of the form $(s; \bot, \Delta, s)$ by construction of $\CANDSTYPI$.
    As $s \ordgeq s$, this satisfies the side condition of \nameref{stm_rt_restore}, so the transition 
    \begin{equation*}
             \Sigma; \Gamma; \Delta \ssemDash \CONF{s; \bot, \ENV{TSV}} 
      \trans \Sigma; \Gamma; \Delta \ssemDash \CONF{\bot, \ENV{TSV}}
    \end{equation*}

    \noindent can be concluded by that rule for any appropriately shaped $\ENV{SV}$ and $(\bot, \Delta, s) \in \CANDSTYPI$ by construction.

  \item If $q_1$ is of the form $\code{throw}; q_1'$, then this satisfies Case~\ref{case:typi1} of Definition~\ref{def:typing_interpretation_stm}.

  \item If $q_1$ is non-empty and not of the form $\code{throw}; q_1'$, then we know is that $(q_1;\bot, \Delta_1, s_1) \in \STYPI^1$, that $\FINISH{q_1, \Delta_1, s_1} = (\Delta, s)$, and that there exists at least one transition 
    \begin{equation*}
      \Sigma; \Gamma; \Delta_1 \ssemDash[{s_1}] \CONF{q_1; \bot, \ENV{TSV}} \trans \Sigma; \Gamma; \Delta_1' \ssemDash[{s_1'}] \CONF{q_1'; \bot, \ENV{T}; \ENV{SV}'}
    \end{equation*}
    \noindent for all appropriately shaped $\ENV{SV}$, with $(q_1';\bot, \Delta_1', s_1') \in \STYPI^1$.
    There are now two possibilities, depending on the shape of $q_1'$:
    \begin{itemize}
      \item If $q_1'$ is empty, the reduct is of the form $\Sigma; \Gamma; \Delta \ssemDash \CONF{\bot, \ENV{T}; \ENV{SV}'}$.
        Then, the transition 
        \begin{equation*}
          \Sigma; \Gamma; \Delta_1 \ssemDash[{s_1}] \CONF{q_1; s; \bot, \ENV{TSV}} \trans \Sigma; \Gamma; \Delta \ssemDash \CONF{s; \bot, \ENV{T}; \ENV{SV}'}
        \end{equation*}

        \noindent can likewise be concluded, and $(s; \bot, \Delta, s) \in \CANDSTYPI$ by construction.

      \item Otherwise, if $q_1'$ is non-empty, then the reduct is instead of the form $\Sigma; \Gamma; \Delta_1' \ssemDash[{s_1'}] \CONF{q_1'; \bot, \ENV{T}; \ENV{SV}'}$ and $(q_1';\bot, \Delta_1', s_1') \in \STYPI^1$.
        Then, the transition 
        \begin{equation*}
          \Sigma; \Gamma; \Delta_1 \ssemDash[{s_1}] \CONF{q_1; s; \bot, \ENV{TSV}} \trans \Sigma; \Gamma; \Delta_1' \ssemDash[s_1'] \CONF{q_1'; s; \bot, \ENV{T}; \ENV{SV}'}
        \end{equation*}

        \noindent can likewise be concluded.
        Finally, by Theorem~\ref{theorem:stacks_term_level}, we have that $\FINISH{q_1';\bot, \Delta_1', s_1'} = (\Delta, s)$, so by construction we also have that $(q_1'; s; \bot, \Delta_1', s_1') \in \CANDSTYPI$.
\qedhere
    \end{itemize}
\end{itemize}
\end{proof}

\begin{lemma}[While]\label{lemma:compatibility_stm_while}
If $\Sigma; \Gamma; \Delta \vDash e : \TBOOL_s$ and $\Sigma; \Gamma; \Delta \ENV{T} \vDash S : \TCMD{s}$,
then $\Sigma; \Gamma; \Delta; \ENV{T} \vDash \code{while $e$ do $S$} : \TCMD{s}$.
\end{lemma}
\begin{proof}
By the definition of $\Sigma; \Gamma; \Delta; \ENV{T} \vDash \code{while $e$ do $S$} : \TCMD{s}$, we must exhibit a typing interpretation $\STYPI'$ containing the triplet $(\code{while $e$ do $S$}; \bot, \Delta, s)$.

From $\Sigma; \Gamma; \Delta; \ENV{T} \vDash S : \TCMD{s}$ we know that there exists a typing interpretation $\STYPI$ such that $(S, \Delta, s) \in \STYPI$.
We then construct the following candidate typing interpretation, where $q'$ may be empty:
\begin{align*}
    \CANDSTYPI & \DEFSYM \SET{ (\code{while $e$ do $S$}; \bot, \Delta, s) }                   \\
               & \UNION  \SET{ (q'; s; \code{while $e$ do $S$}; \bot, \Delta', s') ~\bigg|    %
                 \begin{array}{r @{~} l}                                                      %
                          (q'; \bot, \Delta', s') \in \STYPI 
                   \land \FINISH{q';\bot, \Delta', s'} = (\Delta, s)   
                  \end{array}                                                                 %
                }                                                                             \\
               & \UNION  \SET{ (\bot, \Delta, s) }
\end{align*}

We must now show that $\CANDSTYPI$ in fact is a typing interpretation.
First we consider the two single triplets that we added:
\begin{itemize}
  \item The triplet $(\bot, \Delta, s)$ satisfies Case~\ref{case:typi1} of Definition~\ref{def:typing_interpretation_stm}.

  \item From $\Sigma; \Gamma; \Delta \vDash e : \TBOOL_s$ and Definition~\ref{def:typing_interpretation_expr}, we know that 
    \begin{equation*}
      \Sigma; \Gamma; \Delta \esemDash[{\TBOOL_s}] \CONF{e, \ENV{SV}} \trans b
    \end{equation*}

    \noindent can be concluded for any $\ENV{SV}$ built from $\Sigma; \Gamma; \Delta$ according to Definition~\ref{def:construction_envsv}.
    Depending on the value of $b$, we now have a choice between two different rules by which to conclude the transition.
    By Corollary~\ref{theorem:stacks_runtime_states_abstraction}, either of the two transition can then be concluded for any $\ENV{SV}$ built according to Definition~\ref{def:construction_envsv}.
    We consider each case in turn:
    \begin{itemize}
      \item If $b = \FALSE$, then we can conclude the transition 
        \begin{align*}
                  \Sigma; \Gamma; \Delta \ssemDash \CONF{\code{while $e$ do $S$, \ENV{TSV}} ; \bot, \ENV{TSV}} 
          \trans \Sigma; \Gamma; \Delta \ssemDash \CONF{\bot, \ENV{TSV}}     
        \end{align*}

        \noindent by \nameref{stm_rt_whilefalse}, and $(\bot, \Delta, s) \in \CANDSTYPI$ by construction.

      \item If $b = \TRUE$, then we can conclude the transition
        \begin{align*}
             \quad   \qquad  \Sigma; \Gamma; \Delta \ssemDash \CONF{\code{while $e$ do $S$} ; \bot, \ENV{TSV}} 
          \trans \Sigma; \Gamma; \Delta \ssemDash \CONF{S ; s ; \code{while $e$ do $S$} ; \bot, \ENV{TSV}}  
        \end{align*}

        \noindent Since $(S;\bot, \Delta, s) \in \STYPI$ and $\FINISH{S;\bot, \Delta, s} = (\Delta, s)$, then we also have that $(S ; s ; \code{while $e$ do $S$} ; \bot, \Delta, s) \in \CANDSTYPI$ by construction.
    \end{itemize}
\end{itemize}

Next we consider an arbitrary triplet $(q'; s; \code{while $e$ do $S$} ; \bot, \Delta', s')$ from the set 
\begin{equation*}
  \left\{ (q'; s; \code{while $e$ do $S$}; \bot, \Delta', s') ~\bigg|         %
    \begin{array}{r @{~} l}                                                   %
            (q'; \bot, \Delta', s') \in \STYPI 
      \land \FINISH{q';\bot, \Delta', s'} = (\Delta, s)
    \end{array}                                                               %
  \right\}                                                                    
\end{equation*}

\noindent As $q'$ can be empty, we have three cases to consider:
\begin{itemize}
  \item If $q'$ is empty, then it must have come from the triplet $(\bot, \Delta, s) \in \STYPI$, since $\FINISH{\bot, \Delta, s} = (\Delta, s)$.
    Thus, the triplet is of the form $(s; \code{while $e$ do $S$}; \bot, \Delta, s)$.
    As $s \ordgeq s$, this satisfies the side condition of \nameref{stm_rt_restore}, so the transition 
    \begin{align*}
           \qquad   \Sigma; \Gamma; \Delta \ssemDash \CONF{s; \code{while $e$ do $S$}; \bot, \ENV{TSV}} 
      \trans  \Sigma; \Gamma; \Delta \ssemDash \CONF{\code{while $e$ do $S$}; \bot, \ENV{TSV}}
    \end{align*}

    \noindent can be concluded by that rule for any appropriately shaped $\ENV{SV}$, and we have that $(\code{while $e$ do $S$}; \bot, \Delta, s) \in \CANDSTYPI$ by construction.

  \item If $q'$ is of the form $\code{throw}; q''$, then this satisfies Case~\ref{case:typi1} of Definition~\ref{def:typing_interpretation_stm}.

  \item If $q'$ is non-empty and not of the form $\code{throw}; q''$, then $(q';\bot, \Delta', s') \in \STYPI$, $\FINISH{q';\bot, \Delta', s'} = (\Delta, s)$, and there exists at least one transition 
    \begin{equation*}
      \Sigma; \Gamma; \Delta' \ssemDash[{s'}] \CONF{q';\bot, \ENV{TSV}} \trans \Sigma; \Gamma; \Delta'' \ssemDash[{s''}] \CONF{q'';\bot, \ENV{T}; \ENV{SV}'}
    \end{equation*}

    \noindent for all appropriately shaped $\ENV{SV}$ built from $\Sigma; \Gamma; \Delta'$ according to Definition~\ref{def:construction_envsv}, with $(q'';\bot, \Delta'', s'') \in \STYPI$.    
    There are now two possibilities, depending on the shape of $q''$:
    \begin{itemize}
      \item If $q''$ is empty, the reduct is of the form $\Sigma; \Gamma; \Delta \ssemDash \CONF{\bot, \ENV{T}; \ENV{SV}'}$.
        Then, the transition 
        \begin{align*}
                 & \Sigma; \Gamma; \Delta' \ssemDash[{s'}] \CONF{q';s;\code{while $e$ do $S$};\bot, \ENV{TSV}} \\
          \trans ~& \Sigma; \Gamma; \Delta \ssemDash \CONF{s;\code{while $e$ do $S$};\bot, \ENV{T}; \ENV{SV}'}
        \end{align*}

        \noindent can likewise be concluded, and $(s;\code{while $e$ do $S$};\bot, \Delta, s) \in \CANDSTYPI$ by construction.

      \item Otherwise, if $q''$ is non-empty, then the reduct is instead of the form $\Sigma; \Gamma; \Delta'' \ssemDash[{s''}] \CONF{q''; \bot, \ENV{T}; \ENV{SV}'}$ and $(q'';\bot, \Delta'', s'') \in \STYPI$.
        Then, the transition 
        \begin{align*}
                  & \Sigma; \Gamma; \Delta_1 \ssemDash[{s'}] \CONF{q';s;\code{while $e$ do $S$};\bot, \ENV{TSV}}          \\
          \trans ~& \Sigma; \Gamma; \Delta'' \ssemDash[s''] \CONF{q'';s;\code{while $e$ do $S$};\bot, \ENV{T}; \ENV{SV}'}
        \end{align*}

        \noindent can likewise be concluded.
        Finally, by Theorem~\ref{theorem:stacks_term_level}, we have that $\FINISH{q'';\bot, \Delta'', s''} = (\Delta, s)$, so by construction we also have that $(q'';s;\code{while $e$ do $S$};\bot, \Delta'', s'') \in \CANDSTYPI$.
        \qedhere
    \end{itemize}
\end{itemize}
\end{proof}

\begin{lemma}[Variable declaration]\label{lemma:compatibility_stm_decl}
If $\Sigma; \Gamma; \Delta \vDash e : B_{s_1}$ and $\Sigma; \Gamma; \Delta, x:\TVAR{B_{s_1}}; \ENV{T} \vDash S : \TCMD{s}$,
then $\Sigma; \Gamma; \Delta; \ENV{T} \vDash \code{$\TVAR{B_{s_1}}$ $x$ := $e$ in $S$} : \TCMD{s}$.
\end{lemma}
\begin{proof}
By the definition of $\Sigma; \Gamma; \Delta; \ENV{T} \vDash \code{$\TVAR{B_{s_1}}$ $x$ := $e$ in $S$} : \TCMD{s}$, we must exhibit a typing interpretation containing $(\code{$\TVAR{B_{s_1}}$ $x$ := $e$ in $S$}; \bot, \Delta, s)$.

From $\Sigma; \Gamma; \Delta, x:\TVAR{B_{s_1}}; \ENV{T} \vDash S : \TCMD{s}$ we know there exists a typing interpretation $\STYPI$ such that 
\begin{equation*}
  (S; \bot, (\Delta, x:\TVAR{B_{s_1}}), s) \in \STYPI 
\end{equation*}

\noindent We then construct the following candidate typing interpretation:
\begin{align*}
  \CANDSTYPI & \DEFSYM \SET{ (\code{$\TVAR{B_{s_1}}$ $x$ := $e$ in $S$}; \bot, \Delta, s) }              \\
             & \UNION  \SET{ (q';\DEL{x}; \bot, \Delta', s') ~\bigg|                                     %
               \begin{array}{r @{~} l}                                                                   %
                       (q';\bot, \Delta', s') \in \STYPI                              
                 \land  \FINISH{q';\bot, \Delta', s'} = ((\Delta, x:\TVAR{B_{s_1}}), s) 
               \end{array}                                                                               %
               }                                                                                         \\
             & \UNION  \SET{ (\bot, \Delta, s) }
\end{align*}
We must now show that $\CANDSTYPI$ indeed is a typing interpretation.
First, we consider the two single triplets that we added:
\begin{itemize}
  \item The triplet $(\bot, \Delta, s)$ satisfies Case~\ref{case:typi1} of Definition~\ref{def:typing_interpretation_stm}.

  \item From $\Sigma; \Gamma; \Delta \vDash e : B_{s_1}$ and Definition~\ref{def:typing_interpretation_expr}, we know that 
    \begin{equation*}
      \Sigma; \Gamma; \Delta \esemDash[{B_{s_1}}] \CONF{e, \ENV{SV}} \trans v
    \end{equation*}

    \noindent where $v$ is a value of type $B$ and level $s_1$ for any $\ENV{SV}$ built from $\Sigma; \Gamma; \Delta$ according to Definition~\ref{def:construction_envsv}.
    Thus, the transition 
    \begin{align*}
              & \Sigma; \Gamma; \Delta \ssemDash \CONF{\code{$\TVAR{B_{s_1}}$ $x$ := $e$ in $S$};\bot, \ENV{TSV}} \\
      \trans ~& \Sigma; \Gamma; \Delta, x:\TVAR{B_{s_1}} \ssemDash \CONF{S ; \DEL{x}; \bot, \ENV{TS}; \ENV{V}, x:v}     
    \end{align*}

    \noindent can be concluded by rule \nameref{stm_rt_decv}.
    By Corollary~\ref{theorem:stacks_runtime_states_abstraction}, the transition can be concluded for any $\ENV{SV}$ built from $\Sigma; \Gamma; \Delta$ according to Definition~\ref{def:construction_envsv}.
    There are no other possible transitions.
    Finally, since 
    \begin{align*}
      (S;\bot, (\Delta, x:\TVAR{B_{s_1}}), s)        & \in \STYPI              \\
      \FINISH{S;\bot, (\Delta, x:\TVAR{B_{s_1}}), s} & = ((\Delta, x:\TVAR{B_{s_1}}), s)
    \end{align*}
  
    \noindent we have that $(S;\DEL{x};\bot, (\Delta, x:\TVAR{B_{s_1}}), s) \in \CANDSTYPI$, by construction of $\CANDSTYPI$.
\end{itemize}

Next we consider an arbitrary triplet $(q'; \DEL{x}; \bot, \Delta', s')$ from the set 
\begin{equation*}
 \quad   \left\{ (q';\DEL{x}; \bot, \Delta', s') ~\bigg|                                            %
    \begin{array}{r @{~} l}                                                                  %
         (q';\bot, \Delta', s') \in \STYPI 
     \land \FINISH{q';\bot, \Delta', s'} = ((\Delta, x:\TVAR{B_{s_1}}), s) 
    \end{array}                                                                              %
  \right\}                                                                                   
\end{equation*}

\noindent As $q'$ can be empty, we have three cases to consider:
\begin{itemize}
  \item If $q'$ is empty, then it must have come from the triplet $(\bot, (\Delta, x:\TVAR{B_{s_1}}), s) \in \STYPI$, since $\FINISH{\bot, (\Delta, x:\TVAR{B_{s_1}}), s} = ((\Delta, x:\TVAR{B_{s_1}}), s)$.
    Thus, the triplet in $\CANDSTYPI$ is of the form $(\DEL{x}; \bot, (\Delta, x:\TVAR{B_{s_1}}), s)$.
    The transition 
    \begin{align*}
              \Sigma; \Gamma; \Delta, x:\TVAR{B_{s'}} \ssemDash \CONF{\DEL{x};\bot, \ENV{TS}; \ENV{V}, x:v } 
      \trans \Sigma; \Gamma; \Delta \ssemDash \CONF{\bot, \ENV{TSV}} 
    \end{align*}

    \noindent can be concluded by that rule for any appropriately shaped $\ENV{S}; \ENV{V}, x:v$ by \nameref{stm_rt_delv}, and $(\bot, \Delta, s) \in \CANDSTYPI$ by construction.

  \item If $q'$ is of the form $\code{throw}; q''$, then this satisfies Case~\ref{case:typi1} of Definition~\ref{def:typing_interpretation_stm}.

  \item If $q'$ is non-empty and not of the form $\code{throw}; q''$, then $(q';\bot, \Delta', s') \in \STYPI$, $\FINISH{q';\bot, \Delta', s'} = ((\Delta, x:\TVAR{B_{s_1}}), s)$, and there exists at least one transition 
    \begin{equation*}
      \Sigma; \Gamma; \Delta' \ssemDash[{s'}] \CONF{q';\bot, \ENV{TSV}} \trans \Sigma; \Gamma; \Delta'' \ssemDash[{s''}] \CONF{q'';\bot, \ENV{T}; \ENV{SV}'}
    \end{equation*}

    \noindent for all appropriately shaped $\ENV{SV}$, with $(q'';\bot, \Delta'', s'') \in \STYPI$.
    There are now two possibilities, depending on the shape of $q''$:
    \begin{itemize}
      \item If $q''$ is empty, the reduct is of the form $\Sigma; \Gamma; \Delta \ssemDash \CONF{\bot, \ENV{T}; \ENV{S}'; \ENV{V}', x:v''}$ and $\Delta'$ must have contained the entry $x:\TVAR{B_{s_1}}$ as well; i.e.\@ $\Delta' = \Delta'', x:\TVAR{B_{s_1}}$ for some $\Delta''$.
        Thus, the transition is of the form 
        \begin{align*}
                  & \Sigma; \Gamma; \Delta'', x:\TVAR{B_{s_1}} \ssemDash[{s'}] \CONF{q';\bot, \ENV{TS}; \ENV{V}, x:v'} \\
          \trans ~& \Sigma; \Gamma; \Delta, x:\TVAR{B_{s_1}} \ssemDash \CONF{\bot, \ENV{T}; \ENV{S}'; \ENV{V}', x:v''}        
        \end{align*}
        
        \noindent for some $\Delta''$, $s'$, $\ENV{S}'$, $\ENV{V}'$ and $v''$, since we cannot know which, if any, of these components were affected by the transition.
        Then, the transition
        \begin{align*}
                  & \Sigma; \Gamma; \Delta'', x:\TVAR{B_{s_1}} \ssemDash[{s'}] \CONF{q';\DEL{x};\bot, \ENV{TS}; \ENV{V}, x:v'} \\
          \trans ~& \Sigma; \Gamma; \Delta, x:\TVAR{B_{s_1}} \ssemDash \CONF{\DEL{x};\bot, \ENV{T}; \ENV{S}'; \ENV{V}', x:v''}        
        \end{align*}

        \noindent can likewise be concluded, and $(\DEL{x};\bot, (\Delta, x:\TVAR{B_{s_1}}), s) \in \CANDSTYPI$ by construction.

      \item Otherwise, if $q''$ is non-empty, then the reduct is instead of the form $\Sigma; \Gamma; \Delta'' \ssemDash[{s''}] \CONF{q''; \bot, \ENV{T}; \ENV{SV}'}$ and $(q'';\bot, \Delta'', s'') \in \STYPI$.
        Then, the transition 
        \begin{align*}
                  & \Sigma; \Gamma; \Delta'  \ssemDash[{s'}]  \CONF{q';\DEL{x};\bot, \ENV{TS}; \ENV{V}, x:v'} \\
          \trans ~& \Sigma; \Gamma; \Delta'' \ssemDash[{s''}] \CONF{q'';\DEL{x};\bot, \ENV{T}; \ENV{S}'; \ENV{V}', x:v''}        
        \end{align*}

        \noindent can likewise be concluded.
        Finally, by Theorem~\ref{theorem:stacks_term_level}, we have that $\FINISH{q'';\bot, \Delta'', s''} = ((\Delta, x:\TVAR{B_{s_1}}), s)$, so by construction we also have that $(q'';\DEL{x};\bot, \Delta'', s'') \in \CANDSTYPI$.
        \qedhere
    \end{itemize}
\end{itemize}
\end{proof}

\begin{lemma}[Variable assignment]\label{lemma:compatibility_stm_v_ass}
If $\Sigma; \Gamma; \Delta \vDash e : B_s$ and $\Delta(x) = \TVAR{B_s}$, then $\Sigma; \Gamma; \Delta; \ENV{T} \vDash \code{$x$ := $e$} : \TCMD{s}$.
\end{lemma}
\begin{proof}
By the definition of $\Sigma; \Gamma; \Delta; \ENV{T} \vDash \code{$x$ := $e$} : \TCMD{s}$, we must exhibit a typing interpretation $\STYPI$ containing the triplet $(\code{$x$ := $e$}; \bot, \Delta, s)$.
We construct a candidate typing interpretation as follows:
\begin{equation*}
  \CANDSTYPI \DEFSYM \SET{ (\code{$x$ := $e$}; \bot, \Delta, s), (\bot, \Delta, s) }
\end{equation*}
We must now show that $\CANDSTYPI$ indeed is a typing interpretation.
We consider each of the triplets we have added:
\begin{itemize}
  \item $(\bot, \Delta, s) \in \CANDSTYPI$ by Case~\ref{case:typi1} of Definition~\ref{def:typing_interpretation_stm}.

  \item From $\Sigma; \Gamma; \Delta \vDash e : B_s$ and Definition~\ref{def:typing_interpretation_expr}, we know that 
    \begin{equation*}
      \Sigma; \Gamma; \Delta \esemDash \CONF{e, \ENV{SV}} \trans v
    \end{equation*}

    \noindent where $v$ is a value of type $B$ and level $s$ for any $\ENV{SV}$ built from $\Sigma; \Gamma; \Delta$ according to Definition~\ref{def:construction_envsv}. 
    Furthermore, from $\Delta(x) = \TVAR{B_s}$, we know that $\Delta = \Delta', x:\TVAR{B_s}$ for some $\Delta'$.
    By the construction in Definition~\ref{def:construction_envsv}, we therefore also have that $\ENV{V}$ is such that $\ENV{V} = \ENV{V}', x:v'$ for some $\ENV{V}'$ and some $v'$ of type $B_s$.
    We can then conclude the transition
    \begin{align*}
              \Sigma; \Gamma; \Delta \ssemDash \CONF{\code{$x$ := $e$};\bot, \ENV{TS};\ENV{V}', x:v'} 
      \trans  \Sigma; \Gamma; \Delta \ssemDash \CONF{\bot, \ENV{TS}; \ENV{V}, x:v}     
    \end{align*}

    \noindent by rule \nameref{stm_rt_assv}.
    By Corollary~\ref{theorem:stacks_runtime_states_abstraction}, the transition can then be concluded for any $\ENV{SV}$ built from $\Sigma; \Gamma; \Delta$ according to Definition~\ref{def:construction_envsv}.
    There are no other possible transitions.
    Finally, $(\bot, \Delta, s) \in \CANDSTYPI$ by construction.
\qedhere
\end{itemize}
\end{proof}

\begin{lemma}[Field assignment]\label{lemma:compatibility_stm_f_ass}
If $\Sigma; \Gamma; \Delta \vDash e : B_s$, $\Delta(\code{this}) = \TVAR{I_{s_1}}$, $\Gamma(I)(p) = \TVAR{B_s}$, and $s_1 \ordleq s$, then $\Sigma; \Gamma; \Delta; \ENV{T} \vdash \code{this.$p$ := $e$} : \TCMD{s}$.
\end{lemma}
\begin{proof}
By the definition of $\Sigma; \Gamma; \Delta; \ENV{T} \vDash \code{this.$p$ := $e$} : \TCMD{s}$, we must exhibit a typing interpretation $\STYPI$ containing the triplet $(\code{this.$p$ := $e$}; \bot, \Delta, s)$.
We construct a candidate typing interpretation as follows:
\begin{equation*}
  \CANDSTYPI \DEFSYM \SET{ (\code{this.$p$ := $e$}; \bot, \Delta, s), (\bot, \Delta, s) }
\end{equation*}
We must now show that $\CANDSTYPI$ indeed is a typing interpretation.
We consider each of the triplets we have added:
\begin{itemize}
  \item $(\bot, \Delta, s) \in \CANDSTYPI$ by Case~\ref{case:typi1} of Definition~\ref{def:typing_interpretation_stm}.

  \item From $\Sigma; \Gamma; \Delta \vDash e : B_s$ and Definition~\ref{def:typing_interpretation_expr}, we know that 
    \begin{equation*}
      \Sigma; \Gamma; \Delta \esemDash \CONF{e, \ENV{SV}} \trans v
    \end{equation*}

    \noindent where $v$ is a value of type $B$ and level $s$ for any $\ENV{SV}$ built from $\Sigma; \Gamma; \Delta$ according to Definition~\ref{def:construction_envsv}. 
    From $\Delta(\code{this}) = \TVAR{I_{s_1}}$, we know that $\Delta = \Delta', \code{this}:\TVAR{I_{s_1}}$ for some $\Delta'$, and $s_1 \ordleq s$.
    Likewise, from $\Gamma(I)(p) = \TVAR{B_s}$, we know that the field $p$ exist on the interface $I$ and is declared with type $\TVAR{B_s}$.
    By the construction in Definition~\ref{def:construction_envsv}, we therefore also have that the environments are such that     
    \begin{align*}
      \ENV{V} & = \ENV{V}', \code{this}:X \\
      \ENV{S} & = \ENV{S}', X:\ENV{F}     \\ 
      \ENV{F} & = \ENV{F}', p:v'
    \end{align*}
    
    \noindent for some $\ENV{V}'$, $\ENV{S}'$, $\ENV{F}'$, for some address $X$ of type $I_{s_1}$ and for some value $v'$ of type $B_s$.
    We can then conclude the transition
    \begin{align*}
              & \Sigma; \Gamma; \Delta \ssemDash \CONF{\code{this.$p$ := $e$};\bot, \ENV{T};\ENV{S}\REBIND{X}{\ENV{F}\REBIND{p}{v'}};\ENV{V}} \\
      \trans ~& \Sigma; \Gamma; \Delta \ssemDash \CONF{\bot, \ENV{T};\ENV{S}\REBIND{X}{\ENV{F}\REBIND{p}{v}};\ENV{V}}
    \end{align*}

    \noindent by rule \nameref{stm_rt_assf}.
    By Corollary~\ref{theorem:stacks_runtime_states_abstraction}, the transition can then be concluded for all $\ENV{SV}$ built from $\Sigma; \Gamma; \Delta$ according to Definition~\ref{def:construction_envsv}.
    There are no other possible transitions.
    Finally, $(\bot, \Delta, s) \in \CANDSTYPI$ by construction.
\qedhere
\end{itemize}
\end{proof}

\begin{figure}\centering
\begin{semantics}
  \RULE[st-envt$_1$]
    { }
    { \Sigma; \Gamma; \ENV{T} \vDash \ENV{T}^\EMPTYSET }
  \RULE[st-envt$_2$]
    { 
      \Sigma; \Gamma; \ENV{T} \vDash \ENV{T}'  \AND
      \Sigma; \Gamma; \ENV{T} \vDash_X \ENV{M}
    }
    { \Sigma; \Gamma; \ENV{T} \vDash \ENV{T}', X:\ENV{M} }
  \RULE[st-envm$_1$]
    { }
    { \Sigma; \Gamma; \ENV{T} \vDash_X \ENV{M}^\EMPTYSET }
  \RULE[st-envm$_2$][st_env_envm2](s \ordleq s_1) 
    { 
      \Sigma; \Gamma; \ENV{T}         \vDash_X \ENV{M} \AND 
      \Sigma; \Gamma; \Delta; \ENV{T} \vDash S : \TCMD{s} 
    }
    { \Sigma; \Gamma; \ENV{T} \vDash_X \ENV{M}, f:(x_1, \ldots, x_n, S) }
  \WHERE{ I_{s_1} }           { = \Gamma(X)}          
  \WHERE{ \TVAR{\TINT_{s_2}} }{ = \Gamma(I)(\code{balance}) } 
  \WHERE{\Gamma(I)(f)}        { = \TSPROC{(B_1, s_1'), \ldots, (B_n, s_n')}{s} }
  \WHERE{\Delta}              { = \code{this}:\TVAR{I_{s_1}}, \code{value}:\TVAR{\TINT_{s_2}}, \code{sender}:\TVAR{\ITOP_{\STOP}}, x_1:\TVAR{B_1, s_1'}, \ldots, x_n:\TVAR{B_n, s_n'} }
\end{semantics}
\caption{Semantic type rules for code environments $\ENV{T}$.}
\label{fig:semantic_type_rules_envtm}
\end{figure}

The semantic notion of type safety for the code environment $\ENV{T}$, using Definition~\ref{def:typing_interpretation_stm} is given by the rules in Figure~\ref{fig:semantic_type_rules_envtm}.
As usual, the definition closely follows the syntactic type rules of Figure~\ref{fig:type_rules_env},
except that we now use the double turnstile, and we also have $\ENV{T}$ on the left-hand side.

\begin{lemma}[Method call]\label{lemma:compatibility_stm_call}
Assume that $\Sigma; \Gamma; \ENV{T} \vDash \ENV{T}$, 
 $\Sigma; \Gamma; \Delta \vDash e_1 : I^Y_s$, 
 $\Sigma; \Gamma; \Delta \vDash e_2 : \TINT_s$, 
 $\Sigma; \Gamma; \Delta \vDash e_i' : (B_i', s_i')$ for each $i\in \{1,\ldots,h\}$, 
 $\Delta(\code{this}) = I^X_{s_1}$, 
 $\Gamma(I^X)(\code{balance}) = \TVAR{\TINT_{s_3}}$, 
 $\Gamma(I^Y)(\code{balance}) = \TVAR{\TINT_{s_4}}$, 
 $\Gamma(I^Y)(f) = \TSPROC{(B_1', s_1'), \ldots, (B_h', s_h')}{s}$, and
 $s_1 \ordleq s \ordleq s_3, s_4$. 
Then $\Sigma; \Gamma; \Delta; \ENV{T} \vDash \CALL{e_1}{f}{e_1', \ldots, e_h'}{e_2} : \TCMD{s}$.
\end{lemma}
\begin{proof}
By the definition of $\Sigma; \Gamma; \Delta; \ENV{T} \vDash \CALL{e_1}{f}{e_1', \ldots, e_h'}{e_2} : \TCMD{s}$, we must exhibit a typing interpretation $\STYPI$ which is such that it contains the triplet $(\CALL{e_1}{f}{e_1', \ldots, e_h'}{e_2}; \bot, \Delta, s)$.

From $\Sigma; \Gamma; \Delta \vDash e_1 : I^Y_s$ and Theorem~\ref{theorem:expressions_runtime_safety} we know that $e_1$ evaluates to an address $Y$ of type $I'$ such that $\Sigma \vdash I' \SUBS I^Y$ and of a level $s'$ such that $s' \ordleq s$.
Furthermore, from $\Gamma(I^Y)(f) = \TSPROC{(B_1', s_1'), \ldots, (B_h', s_h')}{s}$, we know that every such interface $I'$ declares a method $f$ with this signature, or a subtype thereof, since every subtype of $I^Y$ must declare at least the same methods.

Finally, from $\Sigma; \Gamma; \ENV{T} \vDash \ENV{T}$ (Figure~\ref{fig:semantic_type_rules_envtm}), we know that, for each such declaration of the method $f$ in $\ENV{T}$, it holds that $\Sigma; \Gamma; \Delta_i; \ENV{T}\vDash S_i : \TCMD{s_i}$, where $S_i$ is the method body, and $\Delta_i$ and $s_i$ are obtained from the method signature in $\Gamma$.
Since both $\ENV{T}$ and $\Gamma$ are finite, the number of interfaces in $\Gamma$ which are subtypes of $I^Y$ is also finite, and the number of implementations of $I^Y$ and any of its subtypes in $\ENV{T}$ is also finite.
We use $i$ to index this collection in the following.
Thus, we know there exists a finite collection of typing interpretations $\STYPI^i$ containing the triplets $(S_i, \Delta_i, s_i)$ for each declaration of $f$.\footnote{The $\Delta_i$ and $s_i$ may differ because of subtyping of the method signature, and because of different types for the magic variable \code{this}. Each of the types in any $\Delta_i$ may be supertypes of the types from the signature of $f$ on $I^Y$, and $s_i \ordgeq s$, by contravariance of the type constructor.}
We then construct a candidate typing interpretation as follows:
\begin{align*}
    \CANDSTYPI & \DEFSYM \SET{ (\CALL{e_1}{f}{e_1', \ldots, e_h'}{e_2};\bot, \Delta, s) }                \\ 
               & \UNION \bigcup_i \SET{ (q_i';(\ENV{V}, \Delta);s;\bot, \Delta_i', s_i') ~\bigg|                  %
               \begin{array}{r @{~} l}                                                                   %
                        (q_i';\bot, \Delta_i', s_i') \in \STYPI^i
                 \land  \FINISH{q_i';\bot, \Delta_i', s_i'} = (\Delta_i, s_i)           
               \end{array}                                                                               %
               }                                                                                       \\
               & \UNION \bigcup_i \SET{ (s;\bot, \Delta, s_i) } \\ 
               & \UNION \SET{ (\bot, \Delta, s) }
\end{align*}

\noindent where $\ENV{V}$ is constructed from $\Sigma$, $\Gamma$ and $\Delta$ according to the method in Definition~\ref{def:construction_envsv}.

We must now show that $\CANDSTYPI$ indeed is a typing interpretation.
First, we consider the three singleton sets we have added:
\begin{itemize}
  \item The triplet $(\bot, \Delta, s)$ satisfies Case~\ref{case:typi1} of Definition~\ref{def:typing_interpretation_stm}.

  \item Consider the triplets $(s;\bot, \Delta, s_i)$ for each $i$.
    For each $s_i$, we know that $s_i \ordgeq s$ from the definition of subtyping of method declarations in the consistency rules for $\Sigma$ (Figure~\ref{fig:consistency_rules_sigma}).
    Thus, this satisfies the side condition of \nameref{stm_rt_restore}, so the transition 
    \begin{equation*}
             \Sigma; \Gamma; \Delta \ssemDash[{s_i}] \CONF{s; \bot, \ENV{TSV}} 
      \trans \Sigma; \Gamma; \Delta \ssemDash \CONF{\bot, \ENV{TSV}}
    \end{equation*}

    \noindent can be concluded by that rule for any appropriately shaped $\ENV{SV}$, and, by construction, $(\bot, \Delta, s) \in \CANDSTYPI$.

  \item Consider the triplet $(\CALL{e_1}{f}{e_1', \ldots, e_h'}{e_2};\bot, \Delta, s)$.
    Pick any $\ENV{SV}$ built from $\Sigma$, $\Gamma$ and $\Delta$ according to Definition~\ref{def:construction_envsv}.
    From 
    \begin{equation*}
    \Sigma; \Gamma; \Delta \vDash e_2 : \TINT_s \qquad  \Sigma; \Gamma; \Delta \vDash e_1' : (B_1', s_1') \qquad \ldots \qquad \Sigma; \Gamma; \Delta \vDash e_h' : (B_h', s_h')
    \end{equation*} 

    \noindent by Definition~\ref{def:typing_interpretation_expr} and Theorem~\ref{theorem:expressions_runtime_safety}, we know that each of the expressions evaluates to a value of the appropriate type. 
    From this and the remaining premises of the lemma, we know that the transition 
    \begin{align*}
\!\!\!\!\!\Sigma; \Gamma; \Delta \ssemDash \CONF{ \CALL{e_1}{f}{\VEC{e}}{e_2} ; \bot, \ENV{TSV}}    
      \trans \Sigma; \Gamma; \Delta_i \ssemDash[{s_i}] \CONF{S_i ; (\ENV{V}, \Delta) ; s ; \bot, \ENV{T}, \ENV{SV}'}  
    \end{align*}

    \noindent can be concluded by \nameref{stm_rt_call} and, as we know that $f \in \DOM{\ENV{T}(Y)}$, the rule \nameref{stm_rt_fcall} cannot be used. 
    Thus, no other transitions can be concluded and $(S_i; (\ENV{V}, \Delta); s; \bot, \Delta_i, s_i) \in \CANDSTYPI$ by construction.
\end{itemize}

Next we consider an arbitrary triplet $(q_i';(\ENV{V}, \Delta);s;\bot, \Delta_i', s_i')$ from the sets
\begin{equation*}
  \left\{ (q_i';(\ENV{V}, \Delta);s;\bot, \Delta_i', s_i') ~\bigg|                 %
    \begin{array}{r @{~} l}                                                        %
            (q_i';\bot, \Delta_i', s_i') \in \STYPI^i             
     \land  \FINISH{q_i';\bot, \Delta_i', s_i'} = (\Delta_i, s_i) 
    \end{array}                                                                    %
  \right\}                                                                                       
\end{equation*}

\noindent As $q_i'$ can be empty, we have three cases to consider:
\begin{itemize}
  \item If $q_i'$ is empty, then it must have come from the triplet $(\bot, \Delta_i, s_i) \in \STYPI^i$, since $\FINISH{\bot, \Delta_i, s_i} = (\Delta_i, s_i)$.
    Thus by construction, the triplet in $\CANDSTYPI$ is of the form $((\ENV{V},\Delta);s;\bot, \Delta_i, s_i)$, and we know that $\Sigma; \Gamma; \Delta \vdash \ENV{V}$.
    The transition 
    \begin{equation*}
            \Sigma; \Gamma; \Delta_i \ssemDash[{s_i}] \CONF{(\ENV{V}, \Delta); s; \bot, \ENV{TS},\ENV{V}'} 
      \trans \Sigma; \Gamma; \Delta \ssemDash[{s_i}] \CONF{s;\bot, \ENV{TSV}} 
    \end{equation*}

    \noindent can therefore be concluded by rule \nameref{stm_rt_return}, for any appropriately shaped $\ENV{S},\ENV{V}'$, and $(s;\bot, \Delta, s_i) \in \CANDSTYPI$ by construction.

  \item If $q_i'$ is of the form $\code{throw}; q_i''$, then this satisfies Case~\ref{case:typi1} of Definition~\ref{def:typing_interpretation_stm}.

  \item If $q_i'$ is non-empty and not of the form $\code{throw}; q_i''$, then  $(q_i';\bot, \Delta_i', s_i') \in \STYPI^i$, $\FINISH{q_i';\bot, \Delta_i', s_i'} = (\Delta_i, s_i)$, and there exists at least one transition 
    \begin{equation*}
      \Sigma; \Gamma; \Delta_i' \ssemDash[{s_i'}] \CONF{q_i';\bot, \ENV{TSV}} \trans \Sigma; \Gamma; \Delta_i'' \ssemDash[{s_i''}] \CONF{q_i'';\bot, \ENV{T}; \ENV{SV}'}
    \end{equation*}

    \noindent for all appropriately shaped $\ENV{SV}$ built from $\Sigma; \Gamma; \Delta_i'$ according to Definition~\ref{def:construction_envsv}, with $(q_i'';\bot, \Delta_i'', s_i'') \in \STYPI^i$.
    There are now two possibilities, depending on $q_i''$:
    \begin{itemize}
      \item If $q_i''$ is empty, the reduct is of the form $\Sigma; \Gamma; \Delta_i \ssemDash[{s_i}] \CONF{\bot, \ENV{T}; \ENV{SV}'}$.
        Then, the transition
        \begin{align*}
\!\!\!\!\!\Sigma; \Gamma; \Delta_i' \ssemDash[{s_i'}] \CONF{q_i';(\ENV{V},\Delta);s;\bot, \ENV{TSV}} 
          \trans  \Sigma; \Gamma; \Delta_i \ssemDash[{s_i}] \CONF{(\ENV{V},\Delta);s;\bot, \ENV{T}; \ENV{SV}'}        
        \end{align*}

        \noindent can likewise be concluded, and $((\ENV{V},\Delta);s;\bot, \Delta_i, s_i) \in \CANDSTYPI$ by construction.

      \item Otherwise, if $q_i''$ is non-empty, then the reduct is instead of the form $\Sigma; \Gamma; \Delta_i'' \ssemDash[{s_i''}] \CONF{q_i''; \bot, \ENV{T}; \ENV{SV}'}$ and $(q_i'';\bot, \Delta_i'', s_i'') \in \STYPI^i$.
        Then, the transition 
        \begin{align*}
\!\!\!\!\! \Sigma; \Gamma; \Delta'  \ssemDash[{s_i'}]  \CONF{q_i';(\ENV{V},\Delta);s;\bot, \ENV{TSV}} 
          \trans  \Sigma; \Gamma; \Delta'' \ssemDash[{s_i''}] \CONF{q_i'';(\ENV{V},\Delta);s;\bot, \ENV{T}; \ENV{SV}'}
        \end{align*}

        \noindent can likewise be concluded.
        Finally, by Theorem~\ref{theorem:stacks_term_level}, we have that $\FINISH{q_i'';\bot, \Delta_i'', s_i''} = (\Delta_i, s_i)$, so by construction we also have that $(q_i'';(\ENV{V},\Delta);s;\bot, \Delta_i'', s_i'') \in \CANDSTYPI$.
\qedhere
    \end{itemize}
\end{itemize}
\end{proof}

Soundness of the rule for the fallback function is obtained through a lemma, whose proof consists in the same steps as those in the proof of Lemma~\ref{lemma:compatibility_stm_call}; hence, we shall omit this.

Before we can show the lemma for delegate calls (Lemma~\ref{lemma:compatibility_stm_dcall}), we need a few auxiliary results, which ensure that changing the types of the two magic variables \code{this} and \code{value} does not permit any unwanted information flows.
In the following three lemmas below (Lemmas~\ref{lemma:dcall_exp_trans}--\ref{lemma:dcall_stm_type}) we use the following notations.
Let
\begin{itemize}
  \item $\Delta_1 \DEFSYM \Delta, \code{this}:\TVAR{I_1, s_1}, \code{value}:\TVAR{\TINT, s_1'}$,
  \item $\Delta_2 \DEFSYM \Delta, \code{this}:\TVAR{I_2, s_2}, \code{value}:\TVAR{\TINT, s_2'}$,
\end{itemize}

\noindent where $\Sigma \vdash (I_2, s_2) \SUBS (I_1, s_1)$ and $s_2' \ordleq s_1'$, and 
\begin{equation*}
  \forall q \in (\DOM{\Gamma(I^Y)} \INTERSECT (\FNAMES \UNION \{\code{balance}\})) . \Gamma(I_1)(q) = \Gamma(I_2)(q) ,
\end{equation*}

\noindent Furthermore, let $\ENV{SV}^1$ be built from $\Gamma, \Delta_1$ and $\ENV{SV}^2$ be built from $\Gamma, \Delta_2$ according to Definition~\ref{def:construction_envsv}.

\begin{lemma}\label{lemma:dcall_exp_trans}
If   $\Sigma; \Gamma; \Delta_1 \esemDash \CONF{e, \ENV{SV}^1} \trans v_1$, 
then $\Sigma; \Gamma; \Delta_2 \esemDash \CONF{e, \ENV{SV}^2} \trans v_2$. 
\end{lemma}
\begin{proof}
By induction on the derivation of the transition.
For the base step, 
the case for \nameref{exp_rt_val} is immediate and
the case for \nameref{exp_rt_var} is immediate as well, except in the two special cases where $x = \code{this}$ or $x = \code{value}$;
then, subtyping is used to coerce the type up to match $B_s$.
For the inductive step, we reason on the last rule used.
When \nameref{exp_rt_op} is the last rule, we conclude by $n$ applications of the induction hypothesis (one for each of the operands $e_1, \ldots, e_n$).
When \nameref{exp_rt_field} is the last rule, then $e = e'.p$; we apply the induction hypothesis to $e'$ and conclude by \nameref{exp_rt_field}.
\end{proof}

\begin{lemma}\label{lemma:dcall_stm_trans}
If   $\Sigma; \Gamma; \Delta_1 \ssemDash \CONF{Q, \ENV{T};\ENV{SV}^1} \trans \Sigma; \Gamma; \Delta_1' \ssemDash[{s'}] \CONF{Q', \ENV{T};\ENV{SV}^{1'}}$,
then $\Sigma; \Gamma; \Delta_2 \ssemDash \CONF{Q, \ENV{T};\ENV{SV}^2} \trans \Sigma; \Gamma; \Delta_2' \ssemDash[{s'}] \CONF{Q', \ENV{T};\ENV{SV}^{2'}}$
\end{lemma}
\begin{proof}
By a case analysis of the rule used to conclude the transition.
\begin{itemize}
  \item For \nameref{stm_rt_skip}, \nameref{stm_rt_whilefalse}, \nameref{stm_rt_delv}, \nameref{stm_rt_return}, and \nameref{stm_rt_restore} the result is immediate.
  \item For \nameref{stm_rt_if}, \nameref{stm_rt_whiletrue}, \nameref{stm_rt_decv}, \nameref{stm_rt_assv}, \nameref{stm_rt_call}, \nameref{stm_rt_dcall} and \nameref{stm_rt_fcall}, we apply Lemma~\ref{lemma:dcall_exp_trans} for the expression(s) occurring in the premises and then conclude by the same rule.
  \item The case for \nameref{stm_rt_assf} is special.
    Here, $S = \code{this.$p$ := e}$, and from the premise we know that 
    \begin{align*}  
      \Delta_1(\code{this}) & = \TVAR{I_1, s_1} \\ 
      \Gamma(I_1)(p)        & = \TVAR{B, s''}   \\ 
      s_1, s                & \ordleq s''
    \end{align*}

    By assumption, $\Delta_2(\code{this}) = \TVAR{I_2, s_2}$ and $\Sigma \vdash (I_2, s_2) \SUBS (I_1, s_1)$, which implies $s_2 \ordleq s_1$, and therefore also that $s_2 \ordleq s''$ by transitivity of $\ordleq$.
    Furthermore, by the requirement 
    \begin{equation*}
      \forall q \in (\DOM{\Gamma(I_1)} \INTERSECT (\FNAMES \UNION \{\code{balance}\})) . \Gamma(I_1)(q) = \Gamma(I_2)(q) ,
    \end{equation*}

    \noindent we know that also $\Gamma(I_2)(p) = \TVAR{B, s''}$.
    Finally, we apply Lemma~\ref{lemma:dcall_exp_trans} for the evaluation of $e$ and conclude by \nameref{stm_rt_assf}.
    \vspace*{-.5cm}
\end{itemize}
\end{proof}

\begin{lemma}\label{lemma:dcall_stm_type}
If   $\Sigma; \Gamma; \Delta_1, \ENV{T} \vDash S : \TCMD s$, 
then $\Sigma; \Gamma; \Delta_2, \ENV{T} \vDash S : \TCMD s$.
\end{lemma}
\begin{proof}
We know there exists a typing interpretation $\STYPI$ containing $(S;\bot, \Delta_1, s)$.
We then construct a new candidate typing interpretation $\CANDSTYPI$ containing $(S;\bot, \Delta_2, s)$ as follows:
\begin{align*}
  \CANDSTYPI & \DEFSYM \SET{ (Q, (\Delta_2, \hat\Delta), s') \;|\; (Q, (\Delta_1, \hat\Delta), s') \in \STYPI  } \\ 
             & \UNION  \SET{ (q_1; (\ENV{V}', (\Delta_2, \hat\Delta)); q_2; \bot, \Delta', s') \;|\; (q_1; (\ENV{V}', (\Delta_1, \hat\Delta)); q_2; \bot, \Delta', s') \in \STYPI }
\end{align*}

\noindent where $\hat\Delta$ represents a (possibly empty) sequence of further type assumptions extending $\Delta_1$ resp.\@ $\Delta_2$.
Clearly, $(S;\bot, \Delta_2, s) \in \CANDSTYPI$, with $Q = S;\bot$ and $\hat\Delta = \EMPTYSET$.
All that remains is now to show that $\CANDSTYPI$ in fact is a typing interpretation.
For this, pick an arbitrary triplet $(Q', \Delta', s') \in \CANDSTYPI$.
We distinguish 3 cases:
\begin{itemize}
  \item If $Q' \in \TSTACKS$ then this satisfies Case~\ref{case:typi1} of Definition~\ref{def:typing_interpretation_stm}.
  \item Otherwise, if $\Delta' = \Delta_2, \hat\Delta$, then we know that $(Q', (\Delta_1, \hat\Delta), s') \in \STYPI$, which implies that the transition 
    \begin{equation*}
             \Sigma; \Gamma; (\Delta_1,\hat\Delta)  \ssemDash[{s'}]  \CONF{Q',  \ENV{T},\ENV{SV}} 
      \trans \Sigma; \Gamma; (\Delta_1,\hat\Delta') \ssemDash[{s''}] \CONF{Q'', \ENV{T},\ENV{SV}^{1'}} 
    \end{equation*}

    \noindent can be concluded by some rule, for any appropriately shaped $\ENV{SV}^{1'}$.
    By Lemma~\ref{lemma:dcall_stm_trans}, a transition
    \begin{equation*}
             \Sigma; \Gamma; (\Delta_2,\hat\Delta)  \ssemDash[{s'}]  \CONF{Q',  \ENV{T},\ENV{SV}} 
      \trans \Sigma; \Gamma; (\Delta_2,\hat\Delta') \ssemDash[{s''}] \CONF{Q'', \ENV{T},\ENV{SV}^{2'}} 
    \end{equation*}

    \noindent can therefore also be concluded, and $(Q'', (\Delta_2,\hat\Delta), s'') \in \CANDSTYPI$ by construction.
    
  \item Otherwise, if $\Delta' \neq \Delta_2, \hat\Delta$, then the triplet must have come from one where $(\ENV{V}', (\Delta_1, \hat\Delta))$ resided on the stack.
    Thus, $Q' = q_1; (\ENV{V}', (\Delta_2, \hat\Delta)); q_2; \bot$, which by construction must have come from $Q = q_1; (\ENV{V}', (\Delta_1, \hat\Delta)); q_2; \bot$.
    There are now two sub-cases:
    \begin{itemize}
      \item If $q_1$ is non-empty, we know that the transition 
      \begin{equation*}
               \Sigma; \Gamma; \Delta'  \ssemDash[{s'}]  \CONF{Q,  \ENV{T},\ENV{SV}} 
        \trans \Sigma; \Gamma; \Delta'' \ssemDash[{s''}] \CONF{Q_1, \ENV{T},\ENV{SV}^{1}} 
      \end{equation*}

      \noindent exists, and we can therefore conclude a similar transition for $Q'$, since only the top-most element on the stack can affect transitions.
      Thus 
      \begin{equation*}
               \Sigma; \Gamma; \Delta'  \ssemDash[{s'}]  \CONF{Q',  \ENV{T},\ENV{SV}} 
        \trans \Sigma; \Gamma; \Delta'' \ssemDash[{s''}] \CONF{Q_2, \ENV{T},\ENV{SV}^{2}} 
      \end{equation*}

      \noindent and $(Q_2, \Delta'', s'') \in \CANDSTYPI$ by construction.

      \item Otherwise, if $q_1$ is empty, we have that $Q' = (\ENV{V}', (\Delta_2, \hat\Delta)); q_2; \bot$, which by construction must have come from $Q = (\ENV{V}', (\Delta_1, \hat\Delta)); q_2; \bot$.
        The transition for $Q$ would therefore have been concluded by \nameref{stm_rt_return}, with the reduct triplet $(q_2;\bot, (\Delta_1, \hat\Delta), s') \in \STYPI$.
        We can therefore conclude a similar transition for $Q'$ by \nameref{stm_rt_return}, and the reduct triplet $(q_2;\bot, (\Delta_2, \hat\Delta), s') \in \CANDSTYPI$ by construction.
    \end{itemize}
\vspace*{-.5cm}
\end{itemize}
\end{proof}

\begin{lemma}[Delegate call]\label{lemma:compatibility_stm_dcall}
Assume that $\Sigma; \Gamma; \ENV{T} \vDash \ENV{T}$, 
 $\Sigma; \Gamma; \Delta \vDash e : I^Y_s$, 
 $\Sigma; \Gamma; \Delta \vDash e_i' : (B_i', s_i')$ for each $i\in \{1,\ldots,h\}$, 
 $\Delta(\code{this}) = I^X_{s_1}$, 
 $\Sigma \vdash I^X \SUBS I^Y$,  
 $\forall q \in (\DOM{\Gamma(I^Y)} \INTERSECT (\FNAMES \UNION \{\code{balance}\})).\ \Gamma(I^Y)(q) = \Gamma(I^X)(q)$,
 $\Gamma(I^Y)(f) = \TSPROC{(B_1', s_1'), \ldots, (B_h', s_h')}{s}$, and
 $s_1 \ordleq s$.
Then $\Sigma; \Gamma; \Delta; \ENV{T} \vdash \DCALL{e}{f}{e_1', \ldots, e_h'} : \TCMD{s}$.
\end{lemma}
\begin{proof}
By the definition of $\Sigma; \Gamma; \Delta; \ENV{T} \vDash \DCALL{e}{f}{e_1', \ldots, e_h'} : \TCMD{s}$, we must exhibit a typing interpretation $\STYPI$ that contains the triplet $(\DCALL{e}{f}{e_1', \ldots, e_h'}; \bot, \Delta, s)$.
The proof is almost exactly the same as the proof for Lemma~\ref{lemma:compatibility_stm_call} above.
The only differences are that we here do not have a value parameter $e_2$, and the magic variables \code{this} and \code{value} are not changed in the new type environments $\Delta_i$. 
As we know from $\Sigma; \Gamma; \ENV{T} \vDash \ENV{T}$ that the body $S$ is type safe w.r.t.\@ a type environment $\Delta_i'$ in which the types of these variables are set as in the proof of Lemma~\ref{lemma:compatibility_stm_call}, i.e.\@ $\Sigma; \Gamma; \Delta_i'; \ENV{T} \vDash S : \TCMD{s}$, we use Lemma~\ref{lemma:dcall_stm_type} to conclude that $\Sigma; \Gamma; \Delta_i; \ENV{T} \vDash S : \TCMD{s}$ as well. 
The rest of the proof is as the proof for Lemma~\ref{lemma:compatibility_stm_call}, so we do not repeat it.
\end{proof}

\begin{proof}[Proof of Theorem~\ref{thm:compat_comms}]
By induction on the derivations of the two type judgments 
\begin{equation*}
  \Sigma; \Gamma; \EMPTYSET \vdash \ENV{T} \qquad\text{and}\qquad \Sigma; \Gamma; \Delta \vdash S : \TCMD{s}, 
\end{equation*}

\noindent and a case analysis of the last rules used.
Both statements have to be shown together, because of their mutual dependency.
Each syntactic construct is typed by a particular syntactic type rule, and the case can then be concluded by the corresponding compatibility lemma, and using Lemma~\ref{lemma:compatibility_stm_coercion}(2) for cases where the syntactic type rule has a side condition of the form $s \ordleq s'$.

As an example, suppose $\Sigma; \Gamma; \Delta \vdash S : \TCMD{s}$ was concluded by \nameref{stm_t_if}.
Then the judgment is of the form 
\begin{equation*}
  \Sigma; \Gamma; \Delta \vdash \code{if $e$ then $S_\TRUE$ else $S_\FALSE$} : \TCMD{s}
\end{equation*}

\noindent and from the premise and side condition of that rule we know that 
\begin{align*}
  \Sigma; \Gamma; \Delta & \vdash e : (\TBOOL_{s'})      \\ 
  \Sigma; \Gamma; \Delta & \vdash S_{\TRUE} : \TCMD{s'}  \\
  \Sigma; \Gamma; \Delta & \vdash S_{\FALSE} : \TCMD{s'} \\ 
  s                      & \ordleq s'
\end{align*}

\noindent By Theorem~\ref{theorem:compatibility_expressions}, we have that 
\begin{equation*}
  \Sigma; \Gamma; \Delta \vdash e : \TBOOL_{s'} \implies \Sigma; \Gamma; \Delta \vDash e : \TBOOL_{s'} 
\end{equation*}

\noindent and, by the induction hypothesis, we obtain that 
\begin{align*}
 \Sigma; \Gamma; \Delta \vDash S_{\TRUE} : \TCMD{s'} \qquad
 \Sigma; \Gamma; \Delta \vDash S_{\FALSE} : \TCMD{s'} \qquad
 \Sigma; \Gamma; \EMPTYSET \vdash \ENV{T} \implies \Sigma; \Gamma; \ENV{T} \vDash \ENV{T}
\end{align*}
This satisfies the hypotheses of Lemma~\ref{lemma:compatibility_stm_if}; hence,
\begin{equation*}
  \Sigma; \Gamma; \Delta; \ENV{T} \vDash \code{if $e$ then $S_\TRUE$ else $S_\FALSE$} : \TCMD{s'}
\end{equation*}

\noindent Finally, from $s \ordleq s'$ and  Lemma~\ref{lemma:compatibility_stm_coercion}(2), we can conclude 
\begin{equation*}
  \Sigma; \Gamma; \Delta; \ENV{T} \vDash \code{if $e$ then $S_\TRUE$ else $S_\FALSE$} : \TCMD{s}
\qedhere
\end{equation*}
\end{proof}

A statement similar to that of Theorem~\ref{thm:compat_comms} can also be shown for stacks $Q$, which finally gives us the same guarantees as a Subject-Reduction result; namely that well-typedness indeed implies safety, as expected.\footnote{There is actually a minor, technical difference: Subject reduction cannot be shown directly, using the untyped semantics, because it does not store information about entering a heightened security context, which we, in the typed semantics, represented by updating the security level on the turnstile $\ssemDash[{s'}]$ and pushing the original security level $s$ onto the stack. For example, in the case of \code{if-else}, execution of the guarded statement $S_b$ must preserve the judgment $\TCMD{s'}$, and not merely $\TCMD{s}$ for all its execution steps, but the untyped semantics does not take this difference into account.}

To conclude, we remark that the (semantic) typing rules given in Figures~\ref{fig:semantic_subtype_rules_stacks_stm}, \ref{fig:semantic_type_rules_stm} and~\ref{fig:semantic_type_rules_stacks} 
constitute our semantic type system for \TINYSOL{}.

\clearpage
\section{Up-to techniques}\label{app:upto}
\begin{proof}[Proof of Theorem~\ref{lemma:validity_upto_union}]
We proceed by showing that $\UPTOSTYPI \UNION \STYPI$ is a typing interpretation according to Definition~\ref{def:typing_interpretation_stm}.

By assumption, we know that $\UPTOSTYPI \progress \UPTOSTYPI \UNION \STYPI$, where $\STYPI$ is a typing interpretation.
Now consider an arbitrary triplet $P \in \UPTOSTYPI \UNION \STYPI$. 
If $P \in \STYPI$, then we already know that $P$ satisfies the requirements of Definition~\ref{def:typing_interpretation_stm}, since $\STYPI$ is a typing interpretation.
Otherwise, if $P \in \UPTOSTYPI$, there are two possibilities:
\begin{itemize}
  \item If $P~\not\ttrans$ then $P$ must satisfy Case~\ref{case:prog1} of Definition~\ref{def:progression}.
    But then $P$ also satisfies Case~\ref{case:typi1} of Definition~\ref{def:typing_interpretation_stm}, since the requirements are the same.

  \item If $P \ttrans P_1'$ then we know by Case~\ref{case:prog2} of Definition~\ref{def:progression} that, for all $P_2'$ such that $P \ttrans P_2'$ (and, in particular, when $P_2' = P_1'$), it holds that $P_2' \in \STYPI$.
    Thus $P_2' \in \UPTOSTYPI \UNION \STYPI$ as well.
    Therefore $P$ satisfies Case~\ref{case:typi2} of Definition~\ref{def:typing_interpretation_stm}.
\end{itemize}

Thus, we conclude that $\UPTOSTYPI \UNION \STYPI$ indeed is a typing interpretation.
By Definition~\ref{def:typing_interpretation_stm}, it therefore also holds that 
\begin{equation*}
  \UPTOSTYPI \UNION \STYPI \subseteq \STYPING
\end{equation*}
 since $\STYPING$ is the union of all typing interpretations, and so $\UPTOSTYPI \subseteq \STYPING$, as required.
\end{proof}

\clearpage
\section{Semantic Typing for Call Integrity}\label{app:CI}

In this appendix, we report the changes needed to the setting of this paper to ensure, through our approach, another security property, specifically designed for smart-contract languages: {\em call integrity}. This property has been introduced by Grishchenko, Maffei and Schneidewind in \cite{grishchenko2018} to rule out a dangerous phenomenon, called {\em reentrancy}, in Ethereum smart contracts. 
In a nutshell, reentrancy is a pattern based on mutual recursion, where one method $f$ calls another method $g$ while also transferring an amount of currency along with the call; if $g$ then immediately calls $f$ back, it may yield a recursion where $f$ continues to transfer funds to $g$ (for more concrete examples suffering from reentrancy phaenomena, see \cite{AGL25,grishchenko2018}). This could be avoided by using call integrity: this property requires that any call to a method in a `trusted' contract (say, $X$) yields the very same sequence of currency flows (i.e.\@ method calls) even if some of the other `untrusted' contracts (or their stored values) are changed.
In a sense, the code and values of the other contracts, which could be controlled by an attacker, must not be able to affect the currency flow from $X$.

The main result in \cite{grishchenko2018} is that call integrity implies the absence of reentrancy (see Theorem~1 in {\em loc.cit.}); thus, call integrity has been identified in the literature as one of the safety properties that smart contracts should have.
In \cite{AGL25} it is proved that the static type system for non-interference (reported in Appendix~\ref{app:syntactic_type_system}) also ensures call integrity for the version of \TINYSOL{} \emph{without} fallback functions.
We now want to extend this result to the present version of \TINYSOL{}, which includes fallback functions, using the semantic type system we developed in this paper.

To this end, we need a few new definitions (we adapt them from \cite{AGL25}; for a more thorough discussion, we refer to the original reference). 
Essentially, some modifications are needed to properly define the labelled transitions, since in \cite{AGL25} a big-step operational semantics was adopted, whereas here we use a small-step one.

\newcommand{\ltrans}[1]{\trans[#1]}

\begin{definition}[Trace semantics]
\label{def:traces}
A \emph{trace} of method invocations is given by
\begin{center}
\begin{syntax}[h]
  \pi \IS \epsilon \OR \TRANSACT{X}{Y}{f}{\VEC{v}}{z}, \pi 
\end{syntax}
\end{center}

\noindent where $X$ is the address of the calling contract, $Y$ is the address of the called contract, $f$ is the method name, and $\VEC{v}$ and $z$ are the actual parameters.
We annotate the small-step operational semantics with information about the invoked method (if any), to yield labelled transitions of the form $\ltrans\pi$. 
To do this, we first modify the rules in Figures~\ref{fig:semantics_stacks1} and~\ref{fig:semantics_stacks2} as follows:
\begin{itemize}
\item in all rules of Figure~\ref{fig:semantics_stacks1}, all occurrences of $\ \trans\ $ in the conclusions of the rules become $\ \ltrans\epsilon\ $;
\item rules \nameref{stm_s_call}, \nameref{stm_s_fcall} and \nameref{stm_s_dcall} respectively become:
\vspace*{.4cm}
\begin{center}\normalfont
\begin{math}
\begin{array}{l}
  \dfrac
    { \text{\em same premises and side conditions as in rule \nameref{stm_s_call}} }
    { \CONF{\CALL{e_1}{f}{\VEC{e}}{e_2};Q , \ENV{TSV}} \ltrans{\TRANSACT{X}{Y}{f}{\VEC{v}}{z}} \CONF{S;\ENV{V};Q , \ENV{T}, \ENV{SV}' }}
\vspace*{.4cm}
\\
  \dfrac
    { \text{\em same premises and side conditions as in rule \nameref{stm_s_fcall}} }
    { \CONF{\CALL{e_1}{f}{\VEC{e}}{e_2};Q , \ENV{TSV}} \ltrans{\TRANSACT{X}{Y}{f}{\VEC{v}}{z}} \CONF{S;\ENV{V};Q , \ENV{T}, \ENV{SV}' }}
\vspace*{.4cm}
\\
  \dfrac
    { \ENV V(\code{this}) = X \qquad  \text{\em same premises and side conditions as in rule \nameref{stm_s_dcall}} }
    { \CONF{\DCALL{e}{f}{\VEC{e}};Q, \ENV{TSV}} \ltrans{\TRANSACT{X}{Y}{f}{\VEC{v}}{0}} \CONF{S;Q , \ENV{T}, \ENV{SV}'} }
\vspace*{.2cm}
\end{array}
\end{math}
\end{center}
\end{itemize}

\noindent Then, the relation $\ltrans\pi$ is obtained as the 
concatenation of the labelled small-step transitions---that is, 
for every $h \geq 0$ and $\pi_1,\ldots\pi_h$ such that every $\pi_i$ is either $\epsilon$ or $\TRANSACT{X_i}{Y_i}{f_i}{\VEC{v_i}}{z_i}$,
we write $\CONF{Q_0,\ENV{TSV}^0} \ltrans{\pi_1,\ldots,\pi_h} \CONF{Q_{h},\ENV{TSV}^{h}}$ whenever  $\CONF{Q_0,\ENV{TSV}^0} \ltrans{\pi_1} \CONF{Q_1,\ENV{TSV}^1} \ltrans{\pi_2} \ldots \ltrans{\pi_h} \CONF{Q_h,\ENV{TSV}^h}$.
\end{definition}

\begin{definition}[Projection]
\label{def:project}
The \emph{projection} of a trace to a specific contract $X$, written $\pi \BARBSYM_X$, is the trace of calls with $X$ as the calling address. 
Formally:
\begin{equation*}
\epsilon \BARBSYM_X \ =\ \epsilon
\qquad\qquad
  (\TRANSACT{Z}{Y}{f}{\VEC{v}}{z}, \pi) \BARBSYM_{X} \ =\ %
  \begin{cases}
    \TRANSACT{X}{Y}{f}{\VEC{v}}{z}, (\pi \BARBSYM_{X}) & \text{if $Z = X$} \tabularnewline
    \pi \BARBSYM_{X}                                   & \text{otherwise} 
  \end{cases} 
\end{equation*}
\end{definition}

\smallskip

Call integrity is formulated in terms of \emph{transactions}, which are essentially sequences of method calls. 
These can be rendered in our setting as particular kinds of stacks, consisting of sequences of method calls interspersed with variable environments $\ENV{V}$, associating the name of the caller, $X$, to the variable \code{this}.
By executing such a stack, the variable environment will first be loaded by rule \nameref{untyped_return}, thus setting up the appropriate binding for \code{this}, and then the method call will be executed by rules \nameref{stm_s_call}/\nameref{stm_s_fcall}, which will use the currently active binding for \code{this} to set the variable \code{sender}.
Thus, we formally define transactions as:
\begin{center}
\begin{syntax}[h]
  \tau \IS \bot \OR \SET{(\code{this}, X)}; \CALL{Y}{f}{\VEC{v}}{z}; \tau
\end{syntax}
\end{center}

\noindent Finally, we need another piece of notation: given a (partial) function $\phi$ and a subset $\Xi$ of its domain, we write $\phi|_{\Xi}$ to denote the restriction of $\phi$ to $\Xi$.

\begin{definition}[Call integrity]\label{def:call_integrity}
Let \ANAMES\ 
be the set of all contracts (addresses),
$\SETNAME{X} \subseteq \ANAMES$ denote a set of trusted contracts and
$\SETNAME{Y} \DEFSYM \ANAMES \SETMINUS \SETNAME{X}$ 
stand for all other contracts.
A contract $X \in \SETNAME{X}$ has \emph{call integrity} for $\SETNAME{Y}$ if,
for every transaction $\tau$
and environments $\ENV{TS}^1$ and $\ENV{TS}^2$ such that
$\ENV{TS}^1|_{\SETNAME{X}} = \ENV{TS}^2|_{\SETNAME{X}}$, 
$\CONF{\tau, \ENV{TS}^1,\ENV V ^\emptyset} \ltrans{\pi_1} \CONF{\bot, \ENV T^1, \ENV{SV}^{1'} }$ and
$\CONF{\tau, \ENV{TS}^2,\ENV V ^\emptyset} \ltrans{\pi_2} \CONF{\bot, \ENV T^2, \ENV{SV}^{2'} }$,
it holds that $\pi_1\! \BARBSYM_X\ = \pi_2 \!\BARBSYM_X$.
\end{definition}

Notice that, like \cite{AGL25,grishchenko2018}, we require that, to satisfy call integrity, either both execution or none of them must terminate. 
This was implicit in the definition provided in \cite{AGL25}, because there a big-step operational semantics was assumed (and so every configuration with an operational evolution was ensured to terminate). 
Without this requirement, in a small-step setting there would be no way to guarantee that both traces have the same length.

Then, the typed operational semantics remains mostly unchanged; the only modification needed is to remove ``$z \neq 0 \Rightarrow$'' from the side conditions of rules \nameref{stm_rt_call} and \nameref{stm_rt_fcall}. 
This makes sense because call integrity is concerned with the (sequence of) method invocations, and these are not affected by the fact that the amount of currency transferred is 0 or not.\footnote{
	The sequence of method calls would be influenced by the currency transfer if we took into account the requirement that the amount transferred is available to the caller (i.e., the currency exchanged is at least the value contained in the caller's \code{balance} field). 
        As we noted in Footnote~\ref{fn:no_exceptions} on Page~\pageref{fn:no_exceptions}, we prefer not to consider this issue for the moment.
}
In contrast, non-interference is concerned with the values contained in some portions of the memories, which may include the \code{balance} fields; in this case, since with a transferred value $z$ that is 0 such fields are not modified, checking that $s' \ordleq s_3,s_4$ is only needed when there is an actual change of balance. 
Notice that this change has no effect on the proofs of Lemma~\ref{lemma:compatibility_stm_call} and of the analogous Lemma corresponding to rule \nameref{stm_st_fcall} in Figure~\ref{fig:semantic_type_rules_fallback}. 
Hence, also in the new framework tailored for call integrity, the typeability in the static type system implies typeability in the semantic one (see Theorem~\ref{thm:compat_comms}).

We can now also prove a result analogous to Theorem~\ref{theorem:stacks_runtime_noninterference}, stating that the typed operational semantics implies call integrity. 
To this end, we also need to provide a typed labelled semantics of the form $\Sigma; \Gamma; \Delta  \ssemDash[{s}]  \CONF{Q,  \ENV{TSV}} \ltrans{\pi} \Sigma; \Gamma; \Delta' \ssemDash[{s'}]  \CONF{Q',  \ENV{T},\ENV{SV}'} $, by modifying the rules in Figures~\ref{fig:typed_semantics_stacks1} and~\ref{fig:typed_semantics_stacks2} as we did for the untyped semantics in Definition~\ref{def:traces}.

\begin{definition}[$s$-contracts]
\label{def:sContracts}
A contract $X$ is an \emph{$s$-contract w.r.t.\@ $\Gamma$} if $\Gamma(X) = I_s$, for some $I$ such that 
  \begin{itemize}
    \item $\forall q \in (\DOM{\Gamma(I)} \cap (\FNAMES \cup \{\code{balance}\}))\ \exists B \SUCHTHAT \Gamma(I)(q) = \TVAR{B_s}$, and
    \item $\forall f \in (\DOM{\Gamma(I)} \cap \MNAMES)\ \exists \widetilde{B_{s'}} \SUCHTHAT \Gamma(I)(f) = \TSPROC{\widetilde{B_{s'}}}{s}$.
  \end{itemize}
\end{definition}

\begin{lemma}
\label{lem:HvsL}
If $\Sigma; \Gamma; \Delta \ssemDash[{s}] \CONF{S; Q, \ENV{TSV}} \ltrans{\pi} \Sigma; \Gamma; \Delta' \ssemDash[{s'}] \CONF{Q, \ENV T; \ENV{SV}'}$, then $\pi \BARBSYM_Z = \epsilon$, for every $s''$-contract $Z$ such that $s \not\ordleq s''$.
\end{lemma}


\begin{proof}
By induction on the length of the sequence of labelled transitions (that generate $\pi$).
We have 4 base cases (i.e., when the length is 1):
\begin{enumerate}
\item $S = \code{skip}$;
\item $S = \code{while $e$ do $S'$}$, for $\Sigma; \Gamma; \Delta \esemDash[(\TBOOL, \hat s)] \CONF{e, \ENV{SV}} \trans \FALSE$;
\item $S = \code{$x$ := $e$}$;
\item $S = \code{this.$p$ := $e$}$.
\end{enumerate}
In all these cases, $\pi=\epsilon$ and the claim trivially holds.
For the inductive step, we have 5 cases to consider:
\begin{enumerate}
\item $S = \code{$\TVAR{B, \hat s}$ $x$ := $e$ in $S'$}$;
\item $S = \code{while $e$ do $S'$}$, for $\Sigma; \Gamma; \Delta \esemDash[(\TBOOL, \hat s)] \CONF{e, \ENV{SV}} \trans \TRUE$;
\item $S=\code{if $e$ then $S_{\TRUE}$ else $S_{\FALSE}$}$;
\item $S = \CALL{e_1}{m}{\VEC{e}}{e_2}$;
\item $S = \DCALL{e_1}{m}{\VEC{e}}$.
\end{enumerate}
The first case is easy, since $s' = s$ and so can be concluded by induction.

In the second and in the third cases, the security level is raised to $\hat s$, i.e., the security level needed to evaluate the guard $e$ (viz., $\Sigma; \Gamma; \Delta \esemDash[(\TBOOL, \hat s)] \CONF{e, \ENV{SV}} \trans b \in \BOOLEANS$); so, by induction, $\pi \BARBSYM_Z = \epsilon$, for every $s''$-contract $Z$ such that $\hat s \not\ordleq s''$. However, the level can be raised from $s$ to $\hat s$ only if $s \ordleq \hat s$; so, every $s''$ such that $s \not\ordleq s''$ is also such that $\hat s \not\ordleq s''$. Hence, $\pi \BARBSYM_Z = \epsilon$, for every $s''$-contract $Z$ such that $s \not\ordleq s''$.

For the fourth case, we have that
$\Sigma; \Gamma; \Delta \ssemDash \CONF{Q , \ENV{TSV}} \ltrans{\TRANSACT{X}{Y}{f}{\VEC{v}}{z}}
\Sigma; \Gamma; \Delta'' \ssemDash[{\hat s}] \CONF{S' ; (\ENV{V}, \Delta) ; s ;  Q' , \ENV{T};\ENV{SV}''}$, where $s \ordleq \hat s$ and
\begin{itemize}
\item $X = \ENV{V}(\code{this})$ and $\Delta(\code{this}) = \TVAR{I^{X}, s_1}$, for $s_1 \ordleq \hat s$;
\item $\Sigma; \Gamma; \Delta \esemDash[(I^{Y}, \hat s)]  \CONF{e_1, \ENV{SV}} \trans Y$ (for $\Gamma(Y) = (I^{Y}, s_2)$);
\item $\Sigma; \Gamma; \Delta \esemDash[(\TINT, \hat s)] \CONF{e_2, \ENV{SV}} \trans z$;
\item $\Sigma; \Gamma; \Delta \esemDash[(\VEC{B}, \VEC{s}')] \CONF{\VEC{e}, \ENV{SV}} \trans \VEC{v}$;
\item $m$ is either $f$ or $m = \code{id}$ and $f = \ENV V(\code{id})$;
\item if $f \not\in  \DOM{\ENV{T}(Y)}$, then $\ENV{T}(Y)(\code{fallback}) = (\epsilon, S')$ and $\Gamma(I^{Y})(\code{fallback}) = \TSPROC{}{\hat s}$, otherwise $\ENV{T}(Y)(f)  = (\VEC{x}, S')$ and $ \Gamma(I^{Y})(f) = \TSPROC{\BSVEC}{\hat s}$;
\item $\Delta''$ is made up by the associations $\code{this}:\TVAR{I^{Y}, s_2}, \code{sender}:\TVAR{I^X, s_1}, \code{value}:\TVAR{\TINT, \hat s}$; moreover, if this is a normal call, it also contains $\VEC x : \BSVEC$~; otherwise, it also contains the two associations $\code{id}:\TVAR{\TIDF, \hat s}, \code{args}:\TVAR{\VEC{B}, \VEC{s'}}$;
\item $\ENV{V}''$ is made up by the associations $\code{this}:Y, \code{sender}:X, \code{value}:z$; moreover, if this is a normal call, it also contains $\VEC x:\VEC v$; otherwise, it also contains the two associations $\code{id}:f, \code{args}:\VEC{v}$;
\item $\ENV{S}'' = \ENV{S}\REBIND{X}{\ENV{F}^{X} [\code{balance -= } z]}\REBIND{Y}{\ENV{F}^{Y} [\code{balance += } z]} $, where 
\begin{itemize}
\item $\ENV{S}^i(X) = \ENV{F}^{X}$ and $\ENV{S}^i(Y) = \ENV{F}^{Y}$, 
\item $\Gamma(I^X)(\code{balance}) = \TVAR{\TINT, s_3}$ and $\Gamma(I^{Y})(\code{balance}) =  \TVAR{\TINT, s_4}$, with $s' \ordleq s_3, s_4$.
\end{itemize}
\end{itemize}
So, $\pi = \TRANSACT{X}{Y}{f}{\VEC{v}}{z},\pi'$, for some $\pi'$. By induction, $\pi' \BARBSYM_Z = \epsilon$, for every $s''$-contract $Z$ such that $\hat s \not\ordleq s''$ and so, like above, $\pi' \BARBSYM_Z = \epsilon$, for every $s''$-contract $Z$ such that $s \not\ordleq s''$, being $s \ordleq \hat s$. Concerning $ \TRANSACT{X}{Y}{f}{\VEC{v}}{z}$, it is impossible that $X$ is an $s''$-contract, for every $s''$ such that $s \not\ordleq s''$: indeed, if this was the case, by the conditions above, we would have that $s'' = s_1 \ordleq \hat s \ordleq s_3 = s''$, and so $s'' = \hat s$, in contradiction with $s \not\ordleq s''$ (being $s \ordleq \hat s$). Thus, $\pi \BARBSYM_Z = \epsilon$, for every $s''$-contract $Z$ such that $s \not\ordleq s''$, as desired.

The fifth case is similar to the previous one. The main difference is the way in which we can conclude that $X$ (the caller) cannot be an $s''$-contract, for every $s''$ such that $s \not\ordleq s''$. By contradiction. 
First of all, we know that $\Sigma \vdash I^X \SUBS I^Y$; this implies that $X$ must have a method $f$ with the same signature as the one in $Y$ but with a possibly {\em higher} security level (since subtyping is contravariant in the security level of methods). Since $X$ is an $s''$-contract, it must be that $\Gamma(I^X)(f) = \TSPROC{\BSVEC}{s''}$; being $\Gamma(I^Y)(f) = \TSPROC{\BSVEC}{\hat s}$, we have that $\hat s \ordleq s''$. This yields the desired contradiction, since we know that $s \ordleq \hat s$, and so also $s \ordleq s''$.
\end{proof}

\begin{theorem}[Runtime call integrity for stacks]\label{theorem:stacks_runtime_CI}
Let $\SECLEVELS \DEFSYM \SET{L, H}$ with $L \ordleq H$.
Choose a bipartition of \ANAMES\ into $\SETNAME X$ and $\SETNAME Y$, and fix a type environment $\Gamma$ with respect to which every element of $\SETNAME{X}$ is a $L$-contract and every element of $\SETNAME{Y}$ is a $H$-contract.
For $i \in \SET{1,2}$, assume that
\begin{itemize}
  \item $\Sigma; \Gamma; \Delta \vdash \ENV{SV}^i$ (Well-typedness),
  \item $\Gamma \vdash \ENV{TSV}^i$ (Consistency),
  \item $\Sigma; \Gamma; \ENV{SV}^i; s \vdash Q$ (Well-formedness)
\end{itemize}
and that 
\begin{itemize}
  \item $\ENV{TS}^1|_{\SETNAME{X}} = \ENV{TS}^2|_{\SETNAME{X}}$, and 
  \item $\Sigma; \Gamma; \Delta \vdash \ENV{V}^1 =_L \ENV{V}^2$.
\end{itemize}

\noindent If 
$\Sigma; \Gamma; \Delta   \ssemDash[{s}]   \CONF{Q,    \ENV{TSV}^i} \ltrans{\pi_i} 
\Sigma; \Gamma; \Delta_i \ssemDash[{s_i}] \CONF{\bot, \ENV{T}^i,\ENV{SV}^{i'}}$
for $i \in \SET{1, 2}$, then $\pi_1 \BARBSYM_X = \pi_2 \BARBSYM_X$, for any $X \in \SETNAME{X}$. 

\end{theorem}
\begin{proof}
By induction on the sum of the lengths of the two given sequences of labelled transitions (that generate $\pi_1$ and $\pi_2$).
The base case (i.e., when the sum is 0) is when $Q = \bot$; in this case, $\pi_1=\pi_2=\epsilon$ and the claim trivially holds.
For the inductive step, we know that $Q \not\in \TSTACKS$ and we distinguish the top-most element of $Q$.

If $Q \in \SET{s';Q'\ ,\ \code{del}(x);Q'\ , \ (\ENV V',\Delta');Q'}$, then, for $i \in \{1,2\}$, we have that
$\Sigma; \Gamma; \Delta  \ssemDash[{s}]  \CONF{Q,  \ENV{TSV}^i} \ltrans\epsilon 
\Sigma; \Gamma; \hat\Delta \ssemDash[{\hat s}]  \CONF{Q',  \ENV{TS}^i,\widehat{\ENV V^i}}$, where
$$
\hat\Delta = \left\{
\begin{array}{ll}
\Delta' & \mbox{if } Q = (\ENV V',\Delta');Q'
\\
\Delta & \mbox{otherwise}
\end{array}
\right.
\quad
\hat s = \left\{
\begin{array}{ll}
s' & \mbox{if } Q = s';Q'
\\
s & \mbox{otherwise}
\end{array}
\right.
\quad
\widehat{\ENV V^i} = \left\{
\begin{array}{ll}
\ENV V' & \mbox{if } Q = (\ENV V',\Delta');Q'
\\
\ENV V^i & \mbox{otherwise}
\end{array}
\right.
$$
We can thus immediately conclude by induction applied to $\Sigma; \Gamma; \hat\Delta \ssemDash[{\hat s}]  \CONF{Q',  \ENV{TS}^i,\widehat{\ENV V^i}} \ltrans{\pi_i} \Sigma; \Gamma; \Delta_i \ssemDash[{s_i}]  \CONF{\bot,  \ENV{T}^i,\ENV{SV}^{i'}}$, for $i \in \{1,2\}$.

The delicate case is when $Q = S;Q'$, for some $S$ and $Q'$. 
If $s=H$, then the claim immediately holds by Lemma~\ref{lem:HvsL}.
So let us assume that $s=L$ and reason by case analysis on $S$.
\begin{itemize}
\item $S = \code{skip}$: in this case, $\Sigma; \Gamma; \Delta  \ssemDash[{s}]  \CONF{Q,  \ENV{TSV}^i} \ltrans\epsilon 
\Sigma; \Gamma; \Delta \ssemDash[{s}]  \CONF{Q',  \ENV{TSV}^i}$, for $i \in \{1,2\}$, and we trivially conclude by induction.

\item $S = \code{$\TVAR{B, s'}$ $x$ := $e$ in $S'$}$: in this case, $\Sigma; \Gamma; \Delta  \ssemDash[{s}]  \CONF{Q,  \ENV{TSV}^i} \ltrans\epsilon \Sigma; \Gamma; \Delta, x:\TVAR{B, s'} \ssemDash \CONF{S' ; \DEL{x} ; Q', \ENV{TS}^i; \ENV{V}^i, x:v_i}$, where
$\Sigma; \Gamma; \Delta \esemDash[(B, s')] \CONF{e, \ENV{SV}^i} \trans v_i$, for $i \in \{1,2\}$.
If $s' =L$, then $\Gamma;\Delta \vdash \ENV S^1 =_L \ENV S^2$, since we assume that $\ENV{S}^1|_{\SETNAME{X}} = \ENV{S}^2|_{\SETNAME{X}}$ and all elements of every $X \in \SETNAME{X}$ are labelled $L$.
Moreover, we assume that $\Gamma;\Delta \vdash \ENV V^1 =_L \ENV V^2$. 
Thus, by Corollary~\ref{corollary:expressions_runtime_noninterference}, we obtain that $v_1 = v_2 = v$, and so 
\begin{equation*}
    \Gamma;\Delta,x:\TVAR{B_L} \vdash \ENV V^1,x:v =_L \ENV V^2,x:v.
\end{equation*}

\noindent Otherwise, $s'=H$ and so $x$ is an $H$-variable; hence, 
\begin{equation*}
  \Gamma;\Delta, x:\TVAR{B_H} \vdash \ENV V^1,x:v_1 =_L \ENV V^2,x:v_2.
\end{equation*}

\noindent In both cases, we can apply the induction hypothesis to $\Sigma; \Gamma; \Delta, x:\TVAR{B, s'} \ssemDash[{s}] \CONF{S' ; \DEL{x} ; Q', \ENV{TS}^i; \ENV{V}^i, x:v_i} \ltrans{\pi_i} \Sigma; \Gamma; \Delta_i \ssemDash[{s_i}]  \CONF{\bot,  \ENV{T}^i,\ENV{SV}^{i'}} $, for $i \in \{1,2\}$, and conclude. 

\item The cases for $S = \code{$x$ := $e$}$ and $S = \code{this.$p$ := $e$}$ are similar to the previous one (just notice that here $\Delta$ is not changed in the first transition).

\item $S=\code{if $e$ then $S_{\TRUE}$ else $S_{\FALSE}$}$: in this case, we have that
$\Sigma; \Gamma; \Delta \ssemDash \CONF{Q , \ENV{TSV}^i} \ltrans\epsilon  
\Sigma; \Gamma; \Delta \ssemDash[{s'}] \CONF{S_{b_i} ; s; Q' , \ENV{TSV}^i}$, for $i \in \{1,2\}$, where
$\Sigma; \Gamma; \Delta \esemDash[(\TBOOL, s')] \CONF{e, \ENV{SV}^i} \trans b_i (\in \BOOLEANS)$ and $s \ordleq s'$.
We distinguish two sub-cases:
\begin{enumerate}
\item $s'= s (=L)$: 
since we assume that $\ENV{S}^1|_{\SETNAME{X}} = \ENV{S}^2|_{\SETNAME{X}}$ and $\Gamma;\Delta \vdash \ENV V^1 =_L \ENV V^2$, by Corollary~\ref{corollary:expressions_runtime_noninterference}, we obtain that $b_1 = b_2 = b \in \BOOLEANS$.
So, we can apply the induction
hypothesis to $\Sigma; \Gamma; \Delta \ssemDash[{s}] \CONF{S_{b} ; s; Q' , \ENV{TSV}^i} \ltrans{\pi_i}
\Sigma; \Gamma; \Delta_i \ssemDash[{s_i}]  \CONF{\bot,  \ENV{T}^i,\ENV{SV}^{i'}} $, for $i \in \{1,2\}$, and conclude.
\item $s'=H$: if $S_{b_1}=S_{b_2}$, we reason like in the sub-case above. Otherwise, we have that each $\pi_i$ is equal to $\pi_i',\pi_i''$, where
$\Sigma; \Gamma; \Delta \ssemDash[{s'}] \CONF{S_{b_i} ; s; Q' , \ENV{TSV}^i} 
\ltrans{\pi_i'}
\Sigma; \Gamma; \Delta \ssemDash[{s}] \CONF{Q' , \ENV{TSV}^{i''}} 
\ltrans{\pi_i''}
\Sigma; \Gamma; \Delta_i \ssemDash[{s_i}]  \CONF{\bot,  \ENV{T}^i,\ENV{SV}^{i'}} $.
Being $s' = H$, by Lemma~\ref{lem:HvsL} we have that $\pi_1' \BARBSYM_X = \pi_2' \BARBSYM_X = \epsilon$, for any $X \in \SETNAME{X}$; then, by induction hypothesis, we obtain that $\pi_1'' \BARBSYM_X = \pi_2'' \BARBSYM_X$, that allows us to conclude.
\end{enumerate}

\item $S=\code{while $e$ do $S'$}$: this case is similar to the previous one. If $e$ evaluates to the same boolean value in the two environments,
we simply apply the induction. Otherwise, $s'=H$ and the computation where $e$ evaluates to $\FALSE$  yields $Q'$ by exhibiting trace $\epsilon$, whereas the the computation where $e$ evaluates to $\TRUE$ will yield $Q'$ after some trace $\pi$. However, since the overall level has been raised to $H$, because of Lemma~\ref{lem:HvsL}, we have that $\pi \BARBSYM_X = \epsilon$, for any $X \in \SETNAME{X}$, and we conclude by induction.

\item $S = \CALL{e_1}{m}{\VEC{e}}{e_2}$: in this case, we have that
$\Sigma; \Gamma; \Delta \ssemDash \CONF{Q , \ENV{TSV}^i} \ltrans{\TRANSACT{X_i}{Y_i}{f_i}{\VEC{v}_i}{z_i}}
\Sigma; \Gamma; \Delta_i' \ssemDash[{s'}] \CONF{S_i ; (\ENV{V}^i, \Delta) ; s ;  Q' , \ENV{T}^i;\ENV{SV}^{i''}}$, for $i \in \{1,2\}$, where $s \ordleq s'$ and
\begin{itemize}
\item $X_i = \ENV{V}^i(\code{this})$ and $\Delta(\code{this}) = \TVAR{I^{X}, s_1}$, for $s_1 \ordleq s'$;
\item $\Sigma; \Gamma; \Delta \esemDash[(I^{Y_i}, s')]  \CONF{e_1, \ENV{SV}^i} \trans Y_i$ (for $\Gamma(Y_i) = (I^{Y_i}, s_2^i)$);
\item $\Sigma; \Gamma; \Delta \esemDash[(\TINT, s')] \CONF{e_2, \ENV{SV}^i} \trans z_i$;
\item $\Sigma; \Gamma; \Delta \esemDash[(\VEC{B}, \VEC{s}')] \CONF{\VEC{e}, \ENV{SV}^i} \trans \VEC{v}_i$;
\item $m$ is either a method name $f_i$ for both $i$'s, or $m = \code{id}$ and $f_i = \ENV V^i(\code{id})$; 
\item if $f_i \not\in  \DOM{\ENV{T}^i(Y_i)}$, then $\ENV{T}^i(Y_i)(\code{fallback}) = (\epsilon, S_i)$ and $\Gamma(I^{Y_i})(\code{fallback}) = \TSPROC{}{s'}$, otherwise $\ENV{T}^i(Y_i)(f_i)  = (\VEC{x}, S_i)$ and $ \Gamma(I^{Y_i})(f_i) = \TSPROC{\BSVEC}{s'}$;
\item $\Delta'_i$ is made up by the associations $\code{this}:\TVAR{I^{Y_i}, s_2^i}, \code{sender}:\TVAR{I^X, s_1}, \code{value}:\TVAR{\TINT, s'}$; moreover, if in environments $i$ this is a normal call, it also contains $\VEC x : \BSVEC$~; otherwise, it also contains the two associations $\code{id}:\TVAR{\TIDF, s'}, \code{args}:\TVAR{\VEC{B}, \VEC{s}'}$;
\item $\ENV{V}^{i''}$ is made up by the associations $\code{this}:Y_i, \code{sender}:X_i, \code{value}:z_i$; moreover, if in environments $i$ this is a normal call, it also contains $\VEC x:\VEC v_i$; otherwise, it also contains the two associations $\code{id}:f_i, \code{args}:\VEC{v}_i$;
\item $\ENV{S}^{i''} = \ENV{S}\REBIND{X_i}{\ENV{F}^{X_i} [\code{balance -= } z_i]}\REBIND{Y_i}{\ENV{F}^{Y_i} [\code{balance += } z_i]} $, where 
\begin{itemize}
\item $\ENV{S}^i(X_i) = \ENV{F}^{X_i}$ and $\ENV{S}^i(Y_i) = \ENV{F}^{Y_i}$, 
\item $\Gamma(I^X)(\code{balance}) = \TVAR{\TINT, s_3}$ and $\Gamma(I^{Y_i})(\code{balance}) =  \TVAR{\TINT, s_4^i}$, with $s' \ordleq s_3, s_4^i $.
\end{itemize}
\end{itemize}

If $s'=L$, then like above, by Corollary~\ref{corollary:expressions_runtime_noninterference}, we obtain that 
$X_1 = X_2$, $Y_1 = Y_2$, $f_1=f_2$, $\VEC v_1 = \VEC v_2$, and $z_1 = z_2$; so, consequently, also that $\Delta_1' = \Delta_2'$, $S_1 = S_2$, $\ENV{S}^{1''} = \ENV{S}^{2''}$ and $\ENV{V}^{1''} = \ENV{V}^{2''}$. Thus, by induction, we can easily conclude.

If $s'=H$, then we reason as follows. First, observe that the two $\pi_i$'s are equal to $\TRANSACT{X_i}{Y_i}{f_i}{\VEC{v}_i}{z_i} , \pi_i' ,\pi_i''$, where
$\Sigma; \Gamma; \Delta_i' \ssemDash[{s'}] \CONF{S_i ; (\ENV{V}^i, \Delta) ; s ;  Q' , \ENV{T}^i;\ENV{SV}^{i''}} 
\ltrans{\pi_i'}
\Sigma; \Gamma; \Delta \ssemDash[{s}] \CONF{Q' , \ENV{TV}^i;\ENV{S}^{i'''}}
\ltrans{\pi_i''}
\Sigma; \Gamma; \Delta_i \ssemDash[{s}] \CONF{\bot , \ENV{T}^i;\ENV{SV}^{i'}}$.
For any $X \in \SETNAME{X}$, we have that 
$\pi_1' \BARBSYM_X = \pi_2' \BARBSYM_X = \epsilon$ (by Lemma~\ref{lem:HvsL}) and
$\pi_1'' \BARBSYM_X = \pi_2'' \BARBSYM_X$ (by induction).
We are left with the first label of the traces.
The crucial observation is that $X_1$ and $X_2$ must belong to $\SETNAME Y$, because their \code{balance} fields are $H$ (given that $H = s' \ordleq s_3$); hence, the first labels of the two traces are excluded from any projection on a contract from $\SETNAME X$ and we conclude.

\item $S = \DCALL{e_1}{m}{\VEC{e}}$: this case is similar to the previous one when no fallback function is invoked.
The only notable difference is that in this case we have no explicit typing information about the \code{balance} fields, since they are not modified by the method invocation. However, we know that $\Gamma(I^{Y_i})(\code{balance}) = \Gamma(I^X)(\code{balance})$. In the case when $s'=H$, this again entails that $X_1,X_2 \in \SETNAME Y$, because $Y_1,Y_2 \in \SETNAME Y$, being that $\Sigma; \Gamma; \Delta \esemDash[(I^{Y_i}, s')]  \CONF{e_1, \ENV{SV}^i} \trans Y_i$.
\vspace*{-.5cm}
\end{itemize}
\end{proof}

\begin{corollary}[Type safety implies call integrity]
Assume the same security labelling as in Theorem~\ref{theorem:stacks_runtime_CI} and that, for $i \in \SET{1,2}$
\begin{itemize}
  \item $\Sigma; \Gamma; \EMPTYSET \vdash \ENV{S}^i$ (Well-typedness), 
  \item $\Gamma \vdash \ENV{TS}^i; \ENV{V}^\EMPTYSET$ (Consistency), 
  \item $\Sigma; \Gamma; \ENV{S}^i; \ENV{V}^\EMPTYSET; s \vdash \tau$ (Well-formedness),
  \item $\Sigma; \Gamma; \ENV{T}^i \vDash \ENV{T}^i$ ($\ENV{T}$ type safety),
  \item $\Sigma; \Gamma; \EMPTYSET; \ENV{T}^i \vDash \tau : \TCMD{s}$ (stack type safety),
\end{itemize}

\noindent If $\ENV{TS}^1|_{\SETNAME{X}} = \ENV{TS}^2|_{\SETNAME{X}}$ and
$\CONF{\tau, \ENV{TS}^i,\ENV{V}^\EMPTYSET} \ltrans{\pi_i} \CONF{\bot, \ENV{T}^i, \ENV{SV}^{i'} }$
for $i \in \SET{1,2}$, then
$\pi_1 \BARBSYM_X = \pi_2 \BARBSYM_X$, for any $X \in \SETNAME{X}$. 
\end{corollary}
\begin{proof}
The crucial observation is that, if $\Sigma; \Gamma; \EMPTYSET; \ENV{T}^i \vDash \tau : \TCMD{s}$ and 
$\CONF{\tau, \ENV{TS}^i,\ENV{V}^\EMPTYSET} \ltrans{\pi_i} \CONF{\bot, \ENV{T}^i, \ENV{SV}^{i'} }$,
then $\Sigma; \Gamma; \EMPTYSET  \ssemDash[{s}]  \CONF{\tau,  \ENV{TS}^i,\ENV{V}^\EMPTYSET} \ltrans{\pi_i}
\Sigma; \Gamma; \Delta_i \ssemDash[{s_i}]  \CONF{\bot,  \ENV{T}^i,\ENV{SV}^{i'}} $; this holds otherwise we would eventually have a non-terminal stack that has no outgoing typed transitions (in contradiction with the type safe assumption).
\end{proof}

To conclude, a challenging avenue to explore is to to drop the condition that both executions have to terminate in Definition~\ref{def:call_integrity} and instead show that if neither of the two executions get stuck (i.e.~if there were no type errors), then we cannot find in finitely many steps a point at which the two call-traces disagree. 
This would extend  the notion of call-integrity provided in~\cite{grishchenko2018} to non-terminating programs.

\clearpage
\section{Summary of symbols and notations}\label{app:notations}

\begin{tabular}{||c c c c||} 
 \hline
 {\bf Name} & {\bf Symbol} & {\bf Set} & {\bf Defined in} \\ [0.5ex] 
 \hline\hline
 Variable names & $x,y$ & $\VNAMES$ & Figure~\ref{fig:syntax_tinysol} \\
 \hline
 Field names & $p$ & $\FNAMES$ & Figure~\ref{fig:syntax_tinysol} \\
 \hline
 Method names & $f,g$ & $\MNAMES$ & Figure~\ref{fig:syntax_tinysol} \\
 \hline
 Address names & $A,X,Y,Z$ & $\ANAMES$ & Figure~\ref{fig:syntax_tinysol} \\
 \hline
 Interface names & $I$ & $\TNAMES$ & Figure~\ref{fig:syntax_tinysol} \\
 \hline
 Names (all the previous ones) & $n$ & $\NAMES$ & Definition~\ref{def:types} \\
 \hline
 Method identifiers & $m$ & $\MNAMES \cup \{\code{id}\}$ & Figure~\ref{fig:syntax_tinysol}\\
 \hline
 \hline
 Integers & $z$ & $\mathbb{Z}$ & Section~\ref{sec:typedSyntax} \\
 \hline
 Booleans & $b$ & $\mathbb{B}$ & Section~\ref{sec:typedSyntax} \\
 \hline
 Values & $v$ & $\VALUES$ & Figure~\ref{fig:syntax_tinysol} \\
 \hline
 \hline
 Field Declarations & $DF$ & $\DEC{F}$ & Figure~\ref{fig:syntax_tinysol} \\ 
 \hline
Method Declarations & $DM$ & $\DEC{M}$ & Figure~\ref{fig:syntax_tinysol} \\ 
 \hline
 Contract Declarations & $DC$ & $\DEC{C}$ & Figure~\ref{fig:syntax_tinysol} \\ 
 \hline
 Interface Fields & $IF$ &  & Definition~\ref{def:interfaces} \\ 
 \hline
 Interface Methods & $IM$ & & Definition~\ref{def:interfaces} \\ 
 \hline
 Interface Declarations & $ID$ & & Definition~\ref{def:interfaces} \\ 
 \hline
 \hline
 Expressions & $e$ & $\EXPR$ & Figure~\ref{fig:syntax_tinysol} \\
 \hline
 Statements & $S$ & $\STM$ & Figure~\ref{fig:syntax_tinysol} \\
 \hline
 Stacks & $Q$ & $\STACKS$ & Section~\ref{sec:typedSyntax} \\
 \hline
 Terminal Stacks &  & $\TSTACKS$ & Section~\ref{sec:typing_interpretation_stm_stack} \\
 \hline
 \hline
 Variable environments & $\ENV V$ & $\SETENV{V}$ & Section~\ref{sec:typedSyntax} \\
 \hline
 Field environments & $\ENV F$ & $\SETENV{F}$ & Section~\ref{sec:typedSyntax} \\
 \hline
 Method environments & $\ENV M$ & $\SETENV{M}$ & Section~\ref{sec:typedSyntax} \\
 \hline
 States & $\ENV S$ & $\SETENV{S}$ & Section~\ref{sec:typedSyntax} \\
 \hline
 Method tables & $\ENV T$ & $\SETENV{T}$ & Section~\ref{sec:typedSyntax} \\
 \hline
 \hline
  Security levels (with order $\ordleq$) & $s$ & $\SECLEVELS$ & Section~\ref{sec:types} \\
 \hline
 Base types & $B$ & $\BASETYPES$ & Definition~\ref{def:types} \\
 \hline
 Types & $T$ & $\TYPES$ & Definition~\ref{def:types} \\
 \hline
 Type environments & $\Gamma,\Delta$ & $\TYPEENVS$ & Definition~\ref{def:types} \\
 \hline
 Interface hierarchies & $\Sigma$ &  & Definition~\ref{def:types} \\
 \hline
 \end{tabular}

\begin{tabular}{||c c c||} 
 \hline
 {\bf Name} & {\bf Symbol} & {\bf Defined in} \\ [0.5ex] 
 \hline
  \hline
 Operation evaluation & $\trans_{\op}$ & Assumed \\
 \hline
 Evaluation of expressions & $\CONF{e,\ENV{SV}}\trans v$ & Figure~\ref{fig:semantics_expressions} \\
 \hline
 Operational semantics  & $ \CONF{Q, \ENV{TSV}} \trans \CONF{Q', \ENV{T},\ENV{SV}'} $ & Figg.~\ref{fig:semantics_stacks1} and~\ref{fig:semantics_stacks2} \\
 \hline
 Typed eval.\@ of expressions & $\Sigma; \Gamma; \Delta \esemDash \CONF{e, \ENV{SV}} \trans v$ & Figure~\ref{fig:typed_semantics_expressions} \\
 \hline
 Typed oper.\@ semantics  & \!\!\!\!\!\!$\Sigma; \Gamma; \Delta \ssemDash \CONF{Q, \ENV{TSV}} \trans $ & Figg.~\ref{fig:typed_semantics_stacks1} and~\ref{fig:typed_semantics_stacks2} \\
  & $\qquad \Sigma; \Gamma; \Delta' \ssemDash[{s'}] \CONF{Q', \ENV{T};\ENV{SV}'} $ &  \\
 \hline
 Type lifted transitions  & $\Sigma; \Gamma; \ENV{T} \vDash (Q, \Delta, s) \ttrans (Q', \Delta', s')$ & Definition~\ref{def:type_lifted_transitions} \\
 \hline
 \hline
 Consistency of $\Sigma$ & $\Sigma; \Gamma \vdash \Sigma'$ & Figure~\ref{fig:consistency_rules_sigma} \\
 \hline
 Consistency of environments & $\Gamma \vdash \ENV{TSV}$ & Figg.~\ref{fig:consistency_rules_envtsv} and~\ref{fig:consistency_rules_envfm}\\
 \hline
 Well-formedness of stacks & $\Sigma; \Gamma; \ENV{SV}; s \vdash Q$ & Appendix~\ref{app:wellformedness_stack}\\
 \hline
 $s$-equivalence & $\Gamma; \Delta \vdash \ENV{SV}^1 =_s \ENV{SV}^2$ & Appendix~\ref{app:equivalence_states} \\
 \hline
 \hline
 Subtyping & $\Sigma \vdash T_1 \SUBS T_2 $ & Sec.~\ref{sec:types} and Fig.~\ref{fig:subtyping_expressions}\\
 \hline
 Syntactic typing of expressions & $\Sigma; \Gamma; \Delta \vdash e:B_s$ & Figure~\ref{fig:type_rules_expr} \\
 \hline
 Syntactic typing of statements & $\Sigma; \Gamma; \Delta \vdash S:\TCMD s$ & Figure~\ref{fig:type_rules_stm} \\
 \hline
 Syntactic typing of stacks & $\Sigma; \Gamma; \Delta \vdash Q:\TCMD s$ & Figure~\ref{fig:type_rules_stacks} \\
 \hline
 Syntactic typing of environments & $\Sigma; \Gamma; \Delta \vdash_{(X)} \ENV{}$ & Figure~\ref{fig:type_rules_env} \\
 \hline
 \hline
  Typing interpr.\@ for expressions & $ \ETYPI / \ETYPING$ & Definition~\ref{def:typing_interpretation_expr}\\
 \hline
  Semantic typing of expressions & $\Sigma; \Gamma; \Delta \vDash e : B_s$ & Definition~\ref{def:typing_interpretation_expr}\\
 \hline
  Typing interpr.\@ for stacks & $ \STYPI / \STYPING$ & Definition~\ref{def:typing_interpretation_stm}\\
 \hline
  Semantic typing of statements & $\Sigma; \Gamma; \Delta; \ENV{T} \vDash S : \TCMD{s}$ & Definition~\ref{def:typing_interpretation_stm}\\
 \hline
  Semantic typing of stacks & $\Sigma; \Gamma; \Delta; \ENV{T} \vDash Q : \TCMD{s}$ & Definition~\ref{def:typing_interpretation_stm}\\
 \hline
 Type safe code environments & $\Sigma; \Gamma; \ENV{T} \vDash \ENV{T}$ & Figure~\ref{fig:semantic_type_rules_envtm}
 \\
  \hline
  Progression & $\UPTOSTYPI \progress \CANDSTYPI$ & Definition~\ref{def:progression} \\
  \hline
  \hline
  Transactions & $\tau$ & Appendix~\ref{app:CI}, pg.~\pageref{def:call_integrity}\\
  \hline
  Traces & $\pi$ & Definition~\ref{def:traces} \\
  \hline
  Trace projection & $\pi  \BARBSYM_X$ & Definition~\ref{def:project} \\
  \hline
  Labelled oper.\@ semantics & $\CONF{Q,\ENV{TSV}} \ltrans{\pi} \CONF{Q',\ENV{T},\ENV{SV}'}$ & Definition~\ref{def:traces} \\
  \hline
  Labelled typed oper. sem. & \!\!\!\!\!\!$\Sigma; \Gamma; \Delta  \ssemDash[{s}]  \CONF{Q,  \ENV{TSV}} \ltrans{\pi} $ & Appendix~\ref{app:CI}, pg.~\pageref{def:sContracts} \\
  & \qquad $\Sigma; \Gamma; \Delta' \ssemDash[{s'}]  \CONF{Q',  \ENV{T},\ENV{SV}'} $ & \\
  \hline
  
\end{tabular}

\clearpage
\section{The PMW attack}\label{Sect:PMW-attack}

\begin{figure}
\begin{minipage}[t]{.5\textwidth}
\begin{lstlisting}[style=tinysol]
contract Lib : @$(I^L,{s_L})$@ { 
  field owner := X;
  @$\lstvdots$@
  func init(newOwner) {
    this.owner := newOwner
  }
}
\end{lstlisting}
\end{minipage}\hfill
\begin{minipage}[t]{.45\textwidth}
\begin{lstlisting}[style=tinysol]
contract Wallet : @$(I^W,{s_W})$@ {
  field owner := Y;
  @\lstvdots@
  func fallback() { 
    dcall Lib.id(args)
  }
}
\end{lstlisting}
\end{minipage}
\vspace*{-.5cm}
\caption{The Parity Multisig Wallet attack in \TINYSOL.}
\label{fig:PMWattack}
\end{figure}

Figure~\ref{fig:PMWattack} reports the code for the PMW attack, as presented in \cite{chinese2021flowtypes} and translated into \TINYSOL{} (irrelevant details have been omitted). 
To trigger the attack, the attacker (with address \code A) would issue a transaction with 
\[
\code{call Wallet.init(A)}
\]

\noindent This call invokes the fallback function of \code{Wallet}, since the contract \code{Wallet} does not declare a method named \code{init}. 
This triggers the execution of \code{Lib.init} in the context of \code{Wallet}, which would overwrite the value \code{Y} stored in \code{Wallet.owner} and change it to \code{A}.

As desirable, there cannot be any typing interpretation containing this example, because we do not permit delegate calls to arbitrary contracts but only to those that are super-types of the calling contract. 
In the example, \code{Lib} is not a super-type of \code{Wallet}, precisely because it contains the method \code{init}, which \code{Wallet} lacks. 
Formally, by assuming the expectable environments, since $\code{init} \not\in \ENV T (\code{Wallet})$ and $\ENV T (\code{Wallet})(\code{fallback}) = (\epsilon, \DCALL{\code{Lib}}{\code{id}}{\code{args}})$, by rule \nameref{stm_rt_fcall} we have that, 
\[
    \begin{array}{l}
      \Sigma; \Gamma; \Delta \ssemDash \CONF{\CALL{\code{Wallet}}{\code{init}}{\code A}{z} ; Q, \ENV{TSV}}  \\
      \trans \Sigma; \Gamma; \Delta' \ssemDash[{s'}] \CONF{ \DCALL{\code{Lib}}{\code{id}}{\code{args}} ; (\ENV{V}, \Delta) ; s ; Q, \ENV{T}; \ENV{SV}'} 
    \end{array}
\]

\noindent where
\[
\begin{array}{ll}
\Delta' = & \code{this}:\TVAR{I^W, s_W}, 
\code{sender}:\TVAR{I^{\code A}, s_A}, 
\code{value}:\TVAR{\TINT, s'}, 
\\
& \code{id}:\TVAR{\TIDF, s'}, 
\code{args}:\TVAR{I^{\code A}, s_A} 
\end{array}
\]

\noindent and
$$\ENV{V}' = \code{this}:\code{Wallet}, \code{sender}:\code A, \code{value}:z, \code{id}:\code{init}, \code{args}:\code A.
$$
Now, the point is that, by rule \nameref{stm_rt_dcall},
\begin{equation}
\label{eq:PMW}
\Sigma; \Gamma; \Delta' \ssemDash[{s'}] \CONF{ \DCALL{\code{Lib}}{\code{id}}{\code{args}} ; (\ENV{V}, \Delta) ; s ; Q, \ENV{T}; \ENV{SV}'} 
\end{equation}
is stuck, 
because $\code{init} \in (\DOM{\Gamma(I^L)} \INTERSECT (\FNAMES \UNION \{\code{balance}\}))$ but 
$\Gamma(I^L)(\code{init}) \neq \Gamma(I^W)(\code{init})$, since the first is defined, whereas the second is not.
Using Lemma~\ref{lemma:untypable}, we can now conclude that there is no typing interpretation $\STYPI$ containing the triplet $(\CALL{\code{Wallet}}{\code{init}}{\code A}{z} ; Q, \Delta, s)$.

Hence, the only way for \eqref{eq:PMW} to evolve (and having a typing interpretation witnessing the safety of the code) is that \code{Wallet} did, in fact, contain a declaration for \code{init} as well. 
However, in this case, the attack would not be possible, because the fallback function of \code{Wallet} would not be invoked by \code{A}'s transaction.

\end{document}